\documentclass[12pt,letterpaper]{article}%
\usepackage{amsmath}
\usepackage{rotating}
\usepackage{version}
\usepackage{amsfonts}
\usepackage{amssymb}
\usepackage{graphicx}
\usepackage{float}
\usepackage{url}
\usepackage[round,authoryear,compress]{natbib}
\usepackage{color}
\usepackage{xr}
\usepackage[normalem]{ulem}
\usepackage[labelfont={bf}]{caption}
\usepackage{wrapfig}
\usepackage{titletoc}
\usepackage[singlespacing,nodisplayskipstretch]{setspace}
\usepackage[pdftex,letterpaper,portrait,twocolumn=false,nomarginpar]{geometry}%
\newtheorem{theorem}{Theorem}

\newtheorem{assumption}{Assumption}

\newtheorem{definition}{Definition}

\newtheorem{lemma}{Lemma}

\newtheorem{remark}{Remark}

\newenvironment{proof}[1][Proof]{\noindent\textbf{#1.} }{\ \rule{0.5em}{0.5em}}
\begin{document}

\title{Identifying and exploiting alpha in linear asset pricing models with strong,
semi-strong, and latent factors\thanks{We are grateful to Juan Espinosa Torres
and Jing Kong for carrying out the computations and to Hayun Song for
excellent research assistance, all are PhD students at USC. We are also
grateful for helpful comments from the editor, Fabio Trojani, two anonymous
referees, Natalia Bailey, Alessio Sancetta, Takashi Yamagata and participants
at the 2023 Kansas Econometrics Institute, the 2023 NBER Summer Insitute, the
2023 IAAE Annual Conference in Oslo, University of California, Riverside,
April 2024, and ECMFE Workshop at University of Essex, October 2024.}}
\author{M. Hashem Pesaran\\University of Southern California,\ and Trinity College, Cambridge
\and Ron P. Smith\thanks{Corresponding author, email r.smith@bbk.ac.uk. }\\Birkbeck, University of London}
\maketitle

\begin{abstract}
The risk premia of traded factors are the sum of factor means and a parameter
vector we denote by $\boldsymbol{\phi}$\textbf{ }$\ $which is identified from
the cross-section regression of $\alpha_{i}$ on the vector of factor loadings,
$\boldsymbol{\beta}_{i}$.\textbf{ }If $\boldsymbol{\phi}$ is non-zero, then
$\alpha_{i}$ are non-zero and one can construct "phi-portfolios" which exploit
the systematic components of non-zero alpha. We show that for known values of
$\boldsymbol{\beta}_{i}$ and when $\boldsymbol{\phi}\ \ $is non-zero there
exist phi-portfolios that dominate mean-variance portfolios. The paper then
proposes a two-step bias corrected estimator of $\boldsymbol{\phi}$ \ and
derives its asymptotic distribution allowing for idiosyncratic pricing errors,
weak missing factors, and weak error cross-sectional dependence. Small sample
results from extensive Monte Carlo experiments show that the proposed
estimator has the correct size with good power properties. The paper also
provides an empirical application to a large number of U.S. securities with
risk factors selected from a large number of potential risk factors according
to their strength and constructs phi-portfolios and compares their Sharpe
ratios to mean variance and S\&P portfolios.

\end{abstract}


\thispagestyle{empty}%

\newpage%
\pagenumbering{arabic}%
\onehalfspacing
\setlength{\abovedisplayskip}{0pt}
\setlength{\belowdisplayskip}{0pt}
\setlength{\abovedisplayshortskip}{0pt}
\setlength{\belowdisplayshortskip}{0pt}%

\section{Introduction\label{Intro}}

There are two approaches to the estimation of risk premia and testing of
market efficiency, often referred as the beta and the SDF (stochastic discount
factor) methods, see \cite{Jagannathan2002}. This paper adopts the beta
method, and following the literature uses a linear factor pricing model (LFPM)
to explain the time series of excess returns on individual securities,
$r_{it}=R_{it}-r_{t}^{f}$, where $R_{it}$ is the return and $r_{t}^{f}$ the
risk free rate for $i=1,2,...,n;$ $t=1,2,....,T,$ by a set of observed
tradable risk factors. We use individual securities rather than portfolios
since, as we will show, if risk factors are not strong, large $n$ is required
for accurate estimation. \cite{Ang2020} discuss the general issues in the
choice between portfolios and individual stocks. \cite{Pesaran2023} discuss
both the use of portfolios and the relationship between the SDF and LFPM approaches.

The LFPM\ explains the excess return on each security, $r_{it},$ by an
intercept, $\mathit{\alpha}_{i}$, labelled alpha, and a $K\times1$ vector of
traded risk factors, $\mathbf{f}_{t}=(f_{1t},f_{2t},...,f_{Kt})^{\prime}$ with
loadings, $\boldsymbol{\beta}_{i}$:%
\begin{equation}
r_{it}=\mathit{\alpha}_{i}+\boldsymbol{\beta}_{i}^{\prime}\mathbf{f}%
_{t}+u_{it}. \label{LFPM}%
\end{equation}
Under the Arbitrage Pricing Theory (APT) due to
\citet{ROSS1976341}%
, we have
\begin{equation}
E\left(  r_{it}\right)  =c+\boldsymbol{\beta}_{i}^{\prime}\boldsymbol{\lambda
}+\eta_{i}, \label{APT}%
\end{equation}
where $\eta_{i}$ is a firm-specific idiosyncratic pricing error. Ross allowed
$\eta_{i}$ to be non-zero for some $i$ but required that these errors are
bounded in the sense that $\sum_{i=1}^{n}\eta_{i}^{2}<C<\infty$. We relax this
condition and impose weaker restrictions on $\eta_{i}$'s as discussed in
sub-section \ref{pricing}. But the paper's main contribution lies in the fact
that even if there are no idiosyncratic pricing errors, it is still possible
to have non-zero $\mathit{\alpha}_{i}$. To see this, taking unconditional
expectations of (\ref{LFPM}) and using (\ref{APT}) we note that
\[
E(r_{it})=\mathit{\alpha}_{i}+\boldsymbol{\beta}_{i}^{\prime}\mathbf{\mu
}=c+\boldsymbol{\beta}_{i}^{\prime}\boldsymbol{\lambda}+\eta_{i}.
\]
where $E(\mathbf{f}_{t})=\boldsymbol{\mu}$, which in turn yields%
\begin{equation}
\mathit{\alpha}_{i}=c+\boldsymbol{\beta}_{i}^{\prime}\left(
\boldsymbol{\lambda}-\boldsymbol{\mu}\right)  +\eta_{i},\text{ for
}i=1,2,...,n \label{alphai}%
\end{equation}
We denote $\boldsymbol{\lambda}-\boldsymbol{\mu}$ by $\boldsymbol{\phi}$, and
consider the implications of a non-zero $\boldsymbol{\phi}$ for portfolio
optimization. We show that when $\boldsymbol{\phi\neq0}$, one can construct
what we call "phi-portfolios" that exploit the systematic components of
$\mathit{\alpha}_{i}$, as captured by the non-zero elements of
$\boldsymbol{\phi}\ $\ in (\ref{alphai}).\footnote{We are grateful to one of
the anonymous reviewers for suggesting the idea of phi-portfolios, as a way of
motivating and interpreting the role of $\boldsymbol{\phi}$ in asset pricing
models.} In contrast, the idiosyncratic pricing errors, $\eta_{i}$, cannot be
exploited, as it is likely to involve insider trading. But our proposed
phi-portfolios can be implemented and applies irrespective of whether
$\eta_{i}=0$ or not.

The extent to which alpha can be exploited is discussed in sub-section
(\ref{pricing}) where it is shown that $\sum_{i=1}^{n}\left(  \mathit{\alpha
}_{i}-\mathit{\bar{\alpha}}\right)  ^{2}$ depends on the magnitude of
$\boldsymbol{\phi}^{\prime}\boldsymbol{\phi}$ and the strength of the risk
factors. Since this alpha can be exploited by the construction of the
phi-portfolios, this is not consistent with a no arbitrage condition. While it
is true that $E(r_{it})=\boldsymbol{\beta}_{i}^{\prime}\boldsymbol{\lambda}$,
estimating $\boldsymbol{\lambda}$\ does not tell us whether $\boldsymbol{\phi
}\neq0,$\ which can be exploited to obtain Sharpe ratios that are larger than
the Sharpe of the tangency portfolio. The claim that tangency portfolio has
the highest Sharpe ratio rests on the implicit assumption that
$\boldsymbol{\phi=0}$, even if there are no idiosyncratic pricing errors
($\eta_{i}=0$ for all $i$).

The focus of much of the literature has been on testing $\mathit{\alpha}%
_{i}=0$, or on estimating the risk premium $\boldsymbol{\lambda}$. But given
that estimating a non-zero $\boldsymbol{\phi}$ \ provides a way to identify
and exploit the alpha in a linear factor pricing model for large $n,$
$\boldsymbol{\phi}$ is an interesting parameter in itself, which will be the
focus of this paper. Furthermore, the tangency condition used to construct
mean-variance (MV) portfolios has the implicit assumption that $\mathit{\alpha
}_{i}=0$, and the MV portfolio implied by the tangency condition is not
efficient unless $\boldsymbol{\phi=0}$. If $\boldsymbol{\phi\neq0}$ there is
no guarantee that the tangency portfolio will achieve a maximal Sharpe ratio,
and is likely to be dominated by the phi-portfolio. Specifically, abstracting
from model and parameter uncertainties, we show that when $\boldsymbol{\phi
\neq0}$ the Sharpe of the tangency portfolio is bounded in the number of
securities, whilst the Sharpe of the phi portfolio rises in $n$ and there
exists a sufficiently large $n$ that ensures that the Sharpe of the
phi-portfolio will exceed that of the tangency portfolio.

More specifically, first we introduce our proposed phi-portfolio in terms of
the factor loadings $\mathbf{B}_{n}=\left(  \boldsymbol{\beta}_{1}%
,\boldsymbol{\beta}_{2},...,\boldsymbol{\beta}_{n}\right)  ^{\prime}$ and
$\boldsymbol{\phi}$, and compare its limiting properties (as $n\rightarrow
\infty$) with the standard MV portfolio. We show that for known factor
loadings\ and non-zero $\boldsymbol{\phi}$ there exist phi-portfolios with
strictly positive returns, given by $\boldsymbol{\phi}^{\prime}%
\boldsymbol{\phi}>0$, that are fully diversified (their variance tends to zero
with $n$). The rate at which the variance of phi-portfolio returns tends to
zero will depend on the strength of the traded risk factors compared to the
strength of the missing (latent) factors, highlighting the importance of
factor strengths in portfolio analysis. Since the phi-portfolios have Sharpe
ratios that tend to infinity with $n$, they dominate MV portfolios whose
Sharpe ratio is bounded in $n$. Note that unlike MV\ portfolios, the
phi-portfolios do not require knowledge of the inverse of the covariance
matrix of returns, which is particularly difficult to estimate accurately when
$n$ is large. As is usual, these theoretical results apply to population
values where there are no restrictions on trading many individual securities.

Next we consider estimation of and inference about $\boldsymbol{\phi}$ using a
large number of individual securities, taking account of firm-specific pricing
errors, traded factors that are not strong, missing latent factors, and panel
data sets where the time dimension, $T$, is small relative to $n$. Factor
strength plays a central role in our analysis of $\boldsymbol{\phi}$. We use a
measure of factor strength, $\alpha_{k},$ developed in
\citet{bailey2016exponent, bailey2021measurement}%
, which is defined in terms of the proportion of non-zero factor loadings,
$\beta_{ik}$.\footnote{Alpha is used both for the LFPM intercepts and the
measure of factor strength, because these are established usages, but it will
be clear from context and subscripts which is being referred to.} A factor is
strong if this proportion is very close to unity, it is semi-strong if
$1>\alpha_{k}>1/2$, and it is weak if $\alpha_{k}<1/2$. Use of this measure
allows us to be precise about the degree of pervasiveness and show how the
strengths of the observed factors, the missing factors and the pricing errors,
each influence estimation and inference about $\boldsymbol{\phi}$.

In practice, exploitation of a non-zero $\boldsymbol{\phi}$ requires $n$ to be
large and rebalancing such long-short portfolios for so many securities may
incur high transactions costs or not be feasible. In addition, model
uncertainty, estimation uncertainty, time variation in both $\boldsymbol{\beta
}_{i}$ and in conditional volatility pose additional difficulties in
implementing a strategy to exploit the potential returns revealed by
$\boldsymbol{\phi}$. In developing the theory we will abstract from such
practical difficulties, but in the empirical section we illustrate some of
these issues with a comparison of the performance of $\boldsymbol{\phi}$ based
portfolios relative to MV portfolios which would face similar difficulties.

We estimate $\boldsymbol{\phi}=(\phi_{1},\phi_{2},...,\phi_{K})^{\prime}$,
using a two-step estimator. In the first step we estimate the intercepts,
$\hat{\alpha}_{i}$, and the factor loadings, $\boldsymbol{\hat{\beta}}_{i}$,
from least squares regressions of excess returns on an intercept and risk
factors. In the second step, $\boldsymbol{\phi}$ is estimated from the cross
section regression of $\hat{\alpha}_{i}$ on $\boldsymbol{\hat{\beta}}_{i}$. As
with the two-step estimator of $\boldsymbol{\lambda}$, such a two-step
estimator of $\phi_{k}$ will also be biased, and requires bias-correction.
Following
\citet{shanken1992estimation}%
, we consider a bias-corrected version of the two-step estimator of
$\boldsymbol{\phi}$, which we denote by $\boldsymbol{\tilde{\phi}}_{nT}$. We
develop the asymptotic distribution of $\boldsymbol{\tilde{\phi}}_{nT}$ under
quite general set of assumptions regarding the idiosyncratic pricing errors,
error cross-sectional dependence, and the presence of missing (latent)
factors. The paper also investigates the implications of factor strengths for
the precision with which $\boldsymbol{\phi}$ can be estimated. The LFPM,
following
\citet{cham1983arbi}%
, assumes that all the observed factors are strong and the eigenvalues of the
covariance matrix of the errors are bounded.

In developing the arbitrage pricing theory, APT,
\citet{ROSS1976341}%
, whose concerns were primarily theoretical, assumed the factors had mean
zero: $\boldsymbol{\mu}=0,$ so $\boldsymbol{\phi}=\boldsymbol{\lambda}$, is
the risk premium. For traded factors under market efficiency, where
$\boldsymbol{\phi}=\mathbf{0}$, the risk premium is the factor mean
$\boldsymbol{\mu}=\boldsymbol{\lambda}$. Were one interested in estimating
$\boldsymbol{\lambda}$ there may be statistical reasons to estimate $\phi_{k}$
and $\mu_{k}$ separately, then summing them to obtain an estimate of
$\lambda_{k}$. The factor mean, $\mu_{k}=E(f_{kt})$, can be estimated
consistently at the regular $\sqrt{T}$ rate directly using time series data on
the risk factors, $f_{kt}$, for $t=1,2,...,T,$ and does not require knowledge
of the factor loadings or $n$. In contrast, estimation of $\phi_{k}$ requires
panel data to estimate the factor loadings and hence both $n$ and $T$
dimensions are important. In some cases it may be beneficial to use different
time series dimensions, $T_{\mu}$ and $T_{\phi},$ to estimate $\mu_{k}$ and
$\phi_{k}$, respectively. For example, when factor loadings are subject to
breaks it is advisable to use a relatively short sample, and when some factors
are not sufficiently strong a large $n$ is required. Thus one could use large
$T$ to estimate $\mu_{k}$ and small $T$ to estimate $\phi_{k}$ and
$\lambda_{k}$ can be simply estimated by adding the estimates of $\mu_{k}$ and
$\phi_{k}.$ We do not pursue the use of different $T$, but if the same $T$ is
used to estimate the mean of factor $k,$ say $\hat{\mu}_{k,T},$ and the
bias-corrected $\tilde{\phi}_{k,nT}$ then their sum is the same as the
\citet{shanken1992estimation}
bias-corrected risk premium for factor $k,$ $\tilde{\lambda}_{k,nT}.$ This
decomposition can be used to obtain the rate at which $\tilde{\lambda}_{k,nT}$
converges to its true value, $\lambda_{0,k}=\phi_{0,k}+\mu_{0,k}$.

The main theoretical results of the paper are set out around five theorems
under a number of key assumptions, with proofs provided in the mathematical
appendix. The small sample properties of $\boldsymbol{\tilde{\phi}}_{nT}$ are
investigated using extensive Monte Carlo experiments, allowing for a mixture
of strong and semi-strong observed factors, latent factors, pricing errors,
GARCH effects and non-Gaussian errors. Small sample results are in line with
our theoretical findings, and confirm that the bias-corrected estimator,
$\boldsymbol{\tilde{\phi}}_{nT}$, has the correct size and good power
properties for samples with time series dimensions of $T=120$ and $T=240$.
They also show, in accord with the theory, that the precision with which
$\phi_{k}$ is estimated falls with $\alpha_{k}$, the strength of the $k^{th}$ factor.

Our theoretical derivations and Monte Carlo simulations assume that the list
of relevant observed factors is known. But in practice the relevant factors
must be selected. Extending the theory to the high dimensional case where
factors are selected, rather than given, is beyond the scope of the present
paper. Since the rate of convergence of $\boldsymbol{\tilde{\phi}}_{nT}$ to
its true value is given by $\sqrt{T}n^{(\alpha_{k}+\alpha_{\min}-1)/2}$, and
the Monte Carlo confirms the crucial role of factor strengths in estimation
and inference on $\boldsymbol{\phi,}$ it seems sensible to select factors on
the basis of their factor strength. Weak factors whose strength is around
$1/2$ can be ignored and absorbed into the error term.

The above selection procedure is applied in a high dimensional setting with
both a large number of securities ($n$ from $1,090$ to $1,175)$ and a large
number of potential risk factors ($m$ from $177$ to $189)$, taken from the
\citet{chen2022open}%
, which can all be traded. We used monthly data over the period
$1996m1-2022m12$ and considered balanced panels obtained by including
\textit{all} existing stocks in a given month for which there are $T$
observations. We considered $T=120$ and $240$ months, and focus on the latter
which we found to be more reliable, given the large number of securities being
considered. Various procedures could be used to select risk factors for a
given security, $i$. We used Lasso for this purpose and then selected a subset
of these factors that were chosen by a sufficiently large number of securities
in the sample, and whose estimated strength were above the given threshold
value of $\alpha_{k}>0.75$. We refer to this selection procedure as pooled
Lasso. Using this procedure with $T=240$ we ended up with $7$ risk factors for
the sample ending in $2015$, declining to $4$ in $2021$. Interestingly, the
three Fama-French factors were always included in the set of factors selected
by pooled Lasso.

Accordingly, we considered three linear asset pricing models for our portfolio
analysis: the pooled Lasso selected at the start of our evaluation sample,
denoted by PL7, the Fama-French three factors model, FF3, and the Fama-French
five factors model, FF5, which includes two factors not in PL7. The three
models are estimated using rolling samples of size $T=240$, starting with a
sample ending in $2015m12$ and finishing with a sample ending in $2022m11$.
The hypothesis that $\boldsymbol{\phi}=\mathbf{0}$ was rejected for all $84$
rolling samples and all three models, albeit less strongly over the post
Covid-19 period. The test results suggested possible unexploited return
opportunities, and to investigate this possibility further, we constructed
phi-portfolios and compared their Sharpe ratios with the ones based on
standard MV portfolios over the full sample evaluation sample,
$2016m1-2022m12$, and sample ending $2019m12$, that excludes the Covid-19
period. We find that in five out of the six cases (3 models 2 samples) the
phi-portfolio has a higher SR than the corresponding MV portfolio. The
exception is the FF5 pre Covid-19. This illustrates that if $\boldsymbol{\phi
}\neq0,$ it is possible to construct a portfolio that outperforms the mean
variance portfolio. In both the pre Covid-19 sample and the full sample the
highest SR was obtained by the PL7 phi-portfolios, which also outperformed the
S\&P500. The SRs for the sample ending in $2022$ were substantially lower than
the sample ending in $2019$, consistent with a falling value of the
probability that $\boldsymbol{\phi}$ was non-zero.

\textbf{Related literature: }On estimation of risk premia, following Fama and
MacBeth and
\citet{shanken1992estimation}%
, estimation of risk premia is further examined by
\citet{shanken2007estimating}%
,
\citet{kan2013pricing}%
, and
\citet{BAI201531}%
. The survey paper by
\citet{jagannathan2010analysis}
provide further references.

Testing for market efficiency dates back to
\citet{jensen1968performance}
who proposes testing a$_{i}=0$ for each $i$ separately.
\citet{gibbons1989test}
provide a joint test for the case where the errors are Gaussian and $n<T$.
\citet{gagliardini2016time}
develop two-pass regressions of individual stock returns, allowing
time-varying risk premia, and propose a standardised Wald test.
\citet{raponi2019testing}
propose a test of pricing error in cross section regression for fixed number
of time series observations. They use a bias-corrected estimator of
\citet{shanken1992estimation}
to standardise their test statistic.
\citet{ma2020testing}
employ polynomial spline techniques to allow for time variations in factor
loadings when testing for alphas.
\citet{feng2022high}
propose a max-of-square type test of the intercepts instead of the average
used in the literature, and recommend using a combination of the two testing
procedures.
\citet{he2022testing}
propose two statistics, a Wald type statistic which require $n$ and $T$ to be
of the same order of magnitude and a standardised t-ratio.
\citet{kleibergen2009tests}
considers testing in the case where the loadings are small.
\citet{pesaran2012testing, pesaran2023testing}
consider testing that the intercepts in the LFPM are zero when $n$ is large
relative to $T$ and there may be non Gaussian errors and weakly
cross-correlated errors.

A large number of risk factors have been considered in the empirical
literature. We use the
\citet{fama1993common}
three factors in our Monte Carlo design. In our empirical application we use
the five factors proposed by
\citet{fama2015five}
and the large set of factors provided by
\citet{chen2022open}%
.
\citet{harvey2019census}
document over 400 factors published in top finance journals.
\citet{dello2022missing}
argue that despite the hundreds of systematic risk factors considered in the
literature, there is still a sizable pricing error and that this can be
explained by asset specific risk that reflects market frictions and behavioral
biases. There is a large Bayesian literature, including
\citet{chib2020comparing}%
, and
\citet{hwang2022bayesian}
on selecting factor models. The issue of factor selection is also addressed
by
\citet{fama2018choosing}%
.

Strong and weak factors in asset returns are considered by
\citet{ANATOLYEV2022103}%
,
\citet{connor2022semi}%
, and \cite{giglio2023test}.
\citet{beaulieu2020arbitrage}
discuss the lack of identification of risk premia when many of the loadings
are zero. There has also been concern about the consequences of omitted
factors.
\citet{giglio2021asset}
discuss the problem and try to deal with it using a three-pass method which is
valid even when not all factors in the model are specified or observed using
principal components of the test assets.
\citet{onatski2012asymptotics}
and
\citet{lettau2020estimating, lettau2020factors}
provide extensive discussions of weak factor and latent factors, respectively.
More recent contributions include \cite{BaiNg2023Weak} and
\cite{uematsu2023inference}.

There is also a large literature on portfolio construction.
\cite{Herskovic2019} discuss low cost methods of hedging risk factors.
\cite{Preite2024} derive an SDF in which there is compensation for
unsystematic risk within the framework of the APT. \cite{Korsaye2021}
introduce model-free smart SDFs which give rise to non-parametric SDF bounds
for testing asset pricing models.\ \ \cite{Daniel2020} discuss the common
practice of creating characteristic portfolios by sorting on characteristics
associated with average returns and show that these portfolios capture not
only the priced risk associated with the characteristic but also unpriced
risk. \cite{Quaini2023} propose an estimator of tradable factor risk premia.

\textbf{Paper's outline}: The rest of the paper is organized as follows:
Section \ref{LFPMAPT} provides the framework for estimation of
$\boldsymbol{\phi}.$ Section \ref{Assumptions} sets out the assumptions and
states the main theorems. Theorem \ref{TFMbias} shows that the standard
Fama-MacBeth estimator is valid only when there are no pricing errors and
$n/T\rightarrow0$. Theorem \ref{Thzig} shows that the Shanken bias-corrected
estimator of $\lambda_{k}$ continues to be consistent for a fixed $T$ as
$n\rightarrow\infty$, even in presence of weak pricing errors and weak missing
common factors. Theorem \ref{Tfi} provides conditions under which the
bias-corrected estimator, $\boldsymbol{\tilde{\phi}}_{nT}$, is consistent for
$\boldsymbol{\phi}_{0}$, and derives the asymptotic distribution of
$\boldsymbol{\tilde{\phi}}_{nT},$ assuming the observed factors are strong
$(\alpha_{k}=1$ for all $k$). Theorem \ref{Tsemi} extends the results to cases
where one or more risk factors are semi-strong, and establishes the rate at
which $\boldsymbol{\tilde{\phi}}_{nT}$ converges to its true value, assuming
the idiosyncratic pricing errors are sufficiently weak, as discussed after the
theorem. For example, for a factor with strength $\alpha_{k}$ we show that
$\tilde{\phi}_{k,nT}-\phi_{0,k}=O_{p}\left(  T^{-1/2}n^{-(\alpha_{k}%
+\alpha_{\min}-1)/2}\right)  $, and as a consequence
\[
\tilde{\lambda}_{k,nT}-\lambda_{0,k}=O_{p}\left(  T^{-1/2}n^{-(\alpha
_{k}+\alpha_{\min}-1)/2}\right)  +O_{p}\left(  T^{-1/2}\right)  ,
\]
where $\lambda_{0,k}$ is the true value of the risk premia associated to
factor $f_{kt}$, and $\alpha_{\min}$ is the strength of the least strong
factor included. This consistency condition is weaker than the one derived by
\cite{giglio2023test} who also assume $\eta_{i}=0$, for all $i$. Finally,
Theorem \ref{Var} gives conditions for consistent estimation of the asymptotic
variance of $\boldsymbol{\tilde{\phi}}_{nT},$ using a suitable threshold
estimator of the covariance matrix. Section \ref{Simulations} presents the
Monte Carlo (MC) design, its calibration and a summary of the main findings.
Section \ref{FacSel}\ discusses the problem of factor selection from a large
number of potential factors. Section \ref{Empirical} gives the empirical
application using monthly data on a large number of individual US stocks and
risk factors over the period 1996-2021. It selects factors, estimates
$\boldsymbol{\phi}\mathbf{,}$ and compares the performance of MV and
phi-portfolios. Section \ref{Conclusion} provides some concluding remarks.

Detailed mathematical proofs are provided in a mathematical appendix. Further
information on data sources, MC calibration plus some supplementary material
for the empirical application are provided in the online supplement A. To save
space all MC results are given in the online supplement B.

\section{Identification and estimation of $\boldsymbol{\phi}$\label{LFPMAPT}}

Let $R_{it}$ denote the holding period return on traded security $i$, which
can be bought long or short without transaction costs, $r_{t}^{f}$ \ is the
risk free rate, and $r_{it}=R_{it}-r_{t}^{f}$ is the excess return. We start
with the linear factor pricing model (LFPM)
\begin{equation}
r_{it}=\mathit{\alpha}_{i}+\boldsymbol{\beta}_{i}^{\prime}\mathbf{f}%
_{t}+u_{it}, \label{StatM}%
\end{equation}
for $i=1,2,...,n$ and $t=1,2,...,T$, where $r_{it}$ is explained in terms of
the $K\times1$ vector of factors $\mathbf{f}_{t}=(f_{1t},f_{2t},...,f_{Kt}%
)^{\prime}$. The intercept $\mathit{\alpha}_{i}$ and the $K\times1$ vector of
factor loadings, $\boldsymbol{\beta}_{i}=(\beta_{i1},\beta_{i2},...,\beta
_{iK})^{\prime}$, are unknown. The idiosyncratic errors, $u_{it}$ have zero
means and are assumed to be serially uncorrelated. The factors, $\mathbf{f}%
_{t}$, are assumed to be covariance stationary with the constant mean
$\boldsymbol{\mu}=E\left(  \mathbf{f}_{t}\right)  $, and $\boldsymbol{\Sigma
}_{f}=E\left[  \left(  \mathbf{f}_{t}-\boldsymbol{\mu}\right)  \left(
\mathbf{f}_{t}-\boldsymbol{\mu}\right)  ^{\prime}\right]  $.

Under the Arbitrage Pricing Theory (APT) due to Ross (1976) the pricing
errors, $\eta_{i}$ for $i=1,2,...,n$ defined by
\begin{equation}
\eta_{i}=E\left(  r_{it}\right)  -c-\boldsymbol{\beta}_{i}^{\prime
}\boldsymbol{\lambda}\mathbf{,} \label{APTmean}%
\end{equation}
are bounded such that
\begin{equation}
\sum_{i=1}^{n}\eta_{i}^{2}<C<\infty, \label{APTRoss}%
\end{equation}
where $c$ is zero-beta expected excess return, $\boldsymbol{\lambda}$ is the
$K\times1$ vector of risk premia. Under APT, we have
\begin{equation}
E\left(  r_{it}\right)  =c+\boldsymbol{\beta}_{i}^{\prime}\boldsymbol{\lambda
}_{0}+\eta_{i}, \label{Urit}%
\end{equation}
and reduces to the standard beta representation of an unconditional asset
pricing model, if $c+\eta_{i}=0$. In this case, $\boldsymbol{\beta}%
_{i}=-Cov(r_{it},m_{t})/var(m_{t})$, where $m_{t}$ is the stochastic discount
factor (SDF), which satisfies the moment condition $E_{t}(r_{i,t+1}m_{t+1}%
)=0$. See \cite{Pesaran2023}\textbf{ }for a discussion of the link. For
empirical assessment, we consider the more general unconditional pricing model
(\ref{Urit}), and relate to the linear factor pricing model. To this end,
taking unconditional expectations of (\ref{StatM}) and using the
APT\ condition we have
\begin{equation}
E(r_{it})=\mathit{\alpha}_{i}+\boldsymbol{\beta}_{i}^{\prime}\boldsymbol{\mu
}=\boldsymbol{\beta}_{i}^{\prime}\boldsymbol{\lambda}+c+\eta_{i}.
\label{meuiR}%
\end{equation}
which in turn yields%
\begin{equation}
\mathit{\alpha}_{i}=c+\boldsymbol{\beta}_{i}^{\prime}\boldsymbol{\phi}%
+\eta_{i}\text{, for }i=1,2,...,n. \label{ai}%
\end{equation}
where
\begin{equation}
\boldsymbol{\phi}=\boldsymbol{\lambda}-\boldsymbol{\mu}, \label{phi0}%
\end{equation}
Under the APT condition, (\ref{APTRoss}), $\boldsymbol{\phi}$ can be
identified from the cross section regression of $\alpha_{i}$ on
$\boldsymbol{\beta}_{i}$. We relax the APT\ condition and derive restrictions
on the degree to which idiosyncratic pricing errors, $\eta_{i}$, can be
pervasive to achieve identification.

The focus of the literature has been on testing for alpha, $\mathit{\alpha
}_{i}=0$, and the estimation of the risk premia, $\boldsymbol{\lambda}$, using
panel data on excess returns, $\left\{  r_{it},1,2,...,n\text{; }%
t=1,2,...,T\right\}  ,$ and $\mathbf{F}$, the $T\times K$ matrix of
observations on the factors. It is clear that $\boldsymbol{\phi}$ plays an
important role in tests for alpha in LFPM, and a non-zero $\boldsymbol{\phi}$
implies non-zero alphas, which in turn implies exploitable excess profitable
opportunities. More specifically, we show that for known values of
$\boldsymbol{\beta}_{i},$ $i=1,2,...,n\,$\ and when $\boldsymbol{\phi}$ is
non-zero, there exists phi-based portfolios with non-zero means that are fully
diversified (their variance tends to zero with $n$), namely have Sharpe ratios
that tend to infinity with $n$, and hence dominate the MV portfolios. For the
MV portfolios to be efficient it is necessary that $\boldsymbol{\phi=0}$.

\subsection{Why $\boldsymbol{\phi}$ matters: introduction of the
phi-portfolios}

Substitute the expression for $\mathit{\alpha}_{i}$ given by (\ref{ai}) in
(\ref{StatM}) to obtain%

\begin{equation}
r_{it}=c+\boldsymbol{\beta}_{i}^{\prime}\boldsymbol{\phi}+\eta_{i}%
+\boldsymbol{\beta}_{i}^{\prime}\mathbf{f}_{t}+u_{it},\text{ for }i=1,2,...,n,
\label{rit}%
\end{equation}
and write the $n$ return equations more compactly as%
\begin{equation}
\mathbf{r}_{\circ t}=c\mathbf{\tau}_{n}+\mathbf{B}_{n}\boldsymbol{\phi
}+\mathbf{B}_{n}\mathbf{f}_{t}+\boldsymbol{\eta}_{n}+\mathbf{u}_{\circ t},
\label{ritV}%
\end{equation}
where $\mathbf{r}_{\circ t}=(r_{1t},r_{2t},....,r_{nt})^{\prime},$
$\boldsymbol{\tau}_{n}$ is an $n$-dimensional vector of ones, $\mathbf{B}%
_{n}=(\boldsymbol{\beta}_{\circ1},\boldsymbol{\beta}_{\circ2}%
,...,\boldsymbol{\beta}_{\circ K}),$ $\boldsymbol{\beta}_{\circ k}=(\beta
_{1k},\beta_{2k},...,\beta_{nk})^{\prime},$ $\mathbf{u}_{\circ t}%
=(u_{1t},u_{2t},....,u_{nt})^{\prime}$, $\boldsymbol{\eta}_{n}=\left(
\eta_{1},\eta_{2},...,\eta_{n}\right)  ^{\prime}$, and $\mathbf{V}%
_{u}=E\left(  \mathbf{u}_{\circ t}\mathbf{u}_{\circ t}^{\prime}\right)  $.
Suppose that the factors, $\mathbf{f}_{t},$ are traded, and $\boldsymbol{\phi
}^{\prime}\boldsymbol{\phi}>0$. Consider the $n\times1$ vector of
phi-portfolio weights,$\mathbf{w}_{\phi}$, given by%

\begin{equation}
\mathbf{w}_{\phi}=\mathbf{M}_{n}\mathbf{B}_{n}\left(  \mathbf{B}_{n}^{\prime
}\mathbf{M}_{n}\mathbf{B}_{n}\right)  ^{-1}\boldsymbol{\phi}, \label{wphi}%
\end{equation}
where $\mathbf{M}_{n}=\mathbf{I}_{n}-n^{-1}\boldsymbol{\tau}_{n}%
\boldsymbol{\tau}_{n}^{\prime}$. Finally, consider the long-short hedged
portfolio return
\begin{equation}
\rho_{t,\phi}=\mathbf{w}_{\phi}^{\prime}\mathbf{r}_{\circ t}-\boldsymbol{\phi
}^{\prime}\mathbf{f}_{t}=\boldsymbol{\phi}^{\prime}\left[  \left(
\mathbf{B}_{n}^{\prime}\mathbf{M}_{n}\mathbf{B}_{n}\right)  ^{-1}%
\mathbf{B}_{n}^{\prime}\mathbf{M}_{n}\mathbf{r}_{\circ t}-\mathbf{f}%
_{t}\right]  , \label{phiport}%
\end{equation}
and using (\ref{ritV}) note that
\begin{equation}
\rho_{t,\phi}=\boldsymbol{\phi}^{\prime}\boldsymbol{\phi}+\mathbf{w}_{\phi
}^{\prime}\boldsymbol{\eta}_{n}+\mathbf{w}_{\phi}^{\prime}\mathbf{u}_{\circ
t}\mathbf{.} \label{phiport1}%
\end{equation}
For a given vector of pricing errors, $\boldsymbol{\eta}_{n}$, $E\left(
\rho_{t,\phi}\right)  =\boldsymbol{\phi}^{\prime}\boldsymbol{\phi}%
+\mathbf{w}_{\phi}^{\prime}\boldsymbol{\eta}_{n}$, and $Var\left(
\rho_{t,\phi}\right)  =\mathbf{w}_{\phi}^{\prime}\mathbf{V}_{u}\mathbf{w}%
_{\phi}\mathbf{,}$ where $\mathbf{V}_{u}=E(\mathbf{u}_{\circ t}\mathbf{u}%
_{\circ t}^{\prime})$. Using (\ref{wphi}) it follows that
\[
\left\Vert \mathbf{w}_{\phi}^{\prime}\boldsymbol{\eta}_{n}\right\Vert ^{2}%
\leq\left\Vert \boldsymbol{\eta}_{n}\right\Vert ^{2}\left\Vert \mathbf{w}%
_{\phi}\right\Vert ^{2}\leq\left\Vert \boldsymbol{\eta}_{n}\right\Vert
^{2}\left\Vert \boldsymbol{\phi}\right\Vert ^{2}\lambda_{max}\left[  \left(
\mathbf{B}_{n}^{\prime}\mathbf{M}_{n}\mathbf{B}_{n}\right)  ^{-1}\right]
=\frac{\left\Vert \boldsymbol{\eta}_{n}\right\Vert ^{2}\left\Vert
\boldsymbol{\phi}\right\Vert ^{2}}{\lambda_{min}\left(  \mathbf{B}_{n}%
^{\prime}\mathbf{M}_{n}\mathbf{B}_{n}\right)  }.
\]
Similarly,
\begin{equation}
Var\left(  \rho_{t,\phi}\right)  \leq\left\Vert \mathbf{w}_{\phi}\right\Vert
^{2}\lambda_{max}\left(  \mathbf{V}_{u}\right)  =\frac{\lambda_{max}\left(
\mathbf{V}_{u}\right)  \left\Vert \boldsymbol{\phi}\right\Vert ^{2}}%
{\lambda_{min}\left(  \mathbf{B}_{n}^{\prime}\mathbf{M}_{n}\mathbf{B}%
_{n}\right)  }. \label{varphi}%
\end{equation}
Hence, $E\left(  \boldsymbol{\rho}_{t,\phi}\right)  \rightarrow
\boldsymbol{\phi}^{\prime}\boldsymbol{\phi}$, and $Var\left(  \rho_{t,\phi
}\right)  \rightarrow0$, as $n\rightarrow\infty$, if
\begin{equation}
\frac{\left\Vert \boldsymbol{\eta}_{n}\right\Vert ^{2}}{\lambda_{min}\left(
\mathbf{B}_{n}^{\prime}\mathbf{M}_{n}\mathbf{B}_{n}\right)  }\rightarrow
0\text{, and }\frac{\lambda_{max}\left(  \mathbf{V}_{u}\right)  }%
{\lambda_{min}\left(  \mathbf{B}_{n}^{\prime}\mathbf{M}_{n}\mathbf{B}%
_{n}\right)  }\rightarrow0. \label{fulldiv}%
\end{equation}
These conditions are met if $\lambda_{min}\left(  \mathbf{B}_{n}^{\prime
}\mathbf{M}_{n}\mathbf{B}_{n}\right)  \rightarrow\infty$, $\lambda
_{max}\left(  \mathbf{V}_{u}\right)  <C$, and the APT condition (\ref{APTRoss}%
) holds. The first two conditions follow if the LFPM given by (\ref{StatM}) is
an approximate factor model as assumed by
\citet{cham1983arbi}%
, namely when the factors are strong and the errors, $u_{it}$, are weakly
cross correlated.\footnote{The diversification conditions in (\ref{fulldiv})
are met more generally, and accommodate the presence of semi-strong traded
factors in $f_{t}$, and less restrictive conditions on $\mathbf{V}_{u}$.and
the pricing errors $\eta_{i}.$ See Remark \ref{Remfulldiv} below.}

Therefore, knowledge of $\boldsymbol{\phi}^{\prime}\boldsymbol{\phi}$ and its
statistical significance can play an important role in portfolio analysis. To
illustrate this point suppose that $c=0$ and $\boldsymbol{\eta}_{n}%
=\mathbf{0}$, but $\boldsymbol{\phi}^{\prime}\boldsymbol{\phi}>0$, and
consider the hedged return $\rho_{t,\phi}=\boldsymbol{\phi}^{\prime}\left[
\left(  \mathbf{B}_{n}^{\prime}\mathbf{V}_{u}^{-1}\mathbf{B}_{n}\right)
^{-1}\mathbf{B}_{n}\mathbf{V}_{u}^{-1}\mathbf{r}_{\circ t}-\mathbf{f}%
_{t}\right]  $. Then under LFPM we have\footnote{When $c\neq0$, it is not
possible to eliminate $c$ (which is unpriced), and at the same time exploit
the error covariance matrix $\mathbf{V}_{u}$.}%
\begin{align*}
\rho_{t,\phi} &  =\boldsymbol{\phi}^{\prime}\left[  \left(  \mathbf{B}%
_{n}^{\prime}\mathbf{V}_{u}^{-1}\mathbf{B}_{n}\right)  ^{-1}\mathbf{B}%
_{n}^{\prime}\mathbf{V}_{u}^{-1}\left(  \mathbf{B}_{n}\boldsymbol{\phi}%
_{0}+\mathbf{B}_{n}\mathbf{f}_{t}+\mathbf{u}_{\circ t}\right)  -\mathbf{f}%
_{t}\right]  \\
&  =\boldsymbol{\phi}^{\prime}\boldsymbol{\phi}+\boldsymbol{\phi}^{\prime
}\left(  \mathbf{B}_{n}^{\prime}\mathbf{V}_{u}^{-1}\mathbf{B}_{n}\right)
^{-1}\mathbf{B}_{n}^{\prime}\mathbf{V}_{u}^{-1}\mathbf{u}_{\circ t},
\end{align*}
and its squared Sharpe ratio is given by
\begin{equation}
SR_{\phi}^{2}=\frac{\left(  \boldsymbol{\phi}^{\prime}\boldsymbol{\phi
}\right)  ^{2}}{\boldsymbol{\phi}^{\prime}\left(  \mathbf{B}_{n}^{\prime
}\mathbf{V}_{u}^{-1}\mathbf{B}_{n}\right)  ^{-1}\boldsymbol{\phi}%
}.\label{SRsquared1}%
\end{equation}
Since,
\[
\boldsymbol{\phi}^{\prime}\left(  \mathbf{B}_{n}^{\prime}\mathbf{V}_{u}%
^{-1}\mathbf{B}_{n}\right)  ^{-1}\boldsymbol{\phi}\leq\left(  \boldsymbol{\phi
}^{\prime}\boldsymbol{\phi}\right)  \lambda_{\max}\left[  \left(
\mathbf{B}_{n}^{\prime}\mathbf{V}_{u}^{-1}\mathbf{B}_{n}\right)  ^{-1}\right]
=\frac{\boldsymbol{\phi}^{\prime}\boldsymbol{\phi}}{\lambda_{min}\left(
\mathbf{B}_{n}^{\prime}\mathbf{V}_{u}^{-1}\mathbf{B}_{n}\right)  },
\]
then%
\begin{equation}
SR_{\phi}^{2}\geq\left(  \boldsymbol{\phi}^{\prime}\boldsymbol{\phi}\right)
\text{ }\lambda_{min}\left(  \mathbf{B}_{n}^{\prime}\mathbf{V}_{u}%
^{-1}\mathbf{B}_{n}\right)  .\label{SRsquared2}%
\end{equation}
As a result, if $\lambda_{min}\left(  \mathbf{B}_{n}^{\prime}\mathbf{V}%
_{u}^{-1}\mathbf{B}_{n}\right)  \rightarrow\infty$ as $n\rightarrow\infty$,
$SR_{\phi}^{2}$ increases in $n$ without bounds if $\boldsymbol{\phi}^{\prime
}\boldsymbol{\phi}>0$. In contrast, it is easily seen that the Sharpe ratio of
the standard mean-variance (MV) portfolio, defined by $SR_{MV}^{2}%
=\boldsymbol{\mu}_{R}^{\prime}\mathbf{V}_{R}^{-1}\boldsymbol{\mu}_{R}$ is
bounded in $n$, and in consequence $SR_{\phi}^{2}$ will eventually dominate
$SR_{MV}^{2}$ if $\boldsymbol{\phi}^{\prime}\boldsymbol{\phi}>0$. To see this,
note that under the LFPM given by (\ref{rit}) with $c=0$ and $\boldsymbol{\eta
}_{n}=\mathbf{0}$, the MV portfolio is given by $\rho_{MV,t}=\boldsymbol{\mu
}_{R}^{\prime}\mathbf{V}_{R}^{-1}\mathbf{r}_{\circ t}$, where $\boldsymbol{\mu
}_{R}=\mathbf{B}_{n}\boldsymbol{\lambda}$ and $\mathbf{V}_{R}=\mathbf{B}%
_{n}\mathbf{\Sigma}_{f}\mathbf{B}_{n}^{\prime}+\mathbf{V}_{u}$. Also, since
$\mathbf{B}_{n}\mathbf{\Sigma}_{f}\mathbf{B}_{n}^{\prime}$ is rank deficient
then
\[
\mathbf{V}_{R}^{-1}=\mathbf{V}_{u}^{-1}\mathbf{-V}_{u}^{-1}\mathbf{B}\left(
\mathbf{\Sigma}_{f}^{-1}+\mathbf{B}_{n}^{\prime}\mathbf{V}_{u}^{-1}%
\mathbf{B}_{n}\right)  ^{-1}\mathbf{B}_{n}^{\prime}\mathbf{V}_{u}^{-1},
\]
and it follows that
\begin{align}
SR_{MV}^{2} &  =\boldsymbol{\mu}_{R}^{\prime}\mathbf{V}_{R}^{-1}%
\boldsymbol{\mu}_{R}-\boldsymbol{\lambda}^{\prime}\mathbf{\Sigma}_{f}%
^{-1}\left(  \mathbf{\Sigma}_{f}^{-1}+\mathbf{B}_{n}^{\prime}\mathbf{V}%
_{u}^{-1}\mathbf{B}_{n}\right)  ^{-1}\mathbf{\Sigma}_{f}^{-1}%
\boldsymbol{\lambda}\label{SRsqMV}\\
&  \leq\boldsymbol{\lambda}^{\prime}\mathbf{\Sigma}_{f}^{-1}%
\boldsymbol{\lambda}=\left(  \boldsymbol{\mu+\phi}\right)  ^{\prime
}\mathbf{\Sigma}_{f}^{-1}\left(  \boldsymbol{\mu+\phi}\right)  <C\text{.}%
\nonumber
\end{align}
Hence, the Sharpe ratio of the MV portfolio continues to be bounded in $n$
even if $\boldsymbol{\phi\neq0}$. Whether $\boldsymbol{\phi=0}$ or not affects
the magnitude of $SR_{MV}^{2}$ but does not alter the fact that $SR_{MV}^{2}$
will be bounded if $\lambda_{min}\left(  \mathbf{B}_{n}^{\prime}\mathbf{V}%
_{u}^{-1}\mathbf{B}_{n}\right)  \rightarrow\infty$ as $n\rightarrow\infty$.
Furthermore, it is not possible to improve over the MV portfolio by using only
the factor loadings. The optimal portfolio formed using the $K$ beta-based
portfolios, {\small $\boldsymbol{\rho}_{B,t}=$ }$\left(  \mathbf{B}%
_{n}^{\prime}\mathbf{V}_{u}^{-1}\mathbf{B}_{n}\right)  ^{-1}\mathbf{B}%
_{n}^{\prime}\mathbf{V}_{u}^{-1}${\small $\mathbf{r}_{\circ t}$, }has the same
Sharpe ratio as the mean-variance portfolio. See Lemma \ref{Sharpe} for a proof.

In short, to make sure that the mean-variance portfolio is efficient we must
have $\boldsymbol{\phi}^{\prime}\boldsymbol{\phi}=0$, and it is of special
interest to estimate $\boldsymbol{\phi}$ and develop reliable procedures for
testing its statistical significance, using a large number of securities. In
practice, the number of tradeable securities, $n$, might not be sufficiently
large, and there are important specification and estimation uncertainties, and
the Sharpe ratio of the phi-based portfolio, $\boldsymbol{\rho}_{t,\phi}$, is
likely to be bounded even if $\boldsymbol{\phi}^{\prime}\boldsymbol{\phi}>0$.
We turn to these issues in the empirical application provided in\ Section
\ref{Empirical}.

\subsection{Fama-MacBeth and Shanken estimators of risk premia\label{FM}}

It will prove convenient to write (\ref{StatM}) in matrix notation by stacking
the excess returns by $t=1,2,...,T$, for each security $i$
\begin{equation}
\mathbf{r}_{i\circ}=\mathit{\alpha}_{i}\boldsymbol{\tau}_{T}+\mathbf{F}%
\boldsymbol{\beta}_{i}+\mathbf{u}_{i\circ},\text{ for }i=1,2,...,n,
\label{mfi}%
\end{equation}
where $\mathbf{r}_{i\circ}=(r_{i1},r_{i2},...,r_{iT})^{\prime}$,
$\mathbf{F=(f}_{1},\mathbf{f}_{2},...,\mathbf{f}_{T})^{\prime}$,
$\mathbf{u}_{i\circ}=\left(  u_{i1},u_{i2},...,u_{iT}\right)  ^{\prime}$, and
$\boldsymbol{\tau}_{T}$ is a $T\times1$ vector of ones. Similarly, stacking
the excess returns by $i$ for each $t$ we have (\ref{ritV}) which we rewrite
as
\begin{equation}
\mathbf{r}_{\circ t}=\boldsymbol{\alpha}_{n}+\boldsymbol{B}_{n}\mathbf{f}%
_{t}+\mathbf{u}_{\circ t},\text{ for }t=1,2,...,T, \label{mfT}%
\end{equation}
where $\boldsymbol{\alpha}_{n}=(\alpha_{1},\alpha_{2},...,\alpha_{n})^{\prime
}=c\mathbf{\tau}_{n}+\mathbf{B}_{n}\boldsymbol{\phi}+\mathbf{\eta}_{n}$.

The risk premia are usually estimated using a two-pass procedure suggested by
\citet{fama1973risk}%
. The first-pass runs time series regressions of excess returns, $r_{it}$, on
the $K$ observed factors to give estimates of the factor loadings,
$\boldsymbol{\beta}_{i}:$
\begin{equation}
\text{ }\boldsymbol{\hat{\beta}}_{iT}=\left(  \mathbf{F}^{\prime}%
\mathbf{M}_{T}\mathbf{F}\right)  ^{-1}\mathbf{F}^{\prime}\mathbf{M}%
_{T}\mathbf{r}_{i\circ}. \label{betaihat}%
\end{equation}
The second-pass runs a cross section regression of\ average returns, $\bar
{r}_{i\circ}=T^{-1}\sum_{t=1}^{T}r_{it}$ on the estimated factor loadings, to
obtain the FM estimator of $\boldsymbol{\lambda}$, namely
\begin{equation}
\boldsymbol{\hat{\lambda}}_{nT}=\left(  \mathbf{\hat{B}}_{nT}^{\prime
}\mathbf{M}_{n}\mathbf{\hat{B}}_{nT}\right)  ^{-1}\mathbf{\hat{B}}%
_{nT}^{\prime}\mathbf{M}_{n}\mathbf{\bar{r}}_{n\circ}, \label{lambdahat}%
\end{equation}
where $\mathbf{\hat{B}}_{nT}=(\boldsymbol{\hat{\beta}}_{1T},\boldsymbol{\hat
{\beta}}_{2T},...,\boldsymbol{\hat{\beta}}_{nT})^{\prime},$ $\mathbf{\bar{r}%
}_{n\circ}=(\bar{r}_{1T},\bar{r}_{2T},...,\bar{r}_{nT})^{\prime},$
$\mathbf{M}_{T}=\mathbf{I}_{T}-T^{-1}\boldsymbol{\tau}_{T}\boldsymbol{\tau
}_{T}^{\prime}$, $\boldsymbol{\tau}_{T}$ is a $T$-dimensional vector of ones,
$\mathbf{M}_{n}=\mathbf{I}_{n}-n^{-1}\boldsymbol{\tau}_{n}\boldsymbol{\tau
}_{n}^{\prime}$, and $\boldsymbol{\tau}_{n}$ is an n-dimensional vector of ones.

As is well known, when $T$ is finite FM's two-pass estimator is biased due the
errors in estimation of factor loadings that do not vanish. The small $T$ bias
of the two-pass estimator of $\boldsymbol{\lambda}$ has been a source of
concern in the empirical literature. Under standard regularity conditions and
as $n\rightarrow\infty$, we have%
\begin{equation}
\boldsymbol{\hat{\lambda}}_{nT}-\boldsymbol{\lambda}_{0}\rightarrow_{p}\left[
\mathbf{\Sigma}_{\beta\beta}+\frac{\overline{\sigma}^{2}}{T}\left(
\frac{\mathbf{F}^{\prime}\mathbf{M}_{T}\mathbf{F}}{T}\right)  ^{-1}\right]
^{-1}\left(  \mathbf{\Sigma}_{\beta\beta}\mathbf{d}_{fT}-\frac{\overline
{\sigma}^{2}}{T}\left(  \frac{\mathbf{F}^{\prime}\mathbf{M}_{T}\mathbf{F}}%
{T}\right)  ^{-1}\boldsymbol{\lambda}_{0}\right)  , \label{Bias}%
\end{equation}
where $\mathbf{d}_{fT}=\mathbf{\hat{\mu}}_{T}-\mathbf{\mu}_{0},$
$\boldsymbol{\lambda}_{0}$ and $\mathbf{\mu}_{0}$ are the true values of
$\boldsymbol{\lambda}$ and $\mathbf{\mu}$, respectively, $\mathbf{\Sigma
}_{\beta\beta}=\lim_{n\rightarrow\infty}\left(  n^{-1}\mathbf{B}_{n}^{\prime
}\mathbf{M}_{n}\mathbf{B}_{n}\right)  $, and $\overline{\sigma}^{2}%
=\lim_{n\rightarrow\infty}n^{-1}\sum_{i=1}^{n}\sigma_{i}^{2}>0$. Following
\citet{shanken1992estimation}%
, $\overline{\sigma}_{n}^{2}$ can be consistently estimated (for a fixed $T$)
by
\begin{equation}
\widehat{\bar{\sigma}}_{nT}^{2}=\frac{\sum_{t=1}^{T}\sum_{i=1}^{n}\hat{u}%
_{it}^{2}}{n(T-K-1)}, \label{zigbar}%
\end{equation}
where
\begin{equation}
\hat{u}_{it}=r_{it}-\hat{\alpha}_{iT}-\boldsymbol{\hat{\beta}}_{iT}^{\prime
}\mathbf{f}_{t}, \label{res}%
\end{equation}
and as before $\hat{\alpha}_{iT}$ and $\boldsymbol{\hat{\beta}}_{iT}$ are the
OLS estimators of $\mathit{\alpha}_{i}$ and $\boldsymbol{\beta}_{i}$. Using
these results the bias-corrected version of the two-pass estimator is given
by\footnote{See also
\citet{shanken2007estimating}%
,
\citet{kan2013pricing}%
, and
\citet{BAI201531}%
, and the survey paper by
\citet{jagannathan2010analysis}
for further references.}%
\begin{equation}
\boldsymbol{\tilde{\lambda}}_{nT}=\mathbf{H}_{nT\ }^{-1}\left(  \frac
{\mathbf{\hat{B}}_{nT}^{\prime}\mathbf{M}_{n}\mathbf{\bar{r}}_{n\circ}}%
{n}\right)  , \label{lambda_BC}%
\end{equation}
where%
\begin{equation}
\mathbf{H}_{nT\ }=\frac{\mathbf{\hat{B}}_{nT}^{\prime}\mathbf{M}%
_{n}\mathbf{\hat{B}}_{nT}}{n}-T^{-1}\widehat{\bar{\sigma}}_{nT}^{2}\left(
\frac{\mathbf{F}^{\prime}\mathbf{M}_{T}\mathbf{F}}{T}\right)  ^{-1}.
\label{HnT1}%
\end{equation}
When all the risk factors are strong, under certain regularity conditions,
there exists a fixed $T_{0}$ such that for all $T>T_{0}$, then
\begin{equation}
p\lim_{n\rightarrow\infty}\left(  \boldsymbol{\tilde{\lambda}}_{nT}\right)
=\boldsymbol{\lambda}_{T}^{\ast}=\boldsymbol{\lambda}_{0}+\left(
\boldsymbol{\hat{\mu}}_{T}-\boldsymbol{\mu}_{0}\right)  , \label{lambda*}%
\end{equation}
where $\boldsymbol{\mu}_{0}$ indicates the true value of the factor mean.
Shanken refers to $\boldsymbol{\lambda}_{T}^{\ast}$ as "ex-post" risk premia
to be distinguished from $\boldsymbol{\lambda}_{0}$, referred to as "ex ante"
risk premia. See also Section 3.7 of
\citet{jagannathan2010analysis}%
.

In this paper we exploit Shanken's bias correction procedure by applying it to
$\boldsymbol{\phi}=$ $\boldsymbol{\lambda}-\boldsymbol{\mu}$ which we identify
directly using (\ref{ai}) from the regression of $\mathit{\alpha}_{i}$ on
$\boldsymbol{\beta}_{i}$ for $i=1,2,...,n$, assuming the idiosyncratic pricing
errors, $\eta_{i}$, are sufficiently weak relative to the strengths of the
risk factors in a sense which will be made precise below.

\subsection{Estimation of $\boldsymbol{\phi}$\label{ESTphi}}

In view of (\ref{ai}), the estimation of $\boldsymbol{\phi}$ can be carried
out following a two-step procedure whereby in the first step $\mathit{\alpha
}_{i}$ and $\boldsymbol{\beta}_{i}$ are estimated from the least squares
regressions of $r_{it}$ on an intercept and $\mathbf{f}_{t}$, and these are
then used in a second step regression to estimate $\boldsymbol{\phi}$,
namely,
\begin{equation}
\boldsymbol{\hat{\phi}}_{nT}\text{ }=\left(  \mathbf{\hat{B}}_{nT}^{\prime
}\mathbf{M}_{n}\mathbf{\hat{B}}_{nT}\right)  ^{-1}\mathbf{\hat{B}}%
_{nT}^{\prime}\mathbf{M}_{n}\boldsymbol{\hat{\alpha}}_{nT}, \label{phihat1}%
\end{equation}
where $\boldsymbol{\hat{\alpha}}_{nT}=(\hat{\alpha}_{1T},\hat{\alpha}%
_{1T},...,\hat{\alpha}_{nT})^{\prime}=\mathbf{\bar{r}}_{nT}-\mathbf{\hat{B}%
}_{nT}\boldsymbol{\hat{\mu}}_{T},$ and as before $\mathbf{\hat{B}}%
_{nT}=(\boldsymbol{\hat{\beta}}_{1T},\boldsymbol{\hat{\beta}}_{2T}%
,...,\boldsymbol{\hat{\beta}}_{nT})^{\prime}$. This estimator is consistent
for $\boldsymbol{\phi}_{0}$ so long as $n$ and $T\rightarrow\infty$, and
bias-corrections are necessary to ensure the large $n$ consistency of the
estimator when $T$ is fixed. A Shanken type bias-corrected estimator of
$\boldsymbol{\phi}_{0}$ is given by
\begin{equation}
\boldsymbol{\tilde{\phi}}_{nT}=\mathbf{H}_{nT\ }^{-1}\left[  \frac
{\mathbf{\hat{B}}_{nT}^{\prime}\mathbf{M}_{n}\hat{\alpha}_{nT}}{n}%
+T^{-1}\widehat{\bar{\sigma}}_{nT}^{2}\left(  \frac{\mathbf{F}^{\prime
}\mathbf{M}_{T}\mathbf{F}}{T}\right)  ^{-1}\boldsymbol{\hat{\mu}}_{T}\right]
, \label{phitilda}%
\end{equation}
where $\mathbf{H}_{nT\ }$ and $\widehat{\bar{\sigma}}_{nT}^{2}$ are given by
(\ref{HnT1}) and (\ref{zigbar}), respectively. It is also easily established
that
\begin{equation}
\boldsymbol{\tilde{\phi}}_{nT}=\boldsymbol{\tilde{\lambda}}_{nT}%
-\boldsymbol{\hat{\mu}}_{T}, \label{phitilda2}%
\end{equation}
and for a fixed $T$ and as $n\rightarrow\infty$, we have
\[
p\lim_{n\rightarrow\infty}\boldsymbol{\tilde{\phi}}_{nT}=p\lim_{n\rightarrow
\infty}\boldsymbol{\tilde{\lambda}}_{nT}-\boldsymbol{\hat{\mu}}_{T}.
\]
Hence, upon using (\ref{lambda*})%
\begin{equation}
p\lim_{n\rightarrow\infty}\boldsymbol{\tilde{\phi}}_{nT}=\boldsymbol{\lambda
}_{0}+\left(  \boldsymbol{\hat{\mu}}_{T}-\boldsymbol{\mu}_{0}\right)
-\boldsymbol{\hat{\mu}}_{T}=\boldsymbol{\lambda}_{0}-\boldsymbol{\mu}%
_{0}=\boldsymbol{\phi}_{0}\text{, } \label{Conphitilt}%
\end{equation}
and there exists a fixed $T_{0}$ such that for all $T>T_{0},$
$\boldsymbol{\tilde{\phi}}_{nT}$ converges to $\boldsymbol{\phi}_{0}$ as
$n\rightarrow\infty$. \ Also using (\ref{lambda*}) and (\ref{phitilda2}), and
noting that $\boldsymbol{\lambda}_{0}-\boldsymbol{\mu}_{0}=\boldsymbol{\phi
}_{0}$, interestingly we have%
\[
\boldsymbol{\tilde{\lambda}}_{nT}-\boldsymbol{\lambda}_{T}^{\ast
}=\boldsymbol{\tilde{\phi}}_{nT}+\boldsymbol{\hat{\mu}}_{T}%
-\boldsymbol{\lambda}_{T}^{\ast}=\boldsymbol{\tilde{\phi}}_{nT}%
-\boldsymbol{\phi}_{0}.
\]
So inference using the Shanken bias-corrected estimator of
$\boldsymbol{\lambda}$ around $\boldsymbol{\lambda}_{T}^{\ast}$, is the same
as making inference using $\boldsymbol{\tilde{\phi}}_{nT}$ around
$\boldsymbol{\phi}_{0}$.

The asymptotic distribution of $\boldsymbol{\tilde{\phi}}_{nT}$ depends on
both $n$ and $T$. Assuming the observed factors are strong and under certain
regularity conditions, to be introduced below, we have
\begin{equation}
\sqrt{nT}\left(  \boldsymbol{\tilde{\phi}}_{nT}-\boldsymbol{\phi}_{0}\right)
\rightarrow_{d}N\left(  \mathbf{0,\Sigma}_{\beta\beta}^{-1}\mathbf{V}_{\xi
}\mathbf{\Sigma}_{\beta\beta}^{-1}\right)  , \label{Dphitilt}%
\end{equation}
where
\[
\mathbf{V}_{\xi}=\left(  1+\boldsymbol{\lambda}_{0}^{\prime}\mathbf{\Sigma
}_{f}^{-1}\boldsymbol{\lambda}_{0}\right)  p\lim_{n\rightarrow\infty}\left[
n^{-1}\mathbf{B}_{n}^{\prime}\mathbf{M}_{n}\mathbf{V}_{u}\mathbf{M}%
_{n}\mathbf{B}_{n\ }\right]  \emph{.}%
\]
The variance of $\boldsymbol{\tilde{\phi}}_{nT}$ is consistently estimated by%
\begin{equation}
\widehat{Var\left(  \boldsymbol{\tilde{\phi}}_{nT}\right)  }=T^{-1}%
n^{-1}\mathbf{H}_{nT\ }^{-1}\mathbf{\hat{V}}_{\xi,nT}\mathbf{H}_{nT\ }^{-1},
\label{Varphitilt}%
\end{equation}
where $\mathbf{H}_{nT\ }$ is given by (\ref{HnT1}),%
\begin{equation}
\mathbf{\hat{V}}_{\xi,nT}=\left(  1+\hat{s}_{nT}\right)  \left(
n^{-1}\mathbf{\hat{B}}_{n}^{\prime}\mathbf{M}_{n}\mathbf{\tilde{V}}%
_{u}\mathbf{M}_{n}\mathbf{\hat{B}}_{n}\right)  , \label{GnT}%
\end{equation}
and%
\begin{equation}
\hat{s}_{nT}=\boldsymbol{\tilde{\lambda}}_{nT}^{^{\prime}}\left(
\frac{\mathbf{F}^{\prime}\mathbf{M}_{T}\mathbf{F}}{T}\right)  ^{-1}%
\boldsymbol{\tilde{\lambda}}_{nT}, \label{shatnT}%
\end{equation}
$\boldsymbol{\tilde{\lambda}}_{nT}$ is defined by (\ref{lambda_BC}), and
$\mathbf{\tilde{V}}_{u}\mathbf{\ }$is a suitable estimator of $\mathbf{V}%
_{u}=E\left(  \mathbf{u}_{\circ t}\mathbf{u}_{\circ t}^{\prime}\right)  $. How
to estimate $\mathbf{V}_{u}$ and the conditions under which $Var\left(
\boldsymbol{\tilde{\phi}}_{nT}\right)  $ is consistently estimated by
$\widehat{Var\left(  \boldsymbol{\tilde{\phi}}_{nT}\right)  }$ is discussed in
sub-section \ref{Vphihat}.

Note that even when $\mathbf{V}_{u}=\sigma^{2}\mathbf{I}_{T}$ the variance of
$\boldsymbol{\tilde{\phi}}_{nT}$ does not reduce to $\sigma^{2}%
\boldsymbol{\Sigma}_{\beta\beta}^{-1}$, the standard least squares formula
used for the case of known factor loadings. When the loadings are estimated
the scaling term $\left(  1+\boldsymbol{\lambda}_{0}^{\prime}%
\boldsymbol{\Sigma}_{f}^{-1}\boldsymbol{\lambda}_{0}\right)  $ is required and
its neglect can lead to serious over-rejection even if $n/T\rightarrow0$ as
$n$ and $T\rightarrow\infty$.

\subsection{Factor strength \label{facstrength}}

In this paper we deviate from the standard literature and allow the observed
and latent factors to have different degrees of strength, depending on how
pervasively they impact the security returns.
\citet{bailey2021measurement}
define the strength of factor, $f_{kt}$, in terms of the number of its
non-zero factor loadings. For a factor to be strong almost all of its $n$
loadings must differ from zero. Given our focus on estimation of risk premia,
we adopt the following definition which directly relates to the covariance of
$\boldsymbol{\beta}_{i}.$ See also
\citet{chudik2011weak}%
.

\begin{definition}
\label{strengths}(Factor strengths) The strength of factor $f_{kt}$ is
measured by its degree of pervasiveness as defined by the exponent
$\alpha_{_{k}}$ in%
\begin{equation}
\sum_{i=1}^{n}\left(  \beta_{ik}-\bar{\beta}_{k}\right)  ^{2}=\ominus
(n^{\alpha_{k}}), \label{factor strength}%
\end{equation}
and $0<\alpha_{k}\leq1$. We refer to $\left\{  \alpha_{k},\text{
}k=1,2,...,K\right\}  $ as factor strengths. Factor $f_{kt}$ is said to be
strong if $\alpha_{k}=1$, semi-strong if $1>\alpha_{k}>1/2$, and weak if
$0\leq\alpha_{k}\leq1/2$. Condition (\ref{factor strength}) applies
irrespective of whether the loadings, $\beta_{ik}$, are viewed as
deterministic or stochastic.\ 
\end{definition}

In the above definition $\ominus_{p}\left(  n^{\alpha_{k}}\right)  $ denotes
the rate at which additional securities add to the factor's strength and
$\alpha_{k}$ can be viewed as a logarithmic expansion rate in terms of $n$ and
relates to the proportion of non-zero factor loadings. In the literature it is
commonly assumed that the covariance matrix of factor loadings defined by%
\begin{equation}
\boldsymbol{\Sigma}_{\beta\beta}=p\lim_{n\rightarrow\infty}\left[  n^{-1}%
\sum_{i=1}^{n}\left(  \boldsymbol{\beta}_{i}-\boldsymbol{\bar{\beta}}%
_{n}\right)  \left(  \boldsymbol{\beta}_{i}-\boldsymbol{\bar{\beta}}%
_{n}\right)  ^{\prime}\right]  , \label{Zigbeta}%
\end{equation}
is positive definite, where $\boldsymbol{\bar{\beta}}_{n}=n^{-1}\sum_{i=1}%
^{n}\boldsymbol{\beta}_{i}=(\bar{\beta}_{1},\bar{\beta}_{2},...,\bar{\beta
}_{k})^{\prime}$. For $\mathbf{\Sigma}_{\beta\beta}$ to be positive definite
matrix it is \textit{necessary} that all the $K$ risk factors under
consideration are strong in the sense that
\begin{equation}
p\lim_{n\rightarrow\infty}\left[  n^{-1}\sum_{i=1}^{n}\left(  \beta_{ik}%
-\bar{\beta}_{k}\right)  ^{2}\right]  >0\text{, for }k=1,2,...,K.
\label{Sfact}%
\end{equation}
In terms of our definition of factor strength, $\mathbf{\Sigma}_{\beta\beta}%
$\ will be positive definite if all the observed factors are strong, namely if
$\alpha_{k}=1$ for $k=1,2,...,K$. However, such an assumption is quite
restrictive and is unlikely to be satisfied for many risk factors being
considered in the literature.
\citet{bailey2021measurement}
show that, apart from the market factor, only a handful of 144 factors in the
literature considered by
\citet{feng2020taming}
come close to being strong. \cite{giglio2023test} consider the estimation of
PCA-based risk premia in presence of weak factors. However, their definition
of factor strength involves both $n~$and $T$, and is best viewed as a
consistency condition rather than factor strength as such. See the discussion
following Theorem \ref{Tsemi}. Our notion of factor strength, $\alpha_{k}$, is
in line with the recent literature. See, for example, \cite{BaiNg2023Weak} and
\cite{uematsu2023inference}.

\subsection{Missing factor}

We now turn to the structure of the errors, $u_{it}$, in the returns
equations, and consider two possible sources of error cross-sectional
dependence: a missing or latent factor and production networks. The issue of
missing factors has been investigated in the recent literature by
\citet{giglio2021asset}
and
\citet{ANATOLYEV2022103}%
. The issue of production networks has been investigated in the recent
literature by
\citet{herskovic2018networks}%
, who derives two risk factors based on the changes in network concentration
and network sparsity, and
\citet{gofman2020production}%
, who focus on the vertical dimension of production by modelling a supply
chain, in terms of supplier-customer links. They find that the further away a
firm is from final consumers the higher its return. They use this to create a
factor TMB (top minus bottom). Both sources of cross-sectional error
dependence could be important, since network dependence cannot be represented
using latent factor models. See Section 3 of \cite{chudik2011weak}.

To allow for both forms of error cross-sectional dependence we consider the
following decomposition of $u_{it}$
\begin{equation}
u_{it}=\gamma_{i}g_{t}+v_{it}, \label{uit}%
\end{equation}
where $g_{t}$ is the missing (latent) factor and $v_{it}$ is weakly
cross-correlated in the sense of approximate factor models due to
\cite{Chamberlain1983} and \cite{cham1983arbi}. Here we allow for a single
missing factor to simplify the exposition, but note that increasing the number
of missing factors has little impact on our analysis, so long as the number of
missing factors is fixed. Using the normalization $E(g_{t}^{2})=1$, and
assuming that $\gamma_{i}g_{t}$ and $v_{it}$ are independently distributed
then $E(u_{it}u_{jt})=\sigma_{ij}=\gamma_{i}\gamma_{j}+\sigma_{v,ij}$, and as
shown in Lemma \ref{L1}, $n^{-1}\sum_{i=1}^{n}\sum_{j=1}^{n}\left\vert
\sigma_{ij}\right\vert =O(1)$ so long as the strength of $g_{t},$
$\alpha_{\gamma}<1/2$ and $\lambda_{\max}(\mathbf{V}_{v})<\infty.$ This is
despite the fact that $\lambda_{\max}(\mathbf{V}_{u})=O(n^{\alpha_{\gamma}})$,
where $\mathbf{V}_{v}=E\left(  \mathbf{v}_{i}\mathbf{v}_{i}^{\prime}\right)  $
and $\mathbf{V}_{u}=E\left(  \mathbf{u}_{i}\mathbf{u}_{i}^{\prime}\right)
$.\footnote{Note that Chamberlain's approximate factor model specification
requires $\lambda_{\max}(\mathbf{V}_{u})=O(1)$ and is violated if
$\alpha_{\gamma}>0$.} In the Monte Carlo experiments, we consider the
possibility of missing factors, as well as weak spatial and network
cross-dependence that satisfy conditions of approximate factor models.

\subsection{Pricing errors and market efficiency\label{pricing}}

The APT condition (\ref{APTRoss}), given by (18)\ in Theorem II of
\citet{ROSS1976341}%
, ensures that under APT the idiosyncratic pricing errors are sparse. In this
paper we relax the\ Ross's condition to
\begin{equation}
\sum_{i=1}^{n}\eta_{i}^{2}=O(n^{\alpha_{\eta}}), \label{APTg}%
\end{equation}
where the exponent $\alpha_{\eta}$ measures the degrees of pervasiveness of
pricing errors. Deviations from APT are measured in terms of $\alpha_{\eta}$
($0\leq\alpha_{\eta}<1$). We investigate the robustness of our proposed
estimator of $\boldsymbol{\phi}$ to $\alpha_{\eta}$. This extension is
important for tests of market efficiency where the null of interest is
$H_{0}:$ $\mathit{\alpha}_{i}=c$ for all $i$ in (\ref{ai}). We note that under
the alternative hypothesis $H_{1}:$ $\mathit{\alpha}_{i}=c+\boldsymbol{\beta
}_{i}^{\prime}\boldsymbol{\phi}+\eta_{i}$, therefore it is desirable to
develop a test of $\boldsymbol{\phi}=\boldsymbol{0}$ which is robust to a
wider class of pricing errors than those entertained originally by Ross, where
$\alpha_{\eta}=0$.

Under the alternative hypothesis the power of testing $H_{0}$ will depend on
the rate at which $\sum_{i=1}^{n}\left(  \mathit{\alpha}_{i}-\mathit{\bar
{\alpha}}\right)  ^{2}$ rises with $n$. This in turn depends on the degree of
pervasiveness of the idiosyncratic pricing errors, $\alpha_{\eta}$, the
strength of the factors, $\alpha_{j}$, for $j=1,2,...,K$ and the magnitude of
$\boldsymbol{\phi}^{\prime}\boldsymbol{\phi}\,$. Using $\mathit{\alpha}%
_{i}=c+\boldsymbol{\beta}_{i}^{\prime}\boldsymbol{\phi}+\eta_{i}$%
\[
\sum_{i=1}^{n}\left(  \mathit{\alpha}_{i}-\mathit{\bar{\alpha}}\right)
^{2}=\boldsymbol{\phi}^{\prime}\left(  \sum_{i=1}^{n}\left(  \boldsymbol{\beta
}_{i}-\boldsymbol{\bar{\beta}}_{n}\right)  \left(  \boldsymbol{\beta}%
_{i}-\boldsymbol{\bar{\beta}}_{n}\right)  ^{\prime}\right)  \boldsymbol{\phi
}+\sum_{i=1}^{n}\left(  \eta_{i}-\bar{\eta}\right)  ^{2}+2\boldsymbol{\phi
}^{\prime}\sum_{i=1}^{n}\left(  \boldsymbol{\beta}_{i}-\boldsymbol{\bar{\beta
}}_{n}\right)  \eta_{i},
\]
and in the absence of idiosyncratic pricing errors
\begin{equation}
\sum_{i=1}^{n}\left(  \mathit{\alpha}_{i}-\mathit{\bar{\alpha}}\right)
^{2}\geq n\left(  \boldsymbol{\phi}^{\prime}\boldsymbol{\phi}\right)
\lambda_{\min}\left(  n^{-1}\sum_{i=1}^{n}\left(  \boldsymbol{\beta}%
_{i}-\boldsymbol{\bar{\beta}}_{n}\right)  \left(  \boldsymbol{\beta}%
_{i}-\boldsymbol{\bar{\beta}}_{n}\right)  ^{\prime}\right)  . \label{alphaE}%
\end{equation}
When the risk factors are all strong $\lambda_{\min}\left(  n^{-1}\sum
_{i=1}^{n}\left(  \boldsymbol{\beta}_{i}-\boldsymbol{\bar{\beta}}_{n}\right)
\left(  \boldsymbol{\beta}_{i}-\boldsymbol{\bar{\beta}}_{n}\right)  ^{\prime
}\right)  >0$, then $\sum_{i=1}^{n}\left(  \mathit{\alpha}_{i}-\mathit{\bar
{\alpha}}\right)  ^{2}=\ominus(n)$ if and only if $\boldsymbol{\phi}%
\neq\boldsymbol{0}$. Namely, the extent to which alpha can be exploited will
depend on the magnitudes of $\phi_{j}$ and the strength of the risk factors.

\begin{remark}
\label{Remfulldiv}Having formalized the concepts of factor strength, missing
factors, and the less restrictive APT condition given by (\ref{APTg}), it is
now of interest to revisit the conditions under which the phi-portfolio fully
diversifies. Consider (\ref{fulldiv}) and note that
\[
\frac{\left\Vert \boldsymbol{\eta}_{n}\right\Vert ^{2}}{\lambda_{min}\left(
\mathbf{B}_{n}^{\prime}\mathbf{M}_{n}\mathbf{B}_{n}\right)  }=O\left[
n^{-\left(  \alpha_{min}-\alpha_{\eta}\right)  }\right]  \text{, and }%
\frac{\lambda_{max}\left(  \mathbf{V}_{u}\right)  }{\lambda_{min}\left(
\mathbf{B}_{n}^{\prime}\mathbf{M}_{n}\mathbf{B}_{n}\right)  }=O\left[
n^{-\left(  \alpha_{min}-\alpha_{\gamma}\right)  }\right]  .
\]
Using these results in (\ref{phiport1})\ and (\ref{varphi}) and we have
\[
E\left(  \rho_{t,\phi}\right)  =\boldsymbol{\phi}^{\prime}\boldsymbol{\phi
}+O\left[  n^{-\left(  \alpha_{min}-\alpha_{\eta}\right)  }\right]  \text{,
and }Var\left(  \rho_{t,\phi}\right)  =O\left[  n^{-\left(  \alpha
_{min}-\alpha_{\gamma}\right)  }\right]
\]
Therefore, for phi-portfolio to dominate the MV portfolio in addition to
$\boldsymbol{\phi}^{\prime}\boldsymbol{\phi}>0$, it is also required that the
strengths of the traded factors $\alpha_{k}$, $k=1,2,...,K$ are strictly
larger than the strength of the missing factor, $\alpha_{\gamma},$ as well as
the strength of the idiosyncratic pricing errors, $\alpha_{\eta}$.
\end{remark}

\begin{remark}
Similarly, If we allow for idiosyncratic pricing errors, missing factors and
non-strong risk factors the limiting expression for $\sum_{i=1}^{n}\left(
\mathit{\alpha}_{i}-\mathit{\bar{\alpha}}\right)  ^{2}$ becomes%
\begin{equation}
\sum_{i=1}^{n}\left(  \mathit{\alpha}_{i}-\mathit{\bar{\alpha}}\right)
^{2}=\ominus(\sum_{j=1}^{K}n^{\alpha_{j}}\phi_{j}^{2})+O\left(  n^{\alpha
_{\eta}}\right)  +O\left(  n^{\alpha_{\gamma}}\right)  . \label{alphaG}%
\end{equation}
In this more general case for alpha to be exploitable it is necessary that
$\phi_{j}$ associated with the strongest factor is non-zero, and $\alpha
_{\max}=\max_{j}(\alpha_{j})$ is larger than $\alpha_{\eta}$ and
$\alpha_{\gamma}$.
\end{remark}

\section{Assumptions and theorems\label{Assumptions}}

We make the following standard assumptions about $\mathbf{f}_{t},g_{t}$,
$v_{it}$, $\boldsymbol{\beta}_{i}$, $\eta_{i}$, and $\gamma_{i}$ (the drivers
of asset returns):

\begin{assumption}
\label{factors} (Observed common factors) (a) The $K\times1$ vector of
observed risk factors, $\mathbf{f}_{t}$, follows the general linear process
\begin{equation}
\mathbf{f}_{t}=\boldsymbol{\mu}+\sum_{\ell=0}^{\infty}\mathbf{\Psi}_{\ell
}\mathbf{\zeta}_{t-\ell}, \label{ft}%
\end{equation}
where $\left\Vert \boldsymbol{\mu}\right\Vert <C$, $\mathbf{\zeta}%
_{t}\thicksim IID(\mathbf{0},\mathbf{I}_{K})$, and $\mathbf{\Psi}_{\ell}$ are
$K\times K$ exponentially decaying matrices such that $\left\Vert
\mathbf{\Psi}_{\ell}\right\Vert <C\rho^{\ell}$ for some $C>0$ and $0<\rho<1$.
(b) The $T\times K$ data matrix $\mathbf{F=(f}_{1},\mathbf{f}_{2}%
,...,\mathbf{f}_{T})^{\prime}$ is full column rank and there exists $T_{0}$
such that for all $T>T_{0}$, $\boldsymbol{\hat{\Sigma}}_{f}=T^{-1}%
\mathbf{F}^{\prime}\mathbf{M}_{T}\mathbf{F}$ is a positive definite matrix,
$\lambda_{\max}\left[  (T^{-1}\mathbf{F}^{\prime}\mathbf{M}_{T}\mathbf{F)}%
^{-1}\right]  <C$, $\boldsymbol{\hat{\Sigma}}_{f}\rightarrow_{p}%
\boldsymbol{\Sigma}_{f}=E\left(  \mathbf{f}_{t}-\boldsymbol{\mu}_{0}\right)
\left(  \mathbf{f}_{t}-\boldsymbol{\mu}_{0}\right)  ^{\prime}>\mathbf{0},$
where $\boldsymbol{\mu}_{0}$ is the true value of $\boldsymbol{\mu}$.
\end{assumption}

\begin{assumption}
\label{loadings}(Observed factor loadings) (a) The factor loadings $\beta
_{ik}$ for $i=1,2,...,n$ and $k=1,2,...,K$ are stochastically bounded such
that $sup_{ik}E\left(  \beta_{ik}^{2}\right)  <C$,
\begin{equation}
\sum_{i=1}^{n}\left(  \beta_{ik}-\bar{\beta}_{k}\right)  ^{2}=\ominus
_{p}(n^{\alpha_{k}})\text{, for }k=1,2,...,K. \label{As}%
\end{equation}
(b)The $n\times K$ matrix of factor loadings, $\mathbf{B}_{n}%
=(\boldsymbol{\beta}_{\circ1},\boldsymbol{\beta}_{\circ2}%
,...,\boldsymbol{\beta}_{\circ K}),$ where $\boldsymbol{\beta}_{\circ
k}=(\beta_{1k},\beta_{2k},...,\beta_{nk})^{\prime}$ satisfy
\begin{equation}
0<c<\lambda_{min}\left(  \mathbf{D}_{\alpha}^{-1}\mathbf{B}_{n}^{\prime
}\mathbf{M}_{n}\mathbf{B}_{n}\mathbf{D}_{\alpha}^{-1}\right)  <\lambda
_{max}\left(  \mathbf{D}_{\alpha}^{-1}\mathbf{B}_{n}^{\prime}\mathbf{M}%
_{n}\mathbf{B}_{n}\mathbf{D}_{\alpha}^{-1}\right)  <C<\infty, \label{lmax}%
\end{equation}
for some small and large positive constants, $c$ and $C$, where $\mathbf{M}%
_{n}=\mathbf{I}_{n}-n^{-1}\boldsymbol{\tau}_{n}\boldsymbol{\tau}_{n}^{\prime}%
$, $\boldsymbol{\tau}_{n}=(1,1,...,1)^{\prime},$ and $\mathbf{D}_{\alpha}%
\,$\ is the $n\times n$ diagonal matrix%
\begin{equation}
\mathbf{D}_{\alpha}=Diag(n^{\alpha_{1}/2},n^{\alpha_{2}/2},....,n^{\alpha
_{K}/2}). \label{Dn}%
\end{equation}

\end{assumption}

\begin{assumption}
\label{Latent factor} (latent factor) (a) The latent factor, $g_{t}$, in
(\ref{uit}) is distributed independently of $\mathbf{f}_{t^{\prime}},$ for all
$t$ and $t^{\prime}$, $g_{t}$ is serially independent with mean zero,
$E(g_{t})=0,$ $E(g_{t}^{2})=1$, and a finite fourth order moment,
$sup_{t}E(g_{t}^{4})<C$. (b) The loadings $\gamma_{i}$ are such that
$sup_{i}\left\vert \gamma_{i}\right\vert <C$ and
\begin{equation}
\sum_{i=1}^{n}\left\vert \gamma_{i}\right\vert =O(n^{\alpha_{\gamma}}).
\label{normg}%
\end{equation}

\end{assumption}

\begin{assumption}
\label{Errors} (idiosyncratic errors) (a) The errors $\left\{  v_{it}\text{,
}i=1,2,...,n;\text{ }t=1,2,...,T\right\}  $ are distributed independently of
the factors $f_{k,t^{\prime}}$, and $g_{t}$, for all $i,t,t^{\prime}$ and
$k=1,2,...,K$, and their associated loadings $\beta_{ik},$ and $\gamma_{i}$.
They are serially independent with $E(v_{it})=0$ and finite fourth order
moments $E(v_{it}^{4})<\infty$, and covariances $E(v_{it}v_{jt})=\sigma
_{v,ij}$, such that
\begin{equation}
\sup_{i}\sum_{j=1}^{n}\left\vert \sigma_{v,ij}\right\vert <\infty,\text{ and
}\sup_{i}\sum_{j=1}^{n}Cov(v_{it}^{2},v_{jt}^{2})<\infty,\text{ } \label{Covv}%
\end{equation}
with $\lambda_{\min}\left(  \mathbf{V}_{v}\right)  >0$, where $\mathbf{V}%
_{v}=(\sigma_{v,ij})$. (b) The degree of cross-sectional dependence of
$v_{it}$ is sufficiently weak so that
\begin{equation}
T^{-1/2}n^{-1/2}\sum_{t=1}^{T}\sum_{i=1}^{n}(\beta_{ik}-\bar{\beta}_{k}%
)v_{it}\rightarrow_{d}N(0,\omega_{k}^{2})\text{, for }k=1,2,...,K,
\label{Disv}%
\end{equation}
where
\begin{equation}
\omega_{k}^{2}=p\lim_{n\rightarrow\infty}n^{-\alpha_{k}}\sum_{i=1}^{n}%
\sum_{j=1}^{n}(\beta_{ik}-\bar{\beta}_{k})(\beta_{jk}-\bar{\beta}_{k}%
)\sigma_{v,ij}. \label{Varv}%
\end{equation}

\end{assumption}

\begin{assumption}
\label{PriceError} (Pricing errors) The pricing errors, $\eta_{i}$, for
$i=1,2,...,n$ are individually bounded, $sup_{j}\left\vert \eta_{j}\right\vert
<C$\thinspace, and are distributed independently of the factor loadings,
$\beta_{jk}$, and $\gamma_{j}$ for all $i,j$ and $k=1,2,...,K$, as well as
satisfying the condition
\begin{equation}
\sum_{i=1}^{n}\left\vert \eta_{i}\right\vert =O\left(  n^{\alpha_{\eta}%
}\right)  , \label{NormetaD}%
\end{equation}
with $\alpha_{\eta}<1/2$.
\end{assumption}

\begin{remark}
Under Assumption \ref{factors} $E\left(  \mathbf{f}_{t}\right)  =\mathbf{\mu
},$ and $Var(\mathbf{f}_{t})=\boldsymbol{\Sigma}_{f}=\sum_{\ell=0}^{\infty
}\mathbf{\Psi}_{\ell}\mathbf{\Psi}_{\ell}^{\prime}$. Also since $\left\Vert
\boldsymbol{\Sigma}_{f}\right\Vert \leq\sum_{\ell=0}^{\infty}\left\Vert
\mathbf{\Psi}_{\ell}\right\Vert ^{2}$ it then follows from part (a) of
Assumption \ref{factors} that $\left\Vert \boldsymbol{\Sigma}_{f}\right\Vert
<C$.
\end{remark}

\begin{remark}
\label{remarkstrength} Under Assumption \ref{loadings}
\begin{equation}
\mathbf{D}_{\alpha}^{-1}\mathbf{B}_{n}^{\prime}\mathbf{M}_{n}\mathbf{B}%
_{n}\mathbf{D}_{\alpha}^{-1}\rightarrow_{p}\mathbf{\Sigma}_{\beta\beta
}(\boldsymbol{\alpha}\mathbf{)}>0, \label{zigbb}%
\end{equation}
where $\mathbf{\Sigma}_{\beta\beta}(\boldsymbol{\alpha}\mathbf{)}$ is a
$k\times k$ symmetric positive definite matrix which is a function of
$\boldsymbol{\alpha}=\mathbf{(}\alpha_{1},\alpha_{2},...,\alpha_{K})^{\prime}%
$. This follows from (\ref{lmax}) since for any non-zero $n\times1$ vector
$\boldsymbol{\kappa}$,
\[
\boldsymbol{\kappa}^{\prime}\mathbf{D}_{\alpha}^{-1}\mathbf{B}_{n}^{\prime
}\mathbf{M}_{n}\mathbf{B}_{n}\mathbf{D}_{\alpha}^{-1}\boldsymbol{\kappa}%
\geq\left(  \mathbf{\kappa}^{\prime}\mathbf{\kappa}\right)  \lambda
_{min}\left(  \mathbf{D}_{\alpha}^{-1}\mathbf{B}_{n}^{\prime}\mathbf{M}%
_{n}\mathbf{B}_{n}\mathbf{D}_{\alpha}^{-1}\right)  >0.
\]
In the standard case where the factors are all strong ($\alpha_{k}=1$ for all
$k$), the above limit reduces to $n^{-1}\mathbf{B}_{n}^{\prime}\mathbf{M}%
_{n}\mathbf{B}_{n}\rightarrow_{p}\mathbf{\Sigma}_{\beta\beta}(\boldsymbol{\tau
}_{K}\mathbf{)=\Sigma}_{\beta\beta}>0$.
\end{remark}

\begin{remark}
The high level condition (\ref{Disv}) in Assumption \ref{Errors} is required
for establishing the asymptotic normality of the estimator of
$\boldsymbol{\phi}_{0}$, and is clearly met when $v_{it}$ and/or $\beta_{ik}$
are independently distributed. It is also possible to establish (\ref{Disv})
under weaker conditions assuming that $v_{it}$ and/or $\beta_{ik}$ satisfy
some time-series type mixing conditions applied to cross section.
\end{remark}

\begin{remark}
The exponent parameter, $\alpha_{\eta}$, of the pricing condition in
(\ref{NormetaD}), can be viewed as the degree to which pricing errors are
pervasive in large economies (as $n\rightarrow\infty$). Letting
$\boldsymbol{\eta}_{n}=(\eta_{1},\eta_{2},...,\eta_{n})^{\prime}\,$\ we have
\begin{equation}
\sum_{i=1}^{n}\eta_{i}^{2}=\left\Vert \boldsymbol{\eta}_{n}\right\Vert
^{2}\leq\left\Vert \boldsymbol{\eta}_{n}\right\Vert _{\infty}\left\Vert
\boldsymbol{\eta}_{n}\right\Vert _{1}=sup_{j}\left\vert \eta_{j}\right\vert
\left(  \sum_{i=1}^{n}\left\vert \eta_{i}\right\vert \right)  , \label{sqeta}%
\end{equation}
and under Assumption (\ref{PriceError}) it also follows that
\begin{equation}
\sum_{i=1}^{n}\eta_{i}^{2}=O\left(  n^{\alpha_{\eta}}\right)  . \label{Snorm}%
\end{equation}
Similarly%
\begin{equation}
\sum_{i=1}^{n}\gamma_{i}^{2}=O(n^{\alpha_{\gamma}}). \label{Snormg}%
\end{equation}

\end{remark}

\begin{remark}
Whilst (\ref{NormetaD}) implies (\ref{Snorm}), the reverse does not follow. By
allowing for $\alpha_{\eta}>0$ we are relaxing the Ross's boundedness
condition that requires setting $\alpha_{\eta}=0$.
\end{remark}

\begin{remark}
The assumption that the observed and missing factors, $\mathbf{f}_{t}$ and
$g_{t^{\prime}},$ are distributed independently is not restrictive and can be
relaxed. For example, suppose that
\[
g_{t}=\mu_{g}+\boldsymbol{\theta}^{\prime}\mathbf{f}_{t}+v_{gt},
\]
where $\boldsymbol{f}_{t}$ and $v_{gt}$ are independently distributed. Then
using (\ref{uit}) we have%
\[
u_{it}=\gamma_{i}\mu_{g}+\gamma_{i}\left(  \boldsymbol{\theta}^{\prime
}\mathbf{f}_{t}\right)  +\gamma_{i}v_{gt}+v_{it},
\]
and the return equation (\ref{StatM}) can be written as
\[
r_{it}=\left(  \mathit{\alpha}_{i}+\gamma_{i}\mu_{g}\right)  +\left(
\boldsymbol{\beta}_{i}+\gamma_{i}\boldsymbol{\theta}\right)  ^{\prime
}\mathbf{f}_{t}+\gamma_{i}v_{gt}+v_{it},
\]
with $v_{gt}$ now acting as the missing common factor, which, by construction,
is distributed independently of $\mathbf{f}_{t}$.
\end{remark}

\begin{remark}
Assumptions \ref{Latent factor} and \ref{Errors} allow $u_{it}$ to be
cross-sectionally weakly correlated, but require the errors to be serially
uncorrelated. This requirement is not strong for asset pricing models, since
realized returns are only mildly serially correlated and most likely such
serial dependence will be captured by the serial correlation in the observed factors.
\end{remark}

As we shall see, to estimate and conduct inference on the risk premia
associated with the observed factors, $f_{kt}$, we require $\alpha_{k}%
>\alpha_{\gamma}<1/2$, where $\alpha_{\gamma}$ denotes the strength of the
latent factor, $g_{t},$ and similarly defined by $\sum_{i=1}^{n}\gamma_{i}%
^{2}=\ominus(n^{\alpha_{\gamma}})$. Namely, the latent factor must be
sufficiently weak so that ignoring it will be inconsequential, and observed
factors sufficiently strong so that they can be distinguished from the weak
latent factor.

The main theoretical results of the paper are set out around five theorems.
Theorem \ref{TFMbias} considers the Fama-MacBeth two-step estimator and
derives its limiting property as $n$ and $T\rightarrow\infty$. To eliminate
the bias of Fama-MacBeth estimator we require $n/T\rightarrow0$, and to
eliminate the effects of pricing errors we need $Tn^{\alpha_{\eta}%
}/n\rightarrow0$, which results in a contradiction. Thus the Fama-MacBeth
estimator is valid only when there are no pricing errors ($\eta_{i}=0$ for all
$i$) and when $n/T\rightarrow0$. Theorem \ref{Thzig} provides a proof that the
estimator of $\bar{\sigma}_{n}^{2}$ (denoted by $\widehat{\bar{\sigma}}%
_{nT}^{2}$) proposed by
\citet{shanken1992estimation}
continues to be unbiased for a fixed $T$ as $n\rightarrow\infty$, even under
the general setting of the current paper that allows for missing factors as
well as pricing errors. Theorem \ref{Thzig} also establishes that
$\widehat{\bar{\sigma}}_{nT}^{2}-\bar{\sigma}_{n}^{2}\rightarrow
O_{p}(n^{-1/2}T^{-1/2})$, which is essential for establishing the results for
the bias-corrected estimator of $\boldsymbol{\phi}_{0}$, namely
$\boldsymbol{\tilde{\phi}}_{nT}$ given by (\ref{phitilda}), summarized in
Theorem \ref{Tfi}. This theorem provides conditions under which
$\boldsymbol{\tilde{\phi}}_{nT}$ is a consistent estimator of
$\boldsymbol{\phi}_{0}$, and derives its asymptotic distribution assuming the
observed factors are strong, again allowing for pricing errors, a missing
factor, and other forms of weak error cross-sectional dependence. Theorem
\ref{Tsemi} extends the results of Theorem \ref{Tfi} to the case where one or
more of the observed risk factors are semi-strong and shows how factor
strength impacts the precision with which the elements of $\boldsymbol{\phi
}_{0}$ are estimated. Finally, Theorem \ref{Var} presents the conditions under
which the asymptotic variance of $\boldsymbol{\tilde{\phi}}_{nT}$ can be
consistently estimated.

\begin{theorem}
\label{TFMbias}(Small $T$ bias of Fama-MacBeth estimator of
$\boldsymbol{\lambda}$) Consider the multi-factor linear return model
(\ref{mfT}) with the missing factor $g_{t}$ in $u_{it}$ as defined by
(\ref{uit}) and the associated risk premia, $\boldsymbol{\lambda}$, defined by
(\ref{APTmean}). Suppose that Assumptions \ref{factors}, \ref{loadings},
\ref{Errors}, \ref{Latent factor} and \ref{PriceError} hold and all observed
factors are strong. Suppose further that the true value of the risk premia,
$\boldsymbol{\lambda}_{0}$, is estimated by Fama-MacBeth two-pass estimator,
$\boldsymbol{\hat{\lambda}}_{nT}$, defined by (\ref{lambdahat}). Then for any
fixed $T>T_{0}$ such that $\lambda_{\min}\left(  T^{-1}\mathbf{F}^{\prime
}\mathbf{M}_{T}\mathbf{F}\right)  >0,$ we have (as $n\rightarrow\infty$)%
\begin{equation}
\boldsymbol{\hat{\lambda}}_{nT}-\boldsymbol{\lambda}_{0}=\left(
\boldsymbol{\hat{\mu}}_{T}-\boldsymbol{\mu}_{0}\right)  -\frac{\bar{\sigma
}^{2}}{T}\left[  \boldsymbol{\Sigma}_{\beta\beta}+\bar{\sigma}^{2}\frac{1}%
{T}\left(  \frac{\mathbf{F}^{\prime}\mathbf{M}_{T}\mathbf{F}}{T}\right)
^{-1}\right]  ^{-1}\left(  \frac{\mathbf{F}^{\prime}\mathbf{M}_{T}\mathbf{F}%
}{T}\right)  ^{-1}\boldsymbol{\lambda}_{T}^{\ast}+o_{p}(1), \label{bias}%
\end{equation}
where $\boldsymbol{\hat{\mu}}_{T}=T^{-1}\sum_{t=1}^{T}\mathbf{f}_{t}$,%
\[
\mathbf{\Sigma}_{\beta\beta}=\lim_{n\rightarrow\infty}\left(  \frac
{\mathbf{B}_{n}^{\prime}\mathbf{M}_{n}\mathbf{B}_{n}}{n}\right)  \text{, and
}\overline{\sigma}^{2}=\lim_{n\rightarrow\infty}n^{-1}\sum_{i=1}^{n}\sigma
_{i}^{2}>0.
\]

\end{theorem}

The proof is provided in Section \ref{PFMbias} of the mathematical appendix.

To derive the asymptotic distribution of $\boldsymbol{\hat{\lambda}}%
_{nT}-\boldsymbol{\lambda}_{0}$ it is required that both $n$ and
$T\rightarrow\infty$, jointly. Also, noting that
\[
\boldsymbol{\hat{\lambda}}_{nT}-\boldsymbol{\lambda}_{0}=\left(
\boldsymbol{\hat{\mu}}_{T}-\boldsymbol{\mu}_{0}\right)  +\left(
\boldsymbol{\hat{\phi}}_{nT}-\boldsymbol{\phi}_{0}\right)  ,
\]
it is clear that increasing $n$ is not relevant for the distribution of
$\boldsymbol{\hat{\mu}}_{T}-\boldsymbol{\mu}_{0}$, but joint $n$ and $T$
asymptotics are required when investigating the distribution of
$\boldsymbol{\hat{\phi}}_{nT}-\boldsymbol{\phi}_{0}$. Focussing on the latter,
and using result (\ref{phidis1})\ in the Appendix, we have%
\begin{align*}
&  \left(  n^{-1}\mathbf{\hat{B}}_{nT}^{\prime}\mathbf{M}_{n}\mathbf{\hat{B}%
}_{nT}\right)  \sqrt{nT}\left(  \boldsymbol{\hat{\phi}}_{nT}-\boldsymbol{\phi
}_{0}\right)  =n^{-1/2}T^{1/2}\mathbf{B}_{n}^{\prime}\mathbf{M}_{n}%
\boldsymbol{\eta}_{n}+n^{-1/2}T^{1/2}\mathbf{G}_{T}^{\prime}\mathbf{U}%
_{nT}^{\prime}\mathbf{M}_{n}\boldsymbol{\eta}_{n}\\
&  +n^{-1/2}T^{1/2}\mathbf{G}_{T}^{\prime}\mathbf{U}_{nT}^{\prime}%
\mathbf{M}_{n}\mathbf{\bar{u}}_{n\circ}-n^{-1/2}T^{1/2}\mathbf{G}_{T}^{\prime
}\mathbf{U}_{nT}^{\prime}\mathbf{U}_{nT}\mathbf{G}_{T}\boldsymbol{\lambda}%
_{T}^{\ast}.
\end{align*}

Where $\boldsymbol{U}_{nT}=\left(  \mathbf{u}_{1\circ},\mathbf{u}_{2\circ
},...,\mathbf{u}_{n\circ}\right)  ^{\prime},$ $\mathbf{u}_{i\circ}=\left(
u_{i1},u_{i2},...,u_{iT}\right)  ^{\prime},$ $\mathbf{G}_{T}=$ $\mathbf{M}%
_{T}\mathbf{F}\left(  \mathbf{F}^{\prime}\mathbf{M}_{T}\mathbf{F}\right)
^{-1}$, $\overline{\mathbf{u}}_{n\circ}=(\overline{u}_{1\circ},\overline
{u}_{2\circ},...,\overline{u}_{n\circ})^{\prime},$ and $\overline{u}_{i\circ
}=T^{-1}\sum_{t=1}^{T}u_{it}.$ Consider first the terms that include the
pricing errors, $\boldsymbol{\eta}_{n}$, and using the results in Lemma
\ref{L2} note that
\[
n^{-1/2}T^{1/2}\mathbf{B}_{n}^{\prime}\mathbf{M}_{n}\boldsymbol{\eta}%
_{n}=O_{p}\left(  T^{1/2}n^{-1/2+\alpha_{\eta}}\right)  ,\text{ }%
n^{-1/2}T^{1/2}\mathbf{G}_{T}^{\prime}\mathbf{U}_{nT}^{\prime}\mathbf{M}%
_{n}\boldsymbol{\eta}_{n}=O_{p}\left(  n^{-1/2+\frac{\alpha_{\eta}%
+\alpha_{\gamma}}{2}}\right)  .
\]
It is clear that the effects of pricing errors on the distribution of
$\boldsymbol{\hat{\phi}}_{nT}$ vanish only if $T^{1/2}n^{-1/2+\alpha_{\eta}%
}\rightarrow0$, and $\alpha_{\eta}+\alpha_{\gamma}<1$. Also
\begin{align*}
n^{-1/2}T^{1/2}\mathbf{G}_{T}^{\prime}\mathbf{U}_{nT}^{\prime}\mathbf{M}%
_{n}\mathbf{\bar{u}}_{n\circ}  &  =O_{p}\left(  T^{-1/2}\right)  \text{, }\\
n^{-1/2}T^{1/2}\mathbf{G}_{T}^{\prime}\mathbf{U}_{nT}^{\prime}\mathbf{U}%
_{nT}\mathbf{G}_{T}\boldsymbol{\lambda}_{T}^{\ast}  &  =\sqrt{\frac{n}{T}}%
\bar{\sigma}_{n}^{2}\left(  \frac{\mathbf{F}^{\prime}\mathbf{M}_{T}\mathbf{F}%
}{T}\right)  ^{-1}\boldsymbol{\lambda}_{T}^{\ast}+O_{p}\left(  T^{-1/2}%
\right)  .
\end{align*}
Finally, for the first two terms involving $\mathbf{B}_{n}$ and $\mathbf{U}%
_{nT}$ we have%
\begin{equation}
n^{-1/2}T^{1/2}\mathbf{B}_{n}^{\prime}\mathbf{M}_{n}\left(  \mathbf{\bar{u}%
}_{n\circ}-\mathbf{U}_{nT}\mathbf{G}_{T}\boldsymbol{\lambda}_{T}^{\ast
}\right)  =O_{p}\left(  1\right)  . \label{Dis}%
\end{equation}
It is clear that the small $T$ bias of the asymptotic distribution of the
two-step estimator, given by $\sqrt{\frac{n}{T}}\bar{\sigma}_{n}^{2}\left(
\frac{\mathbf{F}^{\prime}\mathbf{M}_{T}\mathbf{F}}{T}\right)  ^{-1}%
\boldsymbol{\lambda}_{T}^{\ast}$, does not vanish unless, $n/T\rightarrow0$.
At the same time for the pricing errors to have no impact on the distribution
of the two-step estimator we must have $Tn^{a_{\eta}}/n\rightarrow0$. Both
conditions cannot be met simultaneously. It is possible to derive the
asymptotic distribution of $\boldsymbol{\hat{\phi}}_{nT}$, and hence that of
$\boldsymbol{\hat{\lambda}}_{nT}$, when $n/T\rightarrow0$ and $\mathbf{\eta
=0}$, but these are quite restrictive conditions, and to avoid them we follow
\citet{shanken1992estimation} and instead consider a bias-corrected version of
$\boldsymbol{\hat{\phi}}_{nT}$, namely $\boldsymbol{\tilde{\phi}}_{nT}$ given
by (\ref{phitilda}). As noted earlier $\boldsymbol{\tilde{\phi}}%
_{nT}=\boldsymbol{\tilde{\lambda}}_{nT}-\boldsymbol{\hat{\mu}}_{T}$, where
$\boldsymbol{\tilde{\lambda}}_{nT}$ is the bias-corrected version of
$\boldsymbol{\hat{\lambda}}_{nT}$ originally proposed by Shanken.

To investigate the asymptotic properties of $\boldsymbol{\tilde{\phi}}_{nT}$
we first need to establish conditions under which $\widehat{\bar{\sigma}}%
_{nT}^{2}$, defined by (\ref{zigbar}), is a consistent estimator of
$\bar{\sigma}_{n}^{2}=n^{-1}\sum_{i=1}^{n}\sigma_{i}^{2}$, which enters the
bias-corrected estimator. The proof of consistency in the literature does not
allow for missing factors or pricing errors and only considers the case where
$T$ is fixed as $n\rightarrow\infty$. For derivation of asymptotic
distribution of $\boldsymbol{\tilde{\phi}}_{nT}$ we also need to consider the
limiting properties of $\widehat{\bar{\sigma}}_{nT}^{2}$ under joint $n$ and
$T$ asymptotics. The following theorem provides the required results for
$\widehat{\bar{\sigma}}_{nT}^{2}$ as an estimator of $\bar{\sigma}_{n}^{2}$.

\begin{theorem}
\label{Thzig} Consider $\widehat{\bar{\sigma}}_{nT}^{2}$, the estimator of
$\bar{\sigma}_{n}^{2}$ given by (see (\ref{zigbar})),%
\begin{equation}
\widehat{\bar{\sigma}}_{nT}^{2}=\frac{\sum_{t=1}^{T}\sum_{i=1}^{n}\hat{u}%
_{it}^{2}}{n(T-K-1)}, \label{AdjZig}%
\end{equation}
and suppose that Assumptions \ref{factors}, \ref{Latent factor}, and
\ref{Errors}, are satisfied. Then for a fixed $T$%
\begin{equation}
\lim_{n\rightarrow\infty}E\left(  \widehat{\bar{\sigma}}_{nT}^{2}\right)
=\bar{\sigma}^{2}, \label{unbiasZig}%
\end{equation}
where $\bar{\sigma}^{2}=\lim_{n\rightarrow\infty}\bar{\sigma}_{n}^{2}$, and
$\bar{\sigma}_{n}^{2}=n^{-1}\sum_{i=1}^{n}\sigma_{i}^{2}$. Furthermore
\begin{equation}
\widehat{\bar{\sigma}}_{nT}^{2}-\bar{\sigma}_{n}^{2}=O_{p}\left(
T^{-1/2}n^{-1/2}\right)  . \label{orderzig}%
\end{equation}

\end{theorem}

For a proof see sub-section \ref{ProofThzig} in the Appendix.

Result (\ref{orderzig}) shows that $\widehat{\bar{\sigma}}_{nT}^{2}$ continues
to be a consistent estimator of $\bar{\sigma}^{2}=\lim_{n\rightarrow\infty
}\bar{\sigma}_{n}^{2}$ for a fixed $T$ as $n\rightarrow\infty$, even in the
presence of pricing errors and a missing common factor. This result also holds
when one or more of the factors are semi-strong.

Equipped with the above result we are now in a position to present the theorem
that sets out the asymptotic distribution of $\boldsymbol{\tilde{\phi}}_{nT}$.

\begin{theorem}
\label{Tfi}Consider, $\boldsymbol{\tilde{\phi}}_{nT}$, the bias-corrected
estimators of $\boldsymbol{\phi}_{0}$ given by (\ref{phitilda}). Suppose
Assumptions \ref{factors}, \ref{loadings}, \ref{Errors}, \ref{Latent factor}
and \ref{PriceError} hold, all the observed factors are strong, ($\alpha
_{k}=1$, for $k=1,2,...,K$), and the strength of the missing factor,
$\alpha_{\gamma}$ defined by (\ref{normg}), satisfies $\alpha_{\gamma}$
$<1/2$.
\begin{align}
\boldsymbol{\tilde{\phi}}_{nT}-\boldsymbol{\phi}_{0}  &  =O_{p}\left(
T^{-1/2}n^{^{-1/2}}\right)  +O_{p}\left(  T^{-1/2}n^{-1+\frac{\alpha_{\eta
}+\alpha_{\gamma}}{2}}\right) \label{phigap}\\
&  +O_{p}\left(  n^{-1+\alpha_{\eta}}\right)  +O_{p}\left(  T^{-1}%
n^{-1/2}\right)  ,\nonumber
\end{align}
where $\alpha_{\eta}\,$denotes the degree of pervasiveness of the pricing
errors defined by (\ref{APTg}). (a) When $T$ is fixed, $\alpha_{\gamma}<1/2$
and $\alpha_{\eta}<1$, then there exists $T_{0}$ such that for all $T>T_{0}$
\begin{equation}
p\lim_{n\rightarrow\infty}\left(  \boldsymbol{\tilde{\phi}}_{nT}\right)
=\boldsymbol{\phi}_{0}\text{.} \label{phicon}%
\end{equation}
Also
\begin{equation}
\sqrt{nT}\left(  \boldsymbol{\tilde{\phi}}_{nT}-\boldsymbol{\phi}_{0}\right)
=\mathbf{\Sigma}_{\beta\beta}^{-1}\boldsymbol{\xi}_{nT}+O_{p}\left(
n^{-\frac{1}{2}+\frac{\alpha_{\eta+\alpha_{\gamma}}}{2}}\right)  +O_{p}\left(
T^{1/2}n^{-1/2+\alpha_{\eta}}\right)  +O_{p}\left(  T^{-1/2}\right)  ,
\label{phiDis}%
\end{equation}
where $\mathbf{\Sigma}_{\beta\beta}=p\lim_{n\rightarrow\infty}\left(
n^{-1}\mathbf{B}_{n}^{\prime}\mathbf{M}_{n}\mathbf{B}_{n\ }\right)  ,$%
\begin{equation}
\boldsymbol{\xi}_{nT}=n^{-1/2}T^{-1/2}\mathbf{B}_{n}^{\prime}\mathbf{M}%
_{n}\mathbf{U}_{nT}\mathbf{a}_{T}, \label{Tegzi}%
\end{equation}
and $\mathbf{a}_{T}=\mathbf{\tau}_{T}-\mathbf{M}_{T}\mathbf{F}(T^{-1}%
\mathbf{F}^{\prime}\mathbf{M}_{T}\mathbf{F)}^{-1}\boldsymbol{\lambda}%
_{T}^{\ast}$. (b) If $\alpha_{\gamma}<1/2$, $\alpha_{\eta}<1/2,$ and
$\sqrt{\frac{T}{n}}n^{\alpha_{\eta}}\rightarrow0$, as $n$ and $T\rightarrow
\infty$ jointly, then%
\begin{equation}
\sqrt{nT}\left(  \boldsymbol{\tilde{\phi}}_{nT}-\boldsymbol{\phi}_{0}\right)
\rightarrow_{d}N\left(  \mathbf{0,\Sigma}_{\beta\beta}^{-1}\mathbf{V}_{\xi
}\mathbf{\Sigma}_{\beta\beta}^{-1}\right)  , \label{Dphi}%
\end{equation}
where
\begin{equation}
\mathbf{V}_{\xi}=\left(  1+\boldsymbol{\lambda}_{0}^{\prime}\mathbf{\Sigma
}_{f}^{-1}\boldsymbol{\lambda}_{0}\right)  \text{ }p\lim_{n\rightarrow\infty
}\left(  n^{-1}\mathbf{B}_{n}^{\prime}\mathbf{M}_{n}\mathbf{V}_{u}%
\mathbf{M}_{n}\mathbf{B}_{n\ }\right)  . \label{Vegzi}%
\end{equation}

\end{theorem}

For a proof see sub-section \ref{ProofTfi} in the Appendix.

Result (\ref{phigap}) establishes the finite $T$ consistency of
$\boldsymbol{\tilde{\phi}}_{nT}$ for $\boldsymbol{\phi}_{0}$ so long as
$\alpha_{\gamma}<1/2$ and $\alpha_{\eta}<1$, thus extending the Shanken result
to a much more general setting. To the best of our knowledge the asymptotic
distribution in (\ref{Dphi}) is new and shows that the asymptotic covariance
matrix of $\boldsymbol{\tilde{\phi}}_{nT}$ includes the term
$\boldsymbol{\lambda}_{T}^{\ast\prime}(T^{-1}\mathbf{F}^{\prime}\mathbf{M}%
_{T}\mathbf{F)}^{-1}\boldsymbol{\lambda}_{T}^{\ast}$, that arises from the
first stage estimation of the factor loadings, and must be included in the
analysis for valid inference. It is also clear that this additional term does
not vanish with $T\rightarrow\infty$, and tends to $\boldsymbol{\lambda}%
_{0}^{\prime}\mathbf{\Sigma}_{f}^{-1}\boldsymbol{\lambda}_{0}\geq\left(
\boldsymbol{\lambda}_{0}^{\prime}\boldsymbol{\lambda}_{0}\right)
\lambda_{\max}\left(  \mathbf{\Sigma}_{f}^{-1}\right)  =\left(
\boldsymbol{\lambda}_{0}^{\prime}\boldsymbol{\lambda}_{0}\right)
\lambda_{\min}\left(  \mathbf{\Sigma}_{f}\right)  >0$, which is strictly
non-zero unless $\boldsymbol{\lambda}_{0}=\mathbf{0}$. Shanken type bias
correction addresses the mean of the asymptotic distribution of
$\boldsymbol{\tilde{\phi}}_{nT}$, but not its covariance.

The $O_{p}\left(  T^{-1/2}\right)  $ term in (\ref{phiDis}) arises from the
sampling errors involved in the estimation of the factor loadings and
$\bar{\sigma}_{n}^{2}$, and tends to zero at the regular $\sqrt{T}$ rate. But
$n$ has to be sufficiently large to eliminate the effects of pricing errors on
identification of $\boldsymbol{\phi}_{0}$, as dictated by condition
$\sqrt{\frac{T}{n}}n^{\alpha_{\eta}}\rightarrow0,$ as $n$ and $T\rightarrow
\infty$.\footnote{The condition $\sqrt{\frac{T}{n}}n^{\alpha_{\eta}%
}\rightarrow0$ can be weakened somewhat to $\sqrt{\frac{T}{n}}n^{\alpha_{\eta
}/2}\rightarrow0$ if we also assume that $\beta_{ik}-\bar{\beta}_{k}$ are
independently distributed over $i$, but will still require $n$ to be larger
than $T$.} The requirement that $T$ need not be too large relative to $n$ for
estimation of $\boldsymbol{\phi}_{0}$ is consistent with separating the
estimation of $\boldsymbol{\phi}_{0}$ from that of $\boldsymbol{\mu}_{0}$,
allowing the use a relatively small $T$ and a large $n$ to estimate
$\boldsymbol{\phi}_{0}$ and a relatively large $T$ when estimating
$\boldsymbol{\mu}_{0}.$

\subsection{What if one or more of the risk factors are semi-strong?}

We now turn to an intermediate case where one or more of the observed factors
are semi-strong, in the sense that their factor strength, $\alpha_{k}$ lies
between $1/2$ and $1$. The case of weak risk factors is already covered in the
proceeding analysis, and such factors can be included in the error term,
$u_{it}$, with little consequence for the estimation of risk premia of the
remaining factors that are strong or semi-strong. Weak factors do not have any
explanatory power and can be dropped from the analysis.

When one or more of the observed factors is semi-strong $\mathbf{\Sigma
}_{\beta\beta}$ is no longer positive definite and Theorem \ref{Tfi} does not
apply, but it is possible to adapt the proofs to establish the limiting
properties of $\tilde{\phi}_{k,nT}\boldsymbol{\ }$(the $k^{th}$ element of
$\boldsymbol{\tilde{\phi}}_{nT}$)\ for different values of \thinspace
$\alpha_{k}\,$.

To this end, analogously to $\boldsymbol{\tilde{\phi}}_{nT}$, we introduce the
following estimator of $\boldsymbol{\phi}_{0}$%
\begin{equation}
\boldsymbol{\tilde{\phi}}_{nT}\left(  \boldsymbol{\alpha}\right)
=\mathbf{H}_{nT}^{-1}\left(  \boldsymbol{\alpha}\right)  \left[
\mathbf{D}_{\alpha}^{-1}\mathbf{\hat{B}}_{nT}^{\prime}\mathbf{M}%
_{n}\mathbf{\hat{\alpha}}_{nT}+\frac{n}{T}\widehat{\bar{\sigma}}_{nT}%
^{2}\mathbf{D}_{\alpha}^{-1}\left(  \frac{\mathbf{F}^{\prime}\mathbf{M}%
_{T}\mathbf{F}}{T}\right)  ^{-1}\boldsymbol{\hat{\mu}}\right]  ,
\label{phialpha0}%
\end{equation}
where
\begin{equation}
\mathbf{H}_{nT}\left(  \boldsymbol{\alpha}\right)  =\mathbf{D}_{\alpha}%
^{-1}\mathbf{\hat{B}}_{nT}^{\prime}\mathbf{M}_{n}\mathbf{\hat{B}}%
_{nT}\mathbf{D}_{\alpha}^{-1}-\frac{n}{T}\widehat{\bar{\sigma}}_{nT}%
^{2}\mathbf{D}_{\alpha}^{-1}\left(  \frac{\mathbf{F}^{\prime}\mathbf{M}%
_{T}\mathbf{F}}{T}\right)  ^{-1}\mathbf{D}_{\alpha}^{-1}. \label{Halpha}%
\end{equation}
It is now easily seen that%
\begin{equation}
\mathbf{D}_{\alpha}\left(  \boldsymbol{\tilde{\phi}}_{nT}\left(
\boldsymbol{\alpha}\right)  -\boldsymbol{\phi}_{0}\right)  =\mathbf{H}%
_{nT}^{-1}\left(  \boldsymbol{\alpha}\right)  \mathbf{q}_{nT}\left(
\boldsymbol{\alpha}\right)  , \label{phialpha}%
\end{equation}
where%
\begin{equation}
\mathbf{q}_{nT}\left(  \boldsymbol{\alpha}\right)  =\mathbf{D}_{\alpha}%
^{-1}\mathbf{\hat{B}}_{nT}^{\prime}\mathbf{M}_{n}\mathbf{\hat{\alpha}}%
_{nT}+\frac{n}{T}\widehat{\bar{\sigma}}_{nT}^{2}\mathbf{D}_{\alpha}%
^{-1}\left(  \frac{\mathbf{F}^{\prime}\mathbf{M}_{T}\mathbf{F}}{T}\right)
^{-1}\boldsymbol{\hat{\mu}}_{T}-\mathbf{H}_{nT}(\boldsymbol{\alpha}%
)\mathbf{D}_{\alpha}\boldsymbol{\phi}_{0}. \label{qnT}%
\end{equation}
$\mathbf{D}_{\alpha}$ is defined by (\ref{Dn}), and $\boldsymbol{\alpha
}=\mathbf{(}\alpha_{1},\alpha_{2},...,\alpha_{K})^{\prime}$. It is easily
established that numerically $\boldsymbol{\tilde{\phi}}_{nT}\left(
\boldsymbol{\alpha}\right)  $ is identical to $\boldsymbol{\tilde{\phi}}_{nT}%
$, and its introduction is primarily for the purpose of establishing the
limiting properties of $\tilde{\phi}_{k,nT}-\phi_{0,k}$ that do depend on
$\alpha_{k}$. Note that
\[
\mathbf{H}_{nT}\left(  \boldsymbol{\alpha}\right)  =n\mathbf{D}_{\alpha}%
^{-1}\mathbf{H}_{nT}\mathbf{D}_{\alpha}^{-1},\text{ and }\mathbf{q}%
_{nT}\left(  \boldsymbol{\alpha}\right)  =n\mathbf{D}_{\alpha}^{-1}%
\mathbf{s}_{nT}%
\]
where $\mathbf{s}_{nT}$ and $\mathbf{H}_{nT}$ are already defined by
(\ref{snT0}) and (\ref{A-HnT}). Using these in (\ref{phialpha}) we have%
\[
\mathbf{D}_{\alpha}\left(  \boldsymbol{\tilde{\phi}}_{nT}\left(
\boldsymbol{\alpha}\right)  -\boldsymbol{\phi}_{0}\right)  =\left(
n\mathbf{D}_{\alpha}^{-1}\mathbf{H}_{nT}\mathbf{D}_{\alpha}^{-1}\right)
^{-1}n\mathbf{D}_{\alpha}^{-1}\mathbf{s}_{nT}=\mathbf{D}_{\alpha}%
\mathbf{H}_{nT}^{-1}\mathbf{s}_{nT},
\]
and it follows that $\boldsymbol{\tilde{\phi}}_{nT}\left(  \boldsymbol{\alpha
}\right)  -\boldsymbol{\phi}_{0}=\mathbf{H}_{nT}^{-1}\mathbf{s}_{nT}%
=\boldsymbol{\tilde{\phi}}_{nT}\left(  \boldsymbol{\tau}_{K}\right)
=\boldsymbol{\tilde{\phi}}_{nT}$. See (\ref{A-phitilda}).

The convergence results for $\boldsymbol{\tilde{\phi}}_{nT}\left(
\boldsymbol{\alpha}\right)  $ are set out in the following theorem.

\begin{theorem}
\label{Tsemi}Consider, $\boldsymbol{\tilde{\phi}}_{nT}\left(
\boldsymbol{\alpha}\right)  $, the bias-corrected estimators of
$\boldsymbol{\phi}_{0}$ given by (\ref{phialpha0}), and suppose Assumptions
\ref{factors}, \ref{loadings}, \ref{Errors}, \ref{Latent factor} and
\ref{PriceError} hold, the strength of observed factors, $\boldsymbol{f}%
_{t}=(f_{1t},f_{2t},...,f_{Kt})^{\prime},$ is given by $\boldsymbol{\alpha
}\mathbf{=(}\alpha_{1},\alpha_{2},...,\alpha_{K})^{\prime}$, and the strength
of the missing factor, $g_{t}$, defined by (\ref{normg}) is $\alpha_{\gamma}$.
Let $\alpha_{\min}=\min_{k}(\alpha_{k})$ and suppose that $\alpha_{\gamma
}<1/2$. Then%
\[
\mathbf{H}_{nT}\left(  \boldsymbol{\alpha}\right)  =\mathbf{D}_{\alpha}%
^{-1}\mathbf{B}_{n}^{\prime}\mathbf{M}_{n}\mathbf{B}_{n}\mathbf{D}_{\alpha
}^{-1}+O_{p}\left(  T^{-1}n^{-\alpha_{\min}+1/2}\right)  ,
\]
where $\mathbf{H}_{nT}\left(  \boldsymbol{\alpha}\right)  $ is given by
(\ref{Halpha}), and by part (b) of Assumption \ref{loadings}, $\mathbf{H}%
_{nT}\left(  \boldsymbol{\alpha}\right)  \rightarrow_{p}\mathbf{\Sigma}%
_{\beta\beta}(\boldsymbol{\alpha}\mathbf{)}>0$, for any fixed $T>T_{0}$ such
that $\lambda_{\max}\left(  \frac{\mathbf{F}^{\prime}\mathbf{M}_{T}\mathbf{F}%
}{T}\right)  ^{-1}<C$ and $\alpha_{\min}>1/2>\alpha_{\gamma}$. Also
\begin{align}
\tilde{\phi}_{k,nT}\left(  \boldsymbol{\alpha}\right)  -\phi_{0,k}  &
=O_{p}\left(  n^{-(\alpha_{k}+\alpha_{\min})/2+1/2}T^{-1/2}\right)
+O_{p}\left(  n^{\frac{-\left(  \alpha_{k}+\alpha_{\min}\right)  +\left(
\alpha_{\eta+\alpha_{\gamma}}\right)  }{2}}T^{-1/2}\right) \label{Psemi}\\
&  +O_{p}\left(  n^{-\left(  \alpha_{k}+\alpha_{\min}\right)  /2+\alpha_{\eta
}}\right)  +O_{p}\left(  n^{-(\alpha_{k}+\alpha_{\min})/2+1/2}T^{-1}\right)
.\nonumber
\end{align}

\end{theorem}

See sub-section \ref{ProofThsemi} of the Appendix for a proof.

The result in (\ref{Psemi}) establishes the consistency of $\tilde{\phi
}_{k,nT}\left(  \boldsymbol{\alpha}\right)  =\tilde{\phi}_{k,nT}$ even if
$f_{kt}$ is semi-strong so long as $n\rightarrow\infty$, and $\alpha_{\min
}>1/2,$ $\alpha_{k}+\alpha_{\min}>\alpha_{\eta}+\alpha_{\gamma}$, and
$\alpha_{k}+\alpha_{\min}>2\alpha_{\eta}$. Clearly, these results reduce to
the case of strong factors where $\alpha_{\min}=\alpha_{k}=1$. Turning to the
asymptotic distribution of $\tilde{\phi}_{k,nT}\left(  \boldsymbol{\alpha
}\right)  $, again only convergence rates are affected, and instead of the
regular rate of $\sqrt{nT}$, we have $\ \sqrt{T}n^{(\alpha_{k}+\alpha_{\min
}-1)/2}$, and using (\ref{Psemi}) we have%
\begin{align}
&  \sqrt{T}n^{(\alpha_{k}+\alpha_{\min}-1)/2}\left(  \tilde{\phi}%
_{k,nT}\left(  \boldsymbol{\alpha}\right)  -\phi_{0,k}\right) \label{phikrate}%
\\
&  =O_{p}(1)+O_{p}\left(  n^{-1/2+\frac{\left(  \alpha_{\eta}+\alpha_{\gamma
}\right)  }{2}}\right)  +O_{p}\left(  \sqrt{T}n^{-1/2+\alpha_{\eta}}\right)
+O_{p}\left(  T^{-1/2}\right)  .\nonumber
\end{align}
The conditions needed for eliminating the effects of the pricing errors are
the same as before and are given by $\alpha_{\eta}+\alpha_{\gamma}<1$ and
$\sqrt{T}n^{-1/2+\alpha_{\eta}}\rightarrow0$. The asymptotic distribution is
unaffected except for the slower rate of convergence alluded to above. It is
also of interest to note that adding semi-strong factors can adversely affect
the convergence rate of the strong factor with $\alpha_{k}=1$. As an example
suppose the asset pricing model contains two factors, one strong, $\alpha
_{1}=1$ and one semi strong with $\alpha_{2}<1$.$\,\ $Then the convergence
rate of $\tilde{\phi}_{1,nT}\left(  \boldsymbol{\alpha}\right)  -\phi_{0,1}$
is given by$\sqrt{T}$ $n^{(1+\alpha_{2}-1)/2}$ which is slower than the rate
we would have obtained for $\tilde{\phi}_{1,nT}\left(  \boldsymbol{\alpha
}\right)  -\phi_{0,1}$ if both factors were strong ($\alpha_{min}=\alpha
_{k}=1)$, namely the regular rate of $\sqrt{nT}$.

Furthermore, when conditions $\alpha_{\eta}+\alpha_{\gamma}<1$ and $\sqrt
{T}n^{-1/2+\alpha_{\eta}}\rightarrow0$ are met we have
\begin{equation}
\tilde{\phi}_{k,nT}\left(  \boldsymbol{\alpha}\right)  -\phi_{0,k}%
=O_{p}\left(  T^{-1/2}n^{-(\alpha_{k}+\alpha_{\min}-1)/2}\right)  ,
\label{phiorder}%
\end{equation}
and $\phi_{0,k}$ is consistently estimated if $T^{-1/2}n^{-(\alpha_{k}%
+\alpha_{\min}-1)/2}\rightarrow0$. Also, using (\ref{phitilda2}), an estimator
of the risk premia, $\lambda_{k}$, is given by $\tilde{\lambda}_{k,nT}\left(
\boldsymbol{\alpha}\right)  =\tilde{\phi}_{k,nT}\left(  \boldsymbol{\alpha
}\right)  +\hat{\mu}_{k,T}$, where $\hat{\mu}_{k,T}=T^{-1}\sum_{t=1}^{T}%
f_{kt}$. Hence%
\[
\tilde{\lambda}_{k,nT}\left(  \boldsymbol{\alpha}\right)  -\lambda
_{0,k}=\left[  \tilde{\phi}_{k,nT}\left(  \boldsymbol{\alpha}\right)
-\phi_{0,k}\right]  +\left(  \hat{\mu}_{k,T}-\mu_{0,k}\right)  .
\]
Under Assumption \ref{factors} $\hat{\mu}_{k,T}-\mu_{0,k}=O_{p}\left(
T^{-1/2}\right)  $, and using (\ref{phiorder}) it then follows that
\begin{equation}
\tilde{\lambda}_{k,nT}\left(  \boldsymbol{\alpha}\right)  -\lambda_{0,k}%
=O_{p}\left(  T^{-1/2}n^{-(\alpha_{k}+\alpha_{\min}-1)/2}\right)
+O_{p}\left(  T^{-1/2}\right)  , \label{ratelambda}%
\end{equation}
and $\tilde{\lambda}_{k,nT}\left(  \boldsymbol{\alpha}\right)  $ is a
consistent estimator of $\lambda_{0,k}$ if $T\rightarrow\infty$ as well as
$T^{-1/2}n^{-(\alpha_{k}+\alpha_{\min}-1)/2}\rightarrow0$. More specifically,
suppose $T=\ominus\left(  n^{d}\right)  $ for some $d>0$, where $\ominus
\left(  \cdot\right)  $ denotes $T$ and $n^{d}$ are of the same order of
magnitude. Then for any $d>0$, the condition for consistency of $\tilde
{\lambda}_{k,nT}\left(  \boldsymbol{\alpha}\right)  $ is given by $\alpha
_{k}+\alpha_{min}+d>1$ and $d>0$. In the case where all risk factors have the
same strength, $\alpha$, the consistency condition reduces to $\alpha>\left(
1-d\right)  /2$, which is weaker than the one derived by \cite{giglio2023test}%
, namely $n/(\left\Vert \beta\right\Vert ^{2}T)\rightarrow0$, where, in terms
of our notation, $\left\Vert \beta\right\Vert ^{2}=\ominus\left(  n^{\alpha
}\right)  $. This latter condition will be met if $\alpha>1-d$. In practice
where $T$ is small relative to $n$, the accuracy of $\tilde{\lambda}%
_{k,nT}\left(  \boldsymbol{\alpha}\right)  $ as an estimator $\lambda_{0,k}$
does depend on $\alpha$, and our weaker condition on $\alpha>\left(
1-d\right)  /2$ is advantageous.

\subsection{Consistent estimation of the variance of $\boldsymbol{\tilde{\phi
}}_{nT}$\label{Vphihat}}

To carry out inference on $\boldsymbol{\phi}_{0}$, or any of its elements
individually, we require a consistent estimator of $Var\left(
\boldsymbol{\tilde{\phi}}_{nT}\right)  $. Using (\ref{Dphi}) and (\ref{Vegzi})
we first note that $\mathbf{\Sigma}_{\beta\beta}$ is consistently estimated by
$\mathbf{H}_{nT\ }$ given by (\ref{HnT1}). Therefore, it is sufficient to find
a suitable estimator of $\mathbf{V}_{u}=(\sigma_{ij})$ such that
$\mathbf{V}_{\xi}$ given by (\ref{Vegzi}) is consistently estimated. Under
suitable sparsity restrictions $\mathbf{V}_{u}$ can be consistently estimated
using the various thresholding procedures advanced in the statistical
literature by
\citet{bickel2008covariance, bickel2008regularized}%
,
\citet{cai2011adaptive}%
, and
\citet[BPS]{BPS2019multiple}%
.
\citet{fan2011high, fan2013large}
also show that the adaptive threshold technique of Cai and Liu applies equally
to the residuals from an approximate factor model. Here we consider the
threshold estimator proposed by BPS which does not require cross-validation
and is shown to have desirable small sample properties. It is given by
$\mathbf{\tilde{V}}_{u}=\left(  \tilde{\sigma}_{ij}\right)  $
\begin{align}
\tilde{\sigma}_{ii}  &  =\hat{\sigma}_{ii}\nonumber\\
\tilde{\sigma}_{ij}  &  =\hat{\sigma}_{ij}\mathbf{1}\left[  \left\vert
\hat{\rho}_{ij}\right\vert >T^{-1/2}c_{\alpha}(n,\delta)\right]  ,\text{
}i=1,2,\ldots,n-1,\text{ }j=i+1,\ldots,n, \label{threshold}%
\end{align}
where
\begin{equation}
\hat{\sigma}_{ij}=\frac{1}{T}\sum_{t=1}^{T}\hat{u}_{it}\hat{u}_{jt},\text{
}\hat{\rho}_{ij}=\frac{\hat{\sigma}_{ij}}{\sqrt{\hat{\sigma}_{ii}\hat{\sigma
}_{jj}}},\text{ \ }\hat{u}_{it}=r_{it}-\hat{\alpha}_{i,T}-\boldsymbol{\hat
{\beta}}_{i,T}^{\prime}\mathbf{f}_{t}, \label{sij}%
\end{equation}
and $c_{p}(n,d)=\Phi^{-1}\left(  1-\frac{p}{2n^{d}}\right)  ,$ is a normal
critical value function, $p$ is the the nominal size of testing of
$\sigma_{ij}=0$, ($i\neq j$) and $d$ is chosen to take account of the
$n(n-1)/2$ multiple tests being carried out. Monte Carlo experiments carried
out by BPS suggest setting $d=2$. The variance estimator given by
(\ref{threshold}) does not require a knowledge of the factor strength and
applies to risk factors of differing degrees.

Under Assumptions \ref{factors}, \ref{Latent factor}, and \ref{Errors},
$\left\Vert \mathbf{V}_{u}\right\Vert =O\left(  n^{\alpha_{\gamma}}\right)  $,
and using results in
\citet{fan2011high, fan2013large}
we have%
\begin{equation}
\left\Vert \mathbf{\tilde{V}}_{u}-\mathbf{V}_{u}\right\Vert =O_{p}\left(
n^{\alpha_{\gamma}}\sqrt{\frac{\ln(n)}{T}}\right)  . \label{NormVu}%
\end{equation}
Consider the following estimator of $\mathbf{V}_{\xi}$
\[
\mathbf{\hat{V}}_{\xi,nT}=\left(  1+\hat{s}_{nT}\right)  \left(
n^{-1}\mathbf{\hat{B}}_{nT}^{\prime}\mathbf{M}_{n}\mathbf{\tilde{V}}%
_{u}\mathbf{M}_{n}\mathbf{\hat{B}}_{nT}\right)  \text{.}%
\]
where $\hat{s}_{nT}=\boldsymbol{\tilde{\lambda}}_{nT}^{^{\prime}}\left(
T^{-1}\mathbf{F}^{\prime}\mathbf{M}_{T}\mathbf{F}\right)  ^{-1}%
\boldsymbol{\tilde{\lambda}}_{nT}$. Under Assumption \ref{factors},
$T^{-1}\mathbf{F}^{\prime}\mathbf{M}_{T}\mathbf{F\rightarrow}_{p}%
\mathbf{\Sigma}_{f}$ $\ $and using the results above we have
$\boldsymbol{\tilde{\lambda}}_{nT}=\boldsymbol{\tilde{\phi}}_{nT}%
+\boldsymbol{\hat{\mu}}_{T}\rightarrow_{p}\boldsymbol{\phi}_{0}%
+\boldsymbol{\mu}_{0}=\boldsymbol{\lambda}_{0}$. Hence, $\hat{s}%
_{nT}\rightarrow_{p}\boldsymbol{\lambda}_{0}^{\prime}\mathbf{\Sigma}_{f}%
^{-1}\boldsymbol{\lambda}_{0}$ as $n,T\rightarrow\infty$, jointly, and it is
sufficient to show that
\begin{equation}
n^{-1}\mathbf{\hat{B}}_{nT}^{\prime}\mathbf{M}_{n}\mathbf{\tilde{V}}%
_{u}\mathbf{M}_{n}\mathbf{\hat{B}}_{nT}-n^{-1}\mathbf{B}_{n}^{\prime
}\mathbf{M}_{n}\mathbf{V}_{u}\mathbf{M}_{n}\mathbf{B}_{n\ }\rightarrow
_{p}\mathbf{0}\text{.} \label{VarCon}%
\end{equation}
The following theorem provides a formal statement of the conditions under
which $\mathbf{\hat{V}}_{\xi,nT}$ is a consistent estimator of $\mathbf{V}%
_{\xi}$.

\begin{theorem}
\label{Var}Suppose Assumptions \ref{factors}, \ref{loadings}, \ref{Errors},
\ref{Latent factor} and \ref{PriceError} hold, and all the observed factors
are strong, ($\alpha_{k}=1$, for $k=1,2,...,K$), and the strength of the
missing factor, $g_{t}$, defined by (\ref{normg}), $\alpha_{\gamma}$ $<1/2$.
Then%
\begin{equation}
\left\Vert \mathbf{\hat{V}}_{\xi,nT}-\mathbf{V}_{\xi}\right\Vert =O_{p}\left(
n^{\alpha_{\gamma}}\sqrt{\frac{\ln(n)}{T}}\right)  , \label{NormVaregzi}%
\end{equation}
where
\begin{equation}
\mathbf{\hat{V}}_{\xi,nT}=\left(  1+\hat{s}_{nT}\right)  \left(
n^{-1}\mathbf{\hat{B}}_{nT}^{\prime}\mathbf{M}_{n}\mathbf{\tilde{V}}%
_{u}\mathbf{M}_{n}\mathbf{\hat{B}}_{nT}\right)  , \label{Varegzihat}%
\end{equation}
$\mathbf{\tilde{V}}_{u}=\left(  \tilde{\sigma}_{ij}\right)  $, $\tilde{\sigma
}_{ij}$ is the threshold estimator of $\sigma_{ij}$ given by (\ref{threshold}%
), and
\begin{equation}
\mathbf{V}_{\xi}=\left(  1+\boldsymbol{\lambda}_{0}^{\prime}\mathbf{\Sigma
}_{f}^{-1}\boldsymbol{\lambda}_{0}\right)  \text{ }p\lim_{n\rightarrow\infty
}\left(  n^{-1}\mathbf{B}_{n}^{\prime}\mathbf{M}_{n}\mathbf{V}_{u}%
\mathbf{M}_{n}\mathbf{B}_{n\ }\right)  , \label{Varegzi}%
\end{equation}

\end{theorem}

For a proof see sub-section \ref{ProofTvar} in the Appendix.

This theorem shows that consistent estimation of $Var\left(
\boldsymbol{\tilde{\phi}}_{nT}\right)  $ can be achieved by using a suitable
threshold estimator of $\mathbf{V}_{u}$, so long as the strength of the
missing factor, $\alpha_{\gamma}$, is sufficiently weak in the sense that
$n^{\alpha_{\gamma}}\sqrt{\ln(n)/T}\rightarrow0$ as $n,T\rightarrow\infty$.

\section{Small sample properties of the estimators and tests for $\boldsymbol{\phi}$\label{Simulations}}

\subsection{Monte Carlo Design}

This section presents Monte Carlo simulations to investigate the small sample
properties of estimators and tests for $\boldsymbol{\phi}_{0}$. In the
empirical application of the next section the factors are selected from a
large list. But here we assume $K=3$ and mimic the 3 Fama-French factors,
namely the market return minus the risk free rate, MKT, the value factor (high
minus low book to market portfolios, HML) and the size factor (small minus big
portfolios, SMB). These are denoted by $f_{kt},$ $k=M,H,S$.\footnote{Data on
factors and the risk free rate are downloaded from Kenneth French's data
library:
https://mba.tuck.dartmouth.edu/pages/faculty/ken.french/data\_library.html}
For further details see Section \ref{EMC} of the online supplement A.

\subsubsection{Loadings and factor strengths}

To calibrate the loadings, $\beta_{ik}$, we used excess returns on a large
number securities observed over the shorter sample covering the 20 years
$2002m1$ $-2021m12$ ($T=240$). Monthly returns for NYSE and NASDAQ stocks code
10 and 11 from CRSP were downloaded from Wharton Research Data Services and
converted to excess returns over the risk free rate, taken from Kenneth
French's webpages, in percent per month. Only stocks with available data for
the full sample were included, yielding a balanced panel, and to avoid
outliers influencing the results, stocks with a kurtosis greater than 16 were
excluded. There were $1289$ stocks before exclusion on the basis of kurtosis
and $1175$ after. The summary statistics giving mean, median, standard
deviation of the estimates of $\beta_{ik}$ and their histograms are provided
in the online supplement A.

For factor strength, we considered a range of DGPs. Given the evidence that
most factors, other than the market factor, are not strong, we focus on the
case where there is one strong factor, namely the market factor with
$\alpha_{M}=1,$ plus two semi-strong factors, with the value factor, $HML$,
being quite strong with $\alpha_{H}=0.85,$ and the size factor, $SML$, being
only moderately strong with $\alpha_{S}=0.65$. These estimates are also
informed by the results provided in
\citet{bailey2021measurement}
who propose methods for estimation of factor strength. For a given factor
strength, $\alpha_{k}$, the associated loadings, $\beta_{ik}$, are generated
as $\boldsymbol{\beta}_{k}=(\beta_{1k},\beta_{2k},,...,\beta_{\alpha_{k}%
},0,0,...0)$ where $n_{\alpha_{k}}=\lfloor n^{\alpha_{k}}\rfloor$ the integer
part of $n^{\alpha_{k}}$, with non-zero and zero values of $\boldsymbol{\beta
}_{k}$ given by
\begin{align*}
\beta_{ik}  &  \sim IIDN(\mu_{\beta_{k}},\sigma_{\beta_{k}}^{2})\text{, for
}i=1,2,....,\lfloor n^{\alpha_{k}}\rfloor,\\
\beta_{ik}  &  =0\text{ for }i=\lfloor n^{\alpha_{k}}\rfloor+1,\lfloor
n^{\alpha_{k}}\rfloor+2,...,n,
\end{align*}
where $\lfloor n^{\alpha_{k}}\rfloor$ denotes the integer part of
$n^{\alpha_{k}}$. \ Since the security returns are randomly generated, it does
not matter how zero and non-zero values of $\beta_{ik}$ are distributed across
$i$. Also, the zero loadings can also be replaced by an exponentially decaying
sequence without any implications for the simulation results.\footnote{See
also footnote 5 of
\citet
*[p.942]{bailey2016exponent}.} We also set
\begin{align*}
\mu_{\beta_{M}}  &  =1\text{, }\sigma_{\beta_{M}}=0.4;\text{ }\mu_{\beta_{H}%
}=0.2\text{, }\sigma_{\beta_{H}}=0.5\\
\mu_{\beta_{S}}  &  =0.6\text{, }\sigma_{\beta_{S}}=0.5,
\end{align*}
which match the mean and standard deviation of the estimates of $\beta_{ik}$.
See above.

\subsubsection{Generation of pricing errors}

The pricing errors in (\ref{APTRoss}) can be considered as firm-specific
characteristics and are set as $\boldsymbol{\eta}_{n}=\left(  \eta_{1}%
,\eta_{2},...,\eta_{n_{\eta}},0,0,...,0\right)  ^{\prime}.$ The non-zero
loadings of $\boldsymbol{\eta}_{n}$ for $i\leq n_{\eta}=\lfloor n^{\alpha
_{\eta}}\rfloor$\ are drawn from $IIDU(0.7,0.9)$, and $\eta_{i}=0$ for
$i=n_{\eta}+1,n_{\eta}+2,....,n.$ We consider $\alpha_{\eta}=(0,0.3).$\textbf{
}When\textbf{ }$\alpha_{\eta}=0$ we have $\eta_{i}=0$ ,\textbf{\ }for all $i.$
As in the case of factor loadings the non-zero values of $\boldsymbol{\eta
}_{n}$ must be randomly allocated to different groups.

\subsubsection{Generation of return equation errors}

The return equation errors, $u_{it}$, are generated following (\ref{uit}) as a
combination of a missing factor, $g_{t}\sim IIDN\left(  0,1\right)  $ plus an
idiosyncratic error, $v_{jt}.$ The loadings $\boldsymbol{\gamma}=\left(
\gamma_{1},\gamma_{2},...,\gamma_{n_{\gamma}},0,0,...,0\right)  ^{\prime}$ of
the missing factor are set as%
\begin{align*}
\gamma_{i}  &  \sim IIDU(0.7,0.9)\text{, for }i=1,2,....,\lfloor
n^{\alpha_{\gamma}}\rfloor,\\
\gamma_{i}  &  =0\text{, for }i=\lfloor n^{\alpha_{\gamma}}\rfloor+1,\lfloor
n^{\alpha_{\gamma}}\rfloor+2,...,n,
\end{align*}
where $\alpha_{\gamma}$ is the strength of the missing factor $g_{t}$. We
consider $\alpha_{\gamma}=1/4$ and $1/2$.

For the idiosyncratic errors, $v_{it}$, we consider spatial as well as a block
diagonal specification, with the spatial specification including a diagonal
specification as the special case. Under the spatial specification the
idiosyncratic errors are generated as the first order spatial autoregressive
model $v_{it}=\rho_{\varepsilon}\sum_{j=1}^{n}w_{ij}v_{jt}+\kappa
\varepsilon_{it},$ which can be written in matrix notation as $\mathbf{v}%
_{t}=\rho_{\varepsilon}\mathbf{Wv}_{t}+\kappa\boldsymbol{\varepsilon}_{t},$
and solved for as $\mathbf{v}_{t}=\kappa\left(  \mathbf{I}_{n}-\rho
_{\varepsilon}\mathbf{W}\right)  ^{-1}\boldsymbol{\varepsilon}_{t}$. Adding
the missing factor now yields%
\begin{equation}
\mathbf{u}_{t}=\boldsymbol{\gamma}\text{ }g_{t}+\kappa\left(  \mathbf{I}%
_{n}-\rho_{\varepsilon}\mathbf{W}\right)  ^{-1}\boldsymbol{\varepsilon}_{t}.
\label{vSAR}%
\end{equation}
The spatial coefficient $\rho_{\varepsilon}$ is such that $|\rho_{\varepsilon
}|<1$, $\mathbf{W=(}w_{ij})$ with $w_{ii}=0$, and $\sum_{j=1}^{n}w_{ij}=1$.
The diagonal case is obtained by setting $\rho_{\varepsilon}=0$, with
$\rho_{\varepsilon}=0.5$ characterizing the SAR specification. The weight
matrix $\mathbf{W}=(w_{ij})$ is set to follow the familiar rook pattern where
all its elements are set to zero except for $w_{i+1,i}=w_{j-1,j}=0.5$ for
$i=1,2,...,n-2$ and $j=3,4...,n$, with $w_{1,2}=w_{n,n-1}=1$.

Under the block error covariance specification, $\mathbf{v}_{t}$ is generated
as $\mathbf{v}_{t}=\kappa\mathbf{\hat{S}}\boldsymbol{\varepsilon}_{t},$ where
$\mathbf{\hat{S}}$ is a block diagonal matrix with its $b^{th}$ block given by
$\mathbf{\hat{S}}_{b}$ for $b=1,2,...,B$, and $\boldsymbol{\varepsilon}%
_{t}=(\mathbf{\varepsilon}_{1t}^{\prime},\mathbf{\varepsilon}_{2t}^{\prime
},...,\mathbf{\varepsilon}_{Bt}^{\prime})^{\prime}$, and
$\boldsymbol{\varepsilon}_{bt}=(\varepsilon_{b,1t},\varepsilon_{b,2t}%
,...,\varepsilon_{b,n_{b},t})^{\prime}$. $\mathbf{\hat{S}}$ is set as a
Cholesky factor of the correlation matrix of $\mathbf{u}_{t}$. Denoting this
correlation matrix by $\mathbf{\hat{R}}_{u},$
\[
\mathbf{\hat{R}}_{u}=\left[  Diag(\mathbf{\hat{V}}_{Bu})\right]
^{-1/2}\mathbf{\hat{V}}_{Bu}\left[  Diag(\mathbf{\hat{V}}_{Bu})\right]
^{-1/2}=Diag(\mathbf{\hat{R}}_{bu},\text{ }b=1,2,...,B),
\]
where $\mathbf{\hat{V}}_{Bu}$ is the threshold estimator of $\mathbf{V}_{u}$
subject to the additional restriction that $\mathbf{V}_{u}$ is block diagonal.
For each block $\mathbf{\hat{R}}_{bu}$ we set the number of distinct non-zero
elements of this block equal to the integer part of $[n_{b}(n_{b}-1)/2]\times
q_{b}$ where $q_{b}$ is the proportion of non-zero distinct elements in block
$b$ of our calibrated sample and computed by the calibration over the sample
$2001m10-2021m9$. The non-zero elements are drawn randomly from $IIDU(0,0.5)$.
Similarly, adding the missing factor, we have
\begin{equation}
\mathbf{u}_{t}=\boldsymbol{\gamma}g_{t}+\kappa\mathbf{\hat{S}}%
\boldsymbol{\varepsilon}_{t}. \label{vblock}%
\end{equation}

The block diagonal structure is intended to capture possible within industry
correlations not picked up by observed or weak missing factors, with each
block representing an industry or sector. To calibrate the block structure
estimates of the pair-wise correlations between the residuals of the return
regressions using the Fama-French three factors of the $T=240$ sample ending
in $2021$ were obtained. Then all the statistically insignificant correlations
were set to zero, allowing for the multiple testing nature of the tests. For
the majority of securities (668 out of the 1168), the pair-wise return
correlations were not statistically significant. The securities with a
relatively large number of non-zero correlations were either in the banking or
energy related industries. Considering stocks by 2-digit SIC classifications,
a division into $B=14$ contiguous groups ranging in size from $33$ to $145$
stocks, seemed sensible. More detail on the process is given in Section
\ref{BlockCov} of the online supplement A.

The primitive errors, $\varepsilon_{it}$ for $i=1,2,...,n$ in (\ref{vSAR}) and
(\ref{vblock})\ are generated as $\varepsilon_{it}=\sqrt{\sigma_{ii}}%
\varpi_{it}$, where $\varpi_{it}\sim IIDN\left(  0,1\right)  $, and
$\varepsilon_{it}=\sqrt{\sigma_{ii}}\left[  \sqrt{\frac{v-2}{v}}\varpi
_{it}\right]  ,$ where $\varpi_{it}\sim IID$ $t(v)$, with $t(v)$ denotes a
standard $t$ distributed variate with $v=5$ degrees of freedom. Also
$\sigma_{ii}\thicksim IID$ \ $0.5(1+\chi_{1}^{2})$ $,\ for$\ $i=1,2,...,n_{b}$
and $b=1,2,...,B$. In this way, it is ensured that $Var(\varepsilon
_{b,it})=\sigma_{b,ii}$, and on average $E\left[  Var(\varepsilon
_{b,it})\right]  =E(\sigma_{ii})=1$, under both Gaussian and t-distributed
errors. Note that $Var\left(  \nu_{b,it}\right)  =v/(v-2)$. All the
experiments are designed to give an $R^{2}$ of about $0.3,$ similar to that
obtained in the empirical applications. For further details see sub-section
\ref{Regfit} of the online supplement A.

\subsubsection{Experiments}

In total, we consider 12 experimental designs: six designs with Gaussian
errors and six with $t(5)$ distributed errors. We considered designs with
GARCH effects, with and without pricing errors, $\eta_{i}$, and with and
without the missing factor, $g_{t}$. We also considered designs with spatial
patterns in the idiosyncratic errors, $v_{it}$. All experiments are
implemented using $R=2,000$ replications. Details of of the 12 experiments are
summarized in Table S-1 of the online supplement B (MC results).

\subsubsection{Alternative estimators of $\mathbf{V}_{u}$}

Subsection \ref{Vphihat} considered consistent estimation of the variance of
$\tilde{\phi}_{nT}$ using $\mathbf{\tilde{V}}_{u}$ a threshold estimator for
$\mathbf{V}_{u},$ given by equation (\ref{threshold})$.$ For comparison
purposes we also considered two other estimators of $\mathbf{V}_{u}.$ These
were the sample covariance matrix $\mathbf{\hat{V}}_{u}=\sum_{t=1}%
^{T}\mathbf{\hat{u}}_{t}\mathbf{\hat{u}}_{t}^{^{\prime}}/T$ and a diagonal
covariance matrix, where the off-diagonal elements of $\mathbf{\hat{V}}_{u},$
$\hat{\sigma}_{ij},$ are set to zero. Thus we have three designs for the
return error covariance matrix, $\mathbf{V}_{u},$ and three different
estimators of it. A comparison of the results for the different covariance
matrices is available on request. The diagonal estimator, as to be expected,
performed poorly when the true covariance matrix was not diagonal,
particularly for the spatial error covariance matrix and when the strength of
the missing factor was close to $1/2$. For these designs the sample and
threshold estimators of the covariance matrix generally performed similarly
and given that there is a theoretical justification for the threshold
estimator and there are structures of the error covariance matrix for which
the sample estimator is unlikely to perform well we report the results using
the threshold estimator in the simulations below.

\subsection{Monte Carlo results}

We focus on a comparison of two-step (defined by (\ref{phihat1}) and the
bias-corrected (BC) estimator (defined by (\ref{phitilda})), and report bias,
root mean square error (RMSE) and size for testing $H_{0j}\,:\phi_{0k}=0,$
$k=M,H,S$ at the five per cent nominal level, for all $n=100,500,1,000,3,000$
and $T=60,$ $120,$ $240$ combinations. The results for all 12 experiments are
summarized in Tables S-A-E1 to S-A-E12 in the online supplement B. In terms of
bias and RMSE the two-step estimator does much better than the bias-corrected
(BC) estimator when $T=60$ and $n=100$, but this gap closes quickly as $n$ is
increased. In fact for $T=60$ and $n=3,000,$ the bias and RMSE of the BC
estimator (at $0.0010$ and $0.0607)$ are much less than those of the two-step
estimator (at -$0.0080$ and $0.1489)$. This pattern continues to hold when
$T=120$ and $240$. Bias correction can cause the RMSE to "blow up" for small
samples, such as $n=100$ and $T=60$, but for $n=500$ and above the
bias-corrected estimator always has a smaller RMSE than the two-step
estimator. As discussed in the theoretical section, having a large $n$ is
important for the properties of the estimators.

But most importantly, the two-step estimator is subject to substantial size
distortions, particularly when $T$ is small relative to $n$. As predicted by
the theory, the degree of over-rejection of the tests based on FM estimator
falls with $T,$ but increases with $n$. For example, the two-step test sizes
rise from $11.1\%$ when $T=60$ and $n=100$ to $60.9\%$ when $T=60$ and
$n=3,000$. Increasing $T$ reduces the size distortion of the two-step
estimator but test sizes are still substantially above the $5\%$ nominal value
when $n$ is large. The strong tendency of the tests based on the
two-step\ estimator to over-reject could be an important contributory factor
leading to false discovery of a large number of apparently significant factors
in the literature. In contrast, sizes of the tests based on the BC estimator,
using the variance estimator given by (\ref{Varphitilt}), are all close to its
nominal value, irrespective of the factor strength or sample size
combinations. We only note some elevated test sizes in the case of the
experimental design 12, and when we consider the semi-strong factors. The
highest test size of $7.85$ per cent is obtained for the least strong factor,
$f_{st}$, when $n=3000$ and $T=60$. See Table S-A-E12 of the online supplement B.

We also experimented with raising $\alpha_{\eta}$ from $0.3$ to $0.5,$ making
the pricing errors much more pervasive and $\rho_{\varepsilon}$ from $0.5$ to
$0.85,$ introducing more spatial correlation$.$ This increased the rejection
rate in experiments 9 and 10.

\subsubsection{Empirical power functions}

Plots of the empirical power functions for testing the null hypothesis
$H_{0j}\,:\phi_{0k}=0,$ $k=M,H,S$, are also provided in the online supplement
B for the 12 experiments (Figures S-A-E1 to S-A-E12). All power functions have
the familiar bell curve shape and tend to unity as $n$ and/or $T$ are
increased, showing the test has satisfactory power, particularly for $n$
sufficiently large even when $T=60$. Again the power functions are quite
similar across the 12 different experiments and show similar patterns for
strong and semi-strong factors. However, this similarity hides the fact that
the test of $\phi_{k}=0$ for the strong factor is much more powerful than
corresponding tests for the semi-strong factors, with the test power declining
as factor strength is reduced.

\subsubsection{Differences in performance of strong and semi-strong factors}

These differences in the effects of factor strength on the power of the test
of $\phi_{k}=0$ are in line with our theoretical results, and are also
reflected in the rate at which the RMSE of the estimators of $\phi_{M}$,
$\phi_{H}$ and $\phi_{S}$ fall with $n$. For example, using results in Tables
S-B-E10-12 in the online supplement B for the bias-corrected estimator in the
case of design 12 with $T=240,$ we note that the ratio of RMSE of $n=3,000$ to
$n=100$ is $17\%$ for the strong factor $\left(  \alpha_{M}=1\right)  $,
$23\%$ for the first semi-strong factor with $\alpha_{H}=0.85$, and $36\%$ for
the second semi-strong factor with $\alpha_{S}=0.65.$ As strength falls one
needs larger cross section samples of securities to attain the same level of
precision. In the case of the two-step estimator there was the same pattern,
but the fall in the RMSE with $n$ was much slower. The ratio for $T=240$ of
RMSE of $n=3,000$ to $n=100$ is $25\%$ for $\phi_{M}$, rather than $17\%$ (for
the bias-corrected estimator); $48\%$ rather than $23\%$ for $\phi_{H}$. 

\subsubsection{Misspecification: Semi-strong versus weak factors}

So far, we have assumed that the DGP is correctly specified with two
semi-strong and no observed weak factors. Here we consider the implications of
incorrectly excluding semi-strong factors or correctly including weak factors
on the small sample properties of the bias-corrected\ estimator of $\phi_{M},$
the coefficient of the strong factor. Using the same DGP (which includes one
strong factor and two semi-strong factors), we carried out additional MC
experiments (designs 1-12) where we also estimated $\phi_{M}$ without the
semi-strong factors being included in the regressions. Comparative results,
with and without the semi-strong factors, are summarized in Tables S-C-E1-3 to
S-C-E10-12 of the online supplement B. We find that incorrectly excluding
semi-strong factors can be quite costly, both in terms of bias and RMSE as
well as size distortions. In terms of RMSE it was almost always better to
estimate the model with the semi strong factors included. The exception was
for the case of $T=60,$ $n=100,$ where including the semi-strong factors
caused the RMSE to blow up. Size distortions resulting from the exclusion of
the semi-strong factors tended to be more pronounced for large $n$ and $T$
samples. These conclusions were not sensitive to the choice of the
experimental design. For these experiments the lesson seems to be that it is
important to have $n$ large and include relevant semi-strong factors provided
that they are sufficiently strong.

When the DGP includes one strong factor ($\alpha_{M}=1$) and two weak factors
($\alpha_{H}=$ $\alpha_{S}=0.5)$, in terms of the bias and RMSE for $\phi_{M}$
it is unambiguously better to exclude the weak factors from the regression,
even though they are in the DGP. Weak factors are best treated as missing and
absorbed in the error term.

\subsection{Main conclusions from MC experiments}

The conclusions from the Monte Carlo simulations are that the bias corrected
estimator of $\boldsymbol{\phi}_{0}\boldsymbol{,}$ generally works well.
Although, it can generate a large RMSE for small $n$ and $T$, this can be
solved by increasing $n.$ This performance is robust to non-Gaussian errors,
GARCH effects, missing weak factors, pricing errors and weak cross-sectional
dependence. Test sizes are generally correct and the power good. The rate at
which RMSEs decline with $n$ depends on the strength of the underlying
factors. Semi-strong factors need much larger values of $n$ for precise
estimation. Tests of the joint significance of $\phi_{M}=\phi_{H}=\phi_{S}=0,$
not reported here, also performed well, as might be expected given the good
power performance of the separate induced tests which are reported. Including
weak factors could be harmful, but there are potential advantages of adding
semi-strong factors, although the issue of how best to select such factors is
an open question to which we now turn.

\section{Factor selection when the number of securities and the number of factors are both large\label{FacSel}}

Our theoretical derivations and Monte Carlo simulations both assume that the
number of risk factors included in the return regressions is fixed and the
factors are known. For the empirical application we face the additional
challenge of selecting a small number of relevant risk factors from a possibly
large number of potential factors, $m$. This problem has been the subject of a
number of recent studies.
\citet{harvey2016and}
propose a multiple testing approach aimed at controlling the false discovery
rate in the process of factor selection, and in a more recent paper
\citet{harvey2020false}
suggest using a double-bootstrap method to calibrate the t-statistic used for
controlling the desired level of the false discovery rate.
\citet{giglio2021asset}
suggest applying the double-selection Lasso procedure by
\citet{belloni2014inference}
to second pass regressions. None of these methods distinguish between strong,
semi-strong or weak factors in their selection process, whereas the theory and
simulations presented above indicate the importance of factor strength for
estimation and inference.

Here we propose an alternative selection procedure where we first estimate the
strength of all the $m$ factors under consideration, and then select factors
with strength above a given threshold, the value of which is informed by the
convergence results of Theorem \ref{Tsemi}. This theorem showed that if factor
$f_{kt}$ has strength $\alpha_{k},$ then for a given $T$ the BC estimator of
$\phi_{k}$ converges to its true value, $\phi_{0k}$, at the rate of
$n^{(\alpha_{k}+\alpha_{\min}-1)/2}$, where $\alpha_{\min}=\min_{i}(\alpha
_{i})$. As is recognized in non-parametric estimation literature, if the rate
of convergence is less than $1/3,$ the gain in precision with $n$ is so slow
that the estimator may not be that useful.\footnote{The
\citet{manski1985semiparametric}
maximum score estimator for a binary response model has $n^{1/3}$ convergence
and this is regarded as very slow and there are suggested modifications such
as
\citet{horowitz1992smoothed}
to increase the rate of convergence to $n^{2/5}$.} To achieve rate of
$n^{1/3}$ we need to set the threshold value of $\alpha_{k}$, denoted by
$\underline{\alpha}$, such that $\underline{\alpha}+\alpha_{\min}>1+2/3$. The
smallest value of such a threshold is obtained when $\alpha_{\min
}=\underline{\alpha}$ or if $\underline{\alpha}>1/2+1/3.$ Given the threshold
the main issue is how to estimate factor strength. This problem is already
addressed in \citet*[BKP]{bailey2021measurement} when $m$ is fixed. In this
setting they base their estimation on the statistical significance of $f_{kt}$
in the first stage time-series regressions of excess returns on all the
factors under consideration, whilst allowing for the $n$ multiple testing
problem which their approach entails. When $m\,$(the number of factors) is
also large the first stage regressions will also be subject to the multiple
testing problem and penalized regression techniques such as Lasso or the one
covariate at a time (OCMT) selection procedure technique proposed by
\citet{chudik2018one}
could be used. Irrespective of selection technique used at the level of
individual security returns, we end up with $n$ different subsets of the $m$
factors under consideration. Factor strengths can then be estimated similarly
to BKP\ from their selection frequencies across the $n$ securities.

To be more specific, denote the set of $m$ factors under consideration by
$\mathcal{S}$ and denote the set of selected factors for security $i$ by
$\hat{S}_{i}$ and their numbers by $\hat{m}_{i}=|\hat{S}_{i}|$. Clearly
$\hat{S}_{i}\subseteq\mathcal{S}$, and $\hat{m}_{i}\leq m$, for $i=1,2,...,n$.
Then compute the proportion of stocks in which the $k^{th}$ factor is
selected, $\hat{\pi}_{k},$ for $k\in\{1,2,...,m\}$ based on $\hat{S}_{i},$
$i=1,2,...,n$ by $\hat{\pi}_{k}=\frac{1}{n}\sum_{i=1}^{n}\mathcal{I}\{k\in
\hat{S}_{i}\}$. Then the strength of the $k^{th}$ factor is measured by%
\begin{equation}
\hat{\alpha}_{k}=%
\begin{cases}
1+\frac{\ln\hat{\pi}_{k}}{\ln n},\text{ if }\hat{\pi}_{k}>0,\\
0,\quad\quad\quad\text{ if }\hat{\pi}_{k}=0.
\end{cases}
\label{alphaj}%
\end{equation}
The transformation from $\hat{\pi}_{k}$ to $\hat{\alpha}_{k}$ is explained and
justified in BKP, where it is shown that considering strength aids
interpretation because it is not dependent on $n$. It is beyond the scope of
the present paper to provide theoretical justification for the proposed factor
selection procedure, but using extensive Monte Carlo experiments
\citet{yoo2022factor}
has shown that the proposed method has desirable small sample properties
whether Lasso or OCMT is used for factor selection at the level of individual
security returns.

\section{An empirical application using a large number of U.S. securities and a large number of risk factors\label{Empirical}}

This section uses the results above in the explanation of monthly returns for
a large number, $n,$ of U.S. securities, by a large active set of $m$
potential risk factors. We first briefly describe the sources and
characteristics of the data for the stock returns and factors, which cover
different sub-samples over the period $1996m1-2022m12.$ We then consider the
selection of a subset of $K$ factors from the active set. Finally we test
$\boldsymbol{\phi}_{0}=0$, and construct and evaluate phi-portfolios and
corresponding mean-variance\textbf{ } portfolios for alternative models.

Monthly returns (inclusive of dividends) for NYSE and NASDAQ stocks from CRSP
with codes $10$ and $11$ were downloaded from Wharton Research Data Services.
They were converted to excess returns by subtracting the risk free rate, which
was taken from Kenneth French's data base. To obtain balanced panels of stock
returns and factors, only variables for which there was data for the full
sample under consideration were used. Excess returns are measured in percent
per month. To avoid outliers influencing the results, stocks with a kurtosis
greater than 16 were excluded. To examine factor selection, four samples were
considered, each had $20$ years of data, $T=240$, ending in $2015m12$,
$2017m12$, $2019m12$, $2021m12$. Filtering out the stocks with kurtosis larger
than $16$ removed about $100$ of the roughly $1200$ stocks. The number of
stocks ($n$) considered for each of the four $T=240$ samples are given in
panel A of Table \ref{FactorSelPL240}. Further detail is given in Section
\ref{DEMP} of the online supplement A.\footnote{Summary statistics for the
excess returns across the different samples are given in Table
\ref{tab:ret_sumtat} of the online supplement A.} For analysis of the
phi-portfolios, the factors selected in the sample ending in $2015m12$ were
used to construct portfolios up to $2022m12.$

\subsection{Factor selection}

For factor selection we used a sample of $T=240$ observations\footnote{Some
results for $T=120$ are included in the online supplement.}. The set of
factors considered combine the $5$ Fama-French factors with the $207$ factors
from the
\citet{chen2022open}%
, Open Source Asset Pricing webpages, both downloaded July 6 2022. Only
factors with data for the full sample were considered so the return
regressions constitute a balanced panel. The number of factors in each of the
four $20$ year samples ending in the years $2015$, $2017$, $2019$, $2021$ is
also given in panel A of Table \ref{FactorSelPL240}, and range between $187$
to $199$. Summary statistics for the factors in the active set are given in
\ref{SumFac} of the online supplement A.

To implement the factor selection procedure set out in Section \ref{FacSel},
Lasso is used to carry out selection in the return regressions for individual
securities and we refer to the factor selection procedure as pooled Lasso
(PL). As is well known, Lasso does not work well with too many highly
correlated regressors, therefore, factors with an absolute correlation with
the market factor greater than $0.70$ were dropped. This still left between
$177$ and $190$ risk factors in the active set $\mathcal{S}$, depending on the
sample period (see panel A of Table \ref{FactorSelPL240}).

Specifically, Lasso was applied $n$ times to the regressions of excess
returns, $r_{it}=R_{it}-r_{t}^{f}$, for $i=1,2,...,n$, on the $177$ to $190$
factors in the active set, $\mathcal{S}$, to select the sub-set $\hat{S}_{i}$
for each $i$ over the four 20-year samples, separately. Following the
literature, the tuning parameters in the Lasso algorithm were set by ten-fold
cross-validation.\footnote{The post-Lasso and one covariate multiple testing
(OCMT) approach of
\citet
*[CKP]{chudik2018one}. were also investigated, but Lasso seemed to work
reasonably well. The details of the Lasso procedure used are given in Section
2.2 of the online supplement of CKP paper.} Interestingly, the market factor
was selected by Lasso for almost all the securities, thus confirming the
pervasive nature of the market factor. No other factor came close to being
selected for all the securities. Lasso tended to choose a lot of non-market
factors and every factor got chosen in at least one return regression. The
mean number of non-market factors chosen by Lasso fell from $11.6$ in the
$2015$ sample to $9.9$ in the $2021$ sample. The median was lower, falling
from $10$ to $8$. There was a long right tail because Lasso tended to choose a
very large number of non-market factors for some securities, ranging from a
maximum of $48$ in the 2021 sample to $54$ in the $2017$ sample.

Apart from the market factor, there are no systematic patterns for the rest of
selected factors in the return equations for the individual securities.
Following the theory set out in Section \ref{FacSel}, the $K$ factors used to
estimate $\boldsymbol{\phi}$ are chosen on the basis of their factor
strength.\footnote{The idea of using factor strength could also be viewed as a
kind of averaging of the factors selected in individual return regressions.} A
minimum threshold of $0.7$ was used. This is below the threshold value of
$\underline{\alpha}=1/2+1/3$ required to achieve the convergence rate of
$n^{1/3}$, and is intended to capture borderline semi-strong factors. We also
consider the values of $0.75,$and $0.80$ that are quite close to
$\underline{\alpha}$. The number of selected factors for different choices of
factor strength threshold is given in panel B of Table \ref{FactorSelPL240},
for the four different samples.

Using the threshold of $0.70$, $17$ factors (inclusive of the market factor)
were selected for the sample ending in $2015$, with the number of selected
factors declining to $15,13$ and $11,$ for the samples ending $2017,$ $2019$
and $2021$, respectively. At the other extreme, setting the threshold at
$0.80$, the number of selected factors dropped to $4$ for the samples ending
in $2015$, $2017$ and $2021$, and $2$ for the sample ending in $2021$. Since
$17$ factors seemed too large and $2$ factors too small, the threshold value
was set at the intermediate value of $0.75$. We considered always conditioning
on the market factor, but since Lasso almost always selected it, this was unnecessary.%

\floatstyle{plaintop}
\restylefloat{table}%
\begin{table}[H]%

\begin{center}%
\begin{tabular}
[c]{ccccc}\hline\hline
$T=240$ with \ end dates & 2021 & 2019 & 2017 & 2015\\\hline
\multicolumn{5}{c}{}\\
\multicolumn{5}{c}{Panel A: Number of stocks and factors under consideration}%
\\
Number of stocks & 1289 & 1276 & 1243 & 1181\\
Number of stocks with kurtosis $<16$ & 1175 & 1143 & 1132 & 1090\\
Number of non-market factors & 187 & 198 & 199 & 197\\
Number of non-market factors with $r<0.70$ & 177 & 189 & 190 & 189\\
\multicolumn{5}{c}{}\\
\multicolumn{5}{c}{Panel B: Number of selected factors by strength
threshold}\\
Number with strength $>0.80$ & 2 & 4 & 4 & 4\\
Number with strength $>0.75$ & 4 & 6 & 7 & 7\\
Number with strength $>0.70$ & 11 & 13 & 15 & 17\\\hline\hline
\end{tabular}

\end{center}

{\footnotesize \textit{Note:} Panel A shows the number of stocks and risk
factors used before and after filtering by the specified criterion. Panel B
shows the number of risk factors selected with strength greater than the
specified threshold level using Lasso to select factors at the level of
individual securities. }%
\caption
{Summary statistics for the number of stocks and number of selected factors using the factor strength threshold  of 0.75, for four twenty years ($T=240$) samples ending in 2021, 2019, 2017 and 2015}%
\label{FactorSelPL240}%
\end{table}%
%

\onehalfspacing
The list of factors selected by pooled Lasso for the four samples are given in
Table \ref{TabRiskF}. The three Fama-French factors, Market, HML and SMB, are
all selected in all four periods.\footnote{We also considered selecting the
risk factors using the generalized one covariate at a time (OCMT) method
proposed by
\citet{sharifvaghefi2022variable}%
. Using GOCMT the Fama-French three factors were again amongst the five
strongest factors selected. The use of GOCMT\ for factor selection is also
investigated by
\citet{yoo2022factor}%
, using Monte Carlo and empirical applications.} Of the Fama-French three,
only the market factor is strong, with estimated strength in excess of $0.98$
across the four periods.\footnote{These results also support the choice of 3
Fama-French factors and their strength used in our Monte Carlo simulations.}
The other factor which is selected across all the four periods is "short
selling". This is proposed by
\citet{dechow2001short}
who argue that short-sellers target firms that are priced high relative to
fundamentals. It measures the extent to which investors are shorting the
market as reflected in Compustat data. Two additional factors are selected in
periods ending in $2019$ and earlier. One is "Beta Tail Risk" proposed by
\citet{kelly2014tail}
which estimates a time-varying tail exponent from the cross section of
returns. The other is "Cash Based Operating Profitability" (CBOP) suggested
by
\citet{ball2016accruals}%
. This is operating profit less accruals, with working capital and R\&D
adjustments. For periods ending in 2017 and 2015 the "Sin Stock" indicator
proposed by
\citet{hong2009price}
is also selected. It takes the value of unity if the stock in question is
involved in producing alcohol, tobacco, and gaming. They find that such stocks
are held less by norm-constrained institutions such as pension plans.

The factor strengths are relatively stable across the periods, with many of
the estimates close to the threshold value of $0.75$. Apart from the market
factor only SMB, Short Selling and Beta Tail Risk factors have strengths in
excess of $0.85$ when averaged across the four periods. From the large number
of factors in the active set we have ended up with relatively few factors that
are reasonably strong and for which $\boldsymbol{\phi}_{0}$ can be estimated
reasonably accurately.\footnote{We do not report estimates of individual
$\phi_{k}.$ Because of correlations between the loadings, the sign, size and
significance of the coefficients are difficult to interpret and for
phi-portfolio construction, discussed below, what matters is $\boldsymbol{\phi
}^{\prime}\boldsymbol{\phi}$ which determines the return on the portfolio.}%

\singlespacing
%

\begin{table}[htbp]%

\begin{center}%
\begin{tabular}
[c]{ccccccccc}\hline\hline
End date & 2021 \  & \multicolumn{2}{c}{2019} & \multicolumn{2}{c}{2017} &
\multicolumn{3}{c}{2015}\\
Selected Factors & \multicolumn{8}{c}{Estimated strength ($\alpha$)}\\\hline
Mkt. & 0.99 &  & 0.98 &  & 0.98 &  & 0.98 & \\
SMB & 0.90 &  & 0.84 &  & 0.86 &  & 0.86 & \\
Short Selling & 0.77 &  & 0.85 &  & 0.85 &  & 0.83 & \\
HML & 0.76 &  & 0.77 &  & 0.76 &  & 0.75 & \\
BetaTailRisk &  &  & 0.87 &  & 0.86 &  & 0.86 & \\
CBOP &  &  & 0.76 &  & 0.77 &  & 0.77 & \\
Sin Stock &  &  & . &  & 0.76 &  & 0.76 & \\\hline\hline
\end{tabular}

\end{center}

{\footnotesize \textit{Note:} The risk factors listed are the market factor
(Mkt.), size (SMB), Short Selling that measures the extent of short sales in
the market, the value factor (HML), the cash-based operating profitability
factor (CBOP), Beta Tail Risk, and Sin Stock which is a binary indicator
taking the value of unity if the stock in question is involved in so called
"Sin" industries producing alcohol, tobacco, and gaming. Further details on
these risk factors are provided by
\citet{chen2022open}%
.}%
\caption{Selected factors with estimated strength in excess of the threshold 0.75 for the samples of size
T=240 ending in 2021, 2019, 2017 and 2015 }\label{TabRiskF}%
\end{table}%
\onehalfspacing

The strengths of the selected factors are also closely related to the average
measures of fit often used in the literature. Here we consider both
$Ave\bar{R}^{2}=n^{-1}\sum_{i=1}^{n}\bar{R}_{i}^{2}$, a simple average of the
fit of the individual return regressions adjusted for degrees of freedom,
$\bar{R}_{i}^{2}=1-(T-K-1)^{-1}\sum_{t=1}^{T}\hat{u}_{it}^{2}/T^{-1}\sum
_{t=1}^{T}\left(  r_{it}-\bar{r}_{i\circ}\right)  ^{2}$, and the adjusted
pooled $R^{2}$ defined by $\overline{PR}^{2}=1-\widehat{\bar{\sigma}}_{nT}%
^{2}/s_{r,nT}^{2}$, where $\widehat{\bar{\sigma}}_{nT}^{2}$ is the
bias-corrected estimator of $\bar{\sigma}_{n}^{2}$ defined by (\ref{AdjZig})
and $=$ $\left(  nT\right)  ^{-1}\sum_{i=1}^{n}\sum_{t=1}^{T}\left(
r_{it}-\bar{r}_{i\circ}\right)  ^{2}$. Both of these measures behave very
similarly, but the pooled version is less sensitive to outliers. As shown in
Appendix \ref{appPRsquared}, for sufficiently large $n$ and $T$,
$\overline{PR}^{2}$ is dominated by the contribution of the most strong
factor(s). Since the only strong factor selected is the market factor in Table
\ref{tab:RsquaredFAandPL}\ we report the $Ave\bar{R}^{2}$ and $\overline
{PR}^{2}$ in the case of return regressions which just include the market
factor and those which include all other factors with strength in excess of
$0.75$. First, we note that the $\overline{PR}^{2}$ values are generally lower
than the $Ave\bar{R}^{2}$. Second, the additional factors do add to the fit,
but their relative contributions vary considerably across sample sizes and
periods. In general, the marginal contribution of non-market factors tend to
be smaller when $T$ is larger, which is consistent with the theory for
adjusted pooled $R^{2}$ set out in Section \ref{appPRsquared} of the online
supplement A.%

\singlespacing
\begin{table}[htbp]%

\begin{center}%
\begin{tabular}
[c]{cccccccccccc}\hline\hline
\multicolumn{3}{c}{End Year} &  & \multicolumn{2}{c}{2021} &
\multicolumn{2}{c}{2019} & \multicolumn{2}{c}{2017} & \multicolumn{2}{c}{2015}%
\\\cline{5-12}%
\multicolumn{3}{c}{No. of stocks, $n$} &  & \multicolumn{2}{c}{1175} &
\multicolumn{2}{c}{1143} & \multicolumn{2}{c}{1132} & \multicolumn{2}{c}{1090}%
\\
\multicolumn{3}{c}{No. of selected factors} &  & \multicolumn{2}{c}{4} &
\multicolumn{2}{c}{6} & \multicolumn{2}{c}{7} & \multicolumn{2}{c}{7}\\\hline
&  &  &  & Mkt. & Selected & Mkt. & Selected & Mkt. & Selected & Mkt. &
Selected\\
&  & $T$ &  &  &  &  &  &  &  &  & \\
$Ave\bar{R}^{2}$ &  & 240 &  & 0.23 & 0.29 & 0.19 & 0.28 & 0.17 & 0.27 &
0.17 & 0.27\\
&  & 120 &  & 0.24 & 0.33 & 0.25 & 0.35 & 0.26 & 0.38 & 0.26 & 0.37\\
&  &  &  &  &  &  &  &  &  &  & \\
$\overline{PR}^{2}$ &  & 240 &  & 0.20 & 0.26 & 0.17 & 0.26 & 0.16 & 0.26 &
0.16 & 0.26\\
&  & 120 &  & 0.19 & 0.27 & 0.20 & 0.29 & 0.23 & 0.34 & 0.23 &
0.34\\\hline\hline
\end{tabular}

\end{center}

{\footnotesize \textit{Note:} This table shows for each of the four end years
and the two sample sizes the adjusted average and pooled }${\footnotesize R}%
^{2}${\footnotesize (}${\footnotesize Ave\bar{R}}^{2}${\footnotesize and
}$\overline{{\footnotesize PR}}^{{\footnotesize 2}}${\footnotesize ) for the
return regressions when using market factor alone or factors selected with
strength higher than }${\footnotesize 0.75}${\footnotesize . The list of
selected factors are given in Table \ref{TabRiskF}. }%
\caption{Average and pooled R squared  for the return regressions when using market factor alone or factors chosen by Pooled Lasso plus 0.75 threshold }\label{tab:RsquaredFAandPL}%
\end{table}%
\onehalfspacing

\subsection{Testing for non-zero $\boldsymbol{\phi}$}

In principle, if $\boldsymbol{\phi}\neq0$ there are potentially exploitable
excess returns. In practice, to construct an effective phi-portfolio a large
number of securities is required and rebalancing such long-short portfolios
for so many securities may not be feasible or may incur high transactions
costs. In addition, model uncertainty, estimation uncertainty, time variation
in both $\boldsymbol{\beta}_{i}$ and in conditional volatility pose additional
difficulties in implementing a strategy to exploit the potential returns
revealed by $\boldsymbol{\phi}$. We will abstract from such practical
difficulties to provide some indication of the performance of phi-portfolios
relative to alternatives which would face similar difficulties.

As our preferred asset pricing model, we consider the seven factors selected
by pooled Lasso using the sample ending in $2015m12,$ which we label as
PL7.\footnote{The selection of the PL7 model was reported in the earlier
version of the paper submitted for publication and was not informed by the
performance the phi-portfolio that we report in this version of the paper.}
Recall that the PL7 includes the 3 Fama-French factors (Mkt., SMB, and HML)
plus the four risk factors, Short Selling, CBOP, Beta Tail Risk, and Sin
Stocks. But given uncertainties that surround the problem of model selection
we also considered the two popular FF factor models, namely FF3 and FF5. The
latter augments FF3 with RMW (robust minus weak operating profitability) and
CMA (conservative minus aggressive investment portfolios). The three models
are estimated using twenty-year rolling windows covering the $84$ months from
$2015m12$ to $2022m11$, so that we can generate out of sample return forecasts
for the months $2016m1-2022m12$.\footnote{Due to entry and exit of securities
the number of securities included in our analysis varied across the rolling
sample periods. We started with $n=1,090$ securities for the first rolling
sample ending in $2015m12$, with the number of available securities with $240$
months of data falling to $953$ by $2017$, $838$ by $2019$, $767$ by $2021$,
and $736$ by $2022$.} For all $3\times84$ model-sample interactions we
computed the following Wald test statistics
\[
W_{t\left\vert T\right.  }^{2}=\boldsymbol{\tilde{\phi}}_{t\left\vert
T\right.  }^{\prime}\left[  \widehat{Var\left(  \boldsymbol{\tilde{\phi}%
}_{t\left\vert T\right.  }\right)  }\right]  ^{-1}\boldsymbol{\tilde{\phi}%
}_{t\left\vert T\right.  },
\]
for testing the null hypothesis $\boldsymbol{\phi}=0,$ where
$\boldsymbol{\tilde{\phi}}_{t\left\vert T\right.  }$ and $\widehat{Var\left(
\boldsymbol{\tilde{\phi}}_{t\left\vert T\right.  }\right)  }$ denote the
rolling versions of (\ref{phitilda}) and (\ref{Varphitilt}).\footnote{The
formulae for the rolling estimates are provided in the sub-section \ref{Roll}
of the online appendix.} For the PL7 the range of the test statistic was from
141.8 to 25.8, as compared to the 5 per cent $\chi^{2}(7)$ critical value of
14.07. Thus the hypothesis that $\boldsymbol{\phi}=0$ is strongly rejected in
all the 84 rolling sample for PL7. This is also true for the FF5 and FF3
models where the test statistic ranged from 162.7 to 25.0, and from 67.9 to
25.8, compared to the 5 per cent critical values of 11.07 and 7.8, respectively.

The rolling values of the Wald statistics for testing $\boldsymbol{\phi}=0$
for the three models are shown in Figure 1. The horizontal line (pink)
represents the critical value of the $\chi_{7}^{2}$ distribution at the 5 per
cent level. The time profiles of these test statistics clearly show that
$\boldsymbol{\phi}=0$ is rejected for all rolling samples and for all three
models. But there is also a clear downward trend showing that the evidence
against $\boldsymbol{\phi}=0$ has been getting weaker over time, irrespective
of the choice of asset pricing model.%

\begin{figure}[ptb]%
\centering
\caption{Rolling chi-squared statistics for testing $\boldsymbol{\phi}=0$
using a window of size $240$ for FF3, FF5, and PL7 models. }%
\includegraphics[
height=2.8357in,
width=4.5861in
]%
{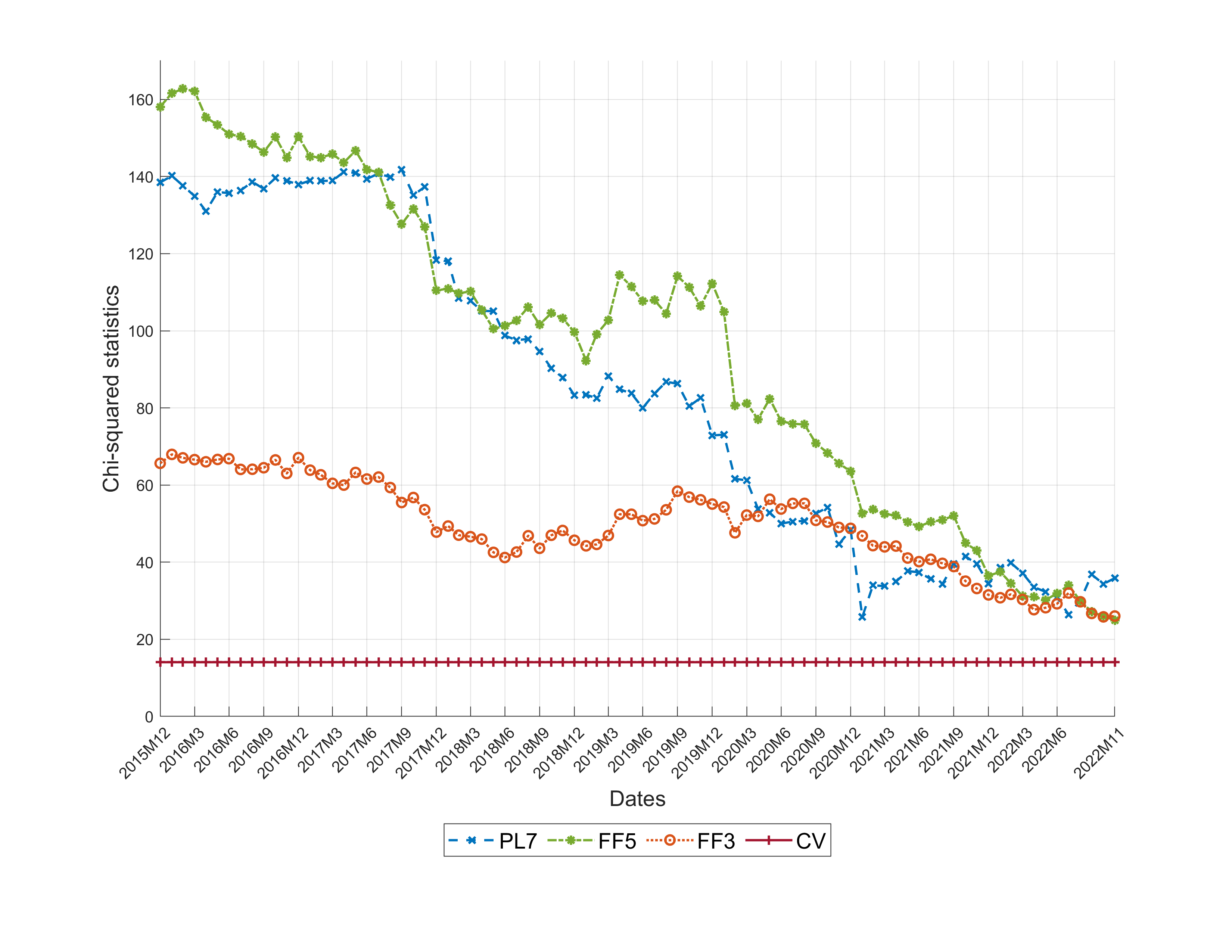}%
\end{figure}
\bigskip

\subsection{Comparative performance of phi and MV portfolios}

Having established that most likely $\boldsymbol{\phi\neq0}$ for the asset
pricing models we have considered, we now turn to the performance of
phi-portfolios based on these models. Using the recursive version of the
phi-portfolio given by (\ref{phiport}), we consider the following
phi-portfolio returns%

\[
\hat{\rho}_{t+1,\phi}=\boldsymbol{\tilde{\phi}}_{t\left\vert T\right.
}\left[  \left(  \mathbf{\hat{B}}_{t\left\vert T\right.  }^{\prime}%
\mathbf{M}_{n}\mathbf{\hat{B}}_{t\left\vert T\right.  }\right)  ^{-1}%
\mathbf{\hat{B}}_{t\left\vert T\right.  }\mathbf{M}_{n}\mathbf{r}_{\circ
,t+1}-\mathbf{f}_{t+1}\right]  ,
\]
for $t=2016m1,2016m2,...,2022m12$, using the rolling estimates $\mathbf{\hat
{B}}_{t\left\vert T\right.  }=\left(  \mathbf{\hat{\beta}}_{1t\left\vert
T\right.  },\mathbf{\hat{\beta}}_{2t\left\vert T\right.  }....,\mathbf{\hat
{\beta}}_{nt\left\vert T\right.  }\right)  ^{\prime}$ with $T=240$, for each
of the three factor models, FF3, FF5 and PL7. We compare the annualised Sharpe
ratios of phi-portfolios with the ones based on associated MV\ portfolios,
given by $\rho_{t+1,MV}=\boldsymbol{\mu}_{R}^{\prime}\mathbf{V}_{R}%
^{-1}\mathbf{r}_{\circ,t+1}$.\footnote{Given our focus on the Sharpe ratios,
we have set the scaling of the MV portfolio to unity.} Although in principle,
MV portfolios can be constructed without a reference to a particular factor
model, reliable estimation of $\boldsymbol{\mu}_{R}$ and $\mathbf{V}_{R}^{-1}$
are challenging when $n$ is relatively large. For example, the rolling sample
covariance matrix estimator of $\mathbf{V}_{R}$, given by $\mathbf{\mathring
{V}}_{R,t\left\vert T\right.  }=T^{-1}\sum_{\tau=t-T+1}^{t}\left(
\mathbf{r}_{\circ,\tau}-\mathbf{\bar{r}}_{\circ,t\left\vert T\right.
}\right)  \left(  \mathbf{r}_{\circ,\tau}-\mathbf{\bar{r}}_{\circ,t\left\vert
T\right.  }\right)  ^{\prime}$, with $\mathbf{\bar{r}}_{\circ,t\left\vert
T\right.  }=T^{-1}\sum_{\tau=t-T+1}^{t}\mathbf{r}_{\circ,\tau}$, will be
singular when $n>T$, and can be very poorly estimated if $T$ is not
sufficiently large relative to $n$. There is a vast literature on consistent
estimation of high dimensional covariance matrices like $\mathbf{V}_{R}.$
\cite{fan2011high} use observed factors while \cite{fan2013large} use
principal components to filter out the effects of strong factors, in both
cases assuming $\mathbf{V}_{u}$ is sparse, and then using a threshold method
to estimate it. Shrinkage estimators of $\mathbf{V}_{R}$ are also proposed in
the literature with a recent survey provided by \cite{LedoitWolfsurvey2022}.
However, the shrinkage estimators require $n$ and $T$ to be of the same order
of magnitude and do not work well when $n$ is much larger than $T$, as in the
present application. We follow \cite{fan2011high} and base our estimation of
$\boldsymbol{\mu}_{R}$ and $\mathbf{V}_{R}$ on the same factor model used to
construct the phi-portfolio. For a given factor model, characterized by
$c\mathbf{,B}$, and $\mathbf{F}$, we compute the MV portfolio returns as%

\[
\hat{\rho}_{t+1.MV}=\mathbf{\hat{\mu}}_{R,t\left\vert T\right.  }^{\prime
}\mathbf{\hat{V}}_{R,t\left\vert T\right.  }^{-1}\mathbf{r}_{\circ,t+1},
\]
where $\mathbf{\hat{\mu}}_{R,t\left\vert T\right.  }=\hat{c}_{t\left\vert
T\right.  }\mathbf{\tau}_{n}+\mathbf{\hat{B}}_{t\left\vert T\right.
}\boldsymbol{\tilde{\lambda}}_{t\left\vert T\right.  }$, $\boldsymbol{\tilde
{\lambda}}_{t\left\vert T\right.  }=\boldsymbol{\tilde{\phi}}_{t\left\vert
T\right.  }+\boldsymbol{\hat{\mu}}_{t\left\vert T\right.  },$
\[
\mathbf{\hat{V}}_{R,t\left\vert T\right.  }=\mathbf{\hat{B}}_{t\left\vert
T\right.  }^{\prime}\left(  \frac{\mathbf{F}_{t\left\vert T\right.  }^{\prime
}\mathbf{M}_{T}\mathbf{F}_{t\left\vert T\right.  }}{T}\right)  ^{-1}%
\mathbf{\hat{B}}_{t\left\vert T\right.  }\mathbf{+\mathbf{\tilde{V}}%
}_{u,t\left\vert T\right.  }\mathbf{,}%
\]
$\boldsymbol{\hat{\mu}}_{t\left\vert T\right.  }=\boldsymbol{\bar{f}%
}_{t\left\vert T\right.  }=T^{-1}\sum_{\tau=t-T+1}^{t}\mathbf{f}_{\tau}$, and
$\mathbf{\mathbf{\tilde{V}}}_{u,t\left\vert T\right.  }$ is the rolling
estimate of $\mathbf{V}_{u}$. The algorithms used to compute the recursive
estimates for the MV portfolio can be found in the sub-section \ref{Roll} in
the online supplement A.

Table \ref{tab:portfolios} presents annualised SR of the phi-portfolios, for
the FF3, FF5, and PL7 models, and their corresponding MV portfolios for two
samples, both beginning in 2016m1, one ending in 2019m12, pre Covid-19, and
one ending in 2022m12. \ In five of the six SR ratios reported in this table,
the phi-portfolio has a higher SR than the corresponding MV portfolio. The
exception is the SR associated to the FF5 model for the pre Covid-19 sample.
This illustrates that if $\boldsymbol{\phi}\neq\mathbf{0},$ it is possible to
construct portfolios that outperform the mean-variance\textbf{ } portfolio.
Amongst the 3 models considered, the phi-portfolio based on the PL7 model
performed best during the pre Covid-19 and the full sample, even beating the
S\&P 500. The Sharpe ratio of the phi-portfolio based on the Pl7 model was
1.95 compared to 0.94 for the S\&P 500 during the pre Covid-19 period, and
fell sharply to 0.65 for the full sample as compared to 0.58 for the S\&P 500.
The SR for the MV portfolio using the same model were 0.87 and 0.42 for the
two samples, respectively.

The sharp decline in the SRs as we add the post Covid-19 years is in line with
the strong downward trend in the Wald statistics for the test of
$\boldsymbol{\phi}=\mathbf{0}$ shown in Figure 1. As is well known SRs have
large standard errors, and in the case of our application that are around 1,
so none of the SRs are significantly different from zero, with the possible
except of the largest SR of 1.95 for phi-portfolio based on PL7 model for the
pre Covid-19 period. We also note that the reported Sharpe ratios do not allow
for transaction costs and the fact that shorting might not be feasible for all
the securities.%

\begin{table}[H]%
%

\caption
{Annualised Sharpe ratios of realized monthly returns for alternative portfolios based on 240 month rolling window estimates.}%
%

\begin{centering}%
%

\begin{tabular}
[c]{cccccc}\hline\hline
& \multicolumn{2}{c}{Mean-variance portfolios} &  &
\multicolumn{2}{c}{Phi-portfolios}\\\cline{2-3}\cline{5-6}
& 2016$m$1 - 2019$m$12 & 2016$m$1 - 2022$m$12 &  & 2016$m$1-2019$m$12 &
2016$m$1-2022$m$12\\
Models & Pre Covid-19 & Full sample &  & Pre Covid-19 & Full sample\\\hline
&  &  &  &  & \\
FF3 & 0.57 & 0.35 &  & 0.86 & 0.51\\
FF5 & 0.49 & 0.43 &  & 0.35 & 0.57\\
Lasso 7 & 0.87 & 0.42 &  & 1.95 & 0.65\\\hline\hline
\end{tabular}
%

\end{centering}%
\label{tab:portfolios}%
\vspace{3mm}%

\emph{{\footnotesize {}\textit{Note}}}{\footnotesize {}: The annualised Sharpe
ratio (SR) is computed as }$\sqrt{12}\bar{\rho}/s_{\rho}${\footnotesize ,
where }$\bar{\rho}$ {\footnotesize is the mean of monthly returns, and
}$s_{\rho}$ {\footnotesize is the standard deviation of monthly returns. For
comparison the SR of the monthly returns on S\&P 500 were 0.94 and 0.58 over
the periods 2016m1-2019m12, and 2016m1-2022m12, respectively.}%

\end{table}%

\section{Concluding remarks\label{Conclusion}%
\onehalfspacing
}

In this paper we have highlighted the importance of decomposing the risk
premia, $\boldsymbol{\lambda}$, into the the factor mean, $\boldsymbol{\mu}$,
and $\boldsymbol{\phi}$, and writing the alpha of security $i$,
$\mathit{\alpha}_{i}$, in terms of $\boldsymbol{\phi}$ and the idiosyncratic
pricing errors. We have shown that when $\boldsymbol{\phi\neq0}$, it is
possible to construct a portfolio, denoted as phi-portfolio, that dominates
the associated mean-variance portfolio when the number securities, $n$, is
sufficiently large and the risk factors are sufficiently strong. Given the
pivotal role played by $\boldsymbol{\phi}$ for estimating the risk premia, for
formation of large $n$ portfolios, and for tests of market efficiency, we have
focussed on estimation of $\boldsymbol{\phi}$, and its asymptotic distribution
under quite a general setting that allows for missing factors and
idiosyncratic pricing errors. Since factor means, $\boldsymbol{\mu}$, can be
estimated at the regular rate of $T^{-1/2}$ from time series data, it is
relatively straightforward to develop a mixed strategy for estimation of
$\boldsymbol{\lambda}$ by adding a time series estimate of $\boldsymbol{\mu}$
to the bias-corrected estimator of $\boldsymbol{\phi}$. If we use the same
time series sample, such an estimator reduces to the Shanken bias-corrected
estimator of $\boldsymbol{\lambda}$. But in practice, given the concern over
the instability of the factor loadings $\beta_{ik}$ over time, one could use
relatively long time series, say $T_{\mu}$, when estimating $\boldsymbol{\mu}%
$, and a shorter time series, say $T_{\phi}<T_{\mu}$, when estimating
$\boldsymbol{\phi}$. The distributional and small sample properties of such a
mixed estimator of risk premia is a topic for further research.

Our theoretical and Monte Carlo results further highlight the important role
played by factor strengths in estimation and inference on $\boldsymbol{\phi}$,
and hence on $\boldsymbol{\lambda}$. For a fixed $T_{\phi}$, factors with
strength below $2/3$ lead to estimates of $\boldsymbol{\phi}$ \ with
convergence rate of $n^{-1/3}$ or worse, and their use in asset pricing models
can be justified only when $n$ is very large. We have also shown that weak
factors, with strength below $1/2$ are best treated as missing and absorbed in
the error term. We have shown that estimation of $\boldsymbol{\phi}$ for
strong or semi-strong factors is robust to weak missing factors, and the
explicit inclusion of weak factors in the empirical analysis is likely to have
adverse spill over effects on the estimates of $\boldsymbol{\phi\ }$\ for
strong and semi-strong factors. In view of these results we have proposed a
factor selection procedure where only factors with strength above $1/2+1/3$
are included in the asset pricing model. Developing a formal statistical
theory for the proposed selection is another topic for future research.

The paper also provides an empirical application to a large number of U.S.
securities with risk factors selected from a large number of potential risk
factors according to their strength, and use a pooled Lasso approach to select
$7$ risk factors out of over $180$. We find strong statistical evidence
against $\boldsymbol{\phi=0}$ for the selected model as well as for the two
popular Fama-French models (FF3 and FF5). Using rolling estimates of
$\boldsymbol{\phi}$ we also construct phi portfolios with better Sharpe ratios
as compared to associated mean-variance\textbf{ } portfolios. But we also warn
that these portfolio comparisons are preliminary and need to be further
investigated by allowing for transaction costs, and the feasibility of the
long-short trading strategies that are involved.

\pagebreak%

\singlespacing
\bibliographystyle{econometrica}
\bibliography{pslapmainV2}
\newpage%

\appendix
%

\numberwithin{equation}{section}%
%

\numberwithin{lemma}{section}%
%

\setcounter{page}{1}
\renewcommand{\thepage}{A\arabic{page}}%

\section{{\protect\small Mathematical Appendix}}

\subsection{{\protect\small Introduction}\label{Appintro}}%

\onehalfspacing

{\small In this mathematical appendix we first introduce the notations used in
our mathematical treatment, and state and establish a number of lemmas. We
then provide detailed proofs of Theorems \ref{TFMbias}, \ref{Thzig},
\ref{Tfi}, \ref{Tsemi} and \ref{Var} presented in the paper. }

{\small \textbf{Notations: }Generic positive finite constants are denoted by
$C$ when large, and $c$ when small. They can take different values at
different instances. $\lambda_{\max}\left(  \boldsymbol{A}\right)  $ and
$\lambda_{\min}\left(  \boldsymbol{A}\right)  $ denote the maximum and minimum
eigenvalues of $\boldsymbol{A}$. $\boldsymbol{A}>0$ denotes that
$\boldsymbol{A}$ is a positive definite matrix. $\left\Vert \boldsymbol{A}%
\right\Vert =\lambda_{\max}^{1/2}(\mathbf{A}^{\prime}\mathbf{A)}$, $\left\Vert
\boldsymbol{A}\right\Vert _{F}=\left[  Tr(\mathbf{A}^{\prime}\mathbf{A)}%
\right]  ^{1/2}$, $\left\Vert \boldsymbol{A}\right\Vert _{p}=\left(
E\left\Vert \boldsymbol{A}\right\Vert ^{p}\right)  ^{1/p}$, for $p\geq2$
denote spectral, Frobenius, and $\ell_{p}$ norms of matrix $\boldsymbol{A}$,
respectively. If $\left\{  f_{n}\right\}  _{n=1}^{\infty}$ is any real
sequence and $\left\{  g_{n}\right\}  _{n=1}^{\infty}$ is a sequences of
positive real numbers, then $f_{n}=O(g_{n})$, if there exists $C$ such that
$\left\vert f_{n}\right\vert /g_{n}\leq C$ for all $n$. $f_{n}=o(g_{n})$ if
$f_{n}/g_{n}\rightarrow0$ as $n\rightarrow\infty$. Similarly, $f_{n}%
=O_{p}(g_{n})$ if $f_{n}/g_{n}$ is stochastically bounded, and $f_{n}%
=o_{p}(g_{n}),$ if $f_{n}/g_{n}\rightarrow_{p}0$, where $\rightarrow_{p}%
$denotes convergence in probability. If $\left\{  f_{n}\right\}
_{n=1}^{\infty}$ and $\left\{  g_{n}\right\}  _{n=1}^{\infty}$ are both
positive sequences of real numbers, then $f_{n}=\ominus\left(  g_{n}\right)  $
if there exists $n_{0}\geq1$ such that $\inf_{n\geq n_{0}}\left(  f_{n}%
/g_{n}\right)  \geq C,$ and $\sup_{n\geq n_{0}}\left(  f_{n}/g_{n}\right)
\leq C$. }

\subsection{{\protect\small Statement of lemmas and their proofs.}}

\begin{lemma}
\label{Sharpe}{\small Suppose the linear factor pricing model (\ref{ritV})
holds subject to the restrictions $c=0$ and $\mathbf{\eta=0}$, $\left\Vert
\mathbf{\lambda}\right\Vert <C$, and $\lambda_{min}\left(  \mathbf{\Sigma}%
_{f}^{-1}\right)  >0$. Consider the $K\times1\,$\ beta-based portfolios,
$\mathbf{\rho}_{B,t}=$ $\mathbf{W}_{B}^{\prime}\mathbf{r}_{\circ t}$, formed
from the factor loadings, $\mathbf{B}_{n}$, where $\mathbf{W}_{B}^{\prime
}=\left(  \mathbf{B}_{n}^{\prime}\mathbf{V}_{u}^{-1}\mathbf{B}_{n}\right)
^{-1}\mathbf{B}_{n}^{\prime}\mathbf{V}_{u}^{-1}$. Then the optimal portfolio
formed using the $K$ beta-based portfolios, $\boldsymbol{\rho}_{B,t}$, has the
same Sharpe ratio as the mean-variance portfolio, given by $\rho
_{MV,t}=\boldsymbol{\mu}_{R}^{\prime}\mathbf{V}_{R}^{-1}\mathbf{r}_{\circ t}$,
where $\boldsymbol{\mu}_{R}=\mathbf{B}_{n}\boldsymbol{\lambda}$ and
$\mathbf{V}_{R}=\mathbf{B}_{n}\mathbf{\Sigma}_{f}\mathbf{B}_{n}^{\prime
}+\mathbf{V}_{u}$:
\begin{equation}
SR_{B}^{2}=SR_{MV}^{2}\leq\mathbf{\lambda}^{\prime}\mathbf{\Sigma}_{f}%
^{-1}\boldsymbol{\lambda}. \label{eqSR}%
\end{equation}
Further if the factors are sufficiently strong such that $\lambda_{min}\left(
\mathbf{B}_{n}^{\prime}\mathbf{V}_{u}^{-1}\mathbf{B}_{n}\right)
\rightarrow\infty$, then%
\begin{equation}
SR_{B}^{2}=SR_{MV}^{2}\rightarrow\mathbf{\lambda}^{\prime}\mathbf{\Sigma}%
_{f}^{-1}\boldsymbol{\lambda}\text{, as }n\rightarrow\infty\text{.}
\label{limSR}%
\end{equation}
}
\end{lemma}

\begin{proof}
{\small Under $c=0$ and $\mathbf{\eta=0}$}, {\small $\mathbf{r}_{\circ
t}=\mathbf{B}\left(  \mathbf{f}_{t}+\boldsymbol{\phi}\right)  +\mathbf{u}%
_{\circ t}$, $\mathbf{\rho}_{B,t}=\mathbf{f}_{t}+\boldsymbol{\phi}%
+\mathbf{W}^{\prime}\mathbf{u}_{\circ t}$, and
\[
E\left(  \mathbf{\rho}_{B,t}\right)  =\boldsymbol{\lambda}\text{, and
}Var\left(  \mathbf{\rho}_{B,t}\right)  =\boldsymbol{\Sigma}_{f}%
+\mathbf{W}^{\prime}\mathbf{V}_{u}\mathbf{W}=\boldsymbol{\Sigma}_{f}+\left(
\mathbf{B}^{\prime}\mathbf{V}_{u}^{-1}\mathbf{B}\right)  ^{-1}.
\]
The best linear combination of these $K$ portfolios is obtained by finding the
$K\times1$ weight vector, $\mathbf{w}_{f},$ that minimizes $Var\left(
\mathbf{w}_{f}^{\prime}\mathbf{\rho}_{B,t}\right)  $ for a given mean,
$\mathbf{w}_{f}^{\prime}\boldsymbol{\lambda}$. The solution to this
optimization problem is given by }
\end{proof}

{\small
\[
\mathbf{w}_{f}=\kappa^{-1}\left(  \boldsymbol{\Sigma}_{f}+\left(
\mathbf{B}_{n}^{\prime}\mathbf{V}_{u}^{-1}\mathbf{B}_{n}\right)  ^{-1}\right)
^{-1}\boldsymbol{\lambda},
\]
where $\kappa$ is a risk aversion coefficient. The squared Sharpe ratio of
$\rho_{B,t}=\mathbf{w}_{f}^{\prime}\mathbf{\rho}_{B,t}$ is given by
\begin{equation}
SR_{B}^{2}=\boldsymbol{\lambda}^{\prime}\left(  \mathbf{\Sigma}_{f}+\left(
\mathbf{B}_{n}^{\prime}\mathbf{V}_{u}^{-1}\mathbf{B}_{n}\right)  ^{-1}\right)
^{-1}\boldsymbol{\lambda}, \label{SR2B1}%
\end{equation}
which can be written equivalently as%
\begin{align}
SR_{B}^{2}  &  =\boldsymbol{\lambda}^{\prime}\left[  \left(  \mathbf{B}%
_{n}^{\prime}\mathbf{V}_{u}^{-1}\mathbf{B}_{n}\right)  ^{-1}\left(
\mathbf{\Sigma}_{f}^{-1}+\mathbf{B}_{n}^{\prime}\mathbf{V}_{u}^{-1}%
\mathbf{B}_{n}\right)  \mathbf{\Sigma}_{f}\right]  ^{-1}\boldsymbol{\lambda
}\label{SR2B}\\
&  =\boldsymbol{\lambda}^{\prime}\mathbf{\Sigma}_{f}^{-1}\left(
\mathbf{\Sigma}_{f}^{-1}+\mathbf{B}_{n}^{\prime}\mathbf{V}_{u}^{-1}%
\mathbf{B}_{n}\right)  ^{-1}\left(  \mathbf{B}_{n}^{\prime}\mathbf{V}_{u}%
^{-1}\mathbf{B}_{n}\right)  \boldsymbol{\lambda}\nonumber
\end{align}
The squared Sharpe ratio of the mean-variance efficient portfolio is given by
\[
SR_{MV}^{2}=\boldsymbol{\lambda}^{\prime}\mathbf{B}_{n}^{\prime}\left(
\mathbf{B}_{n}\mathbf{\Sigma}_{f}\mathbf{B}_{n}^{\prime}+\mathbf{V}%
_{u}\right)  ^{-1}\mathbf{B}_{n}\boldsymbol{\lambda}.
\]
However, since $\mathbf{B}_{n}\mathbf{\Sigma}_{f}\mathbf{B}_{n}^{\prime}$ is
rank deficient then}

{\small
\[
\left(  \mathbf{B}_{n}\mathbf{\Sigma}_{f}\mathbf{B}_{n}^{\prime}%
+\mathbf{V}_{u}\right)  ^{-1}\mathbf{=V}_{u}^{-1}\mathbf{-V}_{u}%
^{-1}\mathbf{B}\left(  \mathbf{\Sigma}_{f}^{-1}+\mathbf{B}_{n}^{\prime
}\mathbf{V}_{u}^{-1}\mathbf{B}_{n}\right)  ^{-1}\mathbf{B}_{n}^{\prime
}\mathbf{V}_{u}^{-1},
\]
and%
\begin{align*}
SR_{MV}^{2}  &  =\boldsymbol{\lambda}^{\prime}\mathbf{B}_{n}^{\prime}\left[
\mathbf{V}_{u}^{-1}\mathbf{-V}_{u}^{-1}\mathbf{B}\left(  \mathbf{\Sigma}%
_{f}^{-1}+\mathbf{B}_{n}^{\prime}\mathbf{V}_{u}^{-1}\mathbf{B}_{n}\right)
^{-1}\mathbf{B}_{n}^{\prime}\mathbf{V}_{u}^{-1}\right]  \mathbf{B}%
_{n}\boldsymbol{\lambda}\\
&  =\boldsymbol{\lambda}^{\prime}\mathbf{B}_{n}^{\prime}\mathbf{V}_{u}%
^{-1}\mathbf{B}_{n}\boldsymbol{\lambda}-\boldsymbol{\lambda}^{\prime
}\mathbf{B}_{n}^{\prime}\mathbf{V}_{u}^{-1}\mathbf{B}_{n}\left(
\mathbf{\Sigma}_{f}^{-1}+\mathbf{B}_{n}^{\prime}\mathbf{V}_{u}^{-1}%
\mathbf{B}_{n}\right)  ^{-1}\mathbf{B}_{n}^{\prime}\mathbf{V}_{u}%
^{-1}\mathbf{B}_{n}\boldsymbol{\lambda}%
\end{align*}
Furthermore,%
\begin{align*}
&  \boldsymbol{\lambda}^{\prime}\mathbf{B}_{n}^{\prime}\mathbf{V}_{u}%
^{-1}\mathbf{B}_{n}\left(  \mathbf{\Sigma}_{f}^{-1}+\mathbf{B}_{n}^{\prime
}\mathbf{V}_{u}^{-1}\mathbf{B}_{n}\right)  ^{-1}\mathbf{B}_{n}^{\prime
}\mathbf{V}_{u}^{-1}\mathbf{B}_{n}\boldsymbol{\lambda}\\
&  =\boldsymbol{\lambda}^{\prime}\left(  \mathbf{B}_{n}^{\prime}\mathbf{V}%
_{u}^{-1}\mathbf{B}_{n}+\mathbf{\Sigma}_{f}^{-1}-\mathbf{\Sigma}_{f}%
^{-1}\right)  \left(  \mathbf{\Sigma}_{f}^{-1}+\mathbf{B}_{n}^{\prime
}\mathbf{V}_{u}^{-1}\mathbf{B}_{n}\right)  ^{-1}\mathbf{B}_{n}^{\prime
}\mathbf{V}_{u}^{-1}\mathbf{B}_{n}\boldsymbol{\lambda}\\
&  =\boldsymbol{\lambda}^{\prime}\mathbf{B}_{n}^{\prime}\mathbf{V}_{u}%
^{-1}\mathbf{B}_{n}\boldsymbol{\lambda}-\boldsymbol{\lambda}^{\prime
}\mathbf{\Sigma}_{f}^{-1}\left(  \mathbf{\Sigma}_{f}^{-1}+\mathbf{B}%
_{n}^{\prime}\mathbf{V}_{u}^{-1}\mathbf{B}_{n}\right)  ^{-1}\mathbf{B}%
_{n}^{\prime}\mathbf{V}_{u}^{-1}\mathbf{B}_{n}\boldsymbol{\lambda}.
\end{align*}
Hence%
\begin{equation}
SR_{MV}^{2}=\boldsymbol{\lambda}^{\prime}\mathbf{\Sigma}_{f}^{-1}\left(
\mathbf{\Sigma}_{f}^{-1}+\mathbf{B}_{n}^{\prime}\mathbf{V}_{u}^{-1}%
\mathbf{B}_{n}\right)  ^{-1}\mathbf{B}_{n}^{\prime}\mathbf{V}_{u}%
^{-1}\mathbf{B}_{n}\boldsymbol{\lambda}. \label{SR2MV}%
\end{equation}
The first part of result (\ref{eqSR}) now follows from a direct comparison of
(\ref{SR2MV}) and (\ref{SR2B}). To establish the second part note that
\begin{align*}
SR_{MV}^{2}-\boldsymbol{\lambda}^{\prime}\mathbf{\Sigma}_{f}^{-1}%
\boldsymbol{\lambda}  &  =\boldsymbol{\lambda}^{\prime}\left[  \mathbf{\Sigma
}_{f}^{-1}\left(  \mathbf{\Sigma}_{f}^{-1}+\mathbf{B}_{n}^{\prime}%
\mathbf{V}_{u}^{-1}\mathbf{B}_{n}\right)  ^{-1}\mathbf{B}_{n}^{\prime
}\mathbf{V}_{u}^{-1}\mathbf{B}_{n}-\mathbf{\Sigma}_{f}^{-1}\right]
\boldsymbol{\lambda}\\
&  =\boldsymbol{\lambda}^{\prime}\left[  \mathbf{\Sigma}_{f}^{-1}\left(
\mathbf{\Sigma}_{f}^{-1}+\mathbf{B}_{n}^{\prime}\mathbf{V}_{u}^{-1}%
\mathbf{B}_{n}\right)  ^{-1}\left(  \mathbf{B}_{n}^{\prime}\mathbf{V}_{u}%
^{-1}\mathbf{B}_{n}+\mathbf{\Sigma}_{f}^{-1}-\mathbf{\Sigma}_{f}^{-1}\right)
-\mathbf{\Sigma}_{f}^{-1}\right]  \boldsymbol{\lambda}\\
&  =\boldsymbol{\lambda}^{\prime}\left[  \mathbf{\Sigma}_{f}^{-1}\left(
\mathbf{I}_{k}-\left(  \mathbf{\Sigma}_{f}^{-1}+\mathbf{B}_{n}^{\prime
}\mathbf{V}_{u}^{-1}\mathbf{B}_{n}\right)  ^{-1}\mathbf{\Sigma}_{f}%
^{-1}\right)  -\mathbf{\Sigma}_{f}^{-1}\right]  \boldsymbol{\lambda}\\
&  =-\boldsymbol{\lambda}^{\prime}\mathbf{\Sigma}_{f}^{-1}\left(
\mathbf{\Sigma}_{f}^{-1}+\mathbf{B}_{n}^{\prime}\mathbf{V}_{u}^{-1}%
\mathbf{B}_{n}\right)  ^{-1}\mathbf{\Sigma}_{f}^{-1}\boldsymbol{\lambda}.
\end{align*}
Since $\boldsymbol{\lambda}^{\prime}\mathbf{\Sigma}_{f}^{-1}\left(
\mathbf{\Sigma}_{f}^{-1}+\mathbf{B}_{n}^{\prime}\mathbf{V}_{u}^{-1}%
\mathbf{B}_{n}\right)  ^{-1}\mathbf{\Sigma}_{f}^{-1}\boldsymbol{\lambda}\geq
0$, it then follows that $SR_{MV}^{2}\leq\boldsymbol{\lambda}^{\prime
}\mathbf{\Sigma}_{f}^{-1}\boldsymbol{\lambda}$. To establish result
(\ref{limSR}), we first note that
\[
\left\vert SR_{MV}^{2}-\boldsymbol{\lambda}^{\prime}\mathbf{\Sigma}_{f}%
^{-1}\boldsymbol{\lambda}\right\vert \leq\left\Vert \boldsymbol{\lambda
}\right\Vert ^{2}\left\Vert \mathbf{\Sigma}_{f}^{-1}\right\Vert ^{2}\left\Vert
\left(  \mathbf{\Sigma}_{f}^{-1}+\mathbf{B}_{n}^{\prime}\mathbf{V}_{u}%
^{-1}\mathbf{B}_{n}\right)  ^{-1}\right\Vert .
\]
where $\left\Vert \boldsymbol{\lambda}\right\Vert <C$, and $\left\Vert
\mathbf{\Sigma}_{f}^{-1}\right\Vert ^{2}=\lambda_{max}\left(  \mathbf{\Sigma
}_{f}^{-1}\right)  =1/\lambda_{min}\left(  \mathbf{\Sigma}_{f}\right)  <C$,
since by Assumption $\lambda_{min}\left(  \mathbf{\Sigma}_{f}\right)  >0$.
Also
\begin{align*}
\left\Vert \left(  \mathbf{\Sigma}_{f}^{-1}+\mathbf{B}_{n}^{\prime}%
\mathbf{V}_{u}^{-1}\mathbf{B}_{n}\right)  ^{-1}\right\Vert  &  =\lambda
_{\text{max}}\left[  \left(  \mathbf{\Sigma}_{f}^{-1}+\mathbf{B}_{n}^{\prime
}\mathbf{V}_{u}^{-1}\mathbf{B}_{n}\right)  ^{-1}\right] \\
&  =\frac{1}{\lambda_{\text{min}}\left(  \mathbf{\Sigma}_{f}^{-1}%
+\mathbf{B}_{n}^{\prime}\mathbf{V}_{u}^{-1}\mathbf{B}_{n}\right)  }.
\end{align*}
Since $\mathbf{\Sigma}_{f}^{-1}$ and $\mathbf{B}_{n}^{\prime}\mathbf{V}%
_{u}^{-1}\mathbf{B}_{n}$ are both symmetric matrices, then (see Section 5.3.2
in Lutkepohl, 1996)
\[
\lambda_{\text{min}}\left(  \mathbf{\Sigma}_{f}^{-1}+\mathbf{B}_{n}^{\prime
}\mathbf{V}_{u}^{-1}\mathbf{B}_{n}\right)  \geq\lambda_{\text{min}}\left(
\mathbf{\Sigma}_{f}^{-1}\right)  +\lambda_{min}\left(  \mathbf{B}_{n}^{\prime
}\mathbf{V}_{u}^{-1}\mathbf{B}_{n}\right)  ,
\]
and
\[
\left\Vert \left(  \mathbf{\Sigma}_{f}^{-1}+\mathbf{B}_{n}^{\prime}%
\mathbf{V}_{u}^{-1}\mathbf{B}_{n}\right)  ^{-1}\right\Vert \leq\frac
{1}{\lambda_{\text{min}}\left(  \mathbf{\Sigma}_{f}^{-1}\right)
+\lambda_{min}\left(  \mathbf{B}_{n}^{\prime}\mathbf{V}_{u}^{-1}\mathbf{B}%
_{n}\right)  }.
\]
Hence
\[
\left\Vert SR_{MV}^{2}-\boldsymbol{\lambda}^{\prime}\mathbf{\Sigma}_{f}%
^{-1}\boldsymbol{\lambda}\right\Vert \leq\frac{\left\Vert \mathbf{\lambda
}\right\Vert ^{2}\left\Vert \mathbf{\Sigma}_{f}^{-1}\right\Vert ^{2}}%
{\lambda_{\text{min}}\left(  \mathbf{\Sigma}_{f}^{-1}\right)  +\lambda
_{min}\left(  \mathbf{B}_{n}^{\prime}\mathbf{V}_{u}^{-1}\mathbf{B}_{n}\right)
},
\]
$\left\Vert SR_{MV}^{2}-\boldsymbol{\lambda}^{\prime}\mathbf{\Sigma}_{f}%
^{-1}\boldsymbol{\lambda}\right\Vert \rightarrow0$, if $\lambda_{min}\left(
\mathbf{B}_{n}^{\prime}\mathbf{V}_{u}^{-1}\mathbf{B}_{n}\right)
\rightarrow\infty$, and result (\ref{limSR}) follows. }

\begin{lemma}
\label{L1}{\small Consider the errors $\left\{  u_{it}%
,i=1,2,...,n;t=1,2,...,T\right\}  $ defined by (\ref{uit}) and suppose that
Assumptions \ref{Latent factor} and \ref{Errors} hold with $\alpha_{\gamma
}<1/2$, and $\mathbf{V}_{u}=(\sigma_{ij})$. Set $\overline{u}_{i\circ}%
=T^{-1}\sum_{t=1}^{T}u_{it},$ and $E(u_{it}^{2})=\sigma_{i}^{2}$. Then
\begin{equation}
\left\Vert \mathbf{V}_{u}\right\Vert =\lambda_{\max}\left(  \mathbf{V}%
_{u}\right)  \leq\sup_{i}\sum_{j=1}^{n}\left\vert \sigma_{ij}\right\vert
=O\left(  n^{\alpha_{\gamma}}\right)  , \label{NormVu1}%
\end{equation}%
\begin{equation}
n^{-1}\sum_{i=1}^{n}\sum_{j=1}^{n}\left\vert \sigma_{ij}\right\vert =O\left(
1\right)  \text{, and }n^{-1}\sum_{i=1}^{n}\sum_{j=1}^{n}\sigma_{ij}^{2}=O(1).
\label{NormVu2}%
\end{equation}
Also, for any $t$ and $t^{\prime}$%
\begin{equation}
a_{n,tt}=\frac{1}{n}\sum_{i=1}^{n}\left(  u_{it}^{2}-\sigma_{i}^{2}\right)
=O_{p}\left(  n^{-1/2}\right)  , \label{aa}%
\end{equation}%
\begin{equation}
Var(a_{n,tt})=\frac{1}{n^{2}}\sum_{i=1}^{n}\sum_{j=1}^{n}Cov(u_{it}^{2}%
,u_{jt}^{2})=O(n^{-1}), \label{Va}%
\end{equation}%
\begin{equation}
a_{n,tt^{\prime}}=n^{-1}\sum_{i=1}^{n}u_{it}u_{it^{\prime}}=O_{p}\left(
n^{-1/2}\right)  \text{, for }t\neq t^{\prime}, \label{aa2}%
\end{equation}%
\begin{equation}
Var\left(  a_{n,tt^{\prime}}\right)  =n^{-2}\sum_{i=1}^{n}\sum_{j=1}^{n}%
\sigma_{ij}^{2}=O\left(  n^{-1}\right)  \text{, for }t\neq t^{\prime},
\label{Vatt'}%
\end{equation}%
\begin{equation}
b_{n,t}=\frac{1}{n}\sum_{i=1}^{n}(u_{it}\overline{u}_{i\circ}-T^{-1}\sigma
_{i}^{2})=O_{p}\left(  T^{-1/2}n^{-1/2}\right)  , \label{ab}%
\end{equation}
and
\begin{equation}
Var(b_{n,t})=O\left(  T^{-1}n^{-1}\right)  . \label{Vb}%
\end{equation}
}
\end{lemma}

\begin{proof}
{\small Result (\ref{NormVu1}) follows noting that under Assumptions
\ref{Latent factor} and \ref{Errors}, $\sigma_{ij}=\gamma_{i}\gamma_{j}%
+\sigma_{v,ij},$ $\sum_{j=1}^{n}\left\vert \sigma_{ij}\right\vert \leq\sup
_{i}\left\vert \gamma_{i}\right\vert \sum_{j=1}^{n}\left\vert \gamma
_{j}\right\vert +\sum_{j=1}^{n}\left\vert \sigma_{v,ij}\right\vert ,$ and by
assumption $\sup_{i}\left\vert \gamma_{i}\right\vert <C$, $sup_{i}\sum
_{j=1}^{n}\left\vert \sigma_{v,ij}\right\vert <C$, and $\sum_{j=1}%
^{n}\left\vert \gamma_{j}\right\vert =O\left(  n^{\alpha_{\gamma}}\right)  $.
To prove (\ref{NormVu2})%
\begin{align*}
n^{-1}\sum_{i=1}^{n}\sum_{j=1}^{n}\left\vert \sigma_{ij}\right\vert  &  \leq
n^{-1}\sum_{i=1}^{n}\sum_{j=1}^{n}\left\vert \gamma_{i}\right\vert \left\vert
\gamma_{j}\right\vert +n^{-1}\sum_{i=1}^{n}\sum_{j=1}^{n}\left\vert
\sigma_{v,ij}\right\vert \\
&  =n^{-1}\left(  \sum_{i=1}^{n}\left\vert \gamma_{i}\right\vert \right)
^{2}+n^{-1}\sum_{i=1}^{n}\sum_{j=1}^{n}\left\vert \sigma_{v,ij}\right\vert .
\end{align*}
By assumption $n^{-1}\sum_{i=1}^{n}\sum_{j=1}^{n}\left\vert \sigma
_{v,ij}\right\vert =O(1)$, and $\sum_{i=1}^{n}\left\vert \gamma_{i}\right\vert
=O(n^{\alpha_{\gamma}})$, then, in view of (\ref{normg}) and (\ref{Covv})
$n^{-1}\sum_{i=1}^{n}\sum_{j=1}^{n}\left\vert \sigma_{ij}\right\vert =O\left(
n^{-1+2\alpha_{\gamma}}\right)  +O(1)=O(1),$ since $\alpha_{\gamma}<1/2$.
Similarly $\sigma_{ij}^{2}=\gamma_{i}^{2}\gamma_{j}^{2}+\sigma_{v,ij}%
^{2}+2\gamma_{i}\gamma_{j}\sigma_{v,ij}$, and
\[
n^{-1}\sum_{i=1}^{n}\sum_{j=1}^{n}\sigma_{ij}^{2}=n^{-1}Tr(\mathbf{V}_{u}%
^{2})=n^{-1}\left(  \sum_{i=1}^{n}\gamma_{i}^{2}\right)  ^{2}+n^{-1}\sum
_{i=1}^{n}\sum_{j=1}^{n}\sigma_{v,ij}^{2}+2n^{-1}\mathbf{\gamma}^{\prime
}\mathbf{V}_{v}\mathbf{\gamma.}%
\]
But
\begin{align}
n^{-1}\sum_{i=1}^{n}\sum_{j=1}^{n}\sigma_{v,ij}^{2}  &  =n^{-1}Tr\left(
\mathbf{V}_{v}^{2}\right)  \leq\lambda_{\max}^{2}\left(  \mathbf{V}%
_{v}\right)  =O(1),\label{Vu2}\\
\boldsymbol{\gamma}^{\prime}\mathbf{V}_{v}\boldsymbol{\gamma}  &  \leq\left(
\boldsymbol{\gamma}^{\prime}\boldsymbol{\gamma}\right)  \lambda_{\max}\left(
\mathbf{V}_{u}\right)  =O\left(  n^{\alpha_{\gamma}}\right)  .\nonumber
\end{align}
Overall
\[
n^{-1}\sum_{i=1}^{n}\sum_{j=1}^{n}\sigma_{ij}^{2}=O\left(  n^{-1+2\alpha
_{\gamma}}\right)  +O(1)+O\left(  n^{-1+\alpha_{\gamma}}\right)  ,
\]
and since $\alpha_{\gamma}<1/2$, then it follows that $n^{-1}\sum_{i=1}%
^{n}\sum_{j=1}^{n}\sigma_{ij}^{2}=O(1)$. To prove (\ref{aa}) we first note
that (using the normalization $E(g_{t}^{2})=1$)
\begin{equation}
u_{it}^{2}-E(u_{it}^{2})=(g_{t}^{2}-1)\gamma_{i}^{2}+\left(  v_{it}^{2}%
-\sigma_{v,ii}\right)  +2g_{t}\gamma_{i}v_{it}, \label{u2it}%
\end{equation}
and
\begin{align*}
a_{n,tt}  &  =(g_{t}^{2}-1)\left(  n^{-1}\sum_{i=1}^{n}\gamma_{i}^{2}\right)
+n^{-1}\sum_{i=1}^{n}\left(  v_{it}^{2}-\sigma_{v,ii}\right)  +2g_{t}\left(
n^{-1}\sum_{i=1}^{n}\gamma_{i}v_{it}\right) \\
&  =O\left(  n^{-1+\alpha_{\gamma}}\right)  +O_{p}\left(  n^{-1/2}\right)
+O\left(  n^{-1+\alpha_{\gamma}/2}\right)  =O_{p}(n^{-1/2})\text{, since
}\alpha_{\gamma}<1/2.
\end{align*}
To establish (\ref{Va}), using (\ref{u2it}) we first note that ( for given
loadings $\gamma_{i}$)
\[
Cov\left(  u_{it}^{2},u_{jt}^{2}\right)  =\gamma_{i}^{2}\gamma_{j}^{2}\left[
E\left(  g_{t}^{2}-1\right)  ^{2}\right]  +Cov(v_{it}^{2},v_{jt}%
^{2})+4E\left(  g_{t}^{2}\right)  \gamma_{i}\gamma_{j}\sigma_{v,ij},
\]%
\begin{align*}
Var(a_{n,tt})  &  =n^{-2}\sum_{i=1}^{n}\sum_{j=1}^{n}Cov(u_{it}^{2},u_{jt}%
^{2})=Var(g_{t}^{2})\left(  n^{-1}\sum_{i=1}^{n}\gamma_{i}^{2}\right)  ^{2}\\
&  +n^{-2}\sum_{i=1}^{n}\sum_{j=1}^{n}Cov(v_{it}^{2},v_{jt}^{2})+4E\left(
g_{t}^{2}\right)  n^{-2}\mathbf{\gamma}^{\prime}\mathbf{V}_{v}\mathbf{\gamma
}=O\left(  n^{-2+2\alpha_{\gamma}}\right)  +O(n^{-1})+O\left(  n^{-2+\alpha
_{\gamma}}\right)  \text{,}%
\end{align*}
and since $\alpha_{\gamma}<1/2$, then $Var(a_{n,tt})=O\left(  n^{-1}\right)
,$ as required (which also corroborate (\ref{aa})). Consider now (\ref{aa2})
and since $u_{it}$ is serially independent then $E\left(  a_{n,tt^{\prime}%
}\right)  =0$ for $t\neq t^{\prime}$, and we have%
\[
E\left(  u_{it}u_{it^{\prime}}u_{jt}u_{jt^{\prime}}\right)  =E\left(
u_{it}u_{jt}u_{it^{\prime}}u_{jt^{\prime}}\right)  =E\left(  u_{it}%
u_{jt}\right)  E\left(  u_{it^{\prime}}u_{jt^{\prime}}\right)  =\sigma
_{ij}^{2}\text{ for }t\neq t^{\prime}\text{,}%
\]
and
\begin{align*}
Var\left(  a_{n,tt^{\prime}}\right)   &  =E\left(  a_{n,tt^{\prime}}%
^{2}\right)  =n^{-2}\sum_{i=1}^{n}\sum_{j=1}^{n}E\left(  u_{it}u_{it^{\prime}%
}u_{jt}u_{jt^{\prime}}\right)  \text{, for }t\neq t^{\prime}\\
&  =n^{-2}\sum_{i=1}^{n}\sum_{j=1}^{n}\sigma_{ij}^{2}=O\left(  n^{-1}\right)
\text{, using (\ref{NormVu2}),}%
\end{align*}
which establishes (\ref{Vatt'}) and (\ref{aa2}). To prove (\ref{ab}) set
$z_{it}=u_{it}\overline{u}_{i\circ}-T^{-1}\sigma_{i}^{2}$, and write
$b_{n,t}=\frac{1}{n}\sum_{i=1}^{n}z_{it}$. Also note that $u_{it}\overline
{u}_{i\circ}=T^{-1}\sum_{s=1}^{T}u_{it}u_{is},$ and given that $\left\{
u_{it}\right\}  $ is serially independent then $E(z_{it})=0$ and
$E(b_{n,t})=0.$ and%
\begin{align}
Var(b_{n,t})  &  =n^{-2}\sum_{i=1}^{n}\sum_{j=1}^{n}E\left(  z_{it}%
z_{jt}\right)  =n^{-2}\sum_{i=1}^{n}\sum_{j=1}^{n}E(u_{it}u_{jt}\overline
{u}_{i\circ}\overline{u}_{j\circ})-T^{-2}\bar{\sigma}_{n}^{4}\nonumber\\
&  =n^{-2}T^{-2}\sum_{i=1}^{n}\sum_{j=1}^{n}\sum_{s=1}^{T}\sum_{s^{\prime}%
=1}^{T}E\left(  u_{it}u_{jt}u_{is}u_{js^{\prime}}\right)  -T^{-2}\bar{\sigma
}_{n}^{4}. \label{Vbt1}%
\end{align}
Also $E\left(  u_{it}u_{jt}u_{is}u_{js^{\prime}}\right)  =0$ for all $t$ if
$s\neq s^{\prime}$. We are left with one case where $s=s^{\prime}=t$, and one
case where $s=s^{\prime}\neq t$. Hence
\begin{align*}
\sum_{i=1}^{n}\sum_{j=1}^{n}\sum_{s=1}^{T}\sum_{s^{\prime}=1}^{T}E\left(
u_{it}u_{jt}u_{is}u_{js^{\prime}}\right)   &  =\sum_{i=1}^{n}\sum_{j=1}%
^{n}\sum_{s=1}^{T}E\left(  u_{it}u_{jt}u_{is}u_{js}\right) \\
&  =\sum_{i=1}^{n}\sum_{j=1}^{n}\sum_{s=1,s=t}^{T}E\left(  u_{it}u_{jt}%
u_{is}u_{js}\right)  +\sum_{i=1}^{n}\sum_{j=1}^{n}\sum_{s=1,s\neq t}%
^{T}E\left(  u_{it}u_{jt}u_{is}u_{js}\right) \\
&  =\sum_{i=1}^{n}\sum_{j=1}^{n}E\left(  u_{it}^{2}u_{jt}^{2}\right)
+\sum_{i=1}^{n}\sum_{j=1}^{n}\sum_{s=1,s\neq t}^{T}E\left(  u_{it}%
u_{jt}\right)  E\left(  u_{is}u_{js}\right) \\
&  =\sum_{i=1}^{n}\sum_{j=1}^{n}\left[  Cov\left(  u_{it}^{2},u_{jt}%
^{2}\right)  +\sigma_{i}^{2}\sigma_{j}^{2}\right]  +(T-1)\sum_{i=1}^{n}%
\sum_{j=1}^{n}\sigma_{ij}^{2},
\end{align*}
and hence
\[
\sum_{i=1}^{n}\sum_{j=1}^{n}\sum_{s=1}^{T}\sum_{s^{\prime}=1}^{T}E\left(
u_{it}u_{jt}u_{is}u_{js^{\prime}}\right)  =\sum_{i=1}^{n}\sum_{j=1}%
^{n}Cov\left(  u_{it}^{2},u_{jt}^{2}\right)  +\left(  \sum_{i=1}^{n}\sigma
_{i}^{2}\right)  ^{2}+(T-1)\sum_{i=1}^{n}\sum_{j=1}^{n}\sigma_{ij}^{2},
\]%
\begin{align*}
Var(b_{n,t})=  &  n^{-2}T^{-2}\sum_{i=1}^{n}\sum_{j=1}^{n}\sum_{s=1}^{T}%
\sum_{s^{\prime}=1}^{T}E\left(  u_{it}u_{jt}u_{is}u_{js^{\prime}}\right)
-T^{-2}\bar{\sigma}_{n}^{4}\\
&  =T^{-2}n^{-2}\sum_{i=1}^{n}\sum_{j=1}^{n}Cov\left(  u_{it}^{2},u_{jt}%
^{2}\right)  +T^{-2}\bar{\sigma}_{n}^{4}+(T-1)\sum_{i=1}^{n}\sum_{j=1}%
^{n}\sigma_{ij}^{2}-T^{-2}\bar{\sigma}_{n}^{4}.
\end{align*}
Using this result in (\ref{Vbt1}) we have%
\[
Var(b_{n,t})=n^{-2}T^{-2}\sum_{i=1}^{n}\sum_{j=1}^{n}Cov\left(  u_{it}%
^{2},u_{jt}^{2}\right)  +\frac{T-1}{T^{2}}\left(  n^{-2}\sum_{i=1}^{n}%
\sum_{j=1}^{n}\sigma_{ij}^{2}\right)  .
\]
Now using (\ref{Va}) and (\ref{NormVu2}) we have $Var(b_{n,t})=O\left(
T^{-1}n^{-1}\right)  $, which establishes (\ref{Vb}), and result (\ref{ab})
follows by Markov inequality. }
\end{proof}

\begin{lemma}
\label{L2}{\small Consider the $n\times T$ error matrix $\boldsymbol{U}%
_{nT}=\left(  \mathbf{u}_{1\circ},\mathbf{u}_{2\circ},...,\mathbf{u}_{n\circ
}\right)  ^{\prime},$ where $\mathbf{u}_{i\circ}=\left(  u_{i1},u_{i2}%
,...,u_{iT}\right)  ^{\prime},$ the $n\times k$ matrix of factor loadings,
$\mathbf{B}_{n}=(\boldsymbol{\beta}_{\circ1},\boldsymbol{\beta}_{\circ
2},...,\boldsymbol{\beta}_{\circ K}),$ where $\boldsymbol{\beta}_{\circ
k}=(\beta_{1k},\beta_{2k},...,\beta_{nk})^{\prime},$ the $n\times1$ vector of
pricing errors $\boldsymbol{\eta}_{n}=(\eta_{1},\eta_{2},...,\eta_{n}%
)^{\prime},$ with the pervasiveness coefficient, $\alpha_{\eta}$, the observed
factors, $\mathbf{f}_{k}=(f_{k1},f_{k2},...,f_{kT})^{\prime}$ are strong (with
$\alpha_{k}=1$, for $k=1,2,...,K)$, the missing factor $g_{t}$, has strength
$\alpha_{\gamma}<1/2$, $\mathbf{G}_{T}=$ $\mathbf{M}_{T}\mathbf{F}\left(
\mathbf{F}^{\prime}\mathbf{M}_{T}\mathbf{F}\right)  ^{-1}$, $\mathbf{F=(f}%
_{1},\mathbf{f}_{2},...,\mathbf{f}_{K})$, $\mathbf{M}_{n}=\mathbf{I}_{n}%
-\frac{1}{n}\boldsymbol{\tau}_{n}\boldsymbol{\tau}_{n}^{\prime},$
$\overline{\mathbf{u}}_{n\circ}=(\overline{u}_{1\circ},\overline{u}_{2\circ
},...,\overline{u}_{n\circ})^{\prime},$ $\overline{u}_{i\circ}=T^{-1}%
\sum_{t=1}^{T}u_{it}$, $\overline{\sigma}_{n}^{2}=n^{-1}\sum_{i=1}^{n}%
\sigma_{i}^{2}$, and $\boldsymbol{\tau}_{n}$ and $\boldsymbol{\tau}_{T}$ are,
respectively, $n\times1$ and $T\times1$ vectors of ones. Suppose that
Assumptions \ref{factors}, \ref{Latent factor}, \ref{Errors}, \ref{loadings}
and \ref{PriceError} hold. Then%
\begin{equation}
n^{-1}\mathbf{B}_{n}^{\prime}\mathbf{M}_{n}\boldsymbol{\eta}_{n}=O_{p}\left(
n^{-1+\alpha_{\eta}}\right)  , \label{ABeta}%
\end{equation}%
\begin{equation}
n^{-1}\mathbf{B}_{n}^{\prime}\mathbf{M}_{n}\overline{\mathbf{u}}_{n\circ
}=O_{p}\left(  T^{-1/2}n^{-1/2}\right)  , \label{ABMubar}%
\end{equation}%
\begin{equation}
n^{-1}\mathbf{B}_{n}^{\prime}\mathbf{M}_{n}\mathbf{U}_{nT}\mathbf{G}_{T}%
=O_{p}\left(  T^{-1/2}n^{-1/2}\right)  , \label{BUG}%
\end{equation}%
\begin{equation}
n^{-1}\mathbf{G}_{T}^{\prime}\mathbf{U}_{nT}^{\prime}\mathbf{M}_{n}%
\boldsymbol{\eta}_{n}=O_{p}\left(  T^{-1/2}n^{-1+\frac{\alpha_{\eta}%
+\alpha_{\gamma}}{2}}\right)  , \label{GUeta}%
\end{equation}%
\begin{equation}
n^{-1}\mathbf{G}_{T}^{\prime}\mathbf{U}_{nT}^{\prime}\mathbf{M}_{n}%
\overline{\mathbf{u}}_{n\circ}=O_{p}\left(  T^{-1}n^{-1/2}\right)  ,
\label{GUubar}%
\end{equation}%
\begin{equation}
\mathbf{G}_{T}^{\prime}\left(  n^{-1}\mathbf{U}_{nT}^{\prime}\mathbf{M}%
_{n}\mathbf{U}_{nT}-\bar{\sigma}_{n}^{2}\mathbf{I}_{T}\right)  \mathbf{G}%
_{T}=O_{p}\left(  T^{-1}n^{-1/2}\right)  . \label{GUUG}%
\end{equation}
}
\end{lemma}

\begin{proof}
{\small To establish (\ref{ABeta}) we first note that the $k^{th}$ element of
$n^{-1}\mathbf{B}_{n}^{\prime}\mathbf{M}_{n}\boldsymbol{\eta}_{n}$ can be
written as
\begin{equation}
\pi_{k,n}=n^{-1}\sum_{i=1}^{n}\left(  \beta_{ik}-\bar{\beta}_{k}\right)
\eta_{i}, \label{Aux1a}%
\end{equation}
where $\bar{\beta}_{k}=n^{-1}\boldsymbol{\tau}_{n}^{\prime}\boldsymbol{\beta
}_{\circ k}$. Since $\eta_{j}$ and $\beta_{ik}$ are distributed independently
for all $i$ and $j$, then $E\left(  \pi_{k,n}\right)  =0$, \footnote{Note that
$\sup_{ik}E\left\vert \beta_{ik}-\bar{\beta}_{k}\right\vert <C$ follows from
$\sup_{i}E\left\vert \beta_{ik}\right\vert ^{2}<C$, required by Assumption
\ref{loadings}.}
\[
E\left(  \left\vert \pi_{k,n}\right\vert \left\vert \boldsymbol{\eta}%
_{n}\right.  \right)  =n^{-1}\sum_{i=1}^{n}E\left\vert \beta_{ik}-\bar{\beta
}_{k}\right\vert \left\vert \eta_{i}\right\vert \leq\left[  \sup
_{i,k}E\left\vert \beta_{ik}-\bar{\beta}_{k}\right\vert \right]  \left(
n^{-1}\sum_{i=1}^{n}\left\vert \eta_{i}\right\vert \right)  =O\left(
n^{-1+\alpha_{\eta}}\right)  ,\text{ }%
\]
for $k=1,2,...,K,$ and by Markov inequality we have $n^{-1}\mathbf{B}%
_{n}^{\prime}\mathbf{M}_{n}\boldsymbol{\eta}_{n}=O_{p}\left(  n^{-1+\alpha
_{\eta}}\right)  $, as required. To establish (\ref{ABMubar}), noting that
\[
n^{-1}\overline{\mathbf{u}}_{n\circ}^{\prime}\mathbf{M}_{n}\mathbf{B}%
_{n}=\left[  n^{-1}\mathbf{\bar{u}}_{n\circ}^{\prime}\left(  \boldsymbol{\beta
}_{\circ1}-\mathbf{\tau}_{n}\bar{\beta}_{1}\right)  ,n^{-1}\mathbf{\bar{u}%
}_{n\circ}^{\prime}\left(  \boldsymbol{\beta}_{\circ2}-\mathbf{\tau}_{n}%
\bar{\beta}_{2}\right)  ,...,n^{-1}\mathbf{\bar{u}}_{n\circ}^{\prime}\left(
\boldsymbol{\beta}_{\circ K}-\mathbf{\tau}_{n}\bar{\beta}_{K}\right)  \right]
,
\]
and the $k^{th}$ element of $n^{-1}\mathbf{B}_{n}^{\prime}\mathbf{M}%
_{n}\overline{\mathbf{u}}_{n\circ}$ is given by $c_{k,nT}=n^{-1}\sum_{i=1}%
^{n}\left(  \beta_{ik}-\bar{\beta}_{k}\right)  \bar{u}_{i\circ}$. We have
$E\left(  c_{k,nT}\right)  =0$, and $Var\left(  c_{k,nT}\left\vert
\boldsymbol{\beta}_{\circ k}\right.  \right)  =n^{-2}\sum_{i=1}^{n}\sum
_{j=1}^{n}\left(  \beta_{ik}-\bar{\beta}_{k}\right)  \left(  \beta_{jk}%
-\bar{\beta}_{k}\right)  E\left(  \overline{u}_{i\circ}\overline{u}_{j\circ
}\right)  ,$ with $E\left(  \overline{u}_{i\circ}\overline{u}_{j\circ}\right)
=T^{-1}\sigma_{ij}$. Hence (recalling that $\mathbf{V}_{u}=(\sigma_{ij})$)%
\begin{equation}
Var\left(  c_{nT,k}\right)  =T^{-1}n^{-2}\sum_{i=1}^{n}\sum_{j=1}^{n}%
\sigma_{ij}E\left[  \left(  \beta_{ik}-\bar{\beta}_{k}\right)  \left(
\beta_{jk}-\bar{\beta}_{k}\right)  \right]  , \label{Vck}%
\end{equation}
and by Cauchy--Schwarz inequality and Assumption \ref{loadings} we have
$E\left[  \left(  \beta_{ik}-\bar{\beta}_{k}\right)  \left(  \beta_{jk}%
-\bar{\beta}_{k}\right)  \right]  \leq\left[  E\left(  \beta_{ik}-\bar{\beta
}_{k}\right)  ^{2}\right]  ^{1/2}\left[  E\left(  \beta_{jk}-\bar{\beta}%
_{k}\right)  ^{2}\right]  ^{1/2}<C$. Hence, $Var\left(  c_{nT,k}\right)  \leq
CT^{-1}n^{-2}\sum_{i=1}^{n}\sum_{j=1}^{n}\left\vert \sigma_{ij}\right\vert $.
Since $\alpha_{\gamma}<1/2$, then by (\ref{NormVu2}) we have $n^{-1}\sum
_{i=1}^{n}\sum_{j=1}^{n}\left\vert \sigma_{ij}\right\vert =O(1)$ and
$Var\left(  c_{nT,k}\right)  =O(T^{-1}n^{-1})$. Thus by Markov inequality it
follows that $c_{nT,k}=O_{p}\left(  T^{-1/2}n^{-1/2}\right)  $, for
$k=1,2,...,K$, and (\ref{ABMubar}) follows. To establish (\ref{BUG}) using
$\mathbf{G}_{T}=$ $\mathbf{M}_{T}\mathbf{F}\left(  \mathbf{F}^{\prime
}\mathbf{M}_{T}\mathbf{F}\right)  ^{-1}$ we have%
\[
n^{-1}\mathbf{B}_{n}^{\prime}\mathbf{M}_{n}\mathbf{U}_{nT}\mathbf{G}%
_{T}=n^{-1}T^{-1}\left(  \mathbf{B}_{n}^{\prime}\mathbf{M}_{n}\mathbf{U}%
_{nT}\mathbf{M}_{T}\mathbf{F}\right)  \left(  T^{-1}\mathbf{F}^{\prime
}\mathbf{M}_{T}\mathbf{F}\right)  ^{-1},
\]
and since $\left(  T^{-1}\mathbf{F}^{\prime}\mathbf{M}_{T}\mathbf{F}\right)
^{-1}$ is a positive definite matrix then it is sufficient to consider the
probability order of the $K\times K$ matrix $T^{-1}n^{-1}\mathbf{B}%
_{n}^{\prime}\mathbf{M}_{n}\mathbf{U}_{nT}\mathbf{M}_{T}\mathbf{F}$ $=$
$(q_{kk^{\prime}})$, where
\[
q_{kk^{\prime}}=T^{-1}n^{-1}\sum_{i=1}^{n}\sum_{t=1}^{T}\left(  \beta
_{ik}-\bar{\beta}_{k}\right)  \left(  f_{k^{\prime}t}-\bar{f}_{k^{\prime}%
}\right)  u_{it}=n^{-1}\sum_{i=1}^{n}\left(  \beta_{ik}-\bar{\beta}%
_{k}\right)  \psi_{iT},
\]
where $\psi_{iT}=T^{-1}\sum_{t=1}^{T}\left(  f_{k^{\prime}t}-\bar
{f}_{k^{\prime}}\right)  u_{it}$. Thus%
\begin{equation}
Var\left(  q_{kk^{\prime}}\left\vert \boldsymbol{\beta}_{\circ k}\right.
\right)  =n^{-2}\sum_{i=1}^{n}\sum_{j=1}^{n}\left(  \beta_{ik}-\bar{\beta}%
_{k}\right)  \left(  \beta_{jk}-\bar{\beta}_{k}\right)  Cov\left(  \psi
_{iT},\psi_{jT}\right)  . \label{Vqkk'}%
\end{equation}
Also%
\[
Cov\left(  \psi_{iT},\psi_{jT}\left\vert \mathbf{F}\right.  \right)  =E\left(
\psi_{iT}\psi_{jT}\left\vert \mathbf{F}\right.  \right)  =T^{-2}\sum_{t=1}%
^{T}\sum_{t^{\prime}=1}^{T}\left(  f_{k^{\prime}t}-\bar{f}_{k^{\prime}%
}\right)  \left(  f_{k^{\prime}t^{\prime}}-\bar{f}_{k^{\prime}}\right)
E\left(  u_{it}u_{jt^{\prime}}\right)  ,
\]
and $E(u_{it}u_{jt^{\prime}})=\sigma_{ij}$ for $t=t^{\prime}$ and $0$
otherwise ($t\neq t^{\prime})$. Then,%
\[
Cov\left(  \psi_{iT},\psi_{jT}\left\vert \mathbf{F}\right.  \right)
=T^{-2}\sum_{t=1}^{T}\left(  f_{k^{\prime}t}-\bar{f}_{k^{\prime}}\right)
^{2}E\left(  u_{it}u_{jt}\right)  =\sigma_{ij}T^{-2}\sum_{t=1}^{T}\left(
f_{k^{\prime}t}-\bar{f}_{k^{\prime}}\right)  ^{2}.
\]
Using this result in (\ref{Vqkk'}) now yields%
\[
Var\left(  q_{kk^{\prime}}\right)  =T^{-1}n^{-1}\left[  T^{-1}\sum_{t=1}%
^{T}E\left(  f_{k^{\prime}t}-\bar{f}_{k^{\prime}}\right)  ^{2}\right]
n^{-1}\sum_{i=1}^{n}\sum_{j=1}^{n}\sigma_{ij}E\left[  \left(  \beta_{ik}%
-\bar{\beta}_{k}\right)  \left(  \beta_{jk}-\bar{\beta}_{k}\right)  \right]
.
\]
The first term of the above is bounded since by assumption $\mathbf{f}_{t}$ is
stationary. The second term is bounded as established above (see the
derivations below (\ref{Vck})). Hence $Var\left(  q_{kk^{\prime}}\right)
=O(T^{-1}n^{-1})$ and result (\ref{BUG}) follows. To establish (\ref{GUeta})
note that $n^{-1}\mathbf{G}_{T}^{\prime}\mathbf{U}_{nT}^{\prime}\mathbf{M}%
_{n}\boldsymbol{\eta}_{n}=\left(  T^{-1}\mathbf{F}^{\prime}\mathbf{M}%
_{T}\mathbf{F}\right)  ^{-1}\left(  n^{-1}T^{-1}\mathbf{F}^{\prime}%
\mathbf{M}_{T}\mathbf{U}_{nT}^{\prime}\mathbf{M}_{n}\boldsymbol{\eta}%
_{n}\right)  ,$ where $\left(  T^{-1}\mathbf{F}^{\prime}\mathbf{M}%
_{T}\mathbf{F}\right)  ^{-1}=O_{p}(1)$. Also the $k^{th}$ element of
$n^{-1}T^{-1}\mathbf{F}^{\prime}\mathbf{M}_{T}\mathbf{U}_{nT}^{\prime
}\mathbf{M}_{n}\boldsymbol{\eta}_{n}$ is given by $p_{k}=T^{-1}\sum_{t=1}%
^{T}\left(  f_{kt}-\bar{f}_{k}\right)  c_{n,t},$ where $c_{n,t}=n^{-1}%
\sum_{i=1}^{n}u_{it}\eta_{i}=n^{-1}\boldsymbol{\eta}_{n}^{\prime}%
\mathbf{u}_{\circ t}$. Under Assumption (\ref{Errors}) $f_{kt}$ and $c_{n,t}$
are distributed independently, $E(p_{n,t})=0$, and $c_{n,t}$ are also serially
uncorrelated we have $\ Var(p_{k})=T^{-2}\sum_{t=1}^{T}E\left(  f_{kt}-\bar
{f}_{k}\right)  ^{2}Var(c_{n,t})$. Also noting that $Var(c_{n,t}\left\vert
\mathbf{\eta}\right.  )=n^{-2}\mathbf{\eta}_{n}^{\prime}\mathbf{V}%
_{u}\mathbf{\eta}_{n}$, then%
\[
Var(p_{k})=\left[  T^{-1}\sum_{t=1}^{T}E\left(  f_{kt}-\bar{f}_{k}\right)
^{2}\right]  \left(  T^{-1}n^{-2}\boldsymbol{\eta}_{n}^{\prime}\mathbf{V}%
_{u}\boldsymbol{\eta}_{n}\right)  .
\]
The first term is bounded, and it follows that
\[
Var(p_{k}\left\vert \boldsymbol{\eta}_{n}\right.  )<CT^{-1}n^{-2}%
\boldsymbol{\eta}_{n}^{\prime}\mathbf{V}_{u}\boldsymbol{\eta}_{n}\leq
CT^{-1}n^{-2}\left(  \boldsymbol{\eta}_{n}^{\prime}\boldsymbol{\eta}%
_{n}\right)  \lambda_{\max}(\mathbf{V}_{u})=O\left(  T^{-1}n^{-2+\alpha_{\eta
}+\alpha_{\gamma}}\right)  ,
\]
and (\ref{GUeta}) follows. To establish (\ref{GUubar}) note that%
\[
n^{-1}\mathbf{G}_{T}^{\prime}\mathbf{U}_{nT}^{\prime}\mathbf{M}_{n}%
\overline{\mathbf{u}}_{n\circ}=\left(  T^{-1}\mathbf{F}^{\prime}\mathbf{M}%
_{T}\mathbf{F}\right)  ^{-1}\left(  n^{-1}T^{-1}\mathbf{F}^{\prime}%
\mathbf{M}_{T}\mathbf{U}_{nT}^{\prime}\mathbf{M}_{n}\overline{\mathbf{u}%
}_{n\circ}\right)  ,
\]
and the $k^{th}$ element of $n^{-1}T^{-1}\mathbf{F}^{\prime}\mathbf{M}%
_{T}\mathbf{U}_{nT}^{\prime}\mathbf{M}_{n}\overline{\mathbf{u}}_{n\circ
}=(d_{k})$ is given by%
\begin{align*}
d_{k}  &  =\frac{1}{nT}\sum_{t=1}^{T}(f_{kt}-\bar{f}_{k})\sum_{i=1}^{n}\left(
u_{it}-\bar{u}_{i\circ}\right)  \bar{u}_{i\circ}=\frac{1}{nT}\sum_{t=1}%
^{T}\sum_{i=1}^{n}(f_{kt}-\bar{f}_{k})(u_{it}\bar{u}_{i\circ}-\bar{u}_{i\circ
}\bar{u}_{i\circ})\\
&  =\frac{1}{T}\sum_{t=1}^{T}(f_{kt}-\bar{f}_{k})\left[  n^{-1}\sum_{i=1}%
^{n}\left(  u_{it}\bar{u}_{i\circ}-T^{-1}\bar{\sigma}_{n}^{2}\right)  \right]
=\frac{1}{T}\sum_{t=1}^{T}(f_{kt}-\bar{f}_{k})b_{n,t},
\end{align*}
where $b_{n,t}=n^{-1}\sum_{i=1}^{n}\left(  u_{it}\bar{u}_{i\circ}-T^{-1}%
\sigma_{i}^{2}\right)  $. By assumption $u_{it}$ and $f_{kt^{\prime}}$ are
distributed independently and $b_{n,t}$ are also serially independent. Then
$E(d_{k})=0,$ and
\[
Var(d_{k})=\frac{1}{T^{2}}\sum_{t=1}^{T}(f_{kt}-\bar{f}_{k})^{2}Var(b_{n,t}).
\]
Now using (\ref{Vb}) in Lemma \ref{L1}, $Var(b_{n,t})=O\left(  T^{-1}%
n^{-1}\right)  ,$ and overall we have $Var(d_{k})=O\left(  T^{-2}%
n^{-1}\right)  ,$ and by Markov inequality $d_{k}=O_{p}\left(  T^{-1}%
n^{-1/2}\right)  $, as required. Finally, consider (\ref{GUUG})\textbf{ }and
note that
\begin{align*}
&  \mathbf{G}_{T}^{\prime}\left(  n^{-1}\mathbf{U}_{nT}^{\prime}\mathbf{M}%
_{n}\mathbf{U}_{nT}-\bar{\sigma}_{n}^{2}\mathbf{I}_{T}\right)  \mathbf{G}%
_{T}\\
&  =T^{-2}\left(  T^{-1}\mathbf{F}^{\prime}\mathbf{M}_{T}\mathbf{F}\right)
^{-1}\mathbf{F}^{\prime}\mathbf{M}_{T}\left[  n^{-1}\mathbf{U}_{nT}^{\prime
}\mathbf{M}_{n}\mathbf{U}_{nT}-\bar{\sigma}_{n}^{2}\mathbf{I}_{T}\right]
\mathbf{M}_{T}\mathbf{F}\left(  T^{-1}\mathbf{F}^{\prime}\mathbf{M}%
_{T}\mathbf{F}\right)  ^{-1}.
\end{align*}
Since by assumption $\left(  T^{-1}\mathbf{F}^{\prime}\mathbf{M}_{T}%
\mathbf{F}\right)  ^{-1}$ is a positive definite matrix for all $T$, and $K$
is fixed then it is sufficient to derive the probability order of the
$(k,k^{\prime})$ element of of the $K\times K$ matrix $\mathbf{\Delta}%
_{nT\ }=(\delta_{kk^{\prime}})$
\begin{align}
\mathbf{\Delta}_{nT\ }  &  =T^{-2}\mathbf{F}^{\prime}\mathbf{M}_{T}\left[
n^{-1}\mathbf{U}_{nT}^{\prime}\mathbf{M}_{n}\mathbf{U}_{nT}-\bar{\sigma}%
_{n}^{2}\mathbf{I}_{T}\right]  \mathbf{M}_{T}\mathbf{F}\nonumber\\
&  =T^{-1}\left[  n^{-1}T^{-1}\mathbf{F}^{\prime}\mathbf{M}_{T}\mathbf{U}%
_{nT}^{\prime}\mathbf{M}_{n}\mathbf{U}_{nT}\mathbf{M}_{T}\mathbf{F}%
-\bar{\sigma}_{n}^{2}T^{-1}\mathbf{F}^{\prime}\mathbf{M}_{T}\mathbf{F}\right]
\nonumber\\
&  =T^{-1}\left(  n^{-1}T^{-1}\mathbf{AM}_{n}\mathbf{A}^{\prime}\mathbf{-}%
\bar{\sigma}_{n}^{2}T^{-1}\mathbf{\mathbf{F}^{\prime}M}_{T}\mathbf{\mathbf{F}%
}\right)  =T^{-1}\left(  \mathbf{S-R}\right)  , \label{dkk'}%
\end{align}
where $\mathbf{A=\mathbf{F}^{\prime}\mathbf{M}}_{T}\mathbf{U}_{nT}^{\prime
}=(a_{ki})$. Denote the $(k,k^{\prime})$ elements of $\mathbf{S}=n^{-1}%
T^{-1}\mathbf{AM}_{n}\mathbf{A}^{\prime}$ and $\mathbf{R}=\bar{\sigma}_{n}%
^{2}T^{-1}\mathbf{\mathbf{F}^{\prime}M}_{T}\mathbf{\mathbf{F}}$ by
$s_{kk^{\prime}}$ and $r_{kk^{\prime}}$, respectively, and note that%
\begin{equation}
r_{kk^{\prime}}=\left(  n^{-1}\sum_{i=1}^{n}\sigma_{i}^{2}\right)  \left[
T^{-1}\sum_{t=1}^{T}(f_{kt}-\bar{f}_{k})(f_{k^{\prime}t}-\bar{f}_{k^{\prime}%
})\right]  \label{rkk'}%
\end{equation}
and $s_{kk^{\prime}}=n^{-1}T^{-1}\sum_{i=1}^{n}\left(  a_{ki}-\bar{a}%
_{k}\right)  a_{k^{\prime}i},$ where $a_{ki}=\sum_{t=1}^{T}\tilde{f}%
_{kt}u_{it}$, and $\bar{a}_{k}=\sum_{t=1}^{T}\tilde{f}_{kt}\bar{u}_{\circ t}$,
where $\tilde{f}_{kt}=f_{kt}-\bar{f}_{k}$. Then%
\begin{align}
s_{kk^{\prime}}  &  =n^{-1}T^{-1}\sum_{i=1}^{n}\left[  \sum_{t=1}^{T}\tilde
{f}_{kt}\left(  u_{it}-\bar{u}_{\circ t}\right)  \right]  \left[  \sum
_{t=1}^{T}\tilde{f}_{k^{\prime}t}u_{it},\right] \nonumber\\
&  =n^{-1}T^{-1}\sum_{t=1}^{T}\sum_{t^{\prime}=1}^{T}\sum_{i=1}^{n}\left(
u_{it}-\bar{u}_{\circ t}\right)  u_{it^{\prime}}\tilde{f}_{kt}\tilde
{f}_{k^{\prime}t^{\prime}}\nonumber\\
&  =n^{-1}T^{-1}\sum_{t=1}^{T}\sum_{t^{\prime}=1}^{T}\sum_{i=1}^{n}%
u_{it}u_{it^{\prime}}\tilde{f}_{kt}\tilde{f}_{k^{\prime}t^{\prime}}-T^{-1}%
\sum_{t=1}^{T}\sum_{t^{\prime}=1}^{T}\bar{u}_{\circ t}\bar{u}_{\circ
t^{\prime}}\tilde{f}_{kt}\tilde{f}_{k^{\prime}t^{\prime}}. \label{skk'}%
\end{align}
Using this result and $r_{kk^{\prime}}$ given by (\ref{rkk'}) in (\ref{dkk'})
now yields%
\begin{align*}
T\delta_{kk^{\prime}}  &  =(s_{kk^{\prime}}-r_{kk^{\prime}})\\
&  =n^{-1}T^{-1}\sum_{t=1}^{T}\sum_{t^{\prime}=1}^{T}\sum_{i=1}^{n}%
u_{it}u_{it^{\prime}}\tilde{f}_{kt}\tilde{f}_{k^{\prime}t^{\prime}}-T^{-1}%
\sum_{t=1}^{T}\sum_{t^{\prime}=1}^{T}\bar{u}_{\circ t}\bar{u}_{\circ
t^{\prime}}\tilde{f}_{kt}\tilde{f}_{k^{\prime}t^{\prime}}\\
&  -\bar{\sigma}_{n}^{2}\left(  T^{-1}\sum_{t=1}^{T}\tilde{f}_{kt}\tilde
{f}_{k^{\prime}t}\right)  -\bar{\sigma}_{n}^{2}\left(  T^{-1}\sum_{t=1}%
^{T}\tilde{f}_{kt}\tilde{f}_{k^{\prime}t}\right)  .
\end{align*}
Also writing the first two terms as sums of the elements with $t=t^{\prime}$
and those with $t\neq t^{\prime}$, we have
\begin{align}
T\delta_{kk^{\prime}}  &  =n^{-1}T^{-1}\sum_{t=1}^{T}\sum_{i=1}^{n}u_{it}%
^{2}\tilde{f}_{kt}\tilde{f}_{k^{\prime}t}-\bar{\sigma}_{n}^{2}\left(
T^{-1}\sum_{t=1}^{T}\tilde{f}_{kt}\tilde{f}_{k^{\prime}t}\right)  -T^{-1}%
\sum_{t=1}^{T}\bar{u}_{\circ t}^{2}\tilde{f}_{kt}\tilde{f}_{k^{\prime}%
t}+\nonumber\\
&  +n^{-1}T^{-1}\sum_{t\neq t^{\prime}}^{T}\sum_{i=1}^{n}u_{it}u_{it^{\prime}%
}\tilde{f}_{kt}\tilde{f}_{k^{\prime}t^{\prime}}-T^{-1}\sum_{t\neq t^{\prime}%
}^{T}\bar{u}_{\circ t^{\prime}}\tilde{f}_{kt}\tilde{f}_{k^{\prime}t^{\prime}%
}\nonumber\\
&  =T^{-1}\sum_{t=1}^{T}a_{n,tt}\tilde{f}_{kt}\tilde{f}_{k^{\prime}t}%
-T^{-1}\sum_{t=1}^{T}\sum_{t^{\prime}=1}^{T}\bar{u}_{\circ t}\bar{u}_{\circ
t^{\prime}}\tilde{f}_{kt}\tilde{f}_{k^{\prime}t}+T^{-1}\sum_{t\neq t^{\prime}%
}^{T}a_{n,tt^{\prime}}\tilde{f}_{kt}\tilde{f}_{k^{\prime}t^{\prime}%
},\nonumber\\
T\delta_{kk^{\prime}}  &  =A_{kk^{\prime}}-B_{kk^{\prime}}+C_{kk^{\prime}},
\label{delta_kk'}%
\end{align}
where $a_{n,tt}=n^{-1}\sum_{i=1}^{n}\left(  u_{it}^{2}-\sigma_{i}^{2}\right)
,$ and $a_{n,tt^{\prime}}=n^{-1}\sum_{i=1}^{n}u_{it}u_{it^{\prime}}$. Since
$u_{it}$ and $f_{kt^{\prime}}$ are distributed independently then $E\left(
A_{kk^{\prime}}\right)  =0$, and
\[
Var\left(  A_{kk^{\prime}}\right)  =T^{-2}\sum_{t=1}^{T}\sum_{t^{\prime}%
=1}^{T}E\left(  a_{n,tt}a_{n,tt^{\prime}}\right)  E\left(  \tilde{f}%
_{kt}\tilde{f}_{kt^{\prime}}\tilde{f}_{k^{\prime}t}\tilde{f}_{k^{\prime
}t^{\prime}}\right)  .
\]
Since $u_{it}$ is serially independent, then $E\left(  a_{n,tt}a_{n,tt^{\prime
}}\right)  =0$ for $t\neq t^{\prime}$ and%
\[
Var\left(  A_{kk^{\prime}}\right)  =T^{-2}\sum_{t=1}^{T}E\left(  a_{n,tt}%
^{2}\right)  E\left(  \tilde{f}_{kt}^{2}\tilde{f}_{k^{\prime}t}^{2}\right)
\leq\left[  \sup_{k,k^{\prime},t}E\left(  \tilde{f}_{kt}^{2}\tilde
{f}_{k^{\prime}t}^{2}\right)  \right]  \left[  T^{-2}\sum_{t=1}^{T}E\left(
a_{n,tt}^{2}\right)  \right]  .
\]
$\sup_{k,k^{\prime},t}E\left(  \tilde{f}_{kt}^{2}\tilde{f}_{k^{\prime}t}%
^{2}\right)  <C$ since by assumption $sup_{t,k}E\left(  \tilde{f}_{kt}%
^{4}\right)  <C$. Also (using (\ref{aa}) of Lemma \ref{L1})
\[
E\left(  a_{n,tt}^{2}\right)  =n^{-2}\sum_{i=1}^{n}\sum_{j=1}^{n}Cov\left(
u_{it}^{2},u_{jt}^{2}\right)  =O\left(  n^{-1}\right)  ,
\]
and we have $A_{kk^{\prime}}=O_{p}\left(  T^{-1}n^{-1/2}\right)  $. Now write
$B_{kk^{\prime}}$ as%
\[
B_{kk^{\prime}}=T^{-1}\sum_{t=1}^{T}\sum_{t^{\prime}=1}^{T}\bar{u}_{\circ
t}\bar{u}_{\circ t^{\prime}}\tilde{f}_{kt}\tilde{f}_{k^{\prime}t}=\left(
T^{-1/2}\sum_{t=1}^{T}\bar{u}_{\circ t}\tilde{f}_{kt}\right)  \left(
T^{-1/2}\sum_{t=1}^{T}\bar{u}_{\circ t}\tilde{f}_{k^{\prime}t}\right)
=q_{\circ k}q_{\circ k^{\prime}},
\]
where $q_{\circ k}=T^{-1/2}\sum_{t=1}^{T}\bar{u}_{\circ t}\tilde{f}_{kt}$.
Also $E(q_{\circ k})=0,$ and $Var\left(  q_{\circ k}\right)  =T^{-1}\sum
_{t=1}^{T}E\left(  \bar{u}_{\circ t}^{2}\right)  E\left(  \tilde{f}_{kt}%
^{2}\right)  =E\left(  \bar{u}_{\circ t}^{2}\right)  \left[  T^{-1}\sum
_{t=1}^{T}E\left(  \tilde{f}_{kt}^{2}\right)  \right]  $, where $E\left(
\bar{u}_{\circ t}^{2}\right)  =n^{-2}\sum_{i=1}^{n}\sum_{j=1}^{n}\sigma
_{ij}=O\left(  n^{-1}\right)  $. Hence, $T^{-1/2}\sum_{t=1}^{T}\bar{u}_{\circ
t}\tilde{f}_{kt}=O_{p}\left(  n^{-1/2}\right)  $, and $B_{kk^{\prime}}%
=O_{p}\left(  n^{-1}\right)  $. Finally, consider $C_{kk^{\prime}}=T^{-1}%
\sum_{t\neq t^{\prime}}^{T}a_{n,tt^{\prime}}\tilde{f}_{kt}\tilde{f}%
_{k^{\prime}t^{\prime}},$ and note that $E\left(  a_{n,tt^{\prime}}\right)
=n^{-1}\sum_{i=1}^{n}E\left(  u_{it}u_{it^{\prime}}\right)  =0$ for all $t\neq
t^{\prime},$which ensures that $E\left(  C_{kk^{\prime}}\right)  =0$.
Further,
\[
Var\left(  C_{kk^{\prime}}\right)  =T^{-2}\sum_{t\neq t^{\prime}}^{T}%
\sum_{s\neq s^{\prime}}^{T}E\left(  a_{n,tt^{\prime}}a_{n,ss^{\prime}}\right)
E\left(  \tilde{f}_{kt}\tilde{f}_{k^{\prime}t^{\prime}}\tilde{f}_{ks}\tilde
{f}_{k^{\prime}s^{\prime}}\right)  =T^{-2}\sum_{t\neq t^{\prime}}^{T}E\left(
a_{n,tt^{\prime}}^{2}\right)  E\left(  \tilde{f}_{kt}^{2}\tilde{f}_{k^{\prime
}t^{\prime}}^{2}\right)  ,
\]
$E\left(  a_{n,tt^{\prime}}^{2}\right)  =n^{-2}\sum_{i=1}^{n}\sum_{j=1}%
^{n}\sigma_{ij}^{2}$. See (\ref{Vatt'}) in Lemma \ref{L1}. Since by assumption
$E\left(  \tilde{f}_{kt}^{4}\right)  <C$, then we have%
\[
Var\left(  C_{kk^{\prime}}\right)  =\left(  n^{-2}\sum_{i=1}^{n}\sum_{j=1}%
^{n}\sigma_{ij}^{2}\right)  \left[  T^{-2}\sum_{t\neq t^{\prime}}^{T}E\left(
\tilde{f}_{kt}^{2}\tilde{f}_{k^{\prime}t^{\prime}}^{2}\right)  \right]
<C\left(  n^{-2}\sum_{i=1}^{n}\sum_{j=1}^{n}\sigma_{ij}^{2}\right)  ,
\]
and using (\ref{NormVu2}), it follows that $Var\left(  C_{kk^{\prime}}\right)
=O\left(  n^{-1}\right)  $, and hence $C_{kk^{\prime}}=O_{p}\left(
n^{-1/2}\right)  $. Using this result and the ones obtained for $A_{kk^{\prime
}}$ and $B_{kk^{\prime}}$ in (\ref{delta_kk'}) now yields}%
\[
{\small \delta_{kk^{\prime}}=T^{-1}\left[  O_{p}\left(  T^{-1}n^{-1/2}\right)
+O_{p}\left(  n^{-1}\right)  +O_{p}\left(  n^{-1/2}\right)  \right]
=O_{p}\left(  T^{-1}n^{-1/2}\right)  ,}%
\]
{\small as required. }
\end{proof}

\section{{\protect\small Proof of theorems in the paper}}

\subsection{{\protect\small Proof of theorem \ref{TFMbias}}\label{PFMbias}}

{\small Consider the two-pass estimator of $\boldsymbol{\lambda}$ defined by
(\ref{lambdahat}), which we reproduce here for convenience
\[
\boldsymbol{\hat{\lambda}}_{nT}=\left(  \mathbf{\hat{B}}_{nT}^{\prime
}\mathbf{M}_{n}\mathbf{\hat{B}}_{nT}\right)  ^{-1}\mathbf{\hat{B}}%
_{nT}^{\prime}\mathbf{M}_{n}\mathbf{\bar{r}}_{n\circ},
\]
where $\mathbf{\hat{B}}_{nT}=(\boldsymbol{\hat{\beta}}_{1},\boldsymbol{\hat
{\beta}}_{2},...,\boldsymbol{\hat{\beta}}_{n})^{\prime},$ $\mathbf{\bar{r}%
}_{n\circ}=(\bar{r}_{1\circ},\bar{r}_{2\circ},...,\bar{r}_{n\circ})^{\prime},$
$\bar{r}_{i\circ}=T^{-1}\sum_{t=1}^{T}r_{it},$%
\begin{equation}
\boldsymbol{\hat{\beta}}_{i}=(\mathbf{F}^{\prime}\mathbf{M}_{T}\mathbf{F)}%
^{-1}\mathbf{F}^{\prime}\mathbf{M}_{T}\mathbf{r}_{i\circ}, \label{betahati}%
\end{equation}
and $\mathbf{r}_{i\circ}=(r_{i1},r_{i2},...,r_{iT})^{\prime}.$ Under the
factor model (\ref{mfi})%
\begin{equation}
\mathbf{r}_{i\circ}=\boldsymbol{\alpha}_{i}\boldsymbol{\tau}_{T}%
+\mathbf{F}\boldsymbol{\beta}_{i}+\mathbf{u}_{i\circ}, \label{ridot}%
\end{equation}
where $\mathbf{u}_{i\circ}=(u_{i1},u_{i2},...,u_{iT})^{\prime},$ and hence
\begin{equation}
\boldsymbol{\hat{\beta}}_{i}=\boldsymbol{\beta}_{i}+(\mathbf{F}^{\prime
}\mathbf{M}_{T}\mathbf{F)}^{-1}\mathbf{F}^{\prime}\mathbf{M}_{T}%
\mathbf{u}_{i\circ}. \label{betahati1}%
\end{equation}
Stacking these results over $i$ yields%
\begin{equation}
\mathbf{\hat{B}}_{nT}=\mathbf{B}_{n}+\mathbf{U}_{nT}\mathbf{G}_{T},
\label{Bhat1}%
\end{equation}
where $\mathbf{U}_{nT}=\mathbf{(u}_{1\circ},\mathbf{u}_{2\circ},...,\mathbf{u}%
_{n\circ})^{\prime},$ and
\begin{equation}
\mathbf{G}_{T}=\mathbf{M}_{T}\mathbf{F}(\mathbf{F}^{\prime}\mathbf{M}%
_{T}\mathbf{F)}^{-1}. \label{GT}%
\end{equation}
Also%
\begin{equation}
\mathbf{r}_{\circ t}=\boldsymbol{\alpha}_{n}+\mathbf{B}_{n}\mathbf{f}%
_{t}+\mathbf{u}_{\circ t}, \label{rdott}%
\end{equation}
where $\mathbf{u}_{\circ t}=(u_{1t},u_{2t},...,u_{nt})^{\prime}$,
}$\boldsymbol{\alpha}${\small $_{n}=(\alpha_{1}$,$\alpha_{2},...,\alpha
_{n})^{\prime}$. Using (\ref{ai}) in the paper,%
\begin{equation}
\boldsymbol{\alpha}_{n}=c\boldsymbol{\tau}_{n}+\mathbf{B}_{n}\boldsymbol{\phi
}_{0}+\boldsymbol{\eta}_{n}, \label{Aai}%
\end{equation}
and (\ref{rdott}) can be written as $\mathbf{r}_{\circ t}=c\boldsymbol{\tau
}_{n}+\mathbf{B}_{n}\left(  \boldsymbol{\phi}_{0}\boldsymbol{+}\mathbf{f}%
_{t}\right)  +\mathbf{u}_{\circ t}+\boldsymbol{\eta}_{n}$. Now averaging over
$t$ yields
\begin{equation}
\mathbf{\bar{r}}_{n\circ}\text{\textbf{ }}=c\boldsymbol{\tau}_{n}%
+\mathbf{B}_{n}\boldsymbol{\lambda}_{T}^{\ast}+\mathbf{\bar{u}}_{n\circ
}\mathbf{+}\boldsymbol{\eta}_{n}, \label{rbar1}%
\end{equation}
where $\mathbf{\bar{r}}_{n\circ}=T^{-1}\sum_{t=1}^{T}\mathbf{r}_{\circ
t}=\left(  \bar{r}_{1\circ},\bar{r}_{2\circ},...,\bar{r}_{n\circ}\right)
^{\prime}$, $\mathbf{\bar{u}}_{n\circ}=T^{-1}\sum_{t=1}^{T}\mathbf{u}_{\circ
t}=\left(  \bar{u}_{1\circ},\bar{u}_{2\circ},...,\bar{u}_{n\circ}\right)
^{\prime}$, and
\begin{equation}
\boldsymbol{\lambda}_{T}^{\ast}=\boldsymbol{\phi}_{0}+\boldsymbol{\hat{\mu}%
}_{T}=\boldsymbol{\lambda}_{0}+\left(  \boldsymbol{\hat{\mu}}_{T}%
-\boldsymbol{\mu}_{0}\right)  . \label{lambdastar}%
\end{equation}
Using (\ref{rbar1}) in (\ref{lambdahat}) we have
\begin{align}
\boldsymbol{\hat{\lambda}}_{nT}  &  =\left(  \mathbf{\hat{B}}_{nT}^{\prime
}\mathbf{M}_{n}\mathbf{\hat{B}}_{nT}\right)  ^{-1}\mathbf{\hat{B}}%
_{nT}^{\prime}\mathbf{M}_{n}\left(  c\boldsymbol{\tau}_{n}+\mathbf{B}%
_{n}\boldsymbol{\lambda}_{T}^{\ast}+\mathbf{\bar{u}}_{n\circ}\mathbf{+}%
\boldsymbol{\eta}_{n}\right) \nonumber\\
&  =\left(  \mathbf{\hat{B}}_{nT}^{\prime}\mathbf{M}_{n}\mathbf{\hat{B}}%
_{nT}\right)  ^{-1}\mathbf{\hat{B}}_{nT}^{\prime}\mathbf{M}_{n}\left[
\mathbf{\hat{B}}_{nT}\boldsymbol{\lambda}_{T}^{\ast}-\left(  \mathbf{\hat{B}%
}_{nT}-\mathbf{B}_{n}\right)  \boldsymbol{\lambda}_{T}^{\ast}+\mathbf{\bar{u}%
}_{n\circ}\mathbf{+}\boldsymbol{\eta}_{n}\right]  . \label{lamb2}%
\end{align}
Also using (\ref{lambdastar}), and recalling that $\boldsymbol{\lambda}%
_{0}=\boldsymbol{\phi}_{0}+\boldsymbol{\mu}_{0}$, we have $\boldsymbol{\hat
{\lambda}}_{nT}-\boldsymbol{\lambda}_{T}^{\ast}=\boldsymbol{\hat{\lambda}%
}_{nT}-\boldsymbol{\lambda}_{0}-\left(  \boldsymbol{\lambda}_{T}^{\ast
}-\boldsymbol{\lambda}_{0}\right)  =\left(  \boldsymbol{\hat{\lambda}}%
_{nT}-\boldsymbol{\lambda}_{0}\right)  -\left(  \boldsymbol{\phi}%
_{0}+\boldsymbol{\hat{\mu}}_{T}-\boldsymbol{\lambda}_{0}\right)  ,$ which
yields $\boldsymbol{\hat{\lambda}}_{nT}-\boldsymbol{\lambda}_{0}%
=\boldsymbol{\hat{\lambda}}_{nT}-\boldsymbol{\lambda}_{T}^{\ast}+\left(
\boldsymbol{\hat{\mu}}_{T}-\boldsymbol{\mu}_{0}\right)  .$ Furthermore,
$\boldsymbol{\hat{\lambda}}_{nT}-\boldsymbol{\lambda}_{T}^{\ast}%
=\boldsymbol{\hat{\phi}}_{nT}-\boldsymbol{\phi}_{0},$where $\boldsymbol{\hat
{\phi}}_{nT}$ is the two-step estimator of $\boldsymbol{\phi}_{0}$ given by
(\ref{phihat1}). This results follows noting that $\boldsymbol{\hat{\lambda}%
}_{nT}=\boldsymbol{\hat{\phi}}_{nT}+\boldsymbol{\hat{\mu}}_{T}$, and
$\boldsymbol{\lambda}_{T}^{\ast}=\boldsymbol{\phi}_{0}+\boldsymbol{\hat{\mu}%
}_{T}$. Therefore,%
\begin{equation}
\boldsymbol{\hat{\lambda}}_{nT}-\boldsymbol{\lambda}_{0}=\left(
\boldsymbol{\hat{\phi}}_{nT}-\boldsymbol{\phi}_{0}\right)  +\left(
\boldsymbol{\hat{\mu}}_{T}-\boldsymbol{\mu}_{0}\right)  . \label{lam0}%
\end{equation}
We focus on deriving the asymptotic distribution of $\boldsymbol{\hat{\lambda
}}_{nT}-\boldsymbol{\lambda}_{T}^{\ast}=\boldsymbol{\hat{\phi}}_{nT}%
-\boldsymbol{\phi}_{0}$ since the panel (cross section) dimension does not
apply to the second component, $\left(  \boldsymbol{\hat{\mu}}_{T}%
-\boldsymbol{\mu}_{0}\right)  $. Now using (\ref{Bhat1}) in (\ref{lamb2}) and
after some simplifications we have%
\[
\left(  \mathbf{\hat{B}}_{nT}^{\prime}\mathbf{M}_{n}\mathbf{\hat{B}}%
_{nT}\right)  \boldsymbol{\hat{\lambda}}_{nT}=\mathbf{\hat{B}}_{nT}^{\prime
}\mathbf{M}_{n}\left[  \mathbf{\hat{B}}_{nT}\boldsymbol{\lambda}_{T}^{\ast
}-\left(  \mathbf{\hat{B}}_{nT}-\mathbf{B}_{n}\right)  \boldsymbol{\lambda
}_{T}^{\ast}+\mathbf{\bar{u}}_{n\circ}\mathbf{+}\boldsymbol{\eta}_{n}\right]
,
\]
or%
\begin{equation}
\left(  n^{-1}\mathbf{\hat{B}}_{nT}^{\prime}\mathbf{M}_{n}\mathbf{\hat{B}%
}_{nT}\right)  \left(  \boldsymbol{\hat{\phi}}_{nT}-\boldsymbol{\phi}%
_{0}\right)  =\mathbf{p}_{nT}, \label{phidis1}%
\end{equation}
where%
\begin{align}
\mathbf{p}_{nT}  &  =n^{-1}\mathbf{B}_{n}^{\prime}\mathbf{M}_{n}%
\boldsymbol{\eta}_{n}+n^{-1}\mathbf{G}_{T}^{\prime}\mathbf{U}_{nT}^{\prime
}\mathbf{M}_{n}\boldsymbol{\eta}_{n}+n^{-1}\mathbf{B}_{n}^{\prime}%
\mathbf{M}_{n}\mathbf{\bar{u}}_{n\circ}+n^{-1}\mathbf{G}_{T}^{\prime
}\mathbf{U}_{nT}^{\prime}\mathbf{M}_{n}\mathbf{\bar{u}}_{n\circ}\label{pnT}\\
&  -n^{-1}\mathbf{B}_{n}^{\prime}\mathbf{M}_{n}\mathbf{U}_{nT}\mathbf{G}%
_{T}\boldsymbol{\lambda}_{T}^{\ast}-n^{-1}\mathbf{G}_{T}^{\prime}%
\mathbf{U}_{nT}^{\prime}\mathbf{M}_{n}\mathbf{U}_{nT}\mathbf{G}_{T}%
\boldsymbol{\lambda}_{T}^{\ast},\nonumber
\end{align}
also
\begin{align}
n^{-1}\mathbf{\hat{B}}_{nT}^{\prime}\mathbf{M}_{n}\mathbf{\hat{B}}_{nT}  &
=n^{-1}\left(  \mathbf{B}_{n}\mathbf{+U}_{nT}\mathbf{G}_{T}\right)  ^{\prime
}\mathbf{M}_{n}\left(  \mathbf{B}_{n}\mathbf{+U}_{nT}\mathbf{G}_{T}\right)
\label{B'MB}\\
&  =n^{-1}\left(  \mathbf{B}_{n}^{\prime}\mathbf{M}_{n}\mathbf{B}_{n}\right)
+n^{-1}\left(  \mathbf{G}_{T}^{\prime}\mathbf{U}_{nT}^{\prime}\mathbf{M}%
_{n}\mathbf{B}_{n}\right)  +\nonumber\\
&  n^{-1}\left(  \mathbf{B}_{n}^{\prime}\mathbf{M}_{n}\mathbf{U}%
_{nT}\mathbf{G}_{T}\right)  +n^{-1}\left(  \mathbf{G}_{T}^{\prime}%
\mathbf{U}_{nT}^{\prime}\mathbf{M}_{n}\mathbf{U}_{nT}\mathbf{G}_{T}\right)
,\nonumber
\end{align}
Now using the results in Lemma \ref{L2} for case where all the observed
factors are strong, for a fixed $T$ and as $n\rightarrow\infty$ we have%
\begin{equation}
n^{-1}\left(  \mathbf{\hat{B}}_{nT}^{\prime}\mathbf{M}_{n}\mathbf{\hat{B}%
}_{nT}\right)  =\boldsymbol{\Sigma}_{\beta\beta}+\bar{\sigma}^{2}%
\mathbf{G}_{T}^{\prime}\mathbf{G}_{T}+o_{p}(1), \label{Bhat'MBhat}%
\end{equation}
where, using (\ref{GT}),%
\begin{equation}
\mathbf{G}_{T}^{\prime}\mathbf{G}_{T}=\frac{1}{T}\left(  \frac{\mathbf{F}%
^{\prime}\mathbf{M}_{T}\mathbf{F}}{T}\right)  ^{-1}. \label{GG}%
\end{equation}
Similarly, for the terms on the right hand side of (\ref{pnT}) we have%
\[
\mathbf{p}_{nT}=-\frac{\bar{\sigma}^{2}}{T}\left(  \frac{\mathbf{F}^{\prime
}\mathbf{M}_{T}\mathbf{F}}{T}\right)  ^{-1}\boldsymbol{\lambda}_{T}^{\ast
}+O_{p}\left(  n^{-1+\alpha_{\eta}}\right)  +O_{p}\left(  T^{-1/2}%
n^{-1/2}\right)  +O_{p}\left(  T^{-1/2}n^{-1+\frac{\alpha_{\eta}%
+\alpha_{\gamma}}{2}}\right)  .
\]
It is now easily seen that for a fixed $T$, and if $\alpha_{\eta}<1$ and
$\alpha_{\gamma}<1/2$, then as $n\rightarrow\infty$
\[
\mathbf{p}_{nT}\rightarrow_{p}-\frac{\bar{\sigma}^{2}}{T}\left(
\frac{\mathbf{F}^{\prime}\mathbf{M}_{T}\mathbf{F}}{T}\right)  ^{-1}%
\boldsymbol{\lambda}_{T}^{\ast}.
\]
Also, for a fixed $T$ by Assumption \ref{factors}, $\frac{\bar{\sigma}^{2}}%
{T}\left(  \frac{\mathbf{FM}_{T}\mathbf{F}}{T}\right)  ^{-1}$ is a positive
definite matrix, and by (\ref{Bhat'MBhat}) $n^{-1}\mathbf{\hat{B}}%
_{nT}^{\prime}\mathbf{M}_{n}\mathbf{\hat{B}}_{nT}\rightarrow_{p}%
\boldsymbol{\Sigma}_{\beta\beta}+\frac{\bar{\sigma}^{2}}{T}\left(
\frac{\mathbf{FM}_{T}\mathbf{F}}{T}\right)  ^{-1}$ which is also a positive
definite matrix, noting that under Assumption \ref{loadings}
$\boldsymbol{\Sigma}_{\beta\beta}$ is a positive definite matrix. Using these
results in (\ref{phidis1}) we now have%
\[
\boldsymbol{\hat{\phi}}_{nT}-\boldsymbol{\phi}_{0}=-\frac{\bar{\sigma}^{2}}%
{T}\left[  \boldsymbol{\Sigma}_{\beta\beta}+\bar{\sigma}^{2}\frac{1}{T}\left(
\frac{\mathbf{F}^{\prime}\mathbf{M}_{T}\mathbf{F}}{T}\right)  ^{-1}\right]
^{-1}\left(  \frac{\mathbf{F}^{\prime}\mathbf{M}_{T}\mathbf{F}}{T}\right)
^{-1}\boldsymbol{\lambda}_{T}^{\ast}+o_{p}(1)\text{, for a fixed }T\text{ as
}n\rightarrow\infty\text{.}%
\]
The bias of estimating $\mathbf{\lambda}_{0}$ by the two-step estimator also
contains the bias of estimating $\boldsymbol{\mu}_{0}$. Using the above result
in (\ref{lam0})\ we now have (for a fixed $T$ and as $n\rightarrow\infty$)
\[
\boldsymbol{\hat{\lambda}}_{nT}-\boldsymbol{\lambda}_{0}=\left(
\boldsymbol{\hat{\mu}}_{T}-\boldsymbol{\mu}_{0}\right)  -\frac{\bar{\sigma
}^{2}}{T}\left[  \boldsymbol{\Sigma}_{\beta\beta}+\bar{\sigma}^{2}\frac{1}%
{T}\left(  \frac{\mathbf{F}^{\prime}\mathbf{M}_{T}\mathbf{F}}{T}\right)
^{-1}\right]  ^{-1}\left(  \frac{\mathbf{F}^{\prime}\mathbf{M}_{T}\mathbf{F}%
}{T}\right)  ^{-1}\boldsymbol{\lambda}_{T}^{\ast}+o_{p}(1),
\]
which establishes Theorem \ref{TFMbias}. }

\subsection{{\protect\small Proof of theorem \ref{Thzig}}\label{ProofThzig}}

{\small Using the expression for $\hat{u}_{it}$ given by (\ref{res}), we have
$\hat{u}_{it}=\mathit{\alpha}_{i}-\hat{\alpha}_{iT}-\left(  \boldsymbol{\hat
{\beta}}_{i,T}-\boldsymbol{\beta}_{i}\right)  ^{\prime}\mathbf{f}_{t}+u_{it}$.
Since $\hat{u}_{it}$ are OLS residuals then for each $i$, we also have
$T^{-1}\sum_{t=1}^{T}\hat{u}_{it}=0$, and the above can be written
equivalently as $\hat{u}_{it}=u_{it}-\bar{u}_{i\circ}-\left(  \boldsymbol{\hat
{\beta}}_{i,T}-\boldsymbol{\beta}_{i}\right)  ^{\prime}\left(  \mathbf{f}%
_{t}-\boldsymbol{\hat{\mu}}_{T}\right)  ,$ for $i=1,2,...,n,$ and stacking
over $i$ now yields
\begin{equation}
\mathbf{\hat{u}}_{t}=\mathbf{u}_{t}-\mathbf{\bar{u}}-\left(  \mathbf{\hat{B}%
}_{nT}\mathbf{-B}_{n}\right)  \left(  \mathbf{f}_{t}-\boldsymbol{\hat{\mu}%
}_{T}\right)  =\mathbf{u}_{t}-\mathbf{\bar{u}}-\mathbf{U}_{nT}\mathbf{G}%
_{T}\left(  \mathbf{f}_{t}-\boldsymbol{\hat{\mu}}_{T}\right)  , \label{uhat}%
\end{equation}
and stacking over $t$%
\[
\mathbf{\hat{U}}_{nT}=\mathbf{U}_{nT}\mathbf{M}_{T}-\mathbf{U}_{nT}%
\mathbf{G}_{T}\mathbf{F}^{\prime}\mathbf{M}_{T}=\mathbf{U}_{nT}\left(
\mathbf{M}_{T}-\mathbf{G}_{T}\mathbf{F}^{\prime}\mathbf{M}_{T}\right)  .
\]
But $\mathbf{G}_{T}\mathbf{=M}_{T}\mathbf{F}\left(  \mathbf{F}^{\prime
}\mathbf{M}_{T}\mathbf{F}\right)  ^{-1}$ , and we have%
\[
\mathbf{\hat{U}}_{nT}=\mathbf{U}_{nT}\mathbf{R}_{T},\text{ }\mathbf{R}%
_{T}=\mathbf{M}_{T}-\mathbf{M}_{T}\mathbf{F}\left(  \mathbf{F}^{\prime
}\mathbf{M}_{T}\mathbf{F}\right)  ^{-1}\mathbf{F}^{\prime}\mathbf{M}_{T},
\]
where $\mathbf{R}_{T}^{2}=\mathbf{R}_{T}=\mathbf{R}_{T}^{\prime}$, $Tr\left(
\mathbf{R}_{T}\right)  =T-1-K$. Then%
\[
\widehat{\bar{\sigma}}_{nT}^{2}=\frac{\sum_{t=1}^{T}\sum_{i=1}^{n}\hat{u}%
_{it}^{2}}{n(T-K-1)}=\frac{Tr\left(  n^{-1}\mathbf{\hat{U}}_{nT}^{\prime
}\mathbf{\hat{U}}_{nT}\right)  }{T-K-1}.
\]
Also
\begin{align*}
n^{-1}T^{-1}E\left[  Tr\left(  \mathbf{U}_{nT}^{\prime}\mathbf{U}_{nT}\right)
\right]   &  =n^{-1}T^{-1}E\left(  \sum_{t=1}^{T}\sum_{i=1}^{n}u_{it}%
^{2}\right)  =n^{-1}\sum_{i=1}^{n}\sigma_{i}^{2}=\bar{\sigma}_{n}^{2},\\
E\left(  n^{-1}\mathbf{U}_{nT}^{\prime}\mathbf{U}_{nT}\right)   &  =n^{-1}%
\sum_{i=1}^{n}E\left(  \mathbf{u}_{i\circ}\mathbf{u}_{i\circ}^{\prime}\right)
=\bar{\sigma}_{n}^{2}\mathbf{I}_{T}.
\end{align*}
Let $v=T-K-1$ and note that%
\begin{align*}
v\widehat{\bar{\sigma}}_{nT}^{2}  &  =Tr\left(  n^{-1}\mathbf{\hat{U}}%
_{nT}^{\prime}\mathbf{\hat{U}}_{nT}\right)  =Tr\left(  n^{-1}\mathbf{U}%
_{nT}^{\prime}\mathbf{U}_{nT}\mathbf{R}_{T}\right)  =\\
&  =Tr\left(  n^{-1}\mathbf{U}_{nT}^{\prime}\mathbf{U}_{nT}\mathbf{M}%
_{T}\right)  -Tr\left(  n^{-1}\mathbf{F}^{\prime}\mathbf{M}_{T}\mathbf{U}%
_{nT}^{\prime}\mathbf{U}_{nT}\mathbf{M}_{T}\mathbf{F}\left(  \mathbf{F}%
^{\prime}\mathbf{M}_{T}\mathbf{F}\right)  ^{-1}\right) \\
&  =Tr\left(  n^{-1}\mathbf{U}_{nT}^{\prime}\mathbf{U}_{nT}\right)
-T^{-1}\mathbf{\tau}_{T}^{\prime}\left(  n^{-1}\mathbf{U}_{nT}^{\prime
}\mathbf{U}_{nT}\right)  \mathbf{\tau}_{T}-Tr\left[  \mathbf{Q}^{\prime
}\left(  n^{-1}\mathbf{U}_{nT}^{\prime}\mathbf{U}_{nT}\right)  \mathbf{Q}%
\right]
\end{align*}
where $\mathbf{Q=M}_{T}\mathbf{F}\left(  T^{-1}\mathbf{F}^{\prime}%
\mathbf{M}_{T}\mathbf{F}\right)  ^{-1/2}$. Consider the first term and note
that%
\begin{equation}
n^{-1}Tr\left(  \mathbf{U}_{nT}^{\prime}\mathbf{U}_{nT}\right)  =\sum
_{t=1}^{T}\left[  n^{-1}\sum_{i=1}^{n}\left(  u_{it}^{2}-\sigma_{i}%
^{2}\right)  \right]  +T\bar{\sigma}_{n}^{2}. \label{TrUU'}%
\end{equation}
Similarly%
\begin{align*}
T^{-1}\mathbf{\tau}_{T}^{\prime}\left(  n^{-1}\mathbf{U}_{nT}^{\prime
}\mathbf{U}_{nT}\right)  \mathbf{\tau}_{T}  &  =T^{-1}n^{-1}\mathbf{\tau}%
_{T}^{\prime}\left[  \mathbf{U}_{nT}^{\prime}\mathbf{U}_{nT}-E\left(
\mathbf{U}_{nT}^{\prime}\mathbf{U}_{nT}\right)  \right]  \mathbf{\tau}%
_{T}+\bar{\sigma}_{n}^{2},\\
Tr\left[  \mathbf{Q}^{\prime}\left(  n^{-1}\mathbf{U}_{nT}^{\prime}%
\mathbf{U}_{nT}\right)  \mathbf{Q}\right]   &  =Tr\left[  \mathbf{Q}^{\prime
}\left[  \mathbf{U}_{nT}^{\prime}\mathbf{U}_{nT}-E\left(  \mathbf{U}%
_{nT}^{\prime}\mathbf{U}_{nT}\right)  \right]  \mathbf{Q}\right]
+K\bar{\sigma}_{n}^{2}.
\end{align*}
Hence
\begin{equation}
\widehat{\bar{\sigma}}_{nT}^{2}-\bar{\sigma}_{n}^{2}=\left(  T/v\right)
\left(  a_{nT}+b_{nT}+c_{nT}\right)  , \label{ziggap}%
\end{equation}
where $a_{nT}=T^{-1}\sum_{t=1}^{T}\left[  n^{-1}\sum_{i=1}^{n}\left(
u_{it}^{2}-\sigma_{i}^{2}\right)  \right]  ,$ $\ b_{nT}=T^{-2}n^{-1}%
\mathbf{\tau}_{T}^{\prime}\left[  \mathbf{U}_{nT}^{\prime}\mathbf{U}%
_{nT}-E\left(  \mathbf{U}_{nT}^{\prime}\mathbf{U}_{nT}\right)  \right]
\mathbf{\tau}_{T},$ and $c_{nT}=T^{-1}n^{-1}Tr\left[  \mathbf{Q}^{\prime
}\left[  \mathbf{U}_{nT}^{\prime}\mathbf{U}_{nT}-E\left(  \mathbf{U}%
_{nT}^{\prime}\mathbf{U}_{nT}\right)  \right]  \mathbf{Q}\right]  $. Due to
the independence of $\mathbf{F\,}$and $\mathbf{U}_{nT}$, we have $E(a_{nT}%
)=0$, $E(b_{nT})=0$ and $E(c_{nT})=0$, and for any fixed $n$ and $T$ $E\left(
\widehat{\bar{\sigma}}_{nT}^{2}-\bar{\sigma}_{n}^{2}\right)  =0$, and for a
fixed $T,$ $\lim_{n\rightarrow\infty}E\left(  \widehat{\bar{\sigma}}_{nT}%
^{2}\right)  =\bar{\sigma}^{2},$ and result (\ref{unbiasZig}) follows. To
establish the probability order of $\widehat{\bar{\sigma}}_{nT}^{2}%
-\bar{\sigma}_{n}^{2}$, we consider the probability orders of $a_{nT},$
$b_{nT}$, and $c_{nT}$ in turn, noting that that $T/v=T/(T-K-1)=O(1)$. For
$a_{nT}$, using result (\ref{aa}) in Lemma \ref{L1}, and noting that $u_{it}$
are serially independent we have%
\begin{equation}
a_{nT}=O_{p}\left(  T^{-1/2}n^{-1/2}\right)  . \label{anT}%
\end{equation}
Consider now $b_{nT}$ and note that $b_{nT}=T^{-2}n^{-1}\sum_{i=1}^{n}\left[
\mathbf{\tau}_{T}^{\prime}\mathbf{u}_{i\circ}\mathbf{u}_{i\circ}^{\prime
}\mathbf{\tau}_{T}-E\left(  \mathbf{\tau}_{T}^{\prime}\mathbf{u}_{i\circ
}\mathbf{u}_{i\circ}^{\prime}\mathbf{\tau}_{T}\right)  \right]  .$ Also
$\mathbf{\tau}_{T}^{\prime}\mathbf{u}_{i\circ}\mathbf{u}_{i\circ}^{\prime
}\mathbf{\tau}_{T}=\sum_{t=1}^{T}\sum_{t^{\prime}=1}^{T}u_{it}u_{it^{\prime}}%
$, and%
\begin{align*}
b_{nT}  &  =T^{-2}n^{-1}\sum_{i=1}^{n}\sum_{t=1}^{T}\sum_{t^{\prime}=1}%
^{T}\left[  u_{it}u_{it^{\prime}}-E\left(  u_{it}u_{it^{\prime}}\right)
\right] \\
&  =T^{-2}\sum_{t=1}^{T}n^{-1}\sum_{i=1}^{n}\left(  u_{it}^{2}-\sigma_{i}%
^{2}\right)  +T^{-2}\sum_{t\neq t^{\prime}}^{T}\left(  n^{-1}\sum_{i=1}%
^{n}u_{it}u_{it^{\prime}}\right) \\
&  =T^{-2}\sum_{t=1}^{T}a_{n,tt}+T^{-2}\sum_{t\neq t^{\prime}}^{T}%
a_{n,tt^{\prime}},
\end{align*}
where $a_{n,tt}$ and $a_{n,tt^{\prime}}$ are both shown in Lemma \ref{L1} to
be $O_{p}(n^{-1/2})$. See equations (\ref{aa}) and (\ref{aa2}). Therefore,
given that $a_{n,tt}$ and $a_{n,tt^{\prime}}$ with $t\neq t^{\prime}$ are also
distributed independently over $t$ we have
\begin{equation}
b_{nT}=O_{p}(n^{-1/2}T^{-1/2})\text{.} \label{bnT}%
\end{equation}
Denote the $k^{th}$ column of $\mathbf{Q}$ by $\mathbf{q}_{k}=(q_{k1}%
,q_{2k},...,q_{Tk})^{\prime}$ (a $T\times1$ vector) and write $c_{nT}$ as%
\begin{align*}
c_{nT}  &  =\sum_{k=1}^{K}T^{-1}\left[  n^{-1}\sum_{i=1}^{n}\mathbf{q}%
_{k}^{\prime}\left[  \mathbf{u}_{i\circ}\mathbf{u}_{i\circ}^{\prime}-E\left(
\mathbf{u}_{i\circ}\mathbf{u}_{i\circ}^{\prime}\right)  \right]
\mathbf{q}_{k}\right] \\
&  =\sum_{k=1}^{K}T^{-1}\left[  n^{-1}\sum_{i=1}^{n}\sum_{t=1}^{T}%
\sum_{t^{\prime}=1}^{T}q_{kt}q_{kt^{\prime}}\left[  u_{it}u_{it^{\prime}%
}-E\left(  u_{it}u_{it^{\prime}}\right)  \right]  \right]  .
\end{align*}
Consider the $k^{th}$ term of the above sum, and note that
\[
c_{nT,k}=T^{-1}n^{-1}\sum_{i=1}^{n}\sum_{t=1}^{T}\sum_{t^{\prime}=1}^{T}%
q_{kt}q_{kt^{\prime}}\left[  u_{it}u_{it^{\prime}}-E\left(  u_{it}%
u_{it^{\prime}}\right)  \right]  =T^{-1}\sum_{t=1}^{T}\sum_{t^{\prime}=1}%
^{T}q_{kt}q_{kt^{\prime}}a_{n,tt^{\prime}},
\]
where $a_{n,tt^{\prime}}=n^{-1}\sum_{i=1}^{n}\left[  u_{it}u_{it^{\prime}%
}-E\left(  u_{it}u_{it^{\prime}}\right)  \right]  $. is defined by (\ref{aa})
and (\ref{aa2}) in Lemma \ref{L1}, with $Var(a_{n,tt^{\prime}})=O\left(
n^{-1}\right)  $ for all $t$ and $t^{\prime}$. Also $a_{n,tt^{\prime}}$ and
$a_{n,ss^{\prime}}$ are distributed independently if $t$ or $t^{\prime}$
differ from $s$ or $s^{\prime}$. Therefore, for all $t$ and $t^{\prime}$ for
all $t$
\[
Var\left(  c_{nT,k}\right)  =T^{-2}\sum_{t=1}^{T}\sum_{t^{\prime}=1}^{T}%
q_{kt}^{2}q_{kt^{\prime}}^{2}Var(a_{n,tt^{\prime}})\leq O\left(
n^{-1}\right)  \left(  T^{-2}\sum_{t=1}^{T}q_{kt}^{2}\right)  .
\]
But it is easily verified that $\mathbf{Q}^{\prime}\mathbf{Q=I}_{K}$ which
yields $\sum_{t=1}^{T}q_{kt}^{2}=1,$ and $Var\left(  c_{nT,k}\right)
=O\left(  T^{-2}n^{-1}\right)  $. Hence, $c_{nT,k}=O_{p}(T^{-1}n^{-1/2})$, for
$k=1,2,...,K$, which establishes that $c_{nT}=O_{p}(T^{-1}n^{-1/2}\dot{)}$.
The order result in (\ref{orderzig}) now follows using this result,
(\ref{anT}) and (\ref{bnT}) in (\ref{ziggap}). }

\subsection{{\protect\small Proof of theorem \ref{Tfi}}\label{ProofTfi}}

{\small The bias-corrected estimator of $\boldsymbol{\phi}_{0}$ is given by
(\ref{phitilda}) which we reproduce here and re-write as
\begin{equation}
\mathbf{H}_{nT\ }\left(  \boldsymbol{\tilde{\phi}}_{nT}-\boldsymbol{\phi}%
_{0}\right)  =\mathbf{s}_{nT}, \label{A-phitilda}%
\end{equation}
where%
\begin{equation}
\mathbf{s}_{nT}=\frac{\mathbf{\hat{B}}_{nT}^{\prime}\mathbf{M}_{n}%
\boldsymbol{\hat{\alpha}}_{nT}}{n}+T^{-1}\widehat{\bar{\sigma}}_{nT}%
^{2}\left(  \frac{\mathbf{F}^{\prime}\mathbf{M}_{T}\mathbf{F}}{T}\right)
^{-1}\mathbf{\hat{\mu}}_{T}-\mathbf{H}_{nT\ }\boldsymbol{\phi}_{0},
\label{snT0}%
\end{equation}%
\begin{equation}
\mathbf{H}_{nT\ }=\frac{\mathbf{\hat{B}}_{nT}^{\prime}\mathbf{M}%
_{n}\mathbf{\hat{B}}_{nT}}{n}-T^{-1}\widehat{\bar{\sigma}}_{nT}^{2}\left(
\frac{\mathbf{F}^{\prime}\mathbf{M}_{T}\mathbf{F}}{T}\right)  ^{-1}.
\label{A-HnT}%
\end{equation}
Also $\mathbf{\hat{a}}_{nT}=\mathbf{\bar{r}}_{n\circ}-\mathbf{\hat{B}}%
_{nT}\mathbf{\hat{\mu}}_{T}$, and $\mathbf{\bar{r}}_{n\circ}$\textbf{
}$=c\boldsymbol{\tau}_{n}+\mathbf{B}_{n}\boldsymbol{\lambda}_{T}^{\ast
}+\mathbf{\bar{u}}_{n\circ}\mathbf{+}\boldsymbol{\eta}_{n}$ (see
(\ref{rbar1})). Using these results and noting that $\boldsymbol{\lambda}%
_{T}^{\ast}=\boldsymbol{\lambda}_{0}+\left(  \boldsymbol{\hat{\mu}}%
_{T}-\boldsymbol{\mu}_{0}\right)  =\boldsymbol{\phi}_{0}+\boldsymbol{\hat{\mu
}}_{T}$, we have
\[
\mathbf{\hat{\alpha}}_{nT}=c\boldsymbol{\tau}_{n}+\mathbf{B}_{n}%
\boldsymbol{\phi}_{0}+\mathbf{\bar{u}}_{n\circ}+\boldsymbol{\eta}_{n}-\left(
\mathbf{\hat{B}}_{nT}-\mathbf{B}_{n}\right)  \boldsymbol{\hat{\mu}}_{T}%
\]
and $\mathbf{\bar{u}}_{n\circ}=(\bar{u}_{1\circ},\bar{u}_{2\circ},...,\bar
{u}_{n\circ})^{\prime}$. Then,%
\begin{align}
\mathbf{\hat{B}}_{nT}^{\prime}\mathbf{M}_{n}\boldsymbol{\hat{\alpha}}_{nT}  &
=\left(  \mathbf{\hat{B}}_{nT}^{\prime}\mathbf{M}_{n}\mathbf{B}_{n}\right)
\boldsymbol{\phi}_{0}+\mathbf{\hat{B}}_{nT}^{\prime}\mathbf{M}_{n}%
\mathbf{\bar{u}}_{n\circ}\label{B'Ma}\\
&  +\mathbf{\hat{B}}_{nT}^{\prime}\mathbf{M}_{n}\boldsymbol{\eta}%
_{n}-\mathbf{\hat{B}}_{nT}^{\prime}\mathbf{M}_{n}\left(  \mathbf{\hat{B}}%
_{nT}-\mathbf{B}_{n}\right)  \boldsymbol{\hat{\mu}}_{T}.\nonumber
\end{align}
Also
\begin{equation}
\mathbf{\hat{B}}_{nT}=\mathbf{B}_{n}\mathbf{+U}_{nT}\mathbf{G}_{T}, \label{UG}%
\end{equation}
where $\mathbf{U}_{nT}=\mathbf{(u}_{1\circ},\mathbf{u}_{2\circ},...,\mathbf{u}%
_{n\circ})^{\prime}$ and $\mathbf{G}_{T}$ is defined by (\ref{GG}). Using
these results together with (\ref{A-HnT}), the right hand side of
(\ref{A-phitilda}) can be written as
\begin{align*}
\mathbf{s}_{nT}  &  =\frac{\mathbf{B}_{n}^{\prime}\mathbf{M}_{n}%
\boldsymbol{\eta}_{n}}{n}+\frac{\mathbf{G}_{T}^{\prime}\mathbf{U}_{nT}%
^{\prime}\mathbf{M}_{n}\boldsymbol{\eta}_{n}}{n}+\frac{\mathbf{G}_{T}^{\prime
}\mathbf{U}_{nT}^{\prime}\mathbf{M}_{n}\mathbf{\bar{u}}_{n\circ}}{n}%
-\frac{\mathbf{B}_{n}^{\prime}\mathbf{M}_{n}\mathbf{U}_{nT}\mathbf{G}%
_{T}\boldsymbol{\lambda}_{T}^{\ast}}{n}\\
&  -\mathbf{G}_{T}^{\prime}\left(  \frac{\mathbf{U}_{nT}^{\prime}%
\mathbf{M}_{n}\mathbf{U}_{nT}}{n}-\widehat{\bar{\sigma}}_{nT}^{2}\right)
\mathbf{G}_{T}\boldsymbol{\lambda}_{T}^{\ast},
\end{align*}
where the last term can be decomposed as%
\[
\mathbf{G}_{T}^{\prime}\left(  \frac{\mathbf{U}_{nT}^{\prime}\mathbf{M}%
_{n}\mathbf{U}_{nT}}{n}-\widehat{\bar{\sigma}}_{nT}^{2}\right)  \mathbf{G}%
_{T}\boldsymbol{\lambda}_{T}^{\ast}=\mathbf{G}_{T}^{\prime}\left(
\frac{\mathbf{U}_{nT}^{\prime}\mathbf{M}_{n}\mathbf{U}_{nT}}{n}-\bar{\sigma
}_{n}^{2}\right)  \mathbf{G}_{T}\boldsymbol{\lambda}_{T}^{\ast}-\left(
\widehat{\bar{\sigma}}_{nT}^{2}-\bar{\sigma}_{n}^{2}\right)  \mathbf{G}%
_{T}^{\prime}\mathbf{G}_{T}\boldsymbol{\lambda}_{T}^{\ast}.
\]
Similarly, using (\ref{B'MB}), we have%
\begin{align*}
\mathbf{H}_{nT\ }  &  =n^{-1}\left(  \mathbf{B}_{n}^{\prime}\mathbf{M}%
_{n}\mathbf{B}_{n}\right)  +n^{-1}\left(  \mathbf{G}_{T}^{\prime}%
\mathbf{U}_{nT}^{\prime}\mathbf{M}_{n}\mathbf{B}_{n}\right)  +\\
&  n^{-1}\left(  \mathbf{B}_{n}^{\prime}\mathbf{M}_{n}\mathbf{U}%
_{nT}\mathbf{G}_{T}\right)  +\mathbf{G}_{T}^{\prime}\left(  \frac
{\mathbf{U}_{nT}^{\prime}\mathbf{M}_{n}\mathbf{U}_{nT}}{n}-\bar{\sigma}%
_{n}^{2}\right)  \mathbf{G}_{T}-T^{-1}\left(  \widehat{\bar{\sigma}}_{nT}%
^{2}-\bar{\sigma}_{n}^{2}\right)  \left(  \frac{\mathbf{F}^{\prime}%
\mathbf{M}_{T}\mathbf{F}}{T}\right)  ^{-1}.
\end{align*}
Using Theorem \ref{Thzig} and since by assumption $T^{-1}\mathbf{F}^{\prime
}\mathbf{M}_{T}\mathbf{F}$ is positive definite, then%
\[
\left(  \widehat{\bar{\sigma}}_{nT}^{2}-\bar{\sigma}_{n}^{2}\right)  \left(
\frac{\mathbf{F}^{\prime}\mathbf{M}_{T}\mathbf{F}}{T}\right)  ^{-1}%
=O_{p}\left(  T^{-1/2}n^{-1/2}\right)  .
\]
Also using results in Lemma \ref{L2}, we have%
\begin{equation}
\mathbf{H}_{nT}=n^{-1}\left(  \mathbf{B}_{n}^{\prime}\mathbf{M}_{n}%
\mathbf{B}_{n}\right)  +O_{p}\left(  T^{-1/2}n^{-1/2}\right)  . \label{AHnT}%
\end{equation}
Hence, $\mathbf{H}_{nT}\rightarrow_{p}\mathbf{\Sigma}_{\beta\beta}$ for a
fixed $T$ as $n\rightarrow\infty$, so long as $\alpha_{\gamma}<1/2$. Note also
that by Assumption $\mathbf{\Sigma}_{\beta\beta}$ is positive definite.
Further%
\begin{align}
\mathbf{s}_{nT}  &  =n^{-1}\left(  \mathbf{B}_{n}^{\prime}\mathbf{M}%
_{n}\mathbf{\bar{u}}_{n\circ}-\mathbf{B}_{n}^{\prime}\mathbf{M}_{n}%
\mathbf{U}_{nT}\mathbf{G}_{T}\boldsymbol{\lambda}_{T}^{\ast}\right)
+n^{-1}\left(  \mathbf{B}_{n}^{\prime}+\mathbf{G}_{T}^{\prime}\mathbf{U}%
_{nT}^{\prime}\right)  \mathbf{M}_{n}\boldsymbol{\eta}_{n}+n^{-1}%
\mathbf{G}_{T}^{\prime}\mathbf{U}_{nT}^{\prime}\mathbf{M}_{n}\mathbf{\bar{u}%
}_{n\circ}\label{snT1}\\
&  -\mathbf{G}_{T}^{\prime}\left(  \frac{\mathbf{U}_{nT}^{\prime}%
\mathbf{M}_{n}\mathbf{U}_{nT}}{n}-\bar{\sigma}_{n}^{2}\right)  \mathbf{G}%
_{T}\boldsymbol{\lambda}_{T}^{\ast}+T^{-1}\left(  \widehat{\bar{\sigma}}%
_{nT}^{2}-\bar{\sigma}_{n}^{2}\right)  \left(  \frac{\mathbf{F}^{\prime
}\mathbf{M}_{T}\mathbf{F}}{T}\right)  ^{-1}\boldsymbol{\lambda}_{T}^{\ast
}.\nonumber
\end{align}
Using (\ref{ABeta}) and (\ref{GUeta}) (in Lemma \ref{L2}) we have%
\[
n^{-1}\mathbf{B}_{n}^{\prime}\mathbf{M}_{n}\boldsymbol{\eta}_{n}%
+n^{-1}\mathbf{G}_{T}^{\prime}\mathbf{U}_{nT}^{\prime}\mathbf{M}%
_{n}\boldsymbol{\eta}_{n}=O_{p}\left(  n^{-1+\alpha_{\eta}}\right)
+O_{p}\left(  T^{-1/2}n^{-1+\frac{\alpha_{\eta+\alpha_{\gamma}}}{2}}\right)
,
\]
and (\ref{GUubar}) and (\ref{GUUG})%
\[
n^{-1}\mathbf{G}_{T}^{\prime}\mathbf{U}_{nT}^{\prime}\mathbf{M}_{n}%
\mathbf{\bar{u}}_{n\circ}-\mathbf{G}_{T}^{\prime}\left(  \frac{\mathbf{U}%
_{nT}^{\prime}\mathbf{M}_{n}\mathbf{U}_{nT}}{n}-\bar{\sigma}_{n}^{2}\right)
\mathbf{G}_{T}\boldsymbol{\lambda}_{T}^{\ast}=O_{p}\left(  T^{-1}%
n^{-1/2}\right)  .
\]
Further by (\ref{orderzig})%
\begin{equation}
T^{-1}\left(  \widehat{\bar{\sigma}}_{nT}^{2}-\bar{\sigma}_{n}^{2}\right)
\left(  \frac{\mathbf{F}^{\prime}\mathbf{M}_{T}\mathbf{F}}{T}\right)
^{-1}\boldsymbol{\lambda}_{T}^{\ast}=O_{p}\left(  T^{-3/2}n^{-1/2}\right)  .
\label{ziggap1}%
\end{equation}
Hence
\begin{align}
\mathbf{s}_{nT}  &  =n^{-1}\left(  \mathbf{B}_{n}^{\prime}\mathbf{M}%
_{n}\mathbf{\bar{u}}_{n\circ}-\mathbf{B}_{n}^{\prime}\mathbf{M}_{n}%
\mathbf{U}_{nT}\mathbf{G}_{T}\boldsymbol{\lambda}_{T}^{\ast}\right)
+O_{p}\left(  n^{-1+\alpha_{\eta}}\right) \label{snT2}\\
&  +O_{p}\left(  T^{-1/2}n^{-1+\frac{\alpha_{\eta+\alpha_{\gamma}}}{2}%
}\right)  +O_{p}\left(  T^{-1}n^{-1/2}\right)  +O_{p}\left(  T^{-3/2}%
n^{-1/2}\right)  .\nonumber
\end{align}
Using this result and (\ref{AHnT}) in (\ref{A-phitilda}) now yields
(\ref{phigap}), as required. To derive the asymptotic distribution of
$\boldsymbol{\tilde{\phi}}_{nT}-\boldsymbol{\phi}_{0}$ since by assumption
$\alpha_{\eta}+\alpha_{\gamma}<1$, then the dominant term of $\mathbf{s}_{nT}$
is given by
\begin{equation}
n^{-1}\left(  \mathbf{B}_{n}^{\prime}\mathbf{M}_{n}\mathbf{\bar{u}}_{n\circ
}-\mathbf{B}_{n}^{\prime}\mathbf{M}_{n}\mathbf{U}_{nT}\mathbf{G}%
_{T}\boldsymbol{\lambda}_{T}^{\ast}\right)  =O_{p}\left(  T^{-1/2}%
n^{-1/2}\right)  , \label{Doms_nT}%
\end{equation}
and to ensure that we end up with a non-degenerate, stable limiting
distribution, $\left(  \boldsymbol{\tilde{\phi}}_{nT}-\boldsymbol{\phi}%
_{0}\right)  $ needs to be scaled by $\sqrt{nT}$ with $n$ and $T\rightarrow
\infty$, jointly. To this end we first note that when $T$ is fixed
$\mathbf{H}_{nT}\rightarrow_{p}\mathbf{\Sigma}_{\beta\beta},$ as
$n\rightarrow\infty$ and we have%
\begin{equation}
\sqrt{nT}\left(  \boldsymbol{\tilde{\phi}}_{nT}-\boldsymbol{\phi}_{0}\right)
\overset{a}{\thicksim}\mathbf{\Sigma}_{\beta\beta}^{-1}\left(  \sqrt
{nT}\mathbf{s}_{nT}\right)  . \label{jointnT}%
\end{equation}
Using (\ref{snT2})%
\[
\sqrt{nT}\mathbf{s}_{nT}=\boldsymbol{\xi}_{nT}+O_{p}\left(  n^{-\frac{1}%
{2}+\frac{\alpha_{\eta+\alpha_{\gamma}}}{2}}\right)  +O_{p}\left(
T^{1/2}n^{-1/2+\alpha_{\eta}}\right)  +O_{p}\left(  T^{-1/2}\right)  ,
\]
where
\begin{equation}
\boldsymbol{\xi}_{nT}=T^{1/2}n^{-1/2}\left(  \mathbf{B}_{n}^{\prime}%
\mathbf{M}_{n}\mathbf{\bar{u}}_{n\circ}-\mathbf{B}_{n}^{\prime}\mathbf{M}%
_{n}\mathbf{U}_{nT}\mathbf{G}_{T}\boldsymbol{\lambda}_{T}^{\ast}\right)  .
\label{egzinT}%
\end{equation}
Using the above results in (\ref{jointnT}) we now have
\begin{equation}
\sqrt{nT}\left(  \boldsymbol{\tilde{\phi}}_{nT}-\boldsymbol{\phi}_{0}\right)
\overset{a}{\thicksim}\mathbf{\Sigma}_{\beta\beta}^{-1}\left[  \boldsymbol{\xi
}_{nT}+O_{p}\left(  n^{-\frac{1}{2}+\frac{\alpha_{\eta+\alpha_{\gamma}}}{2}%
}\right)  +O_{p}\left(  T^{1/2}n^{-1/2+\alpha_{\eta}}\right)  +O_{p}\left(
T^{-1/2}\right)  \right]  . \label{jointnTphiD}%
\end{equation}
Therefore, when condition $(T/n)^{1/2}n^{\alpha_{\eta}}\rightarrow0$, as $n$
and $T\rightarrow\infty$, is met we have%
\[
\sqrt{nT}\left(  \boldsymbol{\tilde{\phi}}_{nT}-\boldsymbol{\phi}_{0}\right)
\overset{a}{\thicksim}\mathbf{\Sigma}_{\beta\beta}^{-1}\boldsymbol{\xi}%
_{nT}+o_{p}(1).
\]
To derive the asymptotic distribution of $\boldsymbol{\xi}_{nT}$, we note that
$\mathbf{\bar{u}}_{n\circ}=T^{-1}\mathbf{U}_{nT}\boldsymbol{\tau}_{T},$ and
$\mathbf{G}_{T}\boldsymbol{\lambda}_{T}^{\ast}=T^{-1}\mathbf{M}_{T}%
\mathbf{F}(T^{-1}\mathbf{F}^{\prime}\mathbf{M}_{T}\mathbf{F)}^{-1}%
\boldsymbol{\lambda}_{T}^{\ast}$. Then, using these results in (\ref{egzinT})
\begin{equation}
\boldsymbol{\xi}_{nT}=\left(  \xi_{k,nT\ }\right)  =n^{-1/2}T^{-1/2}%
\mathbf{B}_{n}^{\prime}\mathbf{M}_{n}\mathbf{U}_{nT}\mathbf{a}_{T},
\label{egzinTv}%
\end{equation}
where $\mathbf{a}_{T}=\boldsymbol{\tau}_{T}-\mathbf{M}_{T}\mathbf{F}%
(T^{-1}\mathbf{F}^{\prime}\mathbf{M}_{T}\mathbf{F)}^{-1}\boldsymbol{\lambda
}_{T}^{\ast}=(a_{t})$. Also%
\begin{equation}
s_{a,T}^{2}=T^{-1}%
{\displaystyle\sum\limits_{t=1}^{T}}
a_{t}^{2}=T^{-1}\mathbf{a}_{T}^{\prime}\mathbf{a}_{T}=1+\boldsymbol{\lambda
}_{T}^{\ast\prime}(T^{-1}\mathbf{F}^{\prime}\mathbf{M}_{T}\mathbf{F)}%
^{-1}\boldsymbol{\lambda}_{T}^{\ast}, \label{saT}%
\end{equation}
where $\boldsymbol{\lambda}_{T}^{\ast}=\boldsymbol{\phi}_{0}+\boldsymbol{\hat
{\mu}}_{T}=\boldsymbol{\lambda}_{0}+\left(  \boldsymbol{\hat{\mu}}%
_{T}-\boldsymbol{\mu}_{0}\right)  $, and $\left(  \boldsymbol{\hat{\mu}}%
_{T}-\boldsymbol{\mu}_{0}\right)  =O_{p}\left(  T^{-1/2}\right)  $. Further
\begin{equation}
s_{a,T}^{2}\geq1\text{ and }s_{a,T}^{2}\leq1+\left(  \boldsymbol{\lambda}%
_{T}^{\ast\prime}\boldsymbol{\lambda}_{T}^{\ast}\right)  \lambda_{\max}\left[
(T^{-1}\mathbf{F}^{\prime}\mathbf{M}_{T}\mathbf{F)}^{-1}\text{ }\right]  <C,
\label{BsaT}%
\end{equation}
and $s_{a}^{2}=\lim_{T\rightarrow\infty}s_{a,T}^{2}=1+\boldsymbol{\lambda}%
_{0}^{\prime}\mathbf{\Sigma}_{f}^{-1}\boldsymbol{\lambda}_{0}^{\prime}.$ The
$k^{th}$ element of $\boldsymbol{\xi}_{nT}$ is given by }%
\[
{\small \xi_{k,nT\ }=n^{-1/2}T^{-1/2}\sum_{i=1}^{n}%
{\displaystyle\sum\limits_{t=1}^{T}}
a_{t}(\beta_{ik}-\bar{\beta}_{k})u_{it},}%
\]
{\small and using (\ref{uit}) we have%
\begin{align*}
\xi_{k,nT\ }  &  =\left(  T^{-1/2}%
{\displaystyle\sum\limits_{t=1}^{T}}
a_{t}g_{t}\right)  \left[  n^{-1/2}\sum_{i=1}^{n}(\beta_{ik}-\bar{\beta}%
_{k})\gamma_{i}\right] \\
&  +n^{-1/2}T^{-1/2}\sum_{i=1}^{n}%
{\displaystyle\sum\limits_{t=1}^{T}}
a_{t}(\beta_{ik}-\bar{\beta}_{k})v_{it}.
\end{align*}
Under Assumption \ref{Latent factor} $g_{t}$ is distributed independently of
$\boldsymbol{f}_{t}$ (and hence of $a_{t}$), as well as being serially
independent. Also $Var(T^{-1/2}%
{\displaystyle\sum\limits_{t=1}^{T}}
a_{t}g_{t})=s_{a,T}^{2}$ (recall that $E(g_{t})=0$ and $E(g_{t}^{2}%
)=1$),\thinspace and we have $T^{-1/2}%
{\displaystyle\sum\limits_{t=1}^{T}}
a_{t}g_{t}=O_{p}(1)$. Further $E\left\vert n^{-1/2}\sum_{i=1}^{n}(\beta
_{ik}-\bar{\beta}_{k})\gamma_{i}\right\vert \leq\sup_{i,k}E\left\vert
\beta_{ik}-\bar{\beta}_{k}\right\vert \left(  n^{-1/2}\sum_{i=1}^{n}\left\vert
\gamma_{i}\right\vert \right)  =O(n^{-1/2+\alpha_{\gamma}})$. Hence%
\begin{align}
\xi_{k,nT\ }  &  =n^{-1/2}T^{-1/2}%
{\displaystyle\sum\limits_{t=1}^{T}}
\sum_{i=1}^{n}a_{t}(\beta_{ik}-\bar{\beta}_{k})v_{it}+O(n^{-1/2+\alpha
_{\gamma}})\nonumber\\
&  =T^{-1/2}%
{\displaystyle\sum\limits_{t=1}^{T}}
a_{t}h_{nt}+O(n^{-1/2+\alpha_{\gamma}}), \label{egziknT}%
\end{align}
where $h_{nt}=n^{-1/2}\sum_{i=1}^{n}(\beta_{ik}-\bar{\beta}_{k})v_{it}$. Under
Assumption \ref{loadings}, $h_{nt}=n^{-1/2}\sum_{i=1}^{n}(\beta_{ik}%
-\bar{\beta}_{k})v_{it}\rightarrow_{d}N(0,\omega_{k}^{2})$, for $k=1,2,...,K$,
where
\[
\omega_{k}^{2}=p\lim_{n\rightarrow\infty}n^{-1}\sum_{i=1}^{n}\sum_{j=1}%
^{n}(\beta_{ik}-\bar{\beta}_{k})(\beta_{jk}-\bar{\beta}_{k})\sigma_{v,ij}>0,
\]
and $\omega_{k}^{2}\leq\sup_{i,k}E(\beta_{ik}-\bar{\beta}_{k})^{2}%
\lim_{n\rightarrow\infty}n^{-1}\sum_{i=1}^{n}\sum_{j=1}^{n}\left\vert
\sigma_{v,ij}\right\vert <C$. Also, since $v_{it}$ are serially independent
then there exists $T_{0}$ such that for any fixed $T>T_{0}$ and as
$n\rightarrow\infty$
\[
T^{-1/2}%
{\displaystyle\sum\limits_{t=1}^{T}}
a_{t}h_{nt}\rightarrow_{d}N\left(  0,\omega_{k}^{2}\left(  1+s_{aT}%
^{2}\right)  \right)  ,
\]
where $s_{aT}^{2}$ is defined by (\ref{saT}). Using this result in
(\ref{egziknT}) and noting that $\alpha_{\gamma}<1/2$, we also have for any
fixed $T$ and as $n\rightarrow\infty,$%
\[
\xi_{k,nT}\rightarrow_{d}N\left(  0,\omega_{k}^{2}\left(  1+s_{aT}^{2}\right)
\right)  \text{, for a fixed }T>T_{0}\text{ and as }n\rightarrow\infty\text{.
}%
\]
This result extends readily to the case where $n$ and $T\rightarrow\infty$,
jointly. In this case
\[
\xi_{k,nT}\rightarrow_{d}N\left(  0,\omega_{k}^{2}\left(  1+s_{a}^{2}\right)
\right)  ,\text{ where }s_{a}^{2}=1+\boldsymbol{\lambda}_{0}^{\prime
}\mathbf{\Sigma}_{f}^{-1}\boldsymbol{\lambda}_{0}.
\]
Similarly, We have (using \thinspace$u_{it}=\gamma_{i}g_{t}+v_{it}$)
\begin{align*}
Cov\left(  \xi_{k,nT},\xi_{k^{\prime},nT}\right)   &  =n^{-1}T^{-1}\sum
_{i=1}^{n}%
{\displaystyle\sum\limits_{j=1}^{n}}
{\displaystyle\sum\limits_{t=1}^{T}}
a_{t}^{2}(\beta_{ik}-\bar{\beta}_{k})(\beta_{jk^{\prime}}-\bar{\beta
}_{k^{\prime}})E\left(  u_{it}u_{jt}\right)  =\\
&  \left(  1+s_{aT}^{2}\right)  \left[  n^{-1}\sum_{i=1}^{n}%
{\displaystyle\sum\limits_{j=1}^{n}}
\gamma_{i}\gamma_{j}(\beta_{ik}-\bar{\beta}_{k})(\beta_{jk^{\prime}}%
-\bar{\beta}_{k^{\prime}})\right] \\
&  +\left(  1+s_{aT}^{2}\right)  \left[  n^{-1}\sum_{i=1}^{n}%
{\displaystyle\sum\limits_{j=1}^{n}}
\sigma_{v,ij}(\beta_{ik}-\bar{\beta}_{k})(\beta_{jk^{\prime}}-\bar{\beta
}_{k^{\prime}})\right]  .
\end{align*}
But%
\begin{align*}
E\left\vert n^{-1}\sum_{i=1}^{n}%
{\displaystyle\sum\limits_{j=1}^{n}}
\gamma_{i}\gamma_{j}(\beta_{ik}-\bar{\beta}_{k})(\beta_{jk^{\prime}}%
-\bar{\beta}_{k^{\prime}})\right\vert  &  \leq\sup_{i,k,k^{\prime}}\left\vert
(\beta_{ik}-\bar{\beta}_{k})(\beta_{jk^{\prime}}-\bar{\beta}_{k^{\prime}%
})\right\vert E(\beta_{ik}-\bar{\beta}_{k})^{2}\left(  n^{-1/2}%
{\displaystyle\sum\limits_{j=1}^{n}}
\left\vert \gamma_{i}\right\vert \right) \\
&  \leq\sup_{i,k}E(\beta_{ik}-\bar{\beta}_{k})^{2}\left(  n^{-1/2}%
{\displaystyle\sum\limits_{j=1}^{n}}
\left\vert \gamma_{i}\right\vert \right)  ^{2}=O\left(  n^{-1+2\alpha_{\gamma
}}\right)  .
\end{align*}
and similarly
\[
\left\vert n^{-1}\sum_{i=1}^{n}%
{\displaystyle\sum\limits_{j=1}^{n}}
\sigma_{v,ij}(\beta_{ik}-\bar{\beta}_{k})(\beta_{jk^{\prime}}-\bar{\beta
}_{k^{\prime}})\right\vert \leq\sup_{i,k}E(\beta_{ik}-\bar{\beta}_{k}%
)^{2}n^{-1}\sum_{i=1}^{n}%
{\displaystyle\sum\limits_{j=1}^{n}}
\left\vert \sigma_{v,ij}\right\vert <C.
\]
Hence $Cov\left(  \xi_{k,nT},\xi_{k^{\prime},nT}\right)  <C$, for all $k$ and
$k^{\prime}$. Using the above results in (\ref{egzinTv}, and noting that $K$
is fixed, we have $\boldsymbol{\xi}_{nT}\rightarrow_{d}N\left(  \mathbf{0,V}%
_{\xi}\right)  $, as $n$ and $T\rightarrow\infty,$where }%
\[
{\small \mathbf{V}_{\xi}=\left(  1+\boldsymbol{\lambda}_{0}^{\prime
}\mathbf{\Sigma}_{f}^{-1}\boldsymbol{\lambda}_{0}\right)  p\lim_{n\rightarrow
\infty}\left(  n^{-1}\mathbf{B}_{n}^{\prime}\mathbf{M}_{n}\mathbf{V}%
_{u}\mathbf{M}_{n}\mathbf{B}_{n\ }\right)  ,}%
\]
{\small noting that $s_{a,T}^{2}\rightarrow_{p}1+\boldsymbol{\lambda}%
_{0}^{\prime}\mathbf{\Sigma}_{f}^{-1}\boldsymbol{\lambda}_{0}$, where
$s_{a,T}^{2}$ is given by (\ref{saT}). Also recall from (\ref{jointnTphiD})
that $\sqrt{nT}\left(  \boldsymbol{\tilde{\phi}}_{nT}-\boldsymbol{\phi}%
_{0}\right)  =\mathbf{\Sigma}_{\beta\beta}^{-1}\boldsymbol{\xi}_{nT}%
+O_{p}\left(  n^{-\frac{1}{2}+\frac{\alpha_{\eta+\alpha_{\gamma}}}{2}}\right)
+O_{p}\left(  T^{1/2}n^{-1/2+\alpha_{\eta}}\right)  +O_{p}\left(
T^{-1/2}\right)  $. Hence, result (\ref{Dphi}) follows since by assumption
$\alpha_{\eta}<1/2$, $\alpha_{\gamma}<1/2$, and $(T/n)^{1/2}n^{\alpha_{\gamma
}}\rightarrow0$. }

\subsection{{\protect\small Proof of theorem \ref{Tsemi}}\label{ProofThsemi}}

{\small Using (\ref{phialpha}) and (\ref{Halpha}) and replacing $T^{-1}\left(
T^{-1}\mathbf{F}^{\prime}\mathbf{M}_{T}\mathbf{F}\right)  ^{-1}$ by
$\mathbf{G}_{T}^{\prime}\mathbf{G}_{T}$ we have (see (\ref{GT}))}

{\small
\begin{equation}
\mathbf{H}_{nT}\left(  \boldsymbol{\alpha}\right)  \mathbf{D}_{\alpha}\left(
\boldsymbol{\tilde{\phi}}_{nT}\left(  \boldsymbol{\alpha}\right)
-\boldsymbol{\phi}_{0}\right)  =\mathbf{q}_{nT}\left(  \boldsymbol{\alpha
}\right)  , \label{Aphigap}%
\end{equation}
where }%
\[
{\small \mathbf{q}_{nT}\left(  \boldsymbol{\alpha}\right)  =\mathbf{D}%
_{\alpha}^{-1}\mathbf{\hat{B}}_{nT}^{\prime}\mathbf{M}_{n}\mathbf{\hat{a}%
}_{nT}+n\widehat{\bar{\sigma}}_{nT}^{2}\mathbf{D}_{\alpha}^{-1}\mathbf{G}%
_{T}^{\prime}\mathbf{G}_{T}\mathbf{\hat{\mu}}_{T}-\mathbf{H}_{nT}\left(
\boldsymbol{\alpha}\right)  \mathbf{D}_{\alpha}\boldsymbol{\phi}_{0},}%
\]
{\small and }%
\[
{\small \mathbf{H}_{nT}\left(  \boldsymbol{\alpha}\right)  =\mathbf{D}%
_{\alpha}^{-1}\mathbf{\hat{B}}_{nT}^{\prime}\mathbf{M}_{n}\mathbf{\hat{B}%
}_{nT}\mathbf{D}_{\alpha}^{-1}-n\widehat{\bar{\sigma}}_{nT}^{2}\mathbf{D}%
_{\alpha}^{-1}\mathbf{G}_{T}^{\prime}\mathbf{G}_{T}\mathbf{D}_{\alpha}^{-1}.}%
\]
{\small But%
\begin{equation}
\mathbf{H}_{nT}\left(  \boldsymbol{\alpha}\right)  =n\mathbf{D}_{\alpha}%
^{-1}\mathbf{H}_{nT}\mathbf{D}_{\alpha}^{-1},\text{ and }\mathbf{q}%
_{nT}\left(  \boldsymbol{\alpha}\right)  =n\mathbf{D}_{\alpha}^{-1}%
\mathbf{s}_{nT}, \label{qnTalpha}%
\end{equation}
where $\mathbf{s}_{nT}$ and $\mathbf{H}_{nT}$ are already defined by
(\ref{snT0}) and (\ref{A-HnT}). Consider first the limiting property of
$\mathbf{H}_{nT}\left(  \boldsymbol{\alpha}\right)  $, and using (\ref{UG})
note that%
\begin{align*}
\mathbf{H}_{nT}\left(  \boldsymbol{\alpha}\right)   &  =\mathbf{D}_{\alpha
}^{-1}\mathbf{B}_{n}^{\prime}\mathbf{M}_{n}\mathbf{B}_{n}\mathbf{D}_{\alpha
}^{-1}+\mathbf{D}_{\alpha}^{-1}\mathbf{G}_{T}^{\prime}\mathbf{U}_{nT}^{\prime
}\mathbf{M}_{n}\mathbf{B}_{n}\mathbf{D}_{\alpha}^{-1}+\\
&  \mathbf{D}_{\alpha}^{-1}\mathbf{B}_{n}^{\prime}\mathbf{M}_{n}%
\mathbf{U}_{nT}\mathbf{G}_{T}\mathbf{D}_{\alpha}^{-1}+\mathbf{D}_{\alpha}%
^{-1}\mathbf{G}_{T}^{\prime}\mathbf{U}_{nT}^{\prime}\mathbf{M}_{n}%
\mathbf{U}_{nT}\mathbf{G}_{T}\mathbf{D}_{\alpha}^{-1}-n\widehat{\bar{\sigma}%
}_{nT}^{2}\mathbf{D}_{\alpha}^{-1}\mathbf{G}_{T}^{\prime}\mathbf{G}%
_{T}\mathbf{D}_{\alpha}^{-1},
\end{align*}
or%
\begin{align*}
\mathbf{H}_{nT}\left(  \boldsymbol{\alpha}\right)   &  =\mathbf{D}_{\alpha
}^{-1}\mathbf{B}_{n}^{\prime}\mathbf{M}_{n}\mathbf{B}_{n}\mathbf{D}_{\alpha
}^{-1}+\mathbf{D}_{\alpha}^{-1}\mathbf{G}_{T}^{\prime}\mathbf{U}_{nT}^{\prime
}\mathbf{M}_{n}\mathbf{B}_{n}\mathbf{D}_{\alpha}^{-1}+\mathbf{D}_{\alpha}%
^{-1}\mathbf{B}_{n}^{\prime}\mathbf{M}_{n}\mathbf{U}_{nT}\mathbf{G}%
_{T}\mathbf{D}_{\alpha}^{-1}\\
&  +n\mathbf{D}_{\alpha}^{-1}\left[  \mathbf{G}_{T}^{\prime}\left(
n^{-1}\mathbf{U}_{nT}^{\prime}\mathbf{M}_{n}\mathbf{U}_{nT}\right)
\mathbf{G}_{T}-\widehat{\bar{\sigma}}_{nT}^{2}\mathbf{G}_{T}^{\prime
}\mathbf{G}_{T}\right]  \mathbf{D}_{\alpha}^{-1}.
\end{align*}
Further%
\begin{align*}
&  n\mathbf{D}_{\alpha}^{-1}\left[  \mathbf{G}_{T}^{\prime}\left(
n^{-1}\mathbf{U}_{nT}^{\prime}\mathbf{M}_{n}\mathbf{U}_{nT}\right)
\mathbf{G}_{T}-\widehat{\bar{\sigma}}_{nT}^{2}\mathbf{G}_{T}^{\prime
}\mathbf{G}_{T}\right]  \mathbf{D}_{\alpha}^{-1}\\
&  =n\mathbf{D}_{\alpha}^{-1}\left[  \mathbf{G}_{T}^{\prime}\left(
n^{-1}\mathbf{U}_{nT}^{\prime}\mathbf{M}_{n}\mathbf{U}_{nT}\right)
\mathbf{G}_{T}-\bar{\sigma}_{n}^{2}\mathbf{G}_{T}^{\prime}\mathbf{G}%
_{T}\right]  \mathbf{D}_{\alpha}^{-1}-\left(  \widehat{\bar{\sigma}}_{nT}%
^{2}-\bar{\sigma}_{n}^{2}\right)  n\mathbf{D}_{\alpha}^{-1}\mathbf{G}%
_{T}^{\prime}\mathbf{G}_{T}\mathbf{D}_{\alpha}^{-1}.
\end{align*}
But by (\ref{orderzig}) $\widehat{\bar{\sigma}}_{nT}^{2}-\bar{\sigma}_{n}%
^{2}=O_{p}(T^{-1/2}n^{-1/2})$, and $\left\Vert \mathbf{D}_{\alpha}%
^{-1}\right\Vert =\lambda_{max}^{1/2}\left(  \mathbf{D}_{\alpha}^{-2}\right)
=n^{-\alpha_{\min}/2}$. Then, using results in Lemma \ref{L2} we have
\[
\left\Vert \mathbf{D}_{\alpha}^{-1}\mathbf{B}_{n}^{\prime}\mathbf{M}%
_{n}\mathbf{U}_{nT}\mathbf{G}_{T}\mathbf{D}_{\alpha}^{-1}\right\Vert \leq
n\left\Vert \mathbf{D}_{\alpha}^{-1}\right\Vert ^{2}\left\Vert n^{-1}%
\mathbf{B}_{n}^{\prime}\mathbf{M}_{n}\mathbf{U}_{nT}\mathbf{G}_{T}\right\Vert
=O_{p}\left(  T^{-1/2}n^{-\alpha_{\min}+1/2}\right)  ,
\]
and%
\begin{align*}
&  \left\Vert n\mathbf{D}_{\alpha}^{-1}\mathbf{G}_{T}^{\prime}\left(
n^{-1}\mathbf{U}_{nT}^{\prime}\mathbf{M}_{n}\mathbf{U}_{nT}-\widehat{\bar
{\sigma}}_{nT}^{2}\mathbf{I}_{T}\right)  \mathbf{G}_{T}\mathbf{D}_{\alpha
}^{-1}\right\Vert \\
&  \leq n\left\Vert \mathbf{D}_{\alpha}^{-1}\right\Vert \left\Vert
\mathbf{G}_{T}^{\prime}\left(  n^{-1}\mathbf{U}_{nT}^{\prime}\mathbf{M}%
_{n}\mathbf{U}_{nT}-\widehat{\bar{\sigma}}_{nT}^{2}\mathbf{I}_{T}\right)
\mathbf{G}_{T}\right\Vert =O_{p}\left(  T^{-1}n^{-\alpha_{\min}+1/2}\right)  .
\end{align*}
Hence
\[
\mathbf{H}_{nT}\left(  \boldsymbol{\alpha}\right)  =\mathbf{D}_{\alpha}%
^{-1}\mathbf{B}_{n}^{\prime}\mathbf{M}_{n}\mathbf{B}_{n}\mathbf{D}_{\alpha
}^{-1}+O_{p}\left(  T^{-1}n^{-\alpha_{\min}+1/2}\right)  ,
\]
and $\mathbf{H}_{nT}\left(  \boldsymbol{\alpha}\right)  \rightarrow
_{p}\mathbf{\Sigma}_{\beta\beta}(\boldsymbol{\alpha}\mathbf{)}$, as
$n\rightarrow\infty$, for a fixed $T$, so long as $\alpha_{\min}%
>1/2>\alpha_{\gamma}$. By Assumption \ref{loadings}, $\lim_{n\rightarrow
\infty}\left(  \mathbf{D}_{\alpha}^{-1}\mathbf{B}_{n}^{\prime}\mathbf{M}%
_{n}\mathbf{B}_{n}\mathbf{D}_{\alpha}^{-1}\right)  =\mathbf{\Sigma}%
_{\beta\beta}(\boldsymbol{\alpha}\mathbf{)}$ is a positive definite matrix.
Using this result in (\ref{Aphigap}) we have%
\begin{equation}
\mathbf{D}_{\alpha}\left(  \boldsymbol{\tilde{\phi}}_{nT}\left(
\boldsymbol{\alpha}\right)  -\boldsymbol{\phi}_{0}\right)
\overset{a}{\thicksim}\mathbf{\Sigma}_{\beta\beta}^{-1}(\boldsymbol{\alpha
}\mathbf{)q}_{nT}(\boldsymbol{\alpha}\mathbf{),} \label{Disalpha}%
\end{equation}
and (since $\mathbf{\Sigma}_{\beta\beta}(\boldsymbol{\alpha}\mathbf{)}$ is
positive definite, $\left\Vert \mathbf{\Sigma}_{\beta\beta}^{-1}%
(\boldsymbol{\alpha}\mathbf{)}\right\Vert <C$)
\[
\left\Vert \boldsymbol{\tilde{\phi}}_{nT}\left(  \boldsymbol{\alpha}\right)
-\boldsymbol{\phi}_{0}\right\Vert \leq\left\Vert \mathbf{D}_{\alpha}%
^{-1}\right\Vert \left\Vert \mathbf{\Sigma}_{\beta\beta}^{-1}%
(\boldsymbol{\alpha}\mathbf{)}\right\Vert \left\Vert \mathbf{q}_{nT}%
(\boldsymbol{\alpha}\mathbf{)}\right\Vert \leq Cn^{-\alpha_{\min}/2}\left\Vert
\mathbf{q}_{nT}(\boldsymbol{\alpha}\mathbf{)}\right\Vert .
\]
Using (\ref{qnTalpha})
\begin{equation}
\left\Vert \mathbf{q}_{nT}(\boldsymbol{\alpha}\mathbf{)}\right\Vert
=\left\Vert n\mathbf{D}_{\alpha}^{-1}\mathbf{s}_{nT}\right\Vert \leq
n^{1-\alpha_{\min}/2}\left\Vert \mathbf{s}_{nT}\right\Vert . \label{q_nTa}%
\end{equation}
Also from (\ref{snT2}) we have%
\begin{align}
\mathbf{s}_{nT}  &  =n^{-1}\left(  \mathbf{B}_{n}^{\prime}\mathbf{M}%
_{n}\mathbf{\bar{u}}_{n\circ}-\mathbf{B}_{n}^{\prime}\mathbf{M}_{n}%
\mathbf{U}_{nT}\mathbf{G}_{T}\boldsymbol{\lambda}_{T}^{\ast}\right)
+O_{p}\left(  n^{-1+\alpha_{\eta}}\right) \label{s_nTa}\\
&  +O_{p}\left(  T^{-1/2}n^{-1+\frac{\alpha_{\eta+\alpha_{\gamma}}}{2}%
}\right)  +O_{p}\left(  T^{-1}n^{-1/2}\right)  .\nonumber
\end{align}
Also using (\ref{Doms_nT}))%
\begin{equation}
n^{-1}\left(  \mathbf{B}_{n}^{\prime}\mathbf{M}_{n}\mathbf{\bar{u}}_{n\circ
}-\mathbf{B}_{n}^{\prime}\mathbf{M}_{n}\mathbf{U}_{nT}\mathbf{G}%
_{T}\boldsymbol{\lambda}_{T}^{\ast}\right)  =O_{p}\left(  T^{-1/2}%
n^{-1/2}\right)  . \label{Dom}%
\end{equation}
Substituting (\ref{Dom}) in (\ref{s_nTa}) and using the result in
(\ref{q_nTa}) we have
\begin{align}
\left\Vert \mathbf{q}_{nT}(\boldsymbol{\alpha}\mathbf{)}\right\Vert  &
=O_{p}\left(  n^{-\alpha_{\min}/2+1/2}T^{-1/2}\right)  +O_{p}\left(
T^{-1/2}n^{\frac{-\alpha_{\min}+\left(  \alpha_{\eta+\alpha_{\gamma}}\right)
}{2}}\right) \label{qknT}\\
&  +O_{p}\left(  n^{-\alpha_{\min}/2+\alpha_{\eta}}\right)  +O_{p}\left(
n^{-\alpha_{\min}/2+1/2}T^{-1}\right)  .\nonumber
\end{align}
Denote the $k^{th}$ element of $\mathbf{q}_{nT}(\boldsymbol{\alpha}\mathbf{)}$
by $q_{k,nT}(\boldsymbol{\alpha}\mathbf{)}$, we also have%
\begin{align*}
q_{k,nT}\left(  \boldsymbol{\alpha}\right)   &  =O_{p}\left(  n^{-\alpha
_{\min}/2+1/2}T^{-1/2}\right)  +O_{p}\left(  T^{-1/2}n^{\frac{-\alpha_{\min
}+\left(  \alpha_{\eta+\alpha_{\gamma}}\right)  }{2}}\right) \\
&  +O_{p}\left(  n^{-\alpha_{\min}/2+\alpha_{\eta}}\right)  +O_{p}\left(
n^{-\alpha_{\min}/2+1/2}T^{-1}\right)  .
\end{align*}
Also note that the $k^{th}$ element of $\mathbf{D}_{\alpha}\left(
\boldsymbol{\tilde{\phi}}_{nT}\left(  \boldsymbol{\alpha}\right)
-\boldsymbol{\phi}_{0}\right)  $ is given by $n^{\alpha_{k}/2}\left(
\tilde{\phi}_{k,nT}\left(  \boldsymbol{\alpha}\right)  -\phi_{0,k}\right)  $.
Hence, in view of (\ref{Disalpha}) and since $\mathbf{\Sigma}_{\beta\beta
}^{-1}(\boldsymbol{\alpha}\mathbf{)}$ is a positive definite matrix then the
probability order of $n^{\alpha_{k}/2}\left(  \tilde{\phi}_{k,nT}\left(
\boldsymbol{\alpha}\right)  -\phi_{0,k}\right)  $ must be the same as that of
$q_{k,nT}$, and hence (as required)%
\begin{align*}
\tilde{\phi}_{k,nT}\left(  \boldsymbol{\alpha}\right)  -\phi_{0,k}  &
=O_{p}\left(  n^{-(\alpha_{k}+\alpha_{\min})/2+1/2}T^{-1/2}\right)
+O_{p}\left(  n^{\frac{-\left(  \alpha_{k}+\alpha_{\min}\right)  +\left(
\alpha_{\eta+\alpha_{\gamma}}\right)  }{2}}T^{-1/2}\right) \\
&  +O_{p}\left(  n^{-\left(  \alpha_{k}+\alpha_{\min}\right)  /2+\alpha_{\eta
}}\right)  +O_{p}\left(  n^{-(\alpha_{k}+\alpha_{\min})/2+1/2}T^{-1}\right)  .
\end{align*}
}

\subsection{{\protect\small Proof of theorem \ref{Var}}\label{ProofTvar}}

{\small Using (\ref{Varegzihat}) and noting that $\hat{s}_{nT}\rightarrow
_{p}\boldsymbol{\lambda}_{0}^{\prime}\mathbf{\Sigma}_{f}^{-1}%
\boldsymbol{\lambda}_{0}$, then
\begin{align}
\mathbf{\hat{V}}_{\xi,nT}-\mathbf{V}_{\xi}  &  =\left(  1+\boldsymbol{\lambda
}_{0}^{\prime}\mathbf{\Sigma}_{f}^{-1}\boldsymbol{\lambda}_{0}\right)  \left[
n^{-1}\mathbf{\hat{B}}_{nT}^{\prime}\mathbf{M}_{n}\mathbf{\tilde{V}}%
_{u}\mathbf{M}_{n}\mathbf{\hat{B}}_{nT}-p\lim_{n\rightarrow\infty}\left(
n^{-1}\mathbf{B}_{n}^{\prime}\mathbf{M}_{n}\mathbf{V}_{u}\mathbf{M}%
_{n}\mathbf{B}_{n\ }\right)  \right]  +o_{p}(1)\nonumber\\
&  =\left(  1+\boldsymbol{\lambda}_{0}^{\prime}\mathbf{\Sigma}_{f}%
^{-1}\boldsymbol{\lambda}_{0}\right)  \left[  n^{-1}\mathbf{\hat{B}}%
_{nT}^{\prime}\mathbf{M}_{n}\mathbf{\tilde{V}}_{u}\mathbf{M}_{n}%
\mathbf{\hat{B}}_{nT}-n^{-1}\mathbf{B}_{n}^{\prime}\mathbf{M}_{n}%
\mathbf{V}_{u}\mathbf{M}_{n}\mathbf{B}_{n\ }\right]  +o_{p}(1). \label{GapVar}%
\end{align}
Also using (\ref{Bhat1}) we have%
\[
n^{-1}\mathbf{\hat{B}}_{nT}^{\prime}\mathbf{M}_{n}\mathbf{\tilde{V}}%
_{u}\mathbf{M}_{n}\mathbf{\hat{B}}_{nT}=n^{-1}\left(  \mathbf{B}%
_{n}+\mathbf{U}_{nT}\mathbf{G}_{n}\right)  ^{\prime}\mathbf{M}_{n}\left(
\mathbf{\tilde{V}}_{u}-\mathbf{V}_{u}+\mathbf{V}_{u}\right)  \mathbf{M}%
_{n}\left(  \mathbf{B}_{n}+\mathbf{U}_{nT}\mathbf{G}_{n}\right)  ,
\]
which, after some algebra, yields%
\[
n^{-1}\mathbf{\hat{B}}_{nT}^{\prime}\mathbf{M}_{n}\mathbf{\tilde{V}}%
_{u}\mathbf{M}_{n}\mathbf{\hat{B}}_{nT}-n^{-1}\mathbf{B}_{n}^{\prime
}\mathbf{M}_{n}\mathbf{V}_{u}\mathbf{M}_{n}\mathbf{B}_{n\ }=\sum_{j=1}%
^{7}\mathbf{A}_{j,nT},
\]
where
\begin{align*}
\mathbf{A}_{1,nT}  &  =n^{-1}\mathbf{B}_{n}^{\prime}\mathbf{M}_{n}\left(
\mathbf{\tilde{V}}_{u}-\mathbf{V}_{u}\right)  \mathbf{M}_{n}\mathbf{B}%
_{n}\text{, \ \ }\mathbf{A}_{2,nT}=n^{-1}\mathbf{G}_{n}^{\prime}%
\mathbf{U}_{nT}^{\prime}\mathbf{M}_{n}\left(  \mathbf{\tilde{V}}%
_{u}-\mathbf{V}_{u}\right)  \mathbf{M}_{n}\mathbf{U}_{nT}\mathbf{G}_{n},\\
\mathbf{A}_{3,nT}  &  =n^{-1}\mathbf{G}_{n}^{\prime}\mathbf{U}_{nT}^{\prime
}\mathbf{M}_{n}\mathbf{V}_{u}\mathbf{M}_{n}\mathbf{U}_{nT}\mathbf{G}%
_{n},\text{ \ }\mathbf{A}_{4,nT}=n^{-1}\mathbf{G}_{n}^{\prime}\mathbf{U}%
_{nT}^{\prime}\mathbf{M}_{n}\left(  \mathbf{\tilde{V}}_{u}-\mathbf{V}%
_{u}\right)  \mathbf{M}_{n}\mathbf{B}_{n},\text{ }\\
\mathbf{A}_{5,nT}  &  =n^{-1}\mathbf{G}_{n}^{\prime}\mathbf{U}_{nT}^{\prime
}\mathbf{M}_{n}\mathbf{V}_{u}\mathbf{M}_{n}\mathbf{B}_{n},\text{ \ }%
\mathbf{A}_{6,nT}=n^{-1}\mathbf{B}_{n}^{\prime}\mathbf{M}_{n}\left(
\mathbf{\tilde{V}}_{u}-\mathbf{V}_{u}\right)  \mathbf{M}_{n}\mathbf{U}%
_{nT}\mathbf{G}_{n},\\
\mathbf{A}_{7,nT}  &  =n^{-1}\mathbf{B}_{n}^{\prime}\mathbf{M}_{n}%
\mathbf{V}_{u}\mathbf{M}_{n}\mathbf{U}_{nT}\mathbf{G}_{n}.
\end{align*}
Considering the above terms in turn we note that
\[
\left\Vert \mathbf{A}_{1,nT}\right\Vert \leq n^{-1}\left\Vert \mathbf{B}%
_{n}^{\prime}\mathbf{M}_{n}\right\Vert ^{2}\left\Vert \mathbf{\tilde{V}}%
_{u}-\mathbf{V}_{u}\right\Vert =\lambda_{\max}\left(  n^{-1}\mathbf{B}%
_{n}^{\prime}\mathbf{M}_{n}\mathbf{B}_{n}\right)  \left\Vert \mathbf{\tilde
{V}}_{u}-\mathbf{V}_{u}\right\Vert .
\]
Also, under Assumption \ref{loadings} $\lambda_{\max}\left(  n^{-1}%
\mathbf{B}_{n}^{\prime}\mathbf{M}_{n}\mathbf{B}_{n}\right)  <C$ and using
(\ref{NormVu}) we have $\left\Vert \mathbf{A}_{1,nT}\right\Vert =O_{p}\left(
n^{\alpha_{\gamma}}\sqrt{\frac{\ln(n)}{T}}\right)  $. Similarly%
\[
\left\Vert \mathbf{A}_{2,nT}\right\Vert \leq n^{-1}\left\Vert \mathbf{G}%
_{n}^{\prime}\mathbf{U}_{nT}^{\prime}\mathbf{M}_{n}\right\Vert ^{2}\left\Vert
\mathbf{\tilde{V}}_{u}-\mathbf{V}_{u}\right\Vert =\lambda_{\max}\left(
n^{-1}\mathbf{G}_{n}^{\prime}\mathbf{U}_{nT}^{\prime}\mathbf{M}_{n}%
\mathbf{U}_{nT}\mathbf{G}_{n}\right)  \left\Vert \mathbf{\tilde{V}}%
_{u}-\mathbf{V}_{u}\right\Vert .
\]
Then using (\ref{GUUG}) $n^{-1}\mathbf{G}_{n}^{\prime}\mathbf{U}_{nT}^{\prime
}\mathbf{M}_{n}\mathbf{U}_{nT}\mathbf{G}_{n}\rightarrow_{p}\bar{\sigma}%
_{n}^{2}\mathbf{G}_{T}^{\prime}\mathbf{G}_{T}=T^{-1}\bar{\sigma}_{n}%
^{2}(T^{-1}\mathbf{F}^{\prime}\mathbf{M}_{T}\mathbf{F)}^{-1}=O(T^{-1})$, and
it follows that $\left\Vert \mathbf{A}_{2,nT}\right\Vert =O_{p}\left(
T^{-1}n^{\alpha_{\gamma}}\sqrt{\frac{\ln(n)}{T}}\right)  $. Turning to the
third term%
\[
\left\Vert \mathbf{A}_{3,nT}\right\Vert \leq n^{-1}\left\Vert \mathbf{G}%
_{n}^{\prime}\mathbf{U}_{nT}^{\prime}\mathbf{M}_{n}\mathbf{U}_{nT}%
\mathbf{G}_{n}\right\Vert \left\Vert \mathbf{V}_{u}\right\Vert =\lambda_{\max
}\left(  n^{-1}\mathbf{G}_{n}^{\prime}\mathbf{U}_{nT}^{\prime}\mathbf{M}%
_{n}\mathbf{U}_{nT}\mathbf{G}_{n}\right)  \left\Vert \mathbf{V}_{u}\right\Vert
.
\]
Also, by Lemma \ref{L1} $\left\Vert \mathbf{V}_{u}\right\Vert =O(n^{\alpha
_{\gamma}})$, and hence $\left\Vert \mathbf{A}_{3,nT}\right\Vert =O_{p}\left(
T^{-1}n^{\alpha_{\gamma}}\right)  $, and $\left\Vert \mathbf{A}_{4,nT}%
\right\Vert \leq\left\Vert n^{-1/2}\mathbf{G}_{n}^{\prime}\mathbf{U}%
_{nT}^{\prime}\mathbf{M}_{n}\right\Vert \left\Vert n^{-1/2}\mathbf{M}%
_{n}\mathbf{B}_{n}\right\Vert \left\Vert \mathbf{\tilde{V}}_{u}-\mathbf{V}%
_{u}\right\Vert $. Also, as shown above $\left\Vert n^{-1/2}\mathbf{M}%
_{n}\mathbf{B}_{n}\right\Vert =O_{p}(1)$, $\left\Vert n^{-1/2}\mathbf{G}%
_{n}^{\prime}\mathbf{U}_{nT}^{\prime}\mathbf{M}_{n}\right\Vert =O_{p}%
(T^{-1/2})$, then
\[
\left\Vert \mathbf{A}_{4,nT}\right\Vert =O_{p}\left(  T^{-1/2}n^{\alpha
_{\gamma}}\sqrt{\frac{\ln(n)}{T}}\right)  ,
\]%
\[
\left\Vert \mathbf{A}_{5,nT}\right\Vert \leq\left\Vert n^{-1/2}\mathbf{G}%
_{n}^{\prime}\mathbf{U}_{nT}^{\prime}\mathbf{M}_{n}\right\Vert \left\Vert
n^{-1/2}\mathbf{M}_{n}\mathbf{B}_{n}\right\Vert \left\Vert \mathbf{V}%
_{u}\right\Vert =O_{p}\left(  T^{-1/2}n^{\alpha_{\gamma}}\right)  ,
\]%
\[
\left\Vert \mathbf{A}_{6,nT}\right\Vert \leq\left\Vert n^{-1/2}\mathbf{G}%
_{n}^{\prime}\mathbf{U}_{nT}^{\prime}\mathbf{M}_{n}\right\Vert \left\Vert
n^{-1/2}\mathbf{M}_{n}\mathbf{B}_{n}\right\Vert \left\Vert \mathbf{\tilde{V}%
}_{u}-\mathbf{V}_{u}\right\Vert =O_{p}\left(  T^{-1/2}n^{\alpha_{\gamma}}%
\sqrt{\frac{\ln(n)}{T}}\right)  ,
\]
and $\left\Vert \mathbf{A}_{7,nT}\right\Vert \leq\left\Vert n^{-1/2}%
\mathbf{B}_{n}^{\prime}\mathbf{M}_{n}\right\Vert \left\Vert \mathbf{V}%
_{u}\right\Vert \left\Vert n^{-1/2}\mathbf{M}_{n}\mathbf{U}_{nT}\mathbf{G}%
_{n}\right\Vert =O_{p}\left(  T^{-1/2}n^{\alpha_{\gamma}}\right)  $. Overall,
\[
\left\Vert n^{-1}\mathbf{\hat{B}}_{n}^{\prime}\mathbf{M}_{n}\mathbf{\tilde{V}%
}_{u}\mathbf{M}_{n}\mathbf{\hat{B}}_{n}-n^{-1}\mathbf{B}_{n}^{\prime
}\mathbf{M}_{n}\mathbf{V}_{u}\mathbf{M}_{n}\mathbf{B}_{n\ }\right\Vert
=O_{p}\left(  n^{\alpha_{\gamma}}\sqrt{\frac{\ln(n)}{T}}\right)  ,
\]
which if used in (\ref{GapVar}) establishes (\ref{NormVaregzi}), as required.
}

{\small \pagebreak}

{\small \bigskip}%
\singlespacing
\thispagestyle{empty}%

\begin{center}
{\small \textbf{Online Supplement A: Data Sources and Calibration of Monte
Carlo Designs } }

{\small \bigskip}

{\small \bigskip To }

{\small \bigskip}

{\small \textbf{Identifying and exploiting alpha in linear asset pricing
models with strong, semi-strong, and latent factors} }

{\small \bigskip}

{\small \bigskip}

{\small by }

{\small \bigskip}

{\small \bigskip}

{\small M. Hashem Pesaran }

{\small University of Southern California, and Trinity College, Cambridge }

{\small \bigskip}

{\small Ron P. Smith }

{\small Birkbeck, University of London }

{\small \bigskip}

{\small \bigskip}

{\small October 2024 }
\end{center}

\pagebreak%
\setcounter{table}{0}%
\setcounter{section}{0}%
\setcounter{figure}{0}%
\setcounter{footnote}{0}%
%

\renewcommand{\thetable}{SA-\arabic{table}}%
\renewcommand{\thefigure}{SA-\arabic{figure}}%
\setcounter{page}{1}%
\renewcommand{\thepage}{SA-\arabic{page}}%
\renewcommand{\thesection}{S\arabic{section}}%
\renewcommand{\thesubsection}{S-\arabic{section}.\arabic{subsection}}%

\section{{\protect\small Introduction}}

{\small This online supplement provides the details of data sources for the
risk factors and the excess returns on securities used to calibrate the Monte
Carlo designs and to carry out the empirical applications reported in Sections
\ref{Simulations} and \ref{Empirical} of the main paper. Section \ref{DMC}
describes the data used to calibrate the Monte Carlo (MC) designs, and Section
\ref{EMC} provides the estimates that formed the basis of the calibration of
the parameters of the MC designs. Section \ref{BlockCov} provides evidence for
the choice of 14 blocks used in calibration of error covariances in the Monte
Carlo experiments. Section \ref{DEMP} describes the data for factors and
excess returns used in the empirical applications, and Section
\ref{appPRsquared} derives the relationship between pooled $R^{2}$ of return
regressions and the factor strengths used in the discussion of the empirical
results. }

\section{{\protect\small Data used to calibrate the Monte Carlo designs}%
\label{DMC}}

\subsection{{\protect\small Factors}}

{\small To calibrate the parameters of the three factor model used in the
Monte Carlo experiments we used monthly Fama-French three factor data series
over the long sample 1963m8-2021m12, downloaded from Kenneth French's
webpage.\footnote{See
https://mba.tuck.dartmouth.edu/pages/faculty/ken.french/data\_library.html}
The factors are the market return minus the risk free rate, denoted by MKT,
the value factor (high book to market minus low portfolios, HML) and the size
factor (small minus big portfolios, SMB). The risk free rate is also
downloaded from French's webpage. First order autoregressions, AR(1), were
estimated for all the three factors using the full data set, 1963m8-2021m12.
Then GARCH(1,1) models were then estimated on the residuals from the fitted
AR(1) regressions. }

\subsection{{\protect\small Excess returns}}

{\small To calibrate the factor loadings and other parameters of the excess
return regressions we used the shorter sample over the 20 years 2002m1 -
2021m12 ($T=240$). Monthly returns (inclusive of dividend payment) over 2002m1
- 2021m12 ($T=240$) for NYSE and NASDAQ stocks with share codes of 10 and 11
from CRSP were downloaded from Wharton Research Data Services and transformed
to firm-specific excess returns using the risk free rate from French's
webpage, and measured in percent, per month. Only stocks with data over the
period 2002m1-2021m12 were used to arrive at a balanced panel with $T=240$
monthly observations and a total number of $n=1,289$, securities. }

{\small To avoid extreme outliers influencing the estimates we first computed
mean, median, standard deviation, skewness and kurtosis for each security over
2002m1 - 2021m12. Table \ref{tab:sumstat} reports the mean, standard
deviations and interquartile range of these statistics over all the $1,289$
securities in our sample. The histograms of these summary statistics (mean,
median, standard deviation, skewness and kurtosis) of the individual stock
returns for the full sample over 2002m1 - 2021m12 are shown in Figure
\ref{fig:Histo_full}. As can be seen from these summary statistics, there are
outlier security returns with very large standard deviations. This is clear
from the cross firm standard deviation of 13.97 which is much larger than the
kurtosis of 9.24 (See Table \ref{tab:sumstat}), resulting from a number of
extreme outliers also seen from the long right tail of the histogram for the
distribution of kurtosis across firms.}

{\small
\begin{figure}[H]
\caption{Histogram and density function of the individual stock returns ($n=1,289$) over 2002m1-2021m12 ($T=240$)}
\label{fig:Histo_full} }

\begin{center}
{\small $%
\begin{array}
[c]{cc}%
{\parbox[b]{2.7337in}{\begin{center}
\fbox{\includegraphics[
height=2.0435in,
width=2.7337in
]%
{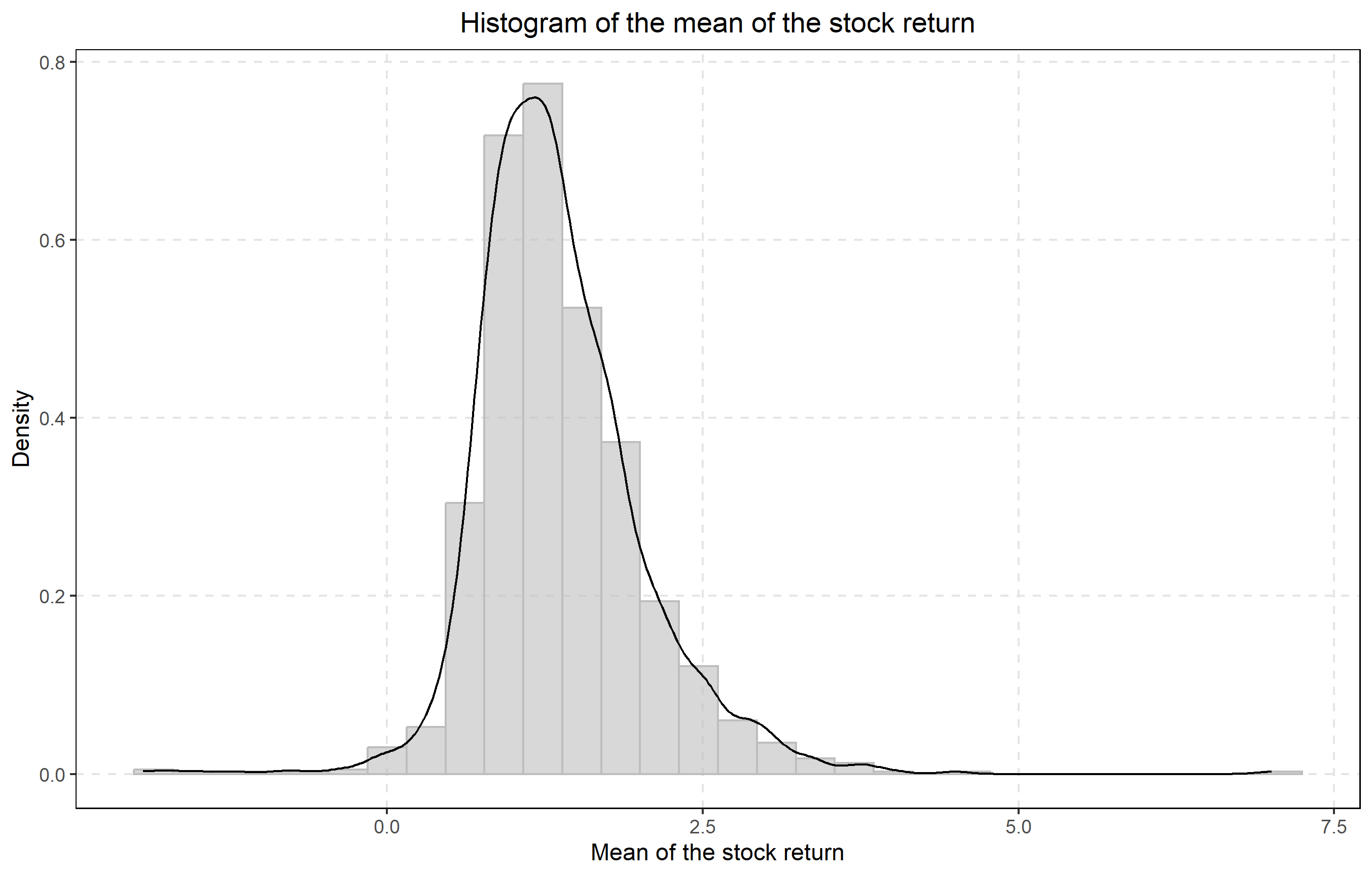}%
}\\
{}%
\end{center}}}
&
{\parbox[b]{2.7337in}{\begin{center}
\fbox{\includegraphics[
height=2.0435in,
width=2.7337in
]%
{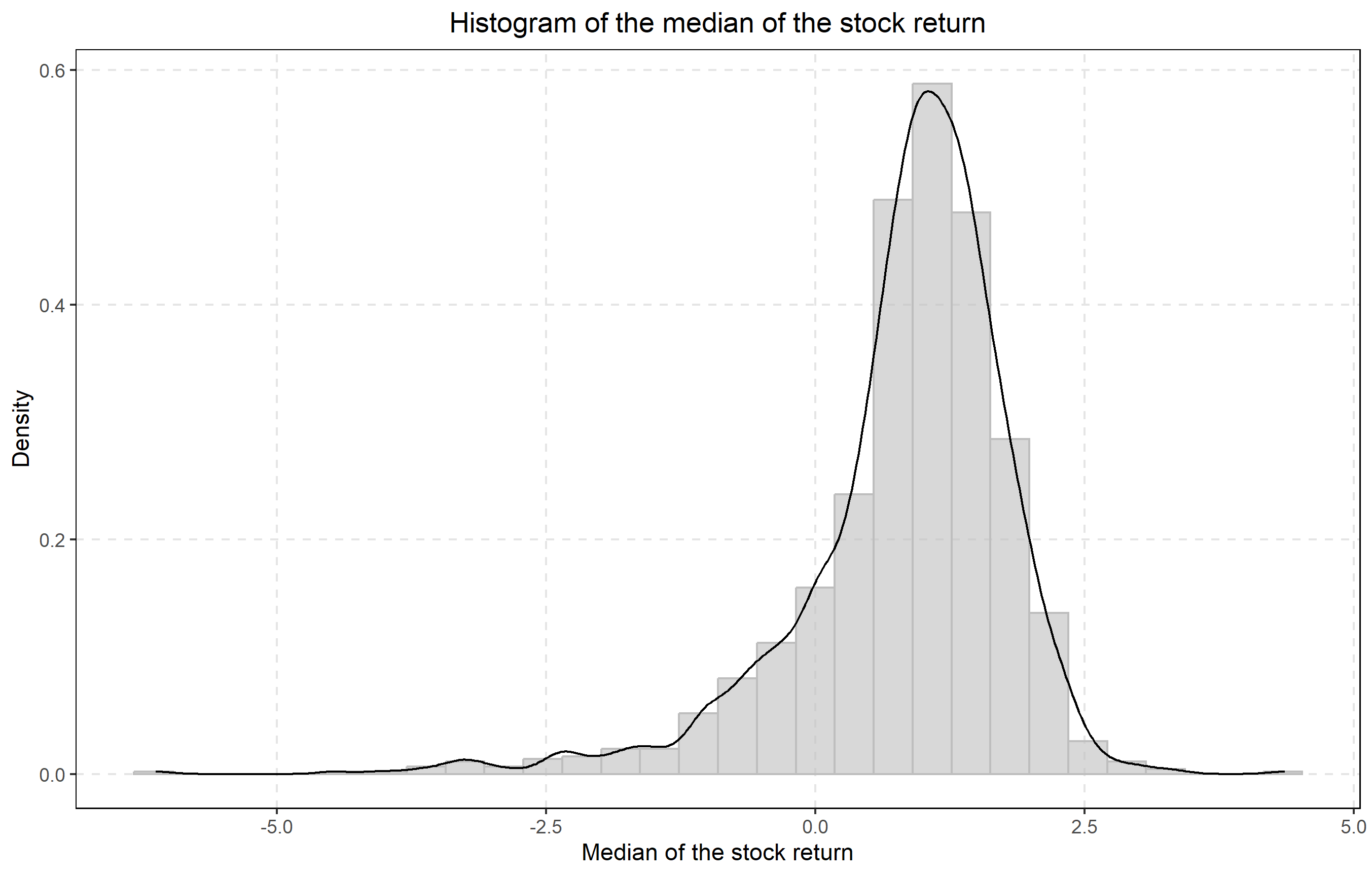}%
}\\
{}%
\end{center}}}
\\%
{\parbox[b]{2.7337in}{\begin{center}
\fbox{\includegraphics[
height=2.0435in,
width=2.7337in
]%
{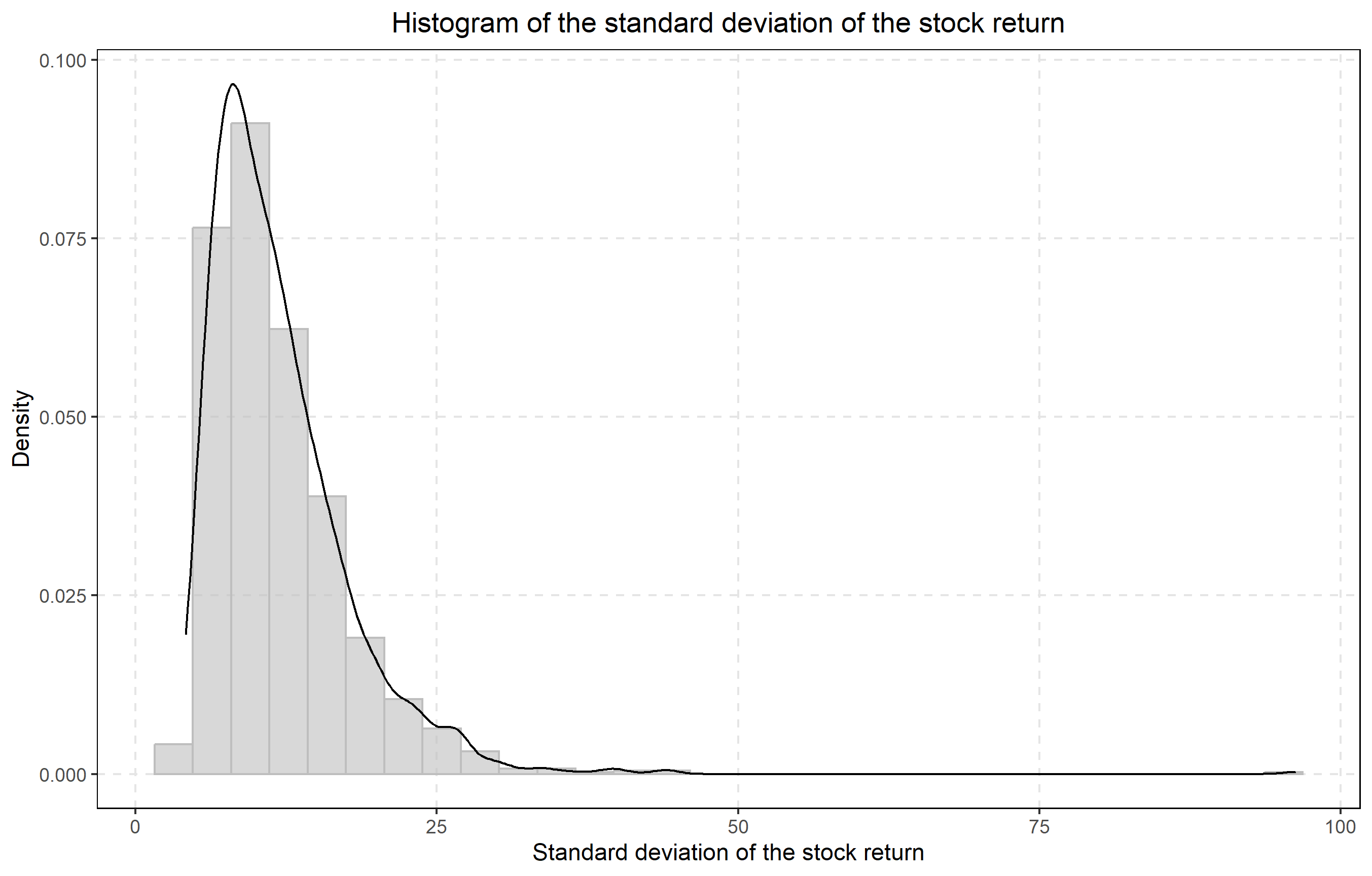}%
}\\
{}%
\end{center}}}
&
{\parbox[b]{2.7337in}{\begin{center}
\fbox{\includegraphics[
height=2.0435in,
width=2.7337in
]%
{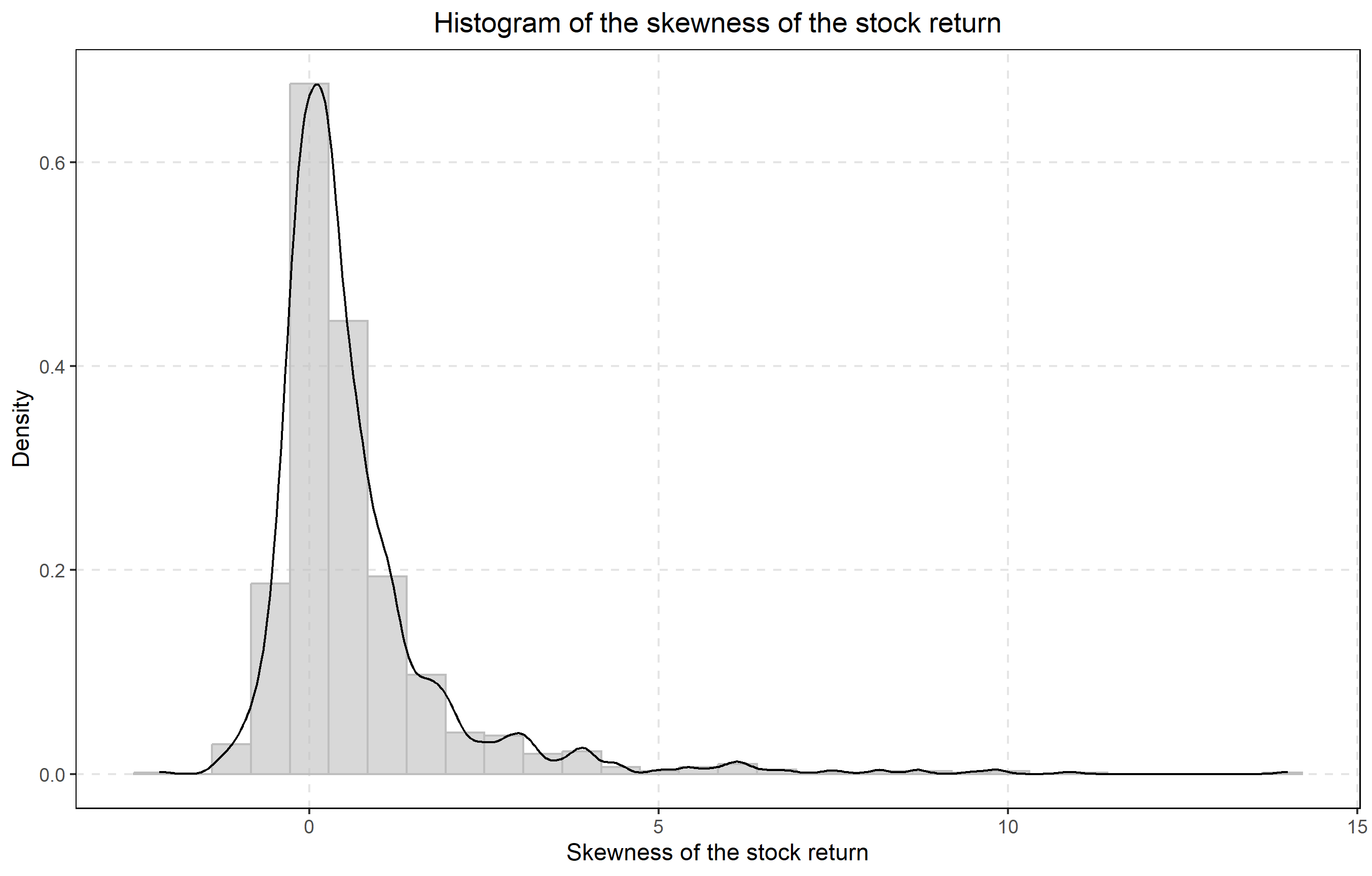}%
}\\
{}%
\end{center}}}
\\%
{\parbox[b]{2.7337in}{\begin{center}
\fbox{\includegraphics[
height=2.0435in,
width=2.7337in
]%
{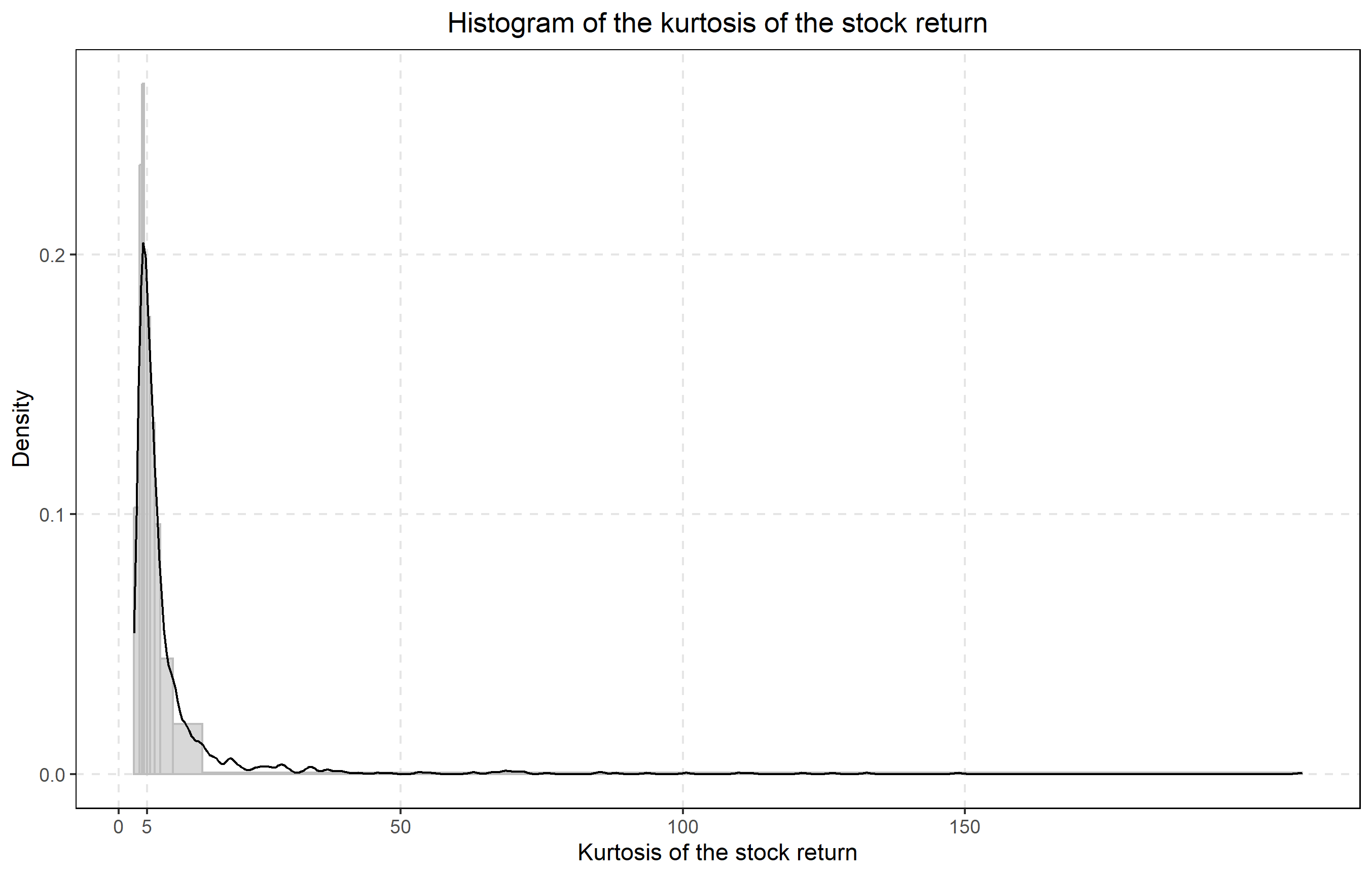}%
}\\
{}%
\end{center}}}
&
\end{array}
$ }
\end{center}

{\small
\end{figure}%
}

\subsection{{\protect\small Stocks with the kurtosis less than or equal to 14
and less than 16}}

{\small To reduce the influence of outlier returns on our results we
considered dropping stock excess returns having kurtosis in excess of 14 and
16. Excluding stocks with kurtosis less than or equal to 14, resulted in a
sample with $n=1,148$ securities, whilst if we use the cut off point of 16 we
ended up with $n=1,175$ securities. The mean, median, standard deviation,
skewness and kurtosis for each security over 2002m1 - 2021m12 of the whole and
two sub-samples, ($1,289,$ $1,148$ and $1,175).$ Then the average, standard
deviation and interquartile range (IQR) for each summary statistic of the
stocks from the two sub-sample are summarized respectively in Table
\ref{tab:sumstat}.}

{\small
\begin{table}[H]%
}

{\small \caption{The average, standard deviation and interquartile range of the
summary statistics of the individual stocks over 2002m1 - 2021m12 ($T=240$)}\label{tab:sumstat}%
}

\begin{center}
{\small
\begin{tabular}
[c]{lrcc}\hline\hline
& Average & Standard deviation & Interquartile range\\\hline
\multicolumn{4}{l}{}\\
\multicolumn{4}{l}{\textit{Panel A}: All stocks $(n=1289)$}\\
stock.mean & 1.3666 & 0.6593 & 0.7511\\
stock.median & 0.8360 & 0.9970 & 0.9310\\
stock.standard deviation & 11.9030 & 6.0082 & 6.5248\\
stock.skewness & 0.6411 & 1.3766 & 0.9290\\
stock.kurtosis & 9.2407 & 13.9786 & 4.0183\\
\multicolumn{4}{l}{}\\
\multicolumn{4}{l}{\textit{Panel B}: Kurtosis $\leq14$ $(n=1148)$}\\
stock.mean & 1.3303 & 0.5898 & 0.7032\\
stock.median & 0.9474 & 0.8816 & 0.8654\\
stock.standard deviation & 10.8837 & 4.3944 & 5.6429\\
stock.skewness & 0.2976 & 0.5996 & 0.7270\\
stock.kurtosis & 5.9885 & 2.3932 & 2.7762\\
\multicolumn{4}{l}{}\\
\multicolumn{4}{l}{\textit{Panel B}: Kurtosis $\leq16$ $(n=1175)$}\\
stock.mean & 1.3336 & 0.5922 & 0.7081\\
stock.median & 0.9363 & 0.8878 & 0.8764\\
stock.standard deviation & 10.9943 & 4.4695 & 5.6836\\
stock.skewness & 0.3219 & 0.6285 & 0.7607\\
stock.kurtosis & 6.1947 & 2.7222 & 2.9819\\\hline\hline
\end{tabular}
}
\end{center}

{\small
\end{table}%
}

{\small Histograms for mean, median, standard deviation, skewness and kurtosis
of the individual stock returns for two sub-samples are shown in Figure
\ref{fig:Histo_kur14} for the sub-sample with the kurtosis less than or equal
to 14 over 2002m1 - 2021m12 ($n=1,148$), and Figure \ref{fig:Histo_kur16} are
for the sub-sample with the kurtosis less than 16 over 2002m1 - 2021m12
($n=1,175$).}

{\small
\begin{figure}[H]%
\caption{Histogram and density function of the individual stock returns with kurtosis less than or equal to 14 over 2002m1-2021m12 ($T=240$) and $n=1,148$}
\label{fig:Histo_kur14}
\[%
\begin{array}
[c]{cc}%
{\parbox[b]{2.7337in}{\begin{center}
\fbox{\includegraphics[
height=2.0435in,
width=2.7337in
]%
{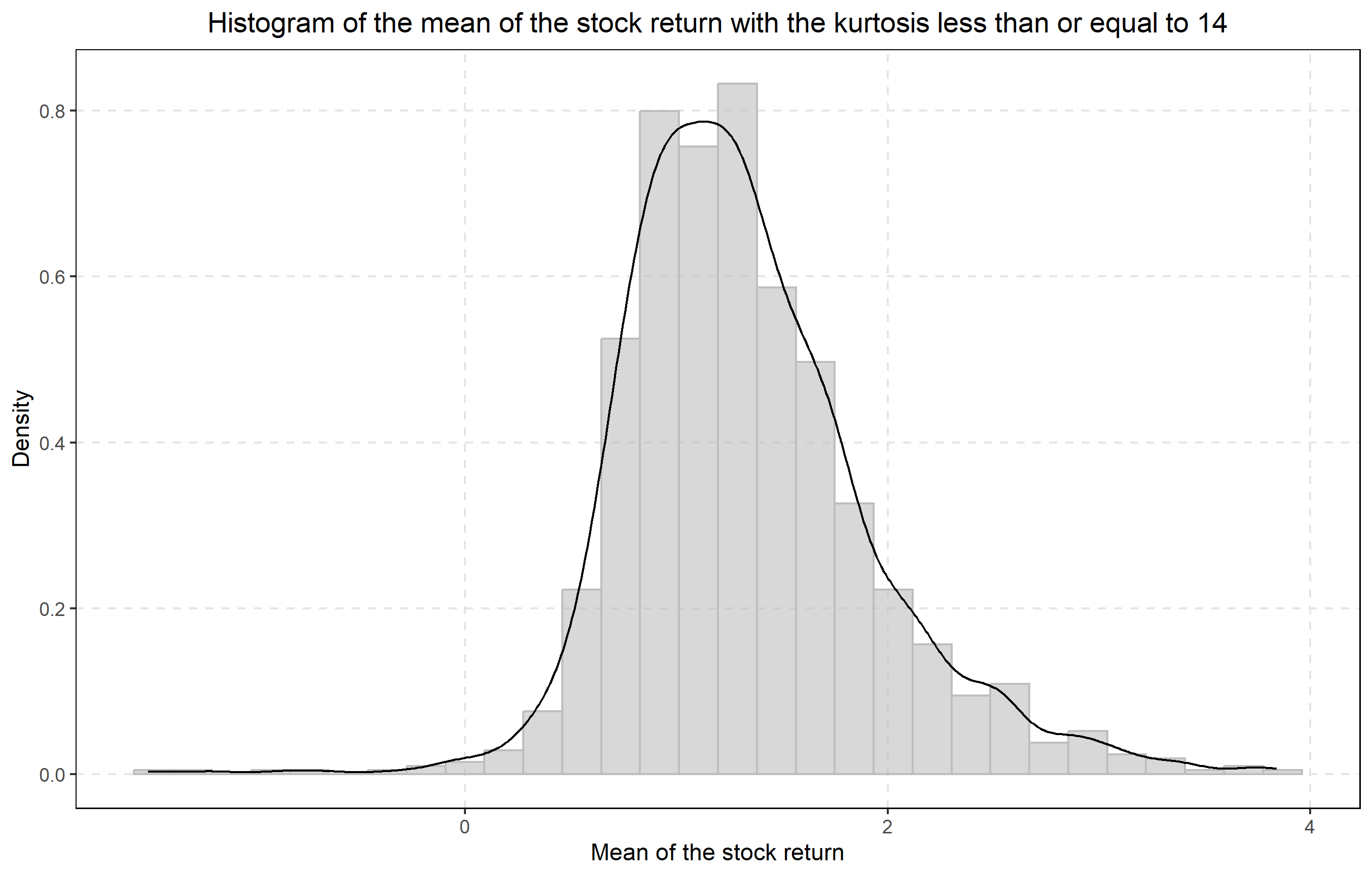}%
}\\
{}%
\end{center}}}
&
{\parbox[b]{2.7337in}{\begin{center}
\fbox{\includegraphics[
height=2.0435in,
width=2.7337in
]%
{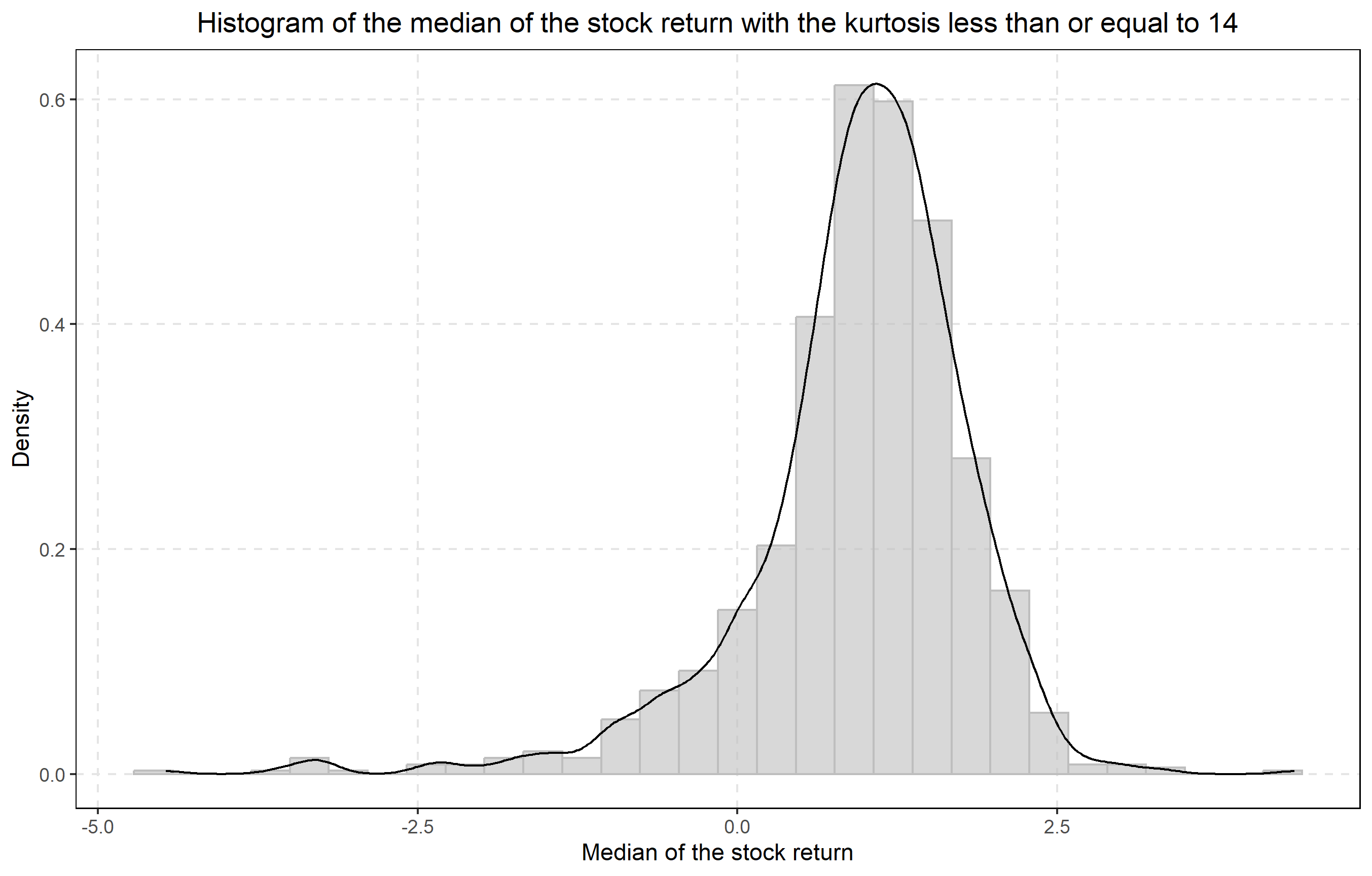}%
}\\
{}%
\end{center}}}
\\%
{\parbox[b]{2.7337in}{\begin{center}
\fbox{\includegraphics[
height=2.0435in,
width=2.7337in
]%
{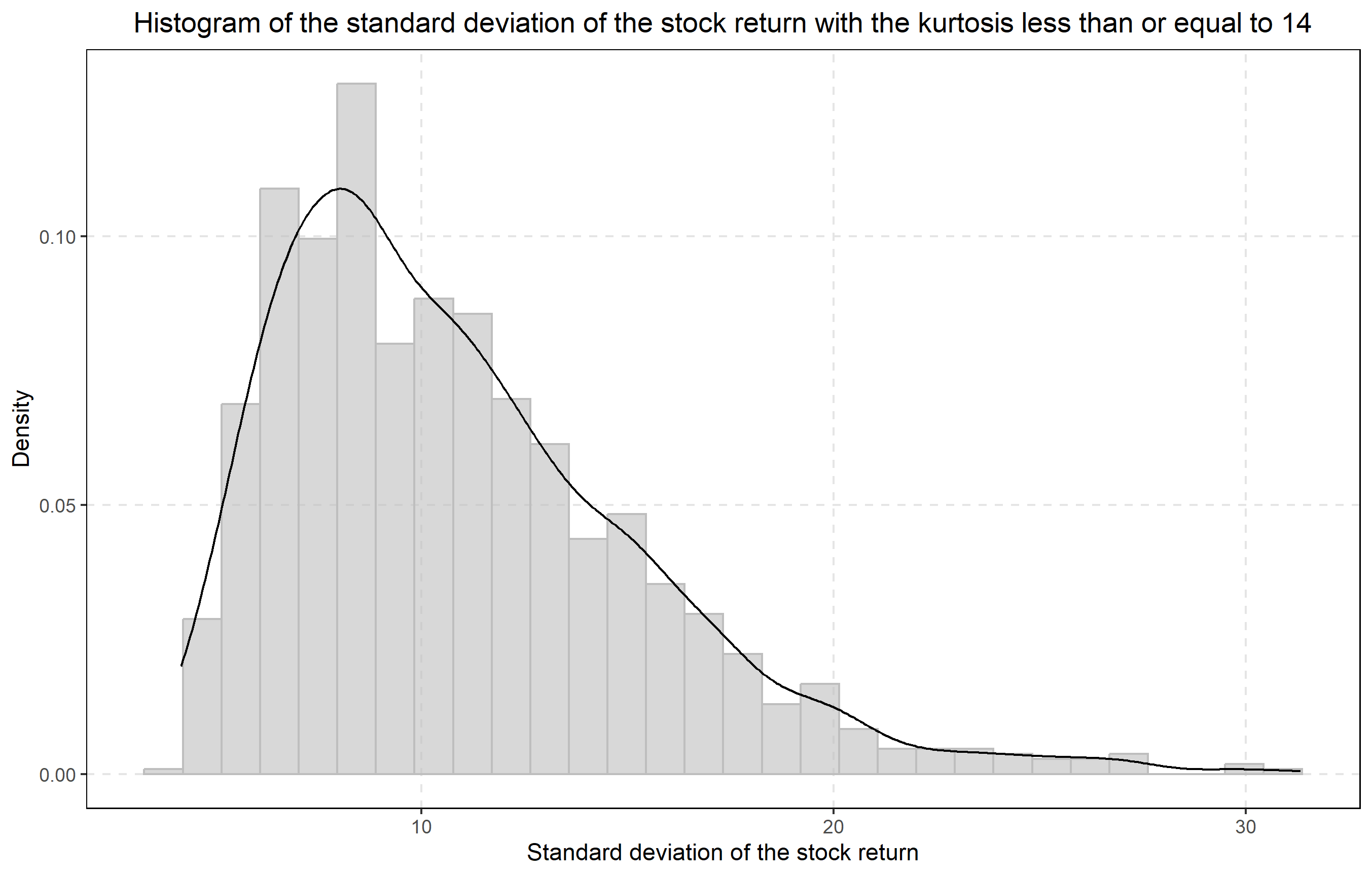}%
}\\
{}%
\end{center}}}
&
{\parbox[b]{2.6904in}{\begin{center}
\fbox{\includegraphics[
height=2.0435in,
width=2.6904in
]%
{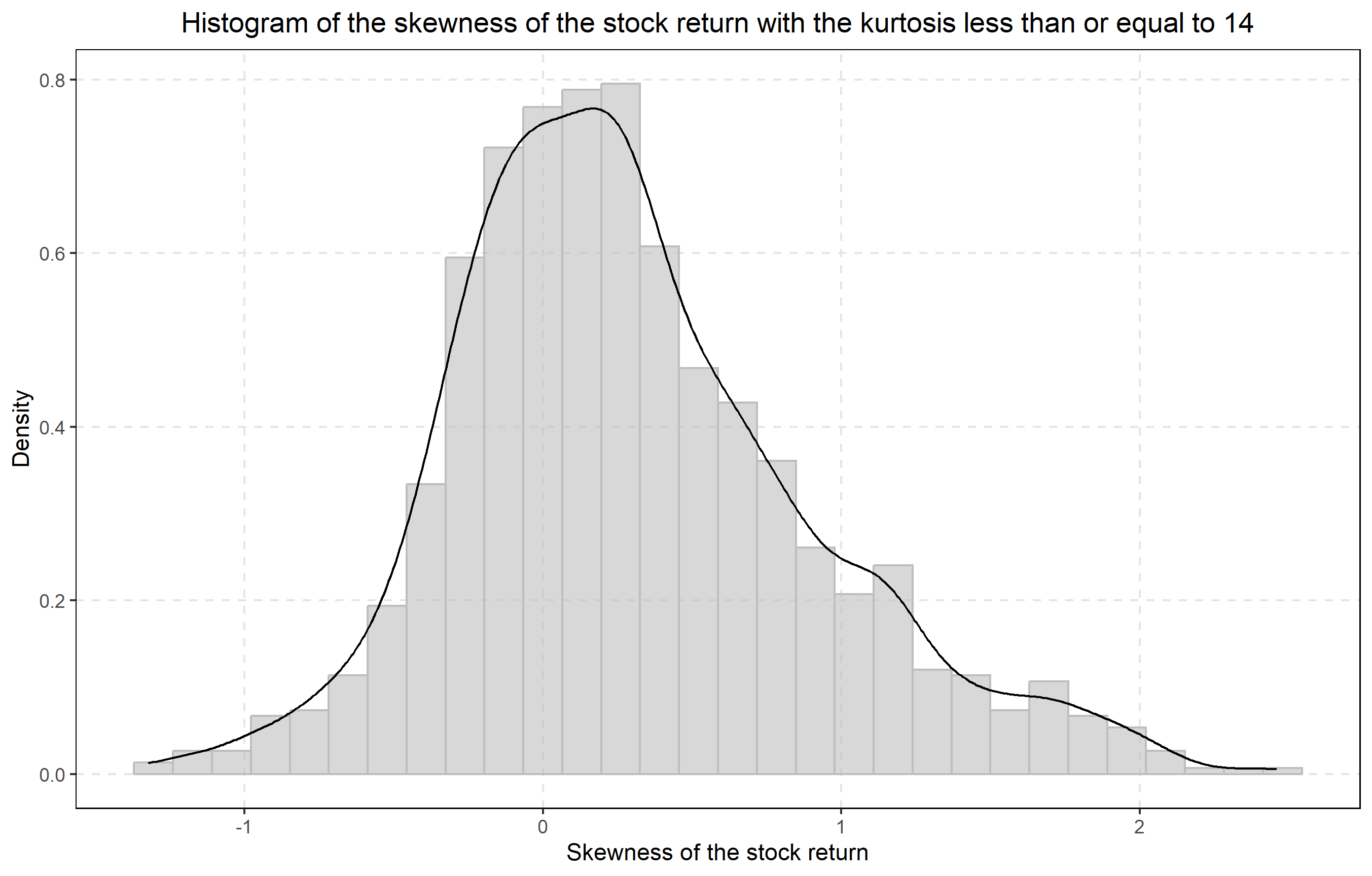}%
}\\
{}%
\end{center}}}
\\%
{\parbox[b]{2.7337in}{\begin{center}
\fbox{\includegraphics[
height=2.0435in,
width=2.7337in
]%
{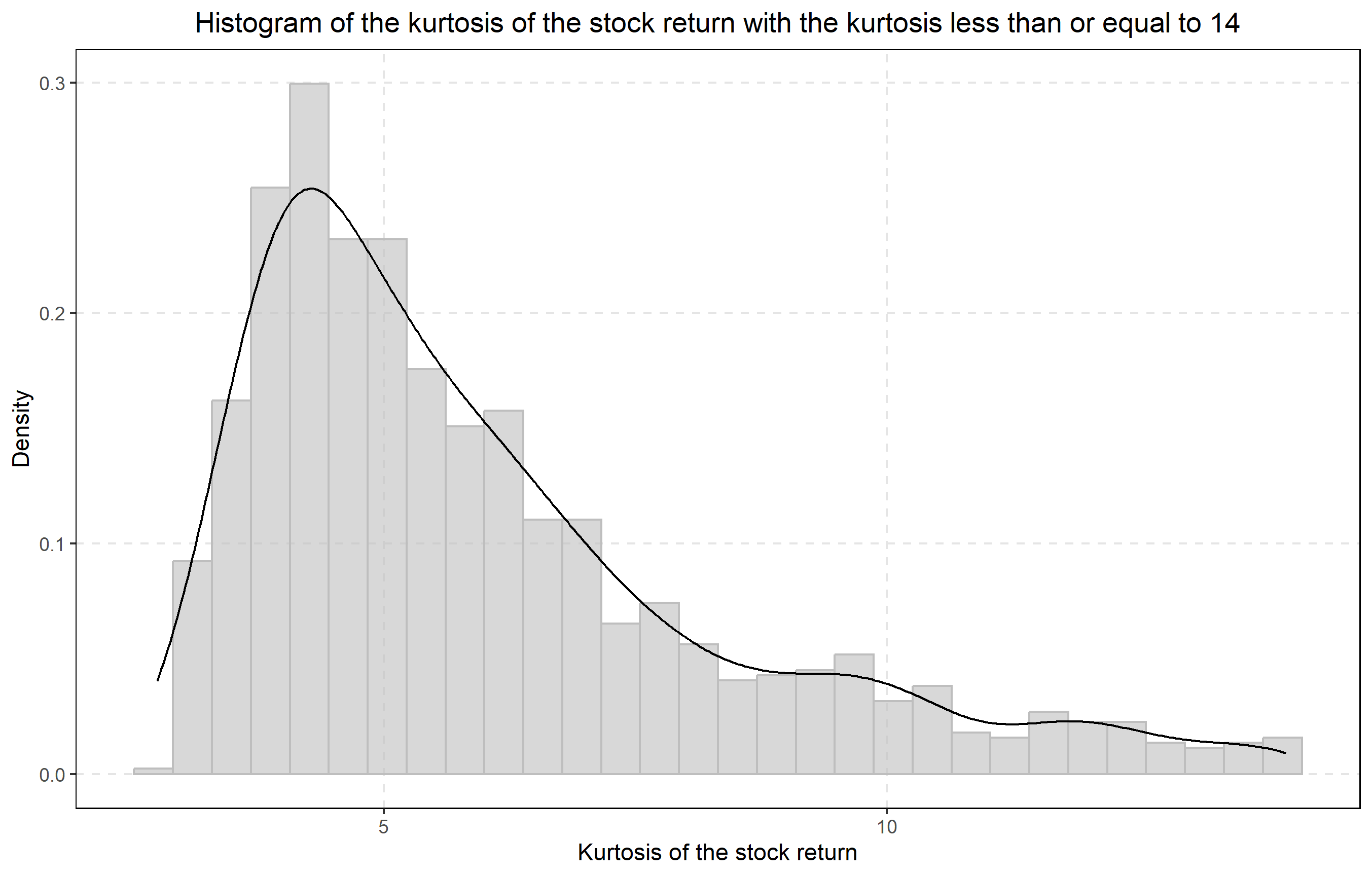}%
}\\
{}%
\end{center}}}
&
\end{array}
\]
}

{\small
\end{figure}%
}

{\small \pagebreak}

{\small
\begin{figure}[H]%
\caption{Histogram and density function of the individual stock returns with kurtosis less than or equal to 16 over 2002m1-2021m12 ($T=240$) and $n=1,148$}
\label{fig:Histo_kur16}}

{\small
\[%
\begin{array}
[c]{cc}%
{\parbox[b]{2.7337in}{\begin{center}
\fbox{\includegraphics[
height=2.0435in,
width=2.7337in
]%
{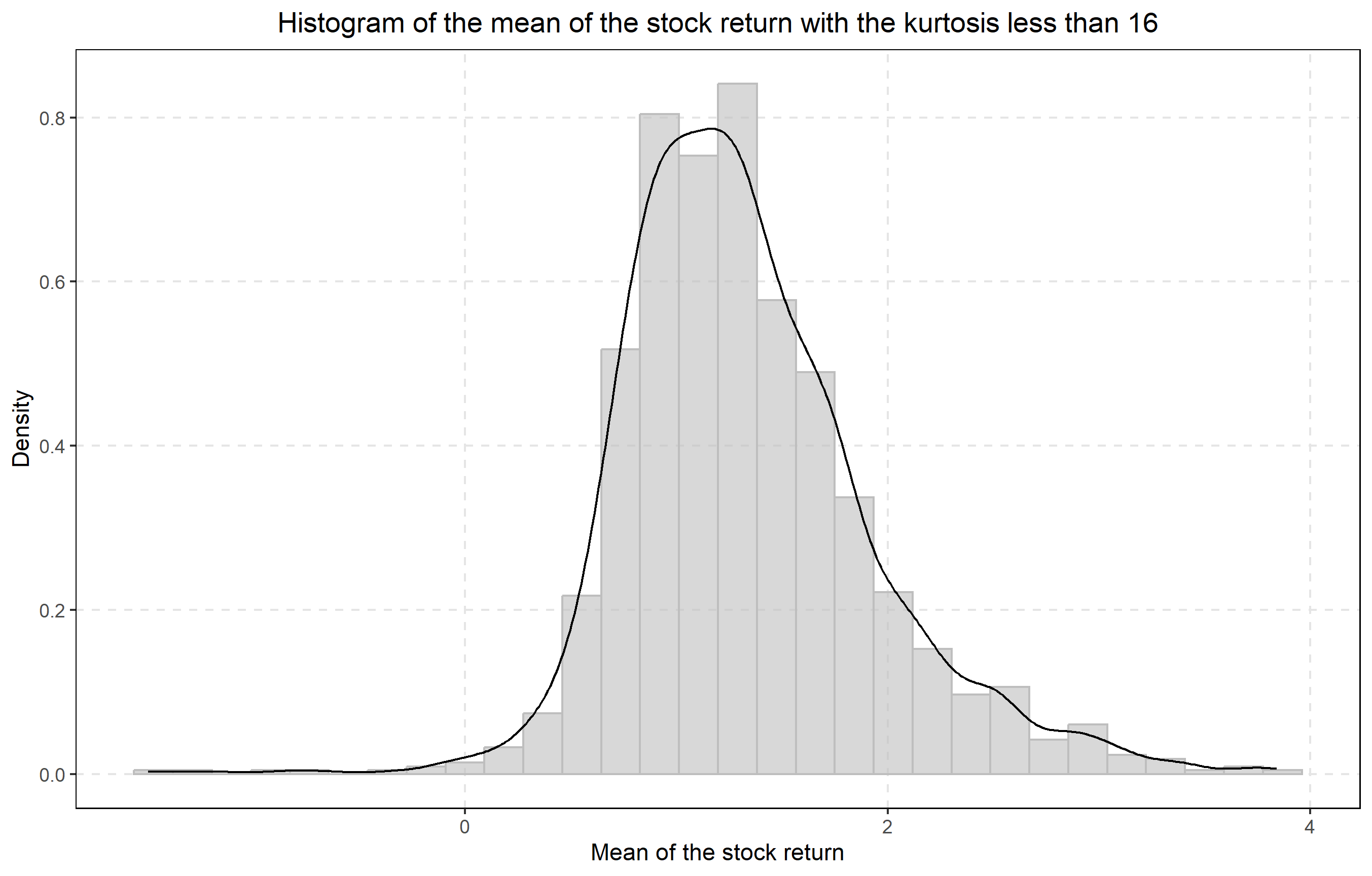}%
}\\
{}%
\end{center}}}
&
{\parbox[b]{2.7337in}{\begin{center}
\fbox{\includegraphics[
height=2.0435in,
width=2.7337in
]%
{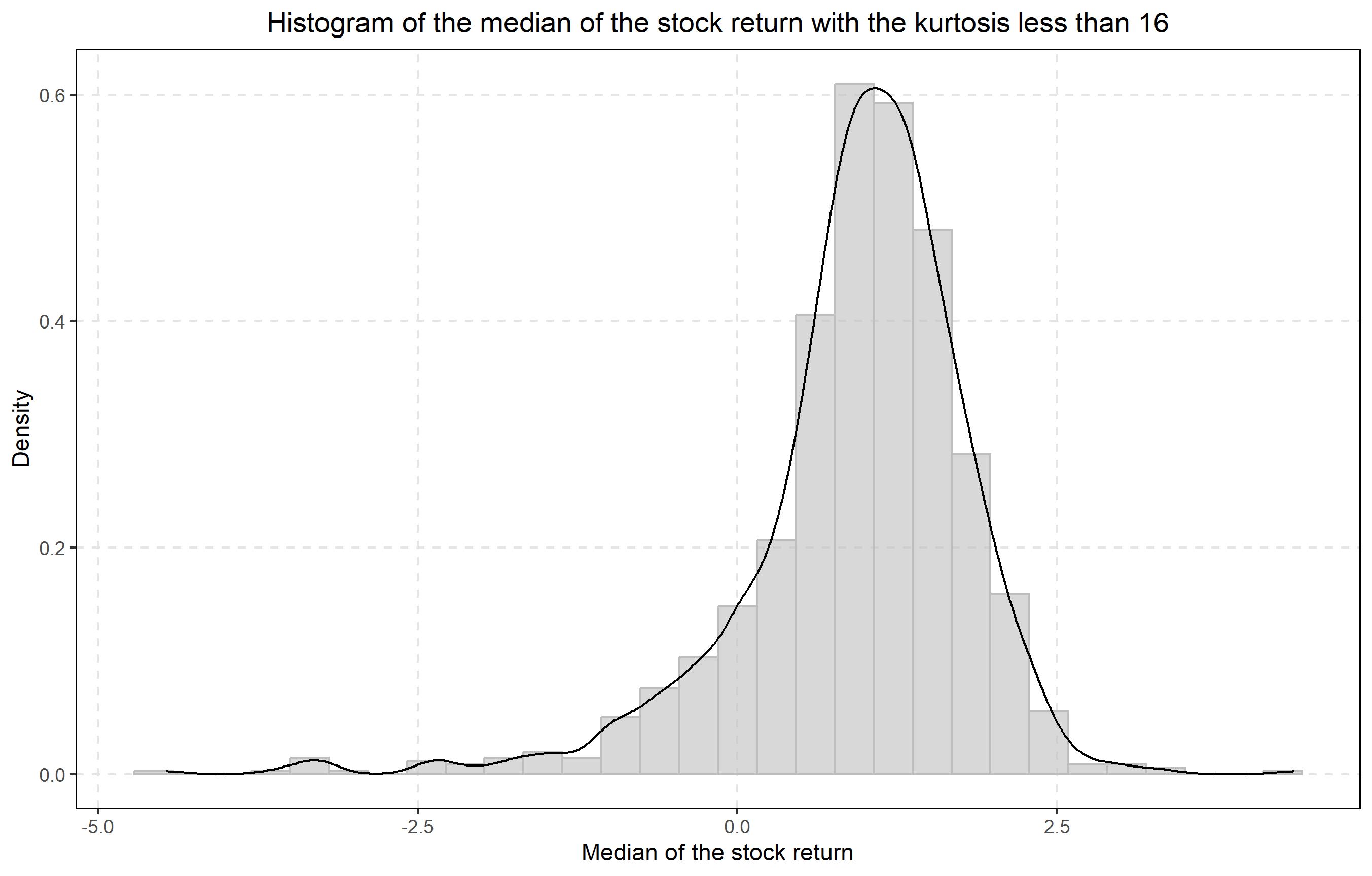}%
}\\
{}%
\end{center}}}
\\%
{\parbox[b]{2.7337in}{\begin{center}
\fbox{\includegraphics[
height=2.0435in,
width=2.7337in
]%
{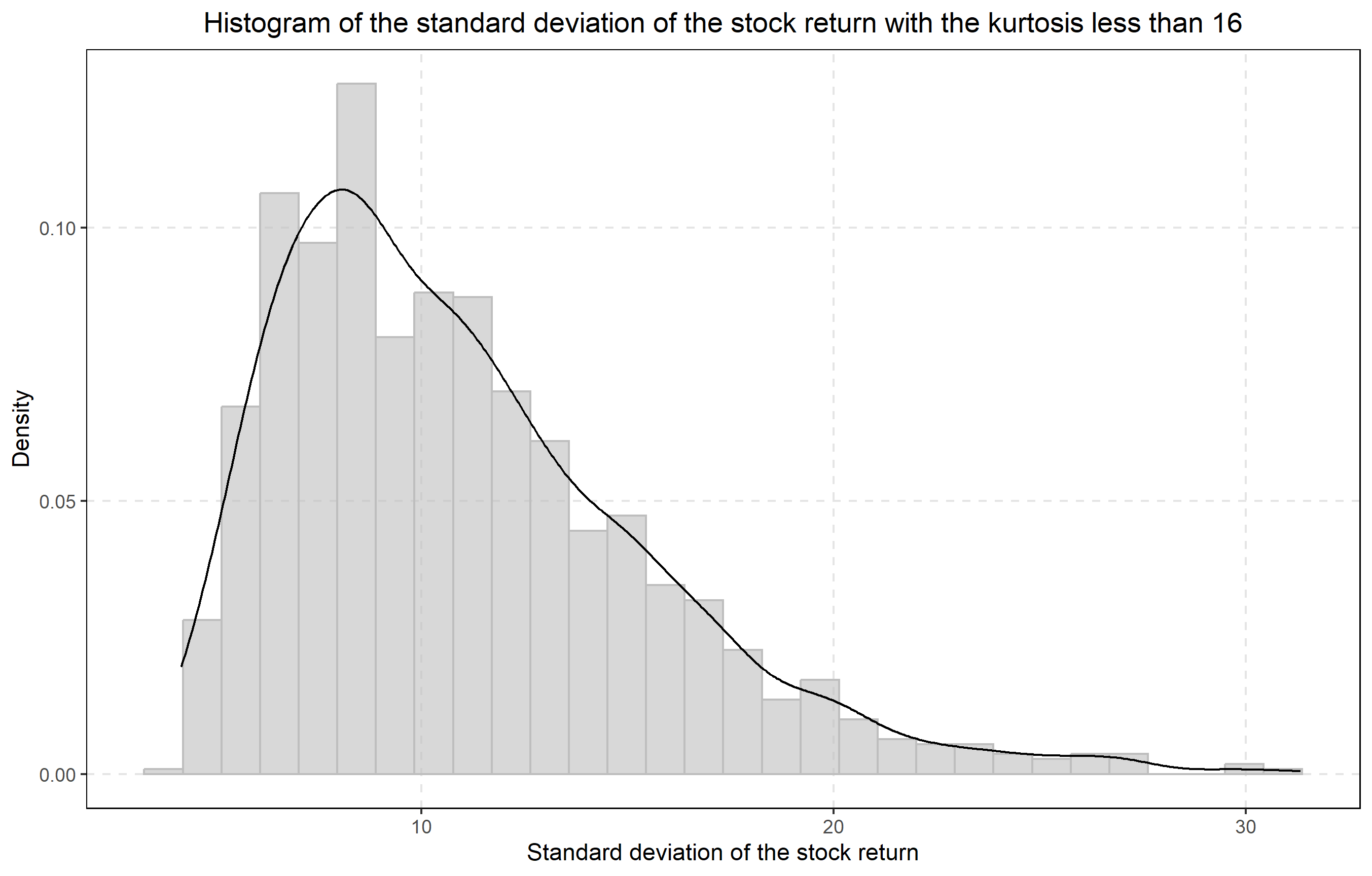}%
}\\
{}%
\end{center}}}
&
{\parbox[b]{2.7337in}{\begin{center}
\fbox{\includegraphics[
height=2.0435in,
width=2.7337in
]%
{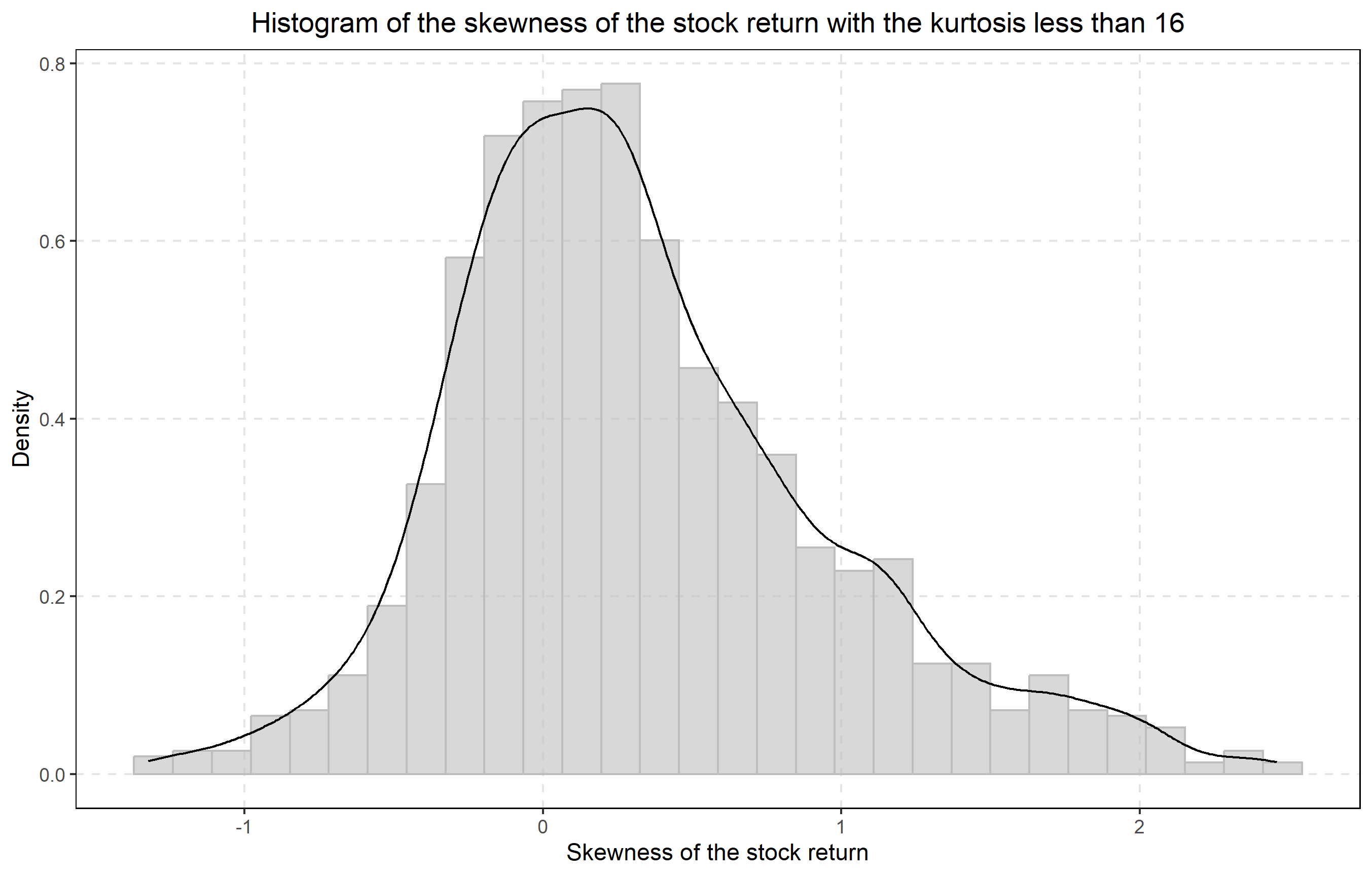}%
}\\
{}%
\end{center}}}
\\%
{\parbox[b]{2.7337in}{\begin{center}
\fbox{\includegraphics[
height=2.0435in,
width=2.7337in
]%
{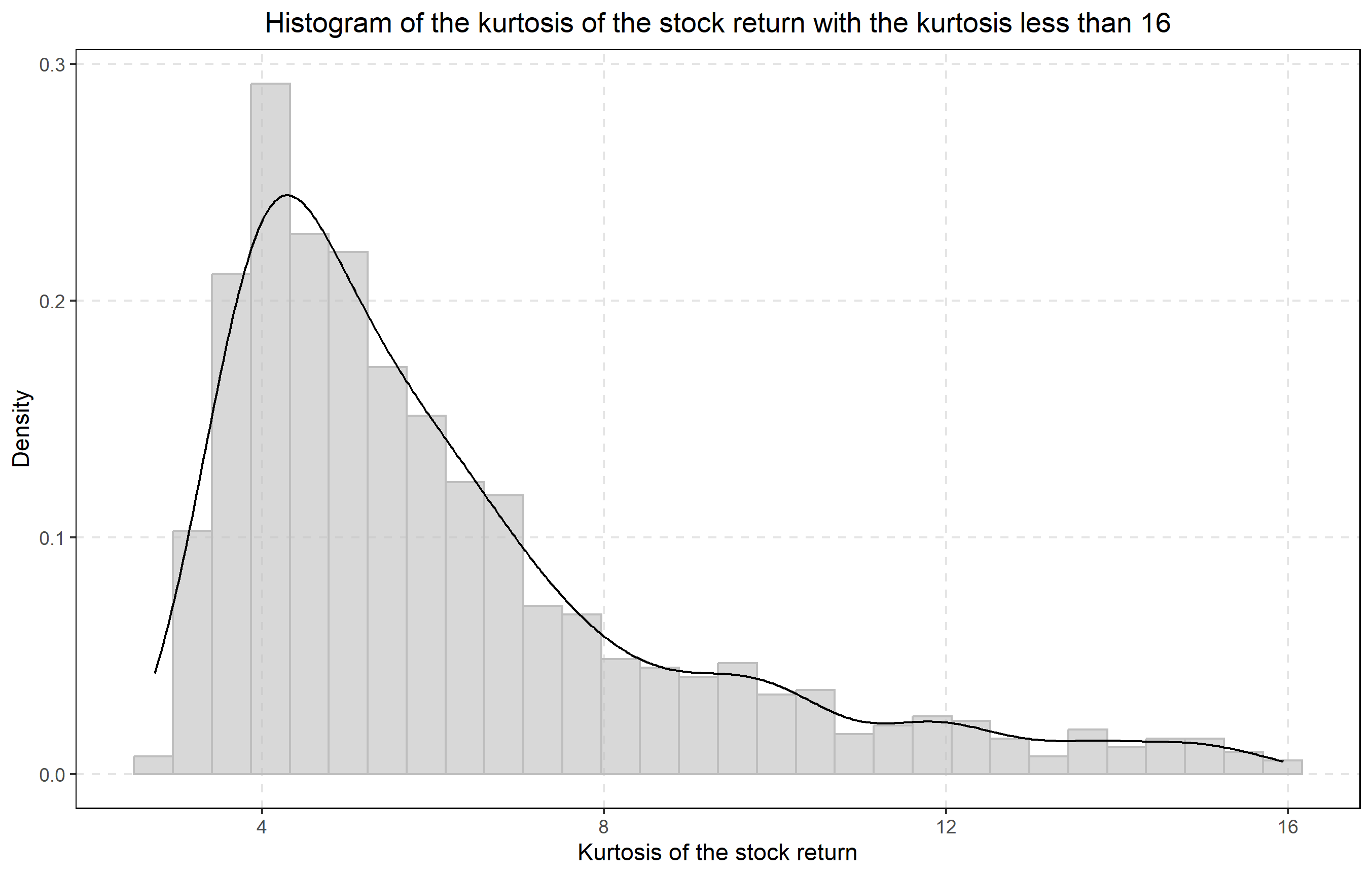}%
}\\
{}%
\end{center}}}
&
\end{array}
\]
}

{\small
\end{figure}%
}

{\small \pagebreak}

\section{{\protect\small The MC design and its calibration\label{EMC}}}

\subsection{{\protect\small Generation of returns}}

{\small Excess returns, $r_{it}$, are generated as
\begin{equation}
r_{it}=\mathit{\alpha}_{i}+\sum_{k=M,H,S}\beta_{ik}f_{kt}+u_{it},
\label{eqn:DGP}%
\end{equation}%
\[
=\mathit{\alpha}_{i}+\boldsymbol{\beta}_{i}^{\prime}\mathbf{f}_{t}+u_{it},
\]
for $i=1,2,..,n;$ $t=1,2,...,T$, with
\begin{equation}
\mathit{\alpha}_{i}=c+\boldsymbol{\beta}_{i}^{\prime}\boldsymbol{\phi}%
+\eta_{i}, \label{APT condition}%
\end{equation}
where $\mathbf{f}_{t}=(f_{Mt},f_{Ht},f_{St})^{\prime}$, and $\boldsymbol{\beta
}_{i}=(\beta_{Mi},\beta_{Hi},\beta_{\operatorname{Si}})^{\prime}$, and
$\boldsymbol{\phi}=(\phi_{M},\phi_{H},\phi_{S})^{\prime}$. }

\subsection{{\protect\small Generation of factors}}

{\small Factors are generated as first-order autoregressive, AR(1), processes
with GARCH(1,1) effects:%
\begin{equation}
f_{kt}=\ \mu_{k}(1-\rho_{k})+\rho_{k}f_{k,t-1}+\left(  1-\rho_{k}^{2}\right)
^{1/2}\sigma_{kt}\zeta_{kt}, \label{eqn:factor}%
\end{equation}%
\begin{equation}
\sigma_{kt}^{2}=\ (1-b_{k}-c_{k})\sigma_{k}^{2}+b_{k}\sigma_{k,t-1}^{2}%
+c_{k}\sigma_{k,t-1}^{2}\zeta_{k,t-1}^{2}, \label{eqn:GARCH}%
\end{equation}
for $k=M,H,S$, starting from $t=-49,...0,1,2,...,T,$ with $f_{k,-50}=0$ (and
$\sigma_{k,-50}=0$ in the case where $c_{k}\neq0$) to minimize the effects of
the initial values on the sample $f_{kt},t=1,2,...,T$ used in the simulations.
}

{\small The data generating process for the factors is calibrated using the
full set of Fama-French three factor data set covering the period
1963m8-2021m12. The calibrated parameter values are $\boldsymbol{\mu}=\left(
\mu_{M},\mu_{H},\mu_{S}\right)  ^{\prime}=\left(  0.59,0.27,0.23\right)
^{\prime},$ $\boldsymbol{\sigma}=\left(  \sigma_{M},\sigma_{H},\sigma
_{S}\right)  ^{\prime}=(4.45,2.86,3.03)^{\prime},$ and $\boldsymbol{\rho
}=\left(  \rho_{M},\rho_{H},\rho_{S}\right)  ^{\prime}%
=(0.06,0.17,0.07)^{\prime}$. Note that $Var(f_{kt})=\sigma_{k}^{2}$. The
parameters of (\ref{APT condition}) are also estimated using the
bias-corrected procedure and are set as $c=0.83$ and $\boldsymbol{\phi
}=\left(  -0.49,-0.35,0.16\right)  ^{\prime}$. To ensure that correlation
across the three factors match the Fama-French data we generated
$\boldsymbol{\zeta}_{t}=(\zeta_{Mt},\zeta_{Ht},\zeta_{St})^{\prime}$ as
$\boldsymbol{\zeta}_{t}=\mathbf{Q}_{\zeta}\boldsymbol{\omega}_{t},$ where
$\mathbf{Q}_{\zeta}$ is the Cholesky factor of $\mathbf{R}_{\zeta}$, the
correlation matrix of $\mathbf{f}_{t}$ given by%
\[
\mathbf{R}_{\zeta}=\left(
\begin{array}
[c]{ccc}%
1 & -0.21 & 0.28\\
-0.21 & 1 & -0.02\\
0.27 & -0.02 & 1
\end{array}
\right)  .
\]
We consider both Gaussian and non-Gaussian errors and generate
$\boldsymbol{\omega}_{t}$ as $IID(\mathbf{0,I}_{3})$, as well as a
multivariate t with $5$ degrees of freedom, namely $t(\mathbf{0}%
,\mathbf{I}_{3},5)$. The remaining parameters are set as $b_{k}=c_{k}=0$ to
generate homoskedastic errors, and $b_{k}=0.8$ and $c_{k}=0.1$ for $k=M,H,S$
to generate GARCH effects. }

\subsection{{\protect\small Estimation of factor models}}

{\small Consider the AR(1) processes with a GARCH(1,1) errors%
\begin{equation}
f_{kt}=\ \mu_{k}(1-\rho_{k})+\rho_{k}f_{k,t-1}+\left(  1-\rho_{k}^{2}\right)
^{1/2}\sigma_{kt}\zeta_{kt}, \label{eqn:factor OL}%
\end{equation}
}

{\small
\begin{equation}
\sigma_{kt}^{2}=\ (1-b_{k}-c_{k})\sigma_{k}^{2}+b_{k}\sigma_{k,t-1}^{2}%
+c_{k}\sigma_{k,t-1}^{2}\zeta_{k,t-1}^{2}, \label{eqn:GARCH OL}%
\end{equation}
where $f_{kt}$, $k=M,H,S$ and $t=1,...,T$, denote the the values of three
factors $MKT,HML,SMB$ in month $t$, respectively. The estimates of GARCH
parameters obtained using the sample 2002m1 - 2021m12 ($T=240$) are summarized
in Table \ref{tab:factor_para_20y}.}

{\small
\begin{table}[H]%
\caption{GARCH parameters for the models of three Fama-French factors for the
sample over 2002m1 - 2021m12 $(T=240)$}\label{tab:factor_para_20y} }

\begin{center}
{\small
\begin{tabular}
[c]{cccccc}\hline\hline
& $\hat\mu$ & $\hat\rho$ & $\hat\sigma$ & $\hat b$ & $\hat c$\\\hline
$MKT$ & 0.8030 & 0.0711 & 4.5703 & 0.6781 & 0.2395\\
& (0.3035) & (0.0648) & ($\cdot$) & (0.0854) & (0.0627)\\
$HML$ & --0.0805 & 0.1816 & 3.1513 & 0.7582 & 0.1987\\
& (0.2194) & (0.0639) & ($\cdot$) & (0.0986) & (0.0686)\\
$SMB$ & 0.1718 & --0.0267 & 2.5815 & 0.8353 & 0.0680\\
& (0.1651) & (0.0649) & ($\cdot$) & (0.1845) & (0.0606)\\\hline\hline
\end{tabular}
}
\end{center}

{\small
\end{table}%
}

{\small The correlation matrix of three factors $MKT$, $HML$, $SMB$ over
2001m1-2021m9 $(T=240)$ is
\[%
\begin{pmatrix}
1.00 & 0.20 & 0.35\\
0.20 & 1.00 & 0.35\\
0.35 & 0.35 & 1.00
\end{pmatrix}
.
\]
}

\subsection{{\protect\small Factor loadings estimates}}

{\small For each of the securities $i=1,2,...,1175$ (with kurtosis below 16),
and $t=1,2,...,240$, OLS regressions excess returns $y_{it}$ for security $i$
was run on an intercept and the three FF factors
\[
y_{it}=r_{it}-r_{t}^{f}=a_{i}+\sum_{k\in\{M,H,S\}}\beta_{ik}f_{kt}+u_{it},
\]
where $r_{it}$ is the return of $i^{th}$ security at time $t$, inclusive of
dividend (if any), and $r_{t}^{f}$ is the risk free rate. The sample mean and
standard deviation of the excess return for each individual stock, denoted as
$\bar{y}_{iT}$ and $s_{yiT}$ are computed as
\begin{equation}
\bar{y}_{iT}=T^{-1}\sum_{t=1}^{T}y_{it}, \label{eqn:mean_y}%
\end{equation}%
\begin{equation}
sd_{iT}(y)=\sqrt{(T-1)^{-1}\sum_{t=1}^{T}(y_{it}-\bar{y}_{iT})^{2}}.
\label{eqn:sd_y}%
\end{equation}
}

{\small The estimates $\hat{a}_{i,T}$, $\hat{\beta}_{ik,T}$, $k=M,H,S$ are
given by
\begin{equation}
\left(  \hat{a}_{i,T},\hat{\beta}_{iM,T},\hat{\beta}_{iH,T},\hat{\beta}%
_{iS,T}\right)  ^{\prime}=\boldsymbol{(F_{0}^{\prime}F}_{0}\boldsymbol{)}%
^{-1}\boldsymbol{F_{0}^{\prime}y_{i\circ}}, \label{eqn:abeta}%
\end{equation}
where $\boldsymbol{F}_{0}=\left(  \boldsymbol{\tau}_{T},\mathbf{F}\right)  $,
$\mathbf{F=}\left(  \boldsymbol{f}_{1},\boldsymbol{f}_{2},...,\boldsymbol{f}%
_{T}\right)  ^{\prime}$, $\boldsymbol{f}_{t}=(f_{Mt},f_{Ht},f_{St})^{\prime}$
and $\mathbf{y}_{i\circ}=\left(  y_{i1},y_{i2},...,y_{iT}\right)  ^{\prime}$.
The standard error of the $i^{th}$ regression, denoted as $s_{iT}$, is given
by
\begin{equation}
\hat{\sigma}_{iT}^{2}=(T-K-1)^{-1}\sum_{t=1}^{T}\hat{u}_{it}^{2},
\label{eqn:S_iT}%
\end{equation}
where $\hat{u}_{it}=y_{it}-\hat{a}_{iT}-\sum_{k\in\{M,H,S\}}\hat{\beta}%
_{ik,T}f_{kt}$. The coefficient of determination of the $i^{th}$ regression,
denoted by $R_{iT}^{2}$, is given by
\begin{equation}
R_{iT}^{2}=1-\frac{\sum_{t=1}^{T}\hat{u}_{it}^{2}}{\sum_{t=1}^{T}(y_{it}%
-\bar{y}_{iT})^{2}}. \label{eqn:R2_iT}%
\end{equation}
We compute the summary statistics: mean, median, standard deviation (S.D.),
skewness, kurtosis, interquartile range, minimum, maximum for the sample mean
and standard deviation of the excess returns, the estimates and the
corresponding standard error, $R$-squared of the regressions over 2002m1 -
2021m12 $(T=240)$, for the $n=1,175$ securities: $\bar{y}_{iT}$, $sd_{i,T}%
(y)$, $\hat{a}_{i,T}$, $\hat{\beta}_{iM,T}$, $\hat{\beta}_{iH,T}$, $\hat
{\beta}_{iS,T}$, $\hat{\sigma}_{iT}^{2}$ and $R_{iT}^{2}$ for $i=1,2,...,1175$%
, computed using \eqref{eqn:mean_y}-\eqref{eqn:R2_iT}. The results are
summarized in Table \ref{tab:betas_20y}.}

{\small
\begin{table}[H]%
\caption{The summary statistics of the estimates, standard error and R-squared
of the panel regression over 2002m1 - 2021m12 $(T=240)$ and $n=1,175$}\label{tab:betas_20y}%
}

\begin{center}
{\small
\begin{tabular}
[c]{lrrrrrrrr}\hline\hline
& mean & median & S.D. & skewness & kurtosis & IQR & min & max\\\hline
$\bar{y}$ & 1.2366 & 1.1572 & 0.5922 & 0.6219 & 4.9697 & 0.7085 & --1.5971 &
3.7434\\
$sd(y)$ & 11.0006 & 10.1115 & 4.4695 & 1.1241 & 4.4498 & 5.6826 & 4.1706 &
31.3322\\
$\hat{a}$ & 0.3749 & 0.3562 & 0.5810 & --0.1980 & 5.5804 & 0.6493 & --2.8882 &
2.9736\\
$\hat{\beta}_{M}$ & 0.9714 & 0.9362 & 0.4279 & 0.5352 & 3.3349 & 0.5810 &
--0.0689 & 2.8590\\
$\hat{\beta}_{H}$ & 0.2235 & 0.2093 & 0.5276 & --0.0916 & 4.7150 & 0.5915 &
--2.4498 & 2.9788\\
$\hat{\beta}_{S}$ & 0.6061 & 0.5790 & 0.5381 & 0.3552 & 3.2197 & 0.7495 &
--0.7504 & 2.7474\\
$\hat{\sigma}$ & 9.4123 & 8.2297 & 4.2794 & 1.3292 & 5.0599 & 5.3811 &
3.4523 & 30.6953\\
$R^{2}$ & 0.2840 & 0.2782 & 0.1383 & 0.2172 & 2.2941 & 0.2131 & 0.0050 &
0.6814\\\hline\hline
\end{tabular}
}
\end{center}

{\small
\end{table}%
}

{\small The histogram for the $\hat{a}_{T}$ and the $\hat{\beta}_{k,T}$ for
$k=M,H,S$ over 2002m1 - 2021m12 $(T=240)$, each using 1175 data points
$\hat{a}_{i,T}$ and the $\hat{\beta}_{ik,T}$ for $k=M,H,S$, $i=1,2,...,1175$,
is shown in the Figure \ref{fig:Betas_20y}.}

{\small
\begin{figure}[H]%
\caption{Histogram and density function of the coefficients of the panel regression over 2002m1-2021m12 ($T=240$) and $n=1,775$}
\label{fig:Betas_20y}
\[%
{\fbox{\includegraphics[
height=2.7397in,
width=3.6417in
]%
{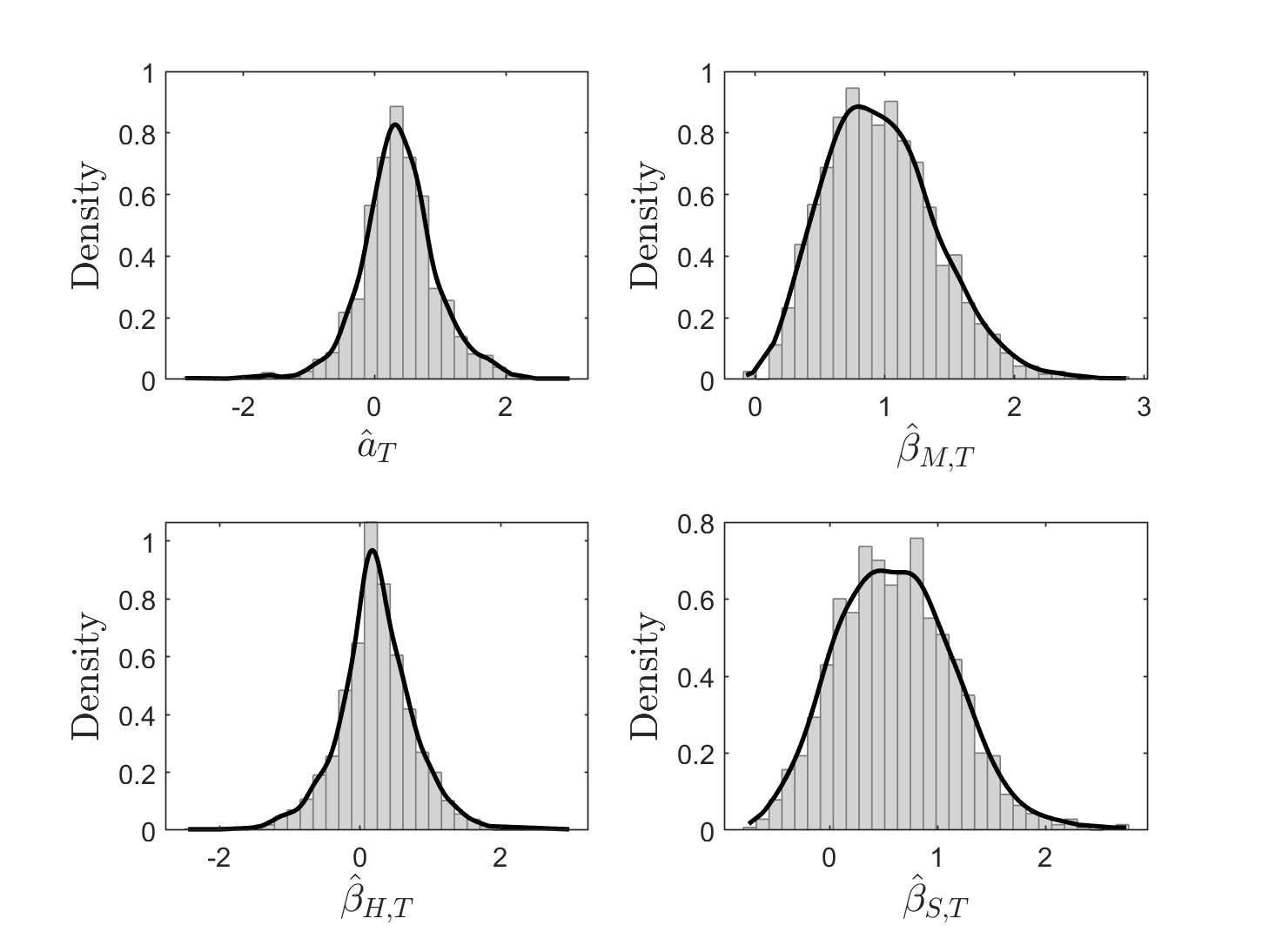}%
}}
\]
}

{\small
\end{figure}%
}

{\small The histogram for the $R_{T}^{2}$ over 2002m1 - 2021m12 $(T=240)$,
using 1,175 data points $R_{iT}^{2}$ for $i=1,2,...,1175$, is shown in the
Figure \ref{fig:R2_20y}.}

{\small
\begin{figure}[H]%
\caption{Histogram and density function of the R-squared of the panel regression over 2002m1-2021m12 ($T=240$) and $n=1,775$}
\label{fig:R2_20y}
\[%
{\fbox{\includegraphics[
height=2.738in,
width=3.6417in
]%
{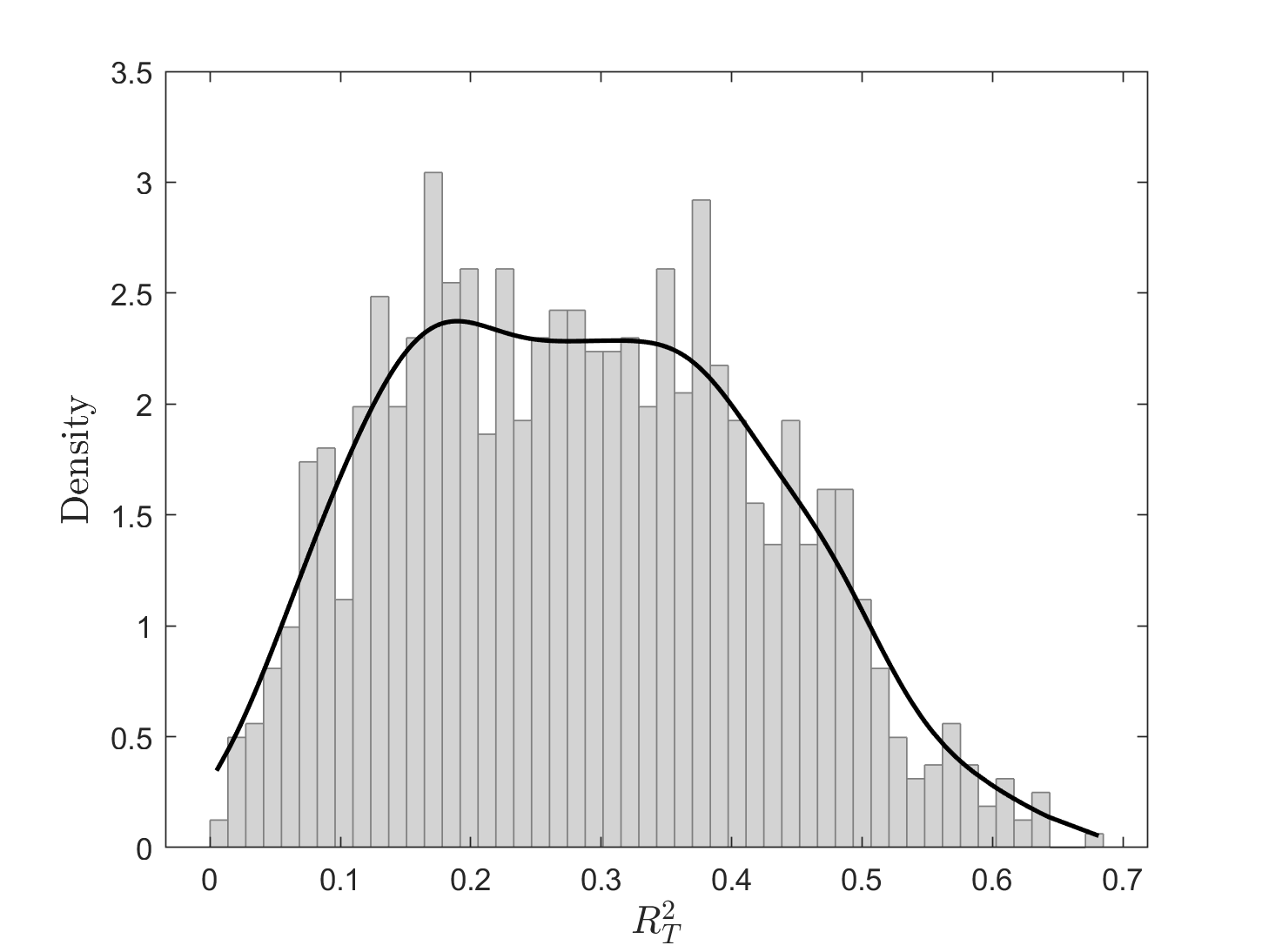}%
}}
\]%
\end{figure}%
}

\subsection{{\protect\small Calibrating the fit of return
regressions\label{Regfit}}}

{\small To see how $\kappa$ controls the regression fit, note that the $n$
return processes (\ref{eqn:DGP}) can be written more compactly in vector form
as%
\[
\mathbf{r}_{t}=\mathbf{a+Bf}_{t}+\mathbf{u}_{t},
\]
where $\mathbf{u}_{t}=\boldsymbol{\gamma}g_{t}+\kappa\mathbf{\hat{S}%
}\boldsymbol{\varepsilon}_{t},$ $\mathbf{a=(}\alpha_{1},\alpha_{2}%
,...,\alpha_{n}\,)^{\prime}$,$\,$with $a_{i}$ given by $a_{i}%
=c+\boldsymbol{\beta}_{i}^{\prime}\boldsymbol{\phi}+\eta_{i},$ and
\[
\mathbf{\hat{S}=}Diag(\mathbf{\hat{S}}_{b}\text{, }b=1,2,...,B).
\]
Overall, the DGP for the return regressions can be written compactly as%
\[
\mathbf{r}_{t}=c\boldsymbol{\tau}_{n}+\mathbf{B}\left(  \mathbf{f}%
_{t}+\boldsymbol{\phi}\right)  +\boldsymbol{\gamma}g_{t}+\kappa\mathbf{\hat
{S}}\boldsymbol{\varepsilon}_{t}+\boldsymbol{\eta}_{n},
\]
where $\boldsymbol{\eta}_{n}=\mathbf{(}\eta_{1},\eta_{2},...,\eta_{n}%
)^{\prime}$. We abstract from pricing errors and weak latent factor and
$\eta_{i}=0$, $g_{t}=0$, and set $\kappa$ such that the pooled $R^{2}$
($PR^{2}$) of return regressions can be controlled to be around $R_{0}%
^{2}=0.30$. We have
\[
PR_{nT}^{2}=1-\frac{n^{-1}T^{-1}\sum_{t=1}^{T}\sum_{i=1}^{n}E\left(
u_{it}^{2}\right)  }{n^{-1}T^{-1}\sum_{t=1}^{T}\sum_{i=1}^{n}Var(r_{it})}.
\]%
\begin{align*}
Var(r_{it})  &  =Var\left(  \boldsymbol{\beta}_{i}^{\prime}\mathbf{f}%
_{t}\right)  +Var(u_{it})=E\left(  \boldsymbol{\beta}_{i}^{\prime
}\mathbf{\Sigma}_{f}\boldsymbol{\beta}_{i}\right)  +E(u_{it}^{2})\\
& \\
&  =\text{Tr}\left[  \mathbf{\Sigma}_{f}E\left(  \boldsymbol{\beta}%
_{i}\boldsymbol{\beta}_{i}^{\prime}\right)  \right]  +E(u_{it}^{2}).
\end{align*}
Denote the $k^{th}$ element of $\boldsymbol{\beta}_{i}$ by $\beta_{ik}$, then
if $\beta_{ik}\sim IIDN(\mu_{\beta_{k}},\sigma_{\beta_{k}}^{2}),$ and
$\beta_{ik}$ are distributed independently over $k=1,2,...,K$, we have%
\[
E\left(  \boldsymbol{\beta}_{i}\boldsymbol{\beta}_{i}^{\prime}\right)
=Diag\left(  \mu_{\beta_{k}}^{2}+\sigma_{\beta_{k}}^{2}\text{, for
}k=1,2,...,K\right)  .
\]
Then,
\[
Var(r_{it})=\mathit{Tr}\left[  \mathbf{\Sigma}_{f}E\left(  \boldsymbol{\beta
}_{i}\boldsymbol{\beta}_{i}^{\prime}\right)  \right]  +E(u_{it}^{2}%
)=\sum_{k=1}^{K}\sigma_{fk}^{2}\left(  \mu_{\beta_{k}}^{2}+\sigma_{\beta_{k}%
}^{2}\right)  +E(u_{it}^{2}).
\]
Also,
\[
\sum_{i=1}^{n}E(u_{it}^{2})=Tr\left[  E\left(  \mathbf{u}_{t}\mathbf{u}%
_{t}^{\prime}\right)  \right]  ,
\]
and (when $g_{t}=0$) we have
\[
E\left(  \mathbf{u}_{t}\mathbf{u}_{t}^{\prime}\right)  =\mathbf{V}_{u}%
=\kappa^{2}\mathbf{\hat{S}}E\left(  \boldsymbol{\varepsilon}_{t}%
\boldsymbol{\varepsilon}_{t}^{\prime}\right)  \mathbf{\hat{S}}^{\prime}%
=\kappa^{2}\mathbf{\hat{S}V}_{\varepsilon}\mathbf{\hat{S}}^{\prime},
\]
where $\mathbf{V}_{\varepsilon}=Diag(\mathbf{V}_{b\varepsilon}^{(r)}%
,b=1,2,...,B)^{\prime}$. Hence,%
\[
n^{-1}T^{-1}\sum_{t=1}^{T}\sum_{i=1}^{n}E\left(  u_{it}^{2}\right)
=n^{-1}Tr(\mathbf{V}_{u})=\kappa^{2}n^{-1}Tr\left(  \mathbf{\hat{S}%
V}_{\varepsilon}\mathbf{\hat{S}}^{\prime}\right)
\]
and%
\[
PR_{nT}^{2}=1-\frac{\kappa^{2}n^{-1}Tr\left(  \mathbf{\hat{S}V}_{\varepsilon
}\mathbf{\hat{S}}^{\prime}\right)  }{\sum_{k=1}^{K}\sigma_{fk}^{2}\left(
\mu_{\beta_{k}}^{2}+\sigma_{\beta_{k}}^{2}\right)  +\kappa^{2}n^{-1}Tr\left(
\mathbf{\hat{S}V}_{\varepsilon}\mathbf{\hat{S}}^{\prime}\right)  }.
\]
To achieve $lim_{n\rightarrow\infty}PR_{nT}^{2}=R_{0}^{2}$, we need to set
(assuming all $K$ factors are strong)
\begin{equation}
\kappa^{2}=\frac{\sum_{k=1}^{K}\sigma_{fk}^{2}\left(  \mu_{\beta_{k}}%
^{2}+\sigma_{\beta_{k}}^{2}\right)  }{n^{-1}Tr\left(  \mathbf{\hat{S}%
V}_{\varepsilon}\mathbf{\hat{S}}^{\prime}\right)  }\left(  \frac{1-R_{0}^{2}%
}{R_{0}^{2}}\right)  . \label{kappa2S}%
\end{equation}
When there are no idiosyncratic error dependence, namely when $\mathbf{\hat
{S}=I}_{n}$, the above expression simplifies to
\begin{equation}
\kappa^{2}=\frac{\sum_{k=1}^{K}\sigma_{fk}^{2}\left(  \mu_{\beta_{k}}%
^{2}+\sigma_{\beta_{k}}^{2}\right)  }{n^{-1}Tr\left(  \mathbf{V}_{\varepsilon
}\right)  }\left(  \frac{1-R_{0}^{2}}{R_{0}^{2}}\right)  .
\label{fitallfactors}%
\end{equation}
}

{\small If we only use the market factor, we have}

{\small
\begin{equation}
\kappa^{2}=\frac{\sigma_{M}^{2}\left(  \mu_{\beta_{M}}^{2}+\sigma_{\beta_{M}%
}^{2}\right)  }{n^{-1}Tr\left(  \mathbf{V}_{\varepsilon}\right)  }\left(
\frac{1-R_{0}^{2}}{R_{0}^{2}}\right)  . \label{fitfirstFactor}%
\end{equation}
We expect that $n^{-1}Tr\left(  \mathbf{V}_{\varepsilon}\right)  \rightarrow
1$, if $E(\sigma_{ii})=1$, as under our DGP. }

\subsection{{\protect\small Estimation of FF factor strengths}}

{\small Denote by $t_{ik,T}=\hat{\beta}_{ik,T}\,/\,\text{s.e.}\left(
\hat{\beta}_{ik,T}\right)  $, the t-statistic corresponding to $\beta_{ik}.$
The total number of factor loadings of factor $k$, that are statistically
significant over $i=1,2,\ldots,n,$ $n=1,175$, is:
\[
\hat{D}_{nT,k}=\sum_{i=1}^{n}\hat{d}_{ik,nT}=\sum_{i=1}^{n}\mathbf{1}\left[
|t_{ik,T}|>c_{p}(n)\right]  \text{,}%
\]
where $\mathbf{1}\left(  A\right)  =1$ if $A>0$, and zero otherwise, and the
critical value function that allows for the multiple testing nature of the
problem, $c_{p}(n,\delta)$, is given by
\begin{equation}
c_{p}(n,\delta)=\Phi^{-1}\left(  1-\frac{p}{2n^{\delta}}\right)  , \label{CVf}%
\end{equation}
where $p$ is the nominal size, set, following
\citet*[BKP]{bailey2021measurement}, as $p=0.1,$ $\delta>0$ is the critical
value exponent, set $\delta=0.25,$ and $\Phi^{-1}(\cdot)$ is the inverse
cumulative distribution function of the standard normal distribution. Let
$\hat{\pi}_{nT,k}$ be the fraction of significant loadings of factor $k$, and
note that $\hat{\pi}_{nT,k}=\hat{D}_{nT,k}/n$. The strength of factor $k$,
denoted by $\alpha_{k0}$, for $k=1,2,\ldots,K$, $K=3$, is estimated by
\begin{equation}
\hat{\alpha}_{k}=%
\begin{cases}
1+\frac{\ln\hat{\pi}_{nT,k}}{\ln n},\text{ if }\hat{\pi}_{nT,k}>0,\\
0,\text{ if }\hat{\pi}_{nT,k}=0.
\end{cases}
\label{phaj}%
\end{equation}
The variance of the estimated strength of factor $k$ is given by
\[
Var(\hat{\alpha}_{k})=(\ln{n})^{-2}\psi_{n}(\alpha_{k0}),
\]
where
\[
\psi_{n}(\alpha_{k0})=p(n-n^{\alpha_{k0}})n^{-\delta-2\alpha_{k0}}\left(
1-\frac{p}{n^{\delta}}\right)  .
\]
So the standard error of the estimated strength of factor $k$ can be computed
by:
\begin{equation}
s.e.(\hat{\alpha}_{k})=\frac{\sqrt{\psi_{n}(\hat{\alpha}_{k})}}{\ln{n}}.
\label{eqn:se_alpha}%
\end{equation}
Mean and variance of the loadings associated with factor $k$ are given by%
\begin{align}
\hat{\mu}_{\beta_{kT}}\left(  \hat{\alpha}_{k}\right)   &  =\frac{\sum
_{i=1}^{n}\mathbf{1}\left[  |t_{ikT}|>c_{p}(n)\right]  \hat{\beta}_{ikT}}%
{\sum_{i=1}^{n}\mathbf{1}\left[  |t_{ikT}|>c_{p}(n)\right]  },
\label{meubetaj}\\
\hat{\sigma}_{\beta_{k}}^{2}\left(  \hat{\alpha}_{k}\right)   &  =\frac
{\sum_{i=1}^{n}\mathbf{1}\left[  |t_{ikT}|>c_{p}(n)\right]  \left(  \hat
{\beta}_{ikT}-\hat{\mu}_{\beta_{kT}}\left(  \hat{\alpha}_{k}\right)  \right)
^{2}}{\sum_{i=1}^{n}\mathbf{1}\left[  |t_{ikT}|>c_{p}(n)\right]  }.
\label{zigsqbetaj}%
\end{align}
In the case where a factor is strong, namely $\alpha_{k}=1$, then it must be
that $\mathbf{1}\left[  |t_{ikT}|>c_{p}(n)\right]  =1$ for all $i$. The
estimated factor strengths $\hat{\alpha}_{k}$ and corresponding standard
errors for $k=M,H,S$, using the sample over 2002m1 - 2021m12 $(T=240)$ and
$n=1,175$ are reported in \ref{tab:str_20y}.}

{\small
\begin{table}[H]%
}

{\small \caption{Strength of three FF factors estimated over 2001m1-2021m9 $(T=240$
and $n=1,175$)}\label{tab:str_20y} }

\begin{center}
{\small
\begin{tabular}
[c]{lccc}\hline\hline
& $M$ & $H$ & $S$\\\hline
$\hat\alpha$ & 0.9941 & 0.8373 & 0.9023\\
& (0.0001) & (0.0014) & (0.0008)\\\hline\hline
\end{tabular}
}
\end{center}

{\small \noindent{\footnotesize \textit{Note:} This table reports the
estimates of the factor strength using \eqref{phaj} and the standard errors
that are reported in () using \eqref{eqn:se_alpha}, for three factors $MKT,
HML$ and $SMB$, using the sample over 2001m1-2021m9 $(T=240)$ and $n=1,175$,
$K=3$. }
\end{table}%
}

\subsection{{\protect\small Estimates of $\boldsymbol{\phi}$ for the FF3
factors}}

{\small
\begin{table}[H]%
}

{\small \caption{The bias-corrected estimates of $c_{0}$ (intercept) and $\phi_{M}$, $\phi_{H}$, and $\phi_{S}$ $\ $for the sample 2002m1 - 2021m12 $(T=240)$ with $n=1175$ securities}
}

\begin{center}
{\small
\begin{tabular}
[c]{ccccc}\hline\hline
& $\hat{c}_{0}$ & $\hat{\phi}_{M}$ & $\hat{\phi}_{H}$ & $\hat{\phi}_{S}%
$\\\hline
& 0.8423 & --0.4808 & --0.3259 & 0.1195\\
&  & (0.0971) & (0.0780) & (0.0846)\\\hline\hline
\end{tabular}
}
\end{center}

{\small
\end{table}%
}

\subsection{{\protect\small Rolling estimates used in portfolio
construction\label{Roll}}}

{\small The rolling estimates for month $t$ are computed using analogous
expressions to those provided in sub-sections \ref{ESTphi} and \ref{Vphihat}
of the main paper with a rolling window size of $T=240$. For ease of
replication, the algorithms used to estimate the rolling estimates for
$t=2015m12,.....,2022m11$ are set out below:
\[
\mathbf{\hat{\beta}}_{it\left\vert T\right.  }=\left[  \sum_{\tau=t-T+1}%
^{t}\left(  \mathbf{f}_{\tau}-\mathbf{\bar{f}}_{t\left\vert T\right.
}\right)  \left(  \mathbf{f}_{\tau}-\mathbf{\bar{f}}_{t\left\vert T\right.
}\right)  ^{\prime}\right]  ^{-1}\sum_{\tau=t-T+1}^{t}\left(  \mathbf{f}%
_{\tau}-\mathbf{\bar{f}}_{t\left\vert T\right.  }\right)  r_{i\tau},
\]%
\[
\mathbf{\bar{f}}_{t\left\vert T\right.  }=T^{-1}\sum_{\tau=t-T+1}%
^{t}\mathbf{f}_{\tau},
\]%
\[
\boldsymbol{\tilde{\phi}}_{t\left\vert T\right.  }=\mathbf{\hat{H}%
}_{t\left\vert T\right.  \ }^{-1}\left[  \frac{\mathbf{\hat{B}}_{t\left\vert
T\right.  }^{\prime}\mathbf{M}_{n}\mathbf{\hat{\alpha}}_{t\left\vert T\right.
}}{n}+T^{-1}\widehat{\bar{\sigma}}_{t\left\vert T\right.  }^{2}\left(
\frac{\mathbf{F}_{t\left\vert T\right.  }^{\prime}\mathbf{M}_{T}%
\mathbf{F}_{t\left\vert T\right.  }}{T}\right)  ^{-1}\mathbf{\bar{f}%
}_{t\left\vert T\right.  }\right]  ,
\]%
\[
\mathbf{\hat{B}}_{t\left\vert T\right.  }=\left(  \mathbf{\hat{\beta}%
}_{1t\left\vert T\right.  },\mathbf{\hat{\beta}}_{2t\left\vert T\right.
}....,\mathbf{\hat{\beta}}_{nt\left\vert T\right.  }\right)  ^{\prime},
\]%
\[
\mathbf{\hat{\alpha}}_{t\left\vert T\right.  }=\mathbf{\bar{r}}_{t\left\vert
T\right.  }-\mathbf{\hat{B}}_{t\left\vert T\right.  }\mathbf{\bar{f}%
}_{t\left\vert T\right.  },
\]%
\[
\mathbf{\hat{H}}_{t\left\vert T\right.  \ }=\frac{\mathbf{\hat{B}%
}_{t\left\vert T\right.  }^{\prime}\mathbf{M}_{n}\mathbf{\hat{B}}_{t\left\vert
T\right.  }}{n}-T^{-1}\widehat{\bar{\sigma}}_{t\left\vert T\right.  }%
^{2}\left(  \frac{\mathbf{F}_{t\left\vert T\right.  }^{\prime}\mathbf{M}%
_{T}\mathbf{F}_{t\left\vert T\right.  }}{T}\right)  ^{-1},
\]%
\[
\widehat{\bar{\sigma}}_{t\left\vert T\right.  }^{2}=\frac{\sum_{t=\tau
-T+1}^{t}\sum_{i=1}^{n}\hat{u}_{i,\tau\left\vert T\right.  }^{2}}{n(T-K-1)},
\]
}

{\small
\[
\hat{u}_{i,\tau\left\vert T\right.  }=r_{i\tau}-\hat{\alpha}_{i,\tau\left\vert
T\right.  }-\mathbf{\hat{\beta}}_{i\tau\left\vert T\right.  }^{\prime
}\mathbf{f}_{\tau}\text{ },\text{ for }\tau=t,t-1,...,t-T+1,
\]%
\[
\mathbf{\bar{r}}_{t\left\vert T\right.  }=\left(  \bar{r}_{1,t\left\vert
T\right.  },\bar{r}_{2,t\left\vert T\right.  },...,\bar{r}_{nt\left\vert
T\right.  }\right)  ^{\prime}\text{, \ }\bar{r}_{i,t\left\vert T\right.
}=T^{-1}\sum_{\tau=t-T+1}^{t}r_{i\tau},
\]%
\[
\widehat{Var\left(  \boldsymbol{\tilde{\phi}}_{t\left\vert T\right.  }\right)
}=T^{-1}n^{-1}\mathbf{H}_{t\left\vert T\right.  \ }^{-1}\mathbf{\hat{V}}%
_{\xi,t\left\vert T\right.  }\mathbf{H}_{t\left\vert T\right.  \ }^{-1},
\]
,%
\[
\mathbf{\hat{V}}_{\xi,t\left\vert T\right.  }=\left(  1+\hat{s}_{t\left\vert
T\right.  }\right)  \left(  n^{-1}\mathbf{\hat{B}}_{t\left\vert T\right.
}^{\prime}\mathbf{M}_{n}\mathbf{\tilde{V}}_{u,t\left\vert T\right.
}\mathbf{M}_{n}\mathbf{\hat{B}}_{t\left\vert T\right.  }\right)  ,
\]%
\[
\hat{s}_{t\left\vert T\right.  }=\boldsymbol{\tilde{\lambda}}_{t\left\vert
T\right.  }^{^{\prime}}\left(  \frac{\mathbf{F}_{t\left\vert T\right.
}^{\prime}\mathbf{M}_{T}\mathbf{F}_{t\left\vert T\right.  }}{T}\right)
^{-1}\boldsymbol{\tilde{\lambda}}_{t\left\vert T\right.  },
\]%
\[
\boldsymbol{\tilde{\lambda}}_{t\left\vert T\right.  }=\boldsymbol{\tilde{\phi
}}_{t\left\vert T\right.  }+\mathbf{\bar{f}}_{t\left\vert T\right.  },
\]%
\[
c_{t\left\vert T\right.  }=\left(  \boldsymbol{\tau}_{n}^{\prime
}\boldsymbol{\tau}_{n}\right)  ^{-1}\boldsymbol{\tau}_{n}^{\prime}%
\mathbf{\hat{\alpha}}_{t\left\vert T\right.  }-\left(  \boldsymbol{\tau}%
_{n}^{\prime}\boldsymbol{\tau}_{n}\right)  ^{-1}\boldsymbol{\tau}_{n}^{\prime
}\mathbf{\hat{B}}_{t\left\vert T\right.  }\boldsymbol{\tilde{\phi}%
}_{t\left\vert T\right.  },
\]
where $\mathbf{M}_{T}=\mathbf{I}_{T}-T^{-1}\boldsymbol{\tau}_{T}%
\boldsymbol{\tau}_{T}^{\prime}$, $\boldsymbol{\tau}_{T}$ is a $T$-dimensional
vector of ones, $\mathbf{M}_{n}=\mathbf{I}_{n}-n^{-1}\boldsymbol{\tau}%
_{n}\boldsymbol{\tau}_{n}^{\prime}$, and $\boldsymbol{\tau}_{n}$ is an
n-dimensional vector of ones. Finally,}

{\small
\[
\mathbf{\tilde{V}}_{u,t\left\vert T\right.  }=\left(  \tilde{\sigma}%
_{ij,t|T}\right)  ,
\]
}

{\small
\begin{align*}
\tilde{\sigma}_{ij,t\left\vert T\right.  }  &  =\hat{\sigma}_{ii,t\left\vert
T\right.  }\\
\tilde{\sigma}_{ij,t\left\vert T\right.  }  &  =\hat{\sigma}_{ij,t\left\vert
T\right.  }\mathbf{1}\left[  \left\vert \hat{\rho}_{ij,t\left\vert T\right.
}\right\vert >T^{-1/2}c_{\alpha}(n,\delta)\right]  ,\text{ }i=1,2,\ldots
,n-1,\text{ }j=i+1,\ldots,n,
\end{align*}
where
\[
\hat{\sigma}_{ij,t\left\vert T\right.  }=\frac{1}{T}\sum_{t=1}^{T}\hat
{u}_{i,\tau\left\vert T\right.  }\hat{u}_{j,\tau\left\vert T\right.  },\text{
}\hat{\rho}_{ij,t\left\vert T\right.  }=\frac{\hat{\sigma}_{ij,t\left\vert
T\right.  }}{\sqrt{\hat{\sigma}_{ii,t\left\vert T\right.  }\hat{\sigma
}_{ii,t\left\vert T\right.  }}},\text{ }%
\]
$\hat{u}_{i,\tau\left\vert T\right.  }=r_{i\tau}-\hat{\alpha}_{i,\tau
\left\vert T\right.  }-\mathbf{\hat{\beta}}_{i\tau\left\vert T\right.
}^{\prime}\mathbf{f}_{\tau}$, and $c_{p}(n,d)=\Phi^{-1}\left(  1-\frac
{p}{2n^{d}}\right)  ,$ is a normal critical value function, $p$ is the the
nominal size of testing of $\sigma_{ij}=0$, ($i\neq j$) and $d=2$ is chosen to
take account of the $n(n-1)/2$ multiple tests being carried out. }

\section{{\protect\small Grouping of securities by their pair-wise correlations\label{BlockCov}}}

{\small The $T=240$ sample ending in $2021$ was used to estimate the pair-wise
correlations of the residuals from the $n=1,168$ returns regressions using the
Fama-French three factors. Then all the statistically insignificant
correlations were set to zero. Significance was determined allowing for the
multiple testing nature of the tests, using the critical value $c_{p}%
(n,\delta)=\Phi^{-1}\left(  1-\frac{p}{2n^{\delta}}\right)  $, with $p=0.05$
and $\delta=2$, since the number of pair-wise correlations is of order
$O(n^{2})$. See also \citet{BPS2019multiple}. For the majority of securities
(668 out of the 1,168), the pair-wise return correlations were not
statistically significant. The securities with a relatively large number of
non-zero correlations were either in the banking or energy related industries.
}

{\small Initially, the securities were grouped using the two digit codes from
the $1987$ standard industrial classification (SIC 1987). But this gave too
many groups, $62$. Many of the groups had very few members: only one security
for $3$ out of the $62$ groups and less than $10$ for $36$ groups. However,
code $60$ (banking) had $145$ securities. Therefore, it was decided to work
with the industrial classification based on a one digit level, and to
aggregate the codes with a small number of securities, taking out two digit
codes where there were large numbers in that code. We ended up with $14$
contiguous groups ranging in size from $33$ to $145$. Average correlations
were low overall, but the average absolute correlation of securities within
the groups was around 10 times that with firms outside the group. See Table
\ref{TabSec}, which gives averages of pair-wise correlations without
thresholding. These estimates suggested that a block diagonal structure with
$14$ blocks was a reasonable characterization which is used in the Monte Carlo
analysis. See Section \ref{Simulations} of the main paper.}

{\small
\begin{table}[H]%
\caption{Sector groupings by SIC codes and within and outside sector average
pair-wise correlations.} \label{TabSec} }

\begin{center}
{\small
\begin{tabular}
[c]{lrcrcrrrrr}\hline\hline
&  & SIC &  & Number &  & \multicolumn{3}{c}{Average correlations} &
\\\cline{7-10}
&  & codes &  & of stocks &  & within &  & outside & \\\hline
Agriculture \& mining &  & 0-17 &  & 51 &  & 0.0848 &  & --0.0043 & \\
Food processing etc &  & 20-27 &  & 82 &  & 0.0382 &  & 0.0074 & \\
Chemicals \& refining &  & 28-29 &  & 78 &  & 0.0234 &  & 0.0027 & \\
Metals &  & 30-34 &  & 65 &  & 0.0439 &  & 0.0070 & \\
Machinery \& equipment &  & 35 &  & 72 &  & 0.0380 &  & 0.0037 & \\
Electrical Equipment &  & 36 &  & 77 &  & 0.0601 &  & 0.0001 & \\
Transport equipment &  & 37 &  & 33 &  & 0.0800 &  & 0.0105 & \\
Misc. manufacturing &  & 38-39 &  & 78 &  & 0.0155 &  & 0.0027 & \\
Transport etc. &  & 40-49 &  & 108 &  & 0.0810 &  & 0.0035 & \\
Wholesale \& retail trade &  & 50-59 &  & 122 &  & 0.0357 &  & 0.0034 & \\
Banking &  & 60 &  & 145 &  & 0.1240 &  & --0.0037 & \\
Other finance &  & 61-67 &  & 98 &  & 0.0273 &  & 0.0061 & \\
Commercial Services &  & 70-79 &  & 114 &  & 0.0149 &  & 0.0031 & \\
Professional Services &  & 80-89 &  & 45 &  & 0.0172 &  & 0.0023 & \\
Total &  &  &  & 1168 &  &  &  &  & \\\hline\hline
\end{tabular}
}
\end{center}

{\small {\footnotesize \textit{Note}: This table gives the average
correlations within and between 14 groups selected based on one and two digits
SIC codes. It shows the number of stocks in each sector, the average pair-wise
correlation of stock returns within the sector as well as the average
pair-wise correlations of returns of stocks in a given sector with those
outside the sector.}%
\end{table}%
}

\section{{\protect\small Data used in the empirical application\label{DEMP}}}

\subsection{{\protect\small Security excess returns}}

{\small Monthly returns (inclusive of dividends) for NYSE and NASDAQ stocks
from CRSP with codes 10 and 11 were downloaded on July 2 2022 from Wharton
Research Data Services. They were converted to excess returns by subtracting
the risk free rate, which was taken from Kenneth French's data base. To obtain
balanced panels of stock returns and factors only variables where there was
data for the full sample under consideration were used. Excess returns are
measured in percent per month. To avoid outliers influencing the results,
stocks with a kurtosis greater than 16 were excluded. }

{\small Four main samples were considered, each had 20 years of data, $T=240$,
ending in 2015m12, 2017m12, 2019m12, 2021m12. Thus the earliest observation
used is for 1996m1. Sub-samples of the main samples of size $T=120$ and $T=60$
ending at the same dates were also examined. Table \ref{tab:ret_sumtat} gives
averages of the summary statistics across the individual stocks for the
various samples. These are very similar over these four periods. Mean returns
were high and substantially greater than the median reflecting the skewness of
returns, which was slightly less in the last period. Filtering out the stocks
with very high kurtosis removed about 100 of the roughly 1,200 stocks in each
period and reduced mean return, standard deviation, skewness as well as
kurtosis. }

{\small The 5 ($T=60$) and 10 ($T=120$) year sub-samples of the main samples,
ending at the same dates, showed very similar patterns. Because of a
requirement for a balanced panel, the shorter the sample the more stocks will
be eligible for inclusion. Compared with around 1,200 in the 20 year sample
there were around 2,000 in the 10 year sub-sample and around 2,500 in the 5
year sub-sample. Again filtering by kurtosis reduced the number of stocks by
about 100. There is more variation in means and medians in the shorter
sub-samples and the shorter the sample the lower the average kurtosis is.}

{\small
\begin{table}[H]%
\caption{Summary statistics for monthly returns in percent for NYSE and NASDAQ
stocks code 10 and 11 for 20 year ($T=240$), 10 year ($T=120$) and 5 year
($T=60$) samples ending at end of specified year}\label{tab:ret_sumtat} }

\begin{center}
{\small
\begin{tabular}
[c]{lrrrrrrrrrr}\hline\hline
& \multicolumn{4}{c}{All stocks} &  & \multicolumn{4}{c}{Stocks with kurtosis
$<16$} & \\\cline{1-5}\cline{7-11}%
End date & 2015 & 2017 & 2019 & 2021 &  & 2015 & 2017 & 2019 & 2021 &
\\\cline{1-5}\cline{7-11}%
\multicolumn{10}{l}{} & \\
\multicolumn{5}{l}{\textit{Panel A}: 20 year period, $T=240$} &  &  &  &  &  &
\\
Mean & 1.38 & 1.35 & 1.30 & 1.37 &  & 1.33 & 1.28 & 1.25 & 1.33 & \\
Median & 0.74 & 0.70 & 0.76 & 0.84 &  & 0.80 & 0.78 & 0.85 & 0.94 & \\
S.D. & 12.44 & 12.61 & 12.22 & 11.90 &  & 11.60 & 11.64 & 11.31 & 10.99 & \\
Skewness & 0.69 & 0.74 & 0.72 & 0.64 &  & 0.42 & 0.44 & 0.40 & 0.32 & \\
Kurtosis & 8.68 & 9.08 & 9.22 & 9.24 &  & 6.10 & 6.29 & 6.29 & 6.19 & \\
n & 1181 & 1243 & 1276 & 1289 &  & 1090 & 1132 & 1143 & 1175 & \\
\multicolumn{10}{l}{} & \\
\multicolumn{5}{l}{\textit{Panel B}: 10 year period, $T=120$} &  &  &  &  &  &
\\
Mean & 1.00 & 1.22 & 1.24 & 1.55 &  & 0.98 & 1.17 & 1.23 & 1.49 & \\
Median & 0.52 & 0.70 & 0.78 & 0.89 &  & 0.59 & 0.77 & 0.84 & 1.01 & \\
S.D. & 12.36 & 12.68 & 10.70 & 11.83 &  & 11.67 & 11.80 & 10.27 & 10.70 & \\
Skewness & 0.48 & 0.50 & 0.44 & 0.50 &  & 0.29 & 0.30 & 0.35 & 0.28 & \\
Kurtosis & 6.93 & 7.08 & 5.32 & 7.02 &  & 5.57 & 5.62 & 4.64 & 5.33 & \\
n & 2045 & 2024 & 1925 & 1871 &  & 1929 & 1907 & 1873 & 1766 & \\
\multicolumn{10}{l}{} & \\
\multicolumn{5}{l}{\textit{Panel C}: 5 year period, $T=60$} &  &  &  &  &  &
\\
Mean & 0.98 & 1.39 & 0.82 & 1.60 &  & 0.96 & 1.35 & 0.80 & 1.47 & \\
Median & 0.42 & 0.77 & 0.20 & 0.43 &  & 0.49 & 0.82 & 0.29 & 0.61 & \\
S.D. & 10.89 & 10.62 & 11.88 & 14.91 &  & 10.41 & 10.15 & 11.28 & 13.28 & \\
Skewness & 0.44 & 0.46 & 0.38 & 0.48 &  & 0.35 & 0.38 & 0.29 & 0.28 & \\
Kurtosis & 4.84 & 4.74 & 4.68 & 6.21 &  & 4.34 & 4.29 & 4.20 & 5.04 & \\
n & 2600 & 2425 & 2497 & 2512 &  & 2541 & 2373 & 2439 & 2388 & \\\hline\hline
\end{tabular}
}
\end{center}

{\small \vspace{1mm} {\footnotesize \textit{Note:} This table shows the
average values of the summary statistics of individual stocks described in
Section 3, and the number of stocks (n), for each sample. }
\end{table}%
}

\subsection{{\protect\small Risk factors}}

{\small For the empirical applications we combined the 5 Fama-French factors
with the 207 factors of \citet{chen2022open}, both downloaded on 6 July 2022.
The available risk factors at the end of each of the four 240 months samples
ending in 2015, 2017, 2019 and 2021 were then screened and factors whose
correlations (in absolute value) with the market factor were larger than 0.70
were dropped. The basic idea was to remove factors that were closely
correlated with market factor. However, the application of this filter only
reduced the number of factors in the active set by around 9-11. See Table 4 of
the main paper. The summary statistics for the factors in the active set are
summarized in Table \ref{SumFac}. It reports mean, median, pair-wise
correlation, standard deviation (S.D.), skewness and kurtosis of the
statistics indicated in the sub-headings of the tables for the $K$ factors
included in the active set for samples of size $T=240$ months ending in
December of 2015, 2017, 2019 and 2021. The summary statistics reported for the
"Mean" on the left panel of Table \ref{SumFac} are based on the time series
means of the individual factors, those under "Median" are based on the time
series medians of the individual factors, and so on.}

{\small
\begin{table}[H]
\caption{Summary statistics for mean, median, standard deviation, pairwise correlations, skewness and kurtosis of factors for $T=240$
samples at the end of specified years (2015, 2017, 2019 and 2021)} }

\begin{center}
{\small
\begin{tabular}
[c]{lrrrrrrrrrrrrrrr}\hline\hline
End date & 2015 &  & 2017 &  & 2019 &  & 2021 &  & 2015 &  & 2017 &  & 2019 &
& 2021\\\cline{1-8}\cline{10-16}%
$m$(\# Factors) & 190 &  & 191 &  & 190 &  & 178 &  & 190 &  & 191 &  & 190 &
& 178\\\cline{2-8}\cline{10-16}
& \multicolumn{7}{c}{Mean} &  & \multicolumn{7}{c}{Median}\\
Mean & 0.51 &  & 0.47 &  & 0.41 &  & 0.32 &  & 0.42 &  & 0.37 &  & 0.31 &  &
0.28\\
Median & 0.45 &  & 0.39 &  & 0.33 &  & 0.26 &  & 0.31 &  & 0.28 &  & 0.22 &  &
0.20\\
S.D. & 0.40 &  & 0.40 &  & 0.38 &  & 0.30 &  & 0.53 &  & 0.50 &  & 0.50 &  &
0.44\\
Skewness & 1.88 &  & 2.50 &  & 2.28 &  & 0.76 &  & 1.59 &  & 1.89 &  & 1.58 &
& 1.38\\
Kurtosis & 10.54 &  & 16.35 &  & 14.42 &  & 3.82 &  & 8.99 &  & 10.37 &  &
9.16 &  & 6.68\\
&  &  &  &  &  &  &  &  &  &  &  &  &  &  & \\\cline{2-8}\cline{10-16}
& \multicolumn{7}{c}{Standard deviation} &  & \multicolumn{7}{c}{Pair-wise
correlation}\\
Mean & 3.97 &  & 3.90 &  & 3.78 &  & 3.33 &  & 0.22 &  & 0.22 &  & 0.22 &  &
0.21\\
Median & 3.51 &  & 3.47 &  & 3.37 &  & 2.99 &  & 0.18 &  & 0.17 &  & 0.17 &  &
0.16\\
S.D. & 2.21 &  & 2.16 &  & 2.12 &  & 1.68 &  & 0.17 &  & 0.17 &  & 0.17 &  &
0.17\\
Skewness & 1.15 &  & 1.23 &  & 1.29 &  & 1.05 &  & 0.86 &  & 0.88 &  & 0.88 &
& 0.92\\
Kurtosis & 4.25 &  & 4.69 &  & 4.96 &  & 4.17 &  & 2.88 &  & 2.93 &  & 2.93 &
& 2.98\\
&  &  &  &  &  &  &  &  &  &  &  &  &  &  & \\\cline{2-8}\cline{10-16}
& \multicolumn{7}{c}{Skewness} &  & \multicolumn{7}{c}{Kurtosis}\\
Mean & 0.14 &  & 0.21 &  & 0.23 &  & 0.04 &  & 8.51 &  & 8.79 &  & 9.42 &  &
7.03\\
Median & 0.13 &  & 0.18 &  & 0.14 &  & 0.06 &  & 6.44 &  & 6.66 &  & 7.09 &  &
5.86\\
S.D. & 1.15 &  & 1.21 &  & 1.27 &  & 1.01 &  & 5.76 &  & 6.55 &  & 7.12 &  &
5.28\\
Skewness & 0.15 &  & 0.55 &  & 0.61 &  & -1.03 &  & 2.34 &  & 3.18 &  & 2.96 &
& 5.01\\
Kurtosis & 3.98 &  & 5.28 &  & 5.24 &  & 7.82 &  & 9.65 &  & 17.71 &  &
15.70 &  & 37.21\\\hline\hline
\end{tabular}
}
\end{center}

{\small \bigskip{\footnotesize \textit{Note:} The T=240 sample was used to
select factors and only factors where the absolute correlation coefficient
with the market factor is less than 0.70 are included.\label{SumFac}}%
\end{table}%
}

\section{{\protect\small \bigskip Pooled R squared and factor strengths\label{appPRsquared}}}

\begin{lemma}
{\small \label{PR2}Consider the factor model
\begin{equation}
r_{it}=\text{a}_{i}+\sum_{k=1}^{K}\beta_{ik}f_{kt}+u_{it}=\text{a}%
_{i}+\boldsymbol{\beta}_{i}^{\prime}\mathbf{f}_{t}+u_{it},\text{ for
}i=1,2,...,n;\text{ }t=1,2,...,T, \label{fac1}%
\end{equation}
and consider the following adjusted pooled measure of fit
\begin{equation}
\overline{PR}^{2}=1-\frac{\widehat{\bar{\sigma}}_{nT}^{2}}{\left(  nT\right)
^{-1}\sum_{i=1}^{n}\sum_{t=1}^{T}\left(  r_{it}-\bar{r}_{i\circ}\right)  ^{2}%
}, \label{AdjPR2A}%
\end{equation}
where $\widehat{\bar{\sigma}}_{nT}^{2}$ is the bias-corrected estimator of
$\left(  nT\right)  ^{-1}\sum_{i=1}^{n}\sum_{t=1}^{T}E\left(  u_{it}%
^{2}\right)  =n\sum_{i=1}^{n}\sigma_{i}^{2}=\bar{\sigma}_{n}^{2}$, given by
(\ref{AdjZig}) $\bar{r}_{i\circ}=T^{-1}\sum_{t=1}^{T}r_{it}$, and $\bar
{u}_{i\circ}=T^{-1}\sum_{t=1}^{T}u_{it}$. Then under Under Assumptions
\ref{Errors}, \ref{factors} and \ref{loadings} we have%
\begin{equation}
\overline{PR}_{nT}^{2}=\sum_{k=1}^{K}\ominus\left(  n^{\alpha_{k}-1}\right)
+O_{p}\left(  T^{-1/2}n^{-1+\frac{\alpha_{max}+\alpha_{\gamma}}{2}}\right)  ,
\label{PR2res}%
\end{equation}
where $\alpha_{k}$ is the strength of factor $f_{tk}$, $\alpha_{max}%
=max_{k}(\alpha_{k})$, and $\alpha_{\gamma}$ is the strength of the missing
factor. }
\end{lemma}

\begin{proof}
{\small Using (\ref{orderzig}) we first recall that
\begin{equation}
\widehat{\bar{\sigma}}_{nT}^{2}-\bar{\sigma}_{n}^{2}=O_{p}\left(
T^{-1/2}n^{-1/2}\right)  . \label{orderzigA}%
\end{equation}
Now averaging (\ref{fac1}) over $t$ and forming deviations of $r_{it}$ from
its time average, $\bar{r}_{i\circ}$, we have (note that $\boldsymbol{\hat
{\mu}}_{T}=T^{-1}\sum_{t=1}^{T}\mathbf{f}_{t}$)%
\[
r_{it}-\bar{r}_{i\circ}=u_{it}-\bar{u}_{i\circ}+\boldsymbol{\beta}_{i}%
^{\prime}\left(  \mathbf{f}_{t}-\boldsymbol{\hat{\mu}}_{T}\right)  .
\]
Using this result we have
\begin{align}
\left(  nT\right)  ^{-1}\sum_{i=1}^{n}\sum_{t=1}^{T}\left(  r_{it}-\bar
{r}_{i\circ}\right)  ^{2}  &  =\left(  nT\right)  ^{-1}\sum_{i=1}^{n}%
\sum_{t=1}^{T}\left(  u_{it}-\bar{u}_{i\circ}\right)  ^{2}+n^{-1}\sum
_{i=1}^{n}\boldsymbol{\beta}_{i}^{\prime}\mathbf{\hat{\Sigma}}_{f}%
\boldsymbol{\beta}_{i}\label{SSrnT}\\
&  -2\left(  nT\right)  ^{-1}\sum_{i=1}^{n}\sum_{t=1}^{T}\left(  u_{it}%
-\bar{u}_{i\circ}\right)  \boldsymbol{\beta}_{i}^{\prime}\left(
\mathbf{f}_{t}-\boldsymbol{\hat{\mu}}_{T}\right)  ,\nonumber
\end{align}
where $\mathbf{\hat{\Sigma}}_{f}=T^{-1}\sum_{t=1}^{T}\left(  \mathbf{f}%
_{t}-\boldsymbol{\hat{\mu}}_{T}\right)  \left(  \mathbf{f}_{t}%
-\boldsymbol{\hat{\mu}}_{T}\right)  ^{\prime}$. For the first term we have
\begin{equation}
\left(  nT\right)  ^{-1}\sum_{i=1}^{n}\sum_{t=1}^{T}\left(  u_{it}-\bar
{u}_{i\circ}\right)  ^{2}-\bar{\sigma}_{n}^{2}=O_{p}(n^{-1/2}T^{-1/2}),
\label{SSunT}%
\end{equation}
which follows from the proof of Theorem \ref{Thzig} by setting
$\boldsymbol{\beta}_{i}=\mathbf{0}$ and $\mathbf{f}_{t}=\mathbf{0}$ in Section
\ref{ProofThzig}. For the cross product term we have
\begin{align*}
&  \left(  nT\right)  ^{-1}\sum_{i=1}^{n}\sum_{t=1}^{T}\left(  u_{it}-\bar
{u}_{i\circ}\right)  \boldsymbol{\beta}_{i}^{\prime}\left(  \mathbf{f}%
_{t}-\boldsymbol{\hat{\mu}}_{T}\right) \\
&  =\left(  nT\right)  ^{-1}\sum_{i=1}^{n}\boldsymbol{\beta}_{i}^{\prime}%
\sum_{t=1}^{T}\left(  \mathbf{f}_{t}-\boldsymbol{\hat{\mu}}_{T}\right)
\left(  u_{it}-\bar{u}_{i\circ}\right)  =\left(  nT\right)  ^{-1}\sum
_{i=1}^{n}\boldsymbol{\beta}_{i}^{\prime}\sum_{t=1}^{T}\left(  \mathbf{f}%
_{t}-\boldsymbol{\hat{\mu}}_{T}\right)  u_{it}\\
&  -\left(  nT\right)  ^{-1}\sum_{i=1}^{n}\bar{u}_{i\circ}\boldsymbol{\beta
}_{i}^{\prime}\sum_{t=1}^{T}\left(  \mathbf{f}_{t}-\boldsymbol{\hat{\mu}}%
_{T}\right)  =\left(  nT\right)  ^{-1}\sum_{i=1}^{n}\boldsymbol{\beta}%
_{i}^{\prime}\sum_{t=1}^{T}\left(  \mathbf{f}_{t}-\boldsymbol{\hat{\mu}}%
_{T}\right)  u_{it}=p_{nT}.
\end{align*}
Also, using (\ref{uit}),
\[
p_{nT}=\left(  nT\right)  ^{-1}\sum_{i=1}^{n}\boldsymbol{\beta}_{i}^{\prime
}\sum_{t=1}^{T}\left(  \mathbf{f}_{t}-\boldsymbol{\hat{\mu}}_{T}\right)
\left(  \gamma_{i}g_{t}+v_{it}\right)  =p_{1,nT}+p_{2,nT},
\]
where%
\[
p_{1,nT}=\left(  n^{-1}\sum_{i=1}^{n}\gamma_{i}\boldsymbol{\beta}_{i}^{\prime
}\right)  \left[  T^{-1}\sum_{t=1}^{T}\left(  \mathbf{f}_{t}-\boldsymbol{\hat
{\mu}}_{T}\right)  g_{t}\right]  ,
\]
and%
\[
p_{2,nT}=\left(  nT\right)  ^{-1}\sum_{i=1}^{n}\boldsymbol{\beta}_{i}^{\prime
}\sum_{t=1}^{T}\left(  \mathbf{f}_{t}-\boldsymbol{\hat{\mu}}_{T}\right)
v_{it}.
\]
Under Assumption \ref{Latent factor}, $\left(  \mathbf{f}_{t}-\boldsymbol{\hat
{\mu}}_{T}\right)  $ and $g_{t}$ are distributed independently and $g_{t}$ are
serially independent with $E\left(  g_{t}\right)  =0$ and $E(g_{t}^{2})=1$,
and it follows that
\[
Var\left[  T^{-1}\sum_{t=1}^{T}\left(  \mathbf{f}_{t}-\boldsymbol{\hat{\mu}%
}_{T}\right)  g_{t}\right]  =E\left[  T^{-2}\sum_{t=1}^{T}\left(
\mathbf{f}_{t}-\boldsymbol{\hat{\mu}}_{T}\right)  ^{2}E\left(  g_{t}%
^{2}\right)  \right]  =T^{-1}E\left(  T^{-1}\mathbf{F}^{\prime}\mathbf{M}%
_{T}\mathbf{F}\right)  =O\left(  T^{-1}\right)  .
\]
Hence, $T^{-1}\sum_{t=1}^{T}\left(  \mathbf{f}_{t}-\boldsymbol{\hat{\mu}}%
_{T}\right)  g_{t}=O_{p}\left(  T^{-1/2}\right)  $. Also (see (\ref{As}) and
(\ref{Snormg}))
\begin{align*}
\left\Vert n^{-1}\sum_{i=1}^{n}\gamma_{i}\boldsymbol{\beta}_{i}^{\prime
}\right\Vert  &  \leq n^{-1}\sum_{i=1}^{n}\left\vert \gamma_{i}\right\vert
\left\Vert \boldsymbol{\beta}_{i}\right\Vert =n^{-1}\sum_{i=1}^{n}\left\vert
\gamma_{i}\right\vert \left(  \boldsymbol{\beta}_{i}^{\prime}\boldsymbol{\beta
}_{i}\right)  ^{1/2}\\
&  \leq\left(  n^{-1}\sum_{i=1}^{n}\left\vert \gamma_{i}\right\vert
^{2}\right)  ^{1/2}\left(  n^{-1}\sum_{i=1}^{n}\boldsymbol{\beta}_{i}^{\prime
}\boldsymbol{\beta}_{i}\right)  ^{1/2}=O_{p}\left(  n^{\frac{-1+\alpha
_{\gamma}}{2}}\right)  O_{p}\left(  n^{\frac{-1+\alpha_{max}}{2}}\right) \\
&  =O_{p}\left(  n^{-1+\frac{\alpha_{max}+\alpha_{\gamma}}{2}}\right)  .
\end{align*}
Hence, $p_{1,nT}=O_{p}\left(  T^{-1/2}n^{-1+\frac{\alpha_{max}+\alpha_{\gamma
}}{2}}\right)  $. Consider $p_{2,nT}$ and recall that under Assumption
\ref{Errors} $v_{it}$ are serially independent, have zero means and are
distributed independently of $\left(  \mathbf{f}_{t}-\boldsymbol{\hat{\mu}%
}_{T}\right)  $ and $\boldsymbol{\beta}_{i}$\thinspace. Then $E\left(
p_{2,nT\ }\right)  =0$ and
\[
Var\left(  p_{2,nT\ }\left\vert \mathbf{F}\right.  \right)  =\frac{1}%
{nT}\left(  \frac{1}{n}\sum_{i=1}^{n}\sum_{j=1}^{n}\sigma_{ij,v}%
\boldsymbol{\beta}_{i}^{\prime}\mathbf{\hat{\Sigma}}_{f}\boldsymbol{\beta}%
_{i}\right)  ,
\]
where $E(v_{it}v_{jt})=\sigma_{ij,v}$. Also,
\[
\left\Vert \frac{1}{n}\sum_{i=1}^{n}\sum_{j=1}^{n}\sigma_{ij,v}%
\boldsymbol{\beta}_{i}^{\prime}\mathbf{\hat{\Sigma}}_{f}\boldsymbol{\beta}%
_{i}\right\Vert \leq\left(  sup_{i}\left\Vert \boldsymbol{\beta}%
_{i}\right\Vert \right)  ^{2}\left\Vert \mathbf{\hat{\Sigma}}_{f}\right\Vert
\left(  \frac{1}{n}\sum_{i=1}^{n}\sum_{j=1}^{n}\left\vert \sigma
_{ij,v}\right\vert \right)  ,
\]
and by assumption $sup_{i}\left\Vert \boldsymbol{\beta}_{i}\right\Vert <C$,
$E\left\Vert \mathbf{\hat{\Sigma}}_{f}\right\Vert <C$, and $n^{-1}\sum
_{i=1}^{n}\sum_{j=1}^{n}\left\vert \sigma_{ij,v}\right\vert =O(1)$. Hence,
$Var\left(  p_{2,nT\ }\right)  =O_{p}(n^{-1}T^{-1})$, and it follows that
$p_{2,nT}=O_{p}(T^{-1/2}n^{-1/2})$, and overall (since $\alpha_{\gamma}<1/2$
and $\alpha_{max}\leq1$)
\begin{equation}
p_{nT}=O_{p}(T^{-1/2}n^{-1/2})+O_{p}\left(  T^{-1/2}n^{-1+\frac{\alpha
_{max}+\alpha_{\gamma}}{2}}\right)  =O_{p}\left(  T^{-1/2}n^{-1+\frac
{\alpha_{max}+\alpha_{\gamma}}{2}}\right)  . \label{R2pnT}%
\end{equation}
Using (\ref{SSunT})\ and (\ref{R2pnT}) in (\ref{SSrnT}), we now have%
\[
\left(  nT\right)  ^{-1}\sum_{i=1}^{n}\sum_{t=1}^{T}\left(  r_{it}-\bar
{r}_{i\circ}\right)  ^{2}=\bar{\sigma}_{n}^{2}+n^{-1}\sum_{i=1}^{n}%
\boldsymbol{\beta}_{i}^{\prime}\mathbf{\hat{\Sigma}}_{f}\boldsymbol{\beta}%
_{i}+O_{p}\left(  T^{-1/2}n^{-1+\frac{\alpha_{max}+\alpha_{\gamma}}{2}%
}\right)  .
\]
Using this result and (\ref{orderzigA}) in (\ref{AdjPR2A}) yields%
\[
\overline{PR}_{nT}^{2}=1-\frac{\bar{\sigma}_{n}^{2}+O_{p}(T^{-1/2}n^{-1/2}%
)}{\bar{\sigma}_{n}^{2}+n^{-1}\sum_{i=1}^{n}\boldsymbol{\beta}_{i}^{\prime
}\mathbf{\hat{\Sigma}}_{f}\boldsymbol{\beta}_{i}+O_{p}(T^{-1/2}n^{-1/2}%
)+O_{p}\left(  T^{-1/2}n^{-1+\frac{\alpha_{max}+\alpha_{\gamma}}{2}}\right)
}.
\]
Since $O_{p}(T^{-1/2}n^{-1/2})$ is dominated by $\left(  T^{-1/2}%
n^{-1+\frac{\alpha_{max}+\alpha_{\gamma}}{2}}\right)  $, we end up with%
\begin{equation}
\overline{PR}_{nT}^{2}=\frac{n^{-1}\sum_{i=1}^{n}\boldsymbol{\beta}%
_{i}^{\prime}\mathbf{\hat{\Sigma}}_{f}\boldsymbol{\beta}_{i}/\bar{\sigma}%
_{n}^{2}+O_{p}\left(  T^{-1/2}n^{-1+\frac{\alpha_{max}+\alpha_{\gamma}}{2}%
}\right)  }{1+n^{-1}\sum_{i=1}^{n}\boldsymbol{\beta}_{i}^{\prime}%
\mathbf{\hat{\Sigma}}_{f}\boldsymbol{\beta}_{i}/\bar{\sigma}_{n}^{2}%
+O_{p}\left(  T^{-1/2}n^{-1+\frac{\alpha_{max}+\alpha_{\gamma}}{2}}\right)  }.
\label{PRnT2}%
\end{equation}
Hence, the order of $\overline{PR}_{nT}^{2}$ is governed by the pooled
signal-to-noise ratio defined by
\[
s_{nT}^{2}=\frac{n^{-1}\sum_{i=1}^{n}\boldsymbol{\beta}_{i}^{\prime
}\mathbf{\hat{\Sigma}}_{f}\boldsymbol{\beta}_{i}}{\bar{\sigma}_{n}^{2}}.
\]
However, under Assumption \ref{factors}
\begin{equation}
\lambda_{min}(\mathbf{\hat{\Sigma}}_{f})\frac{n^{-1}\sum_{i=1}^{n}%
\boldsymbol{\beta}_{i}^{\prime}\boldsymbol{\beta}_{i}}{\bar{\sigma}_{n}^{2}%
}\leq s_{nT}^{2}\leq\lambda_{max}(\mathbf{\hat{\Sigma}}_{f})\frac{n^{-1}%
\sum_{i=1}^{n}\boldsymbol{\beta}_{i}^{\prime}\boldsymbol{\beta}_{i}}%
{\bar{\sigma}_{n}^{2}}, \label{snT}%
\end{equation}
where $c<\lambda_{min}(\mathbf{\hat{\Sigma}}_{f})<\lambda_{max}(\mathbf{\hat
{\Sigma}}_{f})<C.$ Hence,
\[
c\left(  \frac{n^{-1}\sum_{i=1}^{n}\boldsymbol{\beta}_{i}^{\prime
}\boldsymbol{\beta}_{i}}{\bar{\sigma}_{n}^{2}}\right)  \leq s_{nT}^{2}\leq
C\left(  \frac{n^{-1}\sum_{i=1}^{n}\boldsymbol{\beta}_{i}^{\prime
}\boldsymbol{\beta}_{i}}{\bar{\sigma}_{n}^{2}}\right)  ,
\]
and it must be that
\[
s_{nT}^{2}=\ominus\left(  n^{-1}\sum_{i=1}^{n}\boldsymbol{\beta}_{i}^{\prime
}\boldsymbol{\beta}_{i}\right)  =\ominus\left[  \sum_{k=1}^{K}\left(
n^{-1}\sum_{i=1}^{n}\beta_{ik}^{2}\right)  \right]  .
\]
Also, under Assumption \ref{loadings}, $n^{-1}\sum_{i=1}^{n}\beta_{ik}%
^{2}=\ominus\left(  n^{\alpha_{k}-1}\right)  $. Hence,%
\[
s_{nT}^{2}=\sum_{k=1}^{K}\ominus\left(  n^{\alpha_{k}-1}\right)  ,
\]
which in view of (\ref{PRnT2}) now yields (\ref{PR2res}), as desired. }
\end{proof}

{\small
\singlespacing
\noindent
\textbf{References\bigskip}}

{\small
\noindent
Bailey, N., G. Kapetanios and M. H. Pesaran (2021). "Measurement of factor
strength: theory and practice", \textit{Journal of Applied Econometrics, }36,
587-613.}

{\small
\vspace{0.2cm}%
}

{\small
\noindent
Bailey, N., M. H. Pesaran, and L.V. Smith (2019). "A multiple testing approach
to the regularisation of sample correlation matrices", \textit{Journal of
Econometrics,} 208, 507-534.}

{\small
\vspace{0.2cm}%
}

{\small
\noindent
Chen, A. Y. and T. Zimmermann (2022). "Open source cross-sectional asset
pricing", \textit{Critical Finance Review}, 11, 207-264. }

{\small \pagebreak}

{\small \bigskip}%
\singlespacing
\thispagestyle{empty}%

\begin{center}
	{\small \textbf{Online Supplement B: Monte Carlo Results} }
	
	{\small \bigskip}
	
	{\small \bigskip To }
	
	{\small \bigskip}
	
	{\small \textbf{Identifying and exploiting alpha in linear asset pricing
			models with strong, semi-strong, and latent factors} }
	
	{\small \bigskip}
	
	{\small \bigskip}
	
	{\small by }
	
	{\small \bigskip}
	
	{\small \bigskip}
	
	{\small M. Hashem Pesaran }
	
	{\small University of Southern California, and Trinity College, Cambridge }
	
	{\small \bigskip}
	
	{\small Ron P. Smith }
	
	{\small Birkbeck, University of London }
	
	{\small \bigskip}
	
	{\small \bigskip}
	
	{\small October 2024 }
\end{center}

\pagebreak%
\setcounter{table}{0}%
\setcounter{section}{0}%
\setcounter{figure}{0}%
\setcounter{footnote}{0}%
%

\pagebreak%
\setcounter{page}{1}
\renewcommand{\thepage}{SB-\arabic{page}}%

\section{Introduction}

This online supplement provides detailed Monte Carlo results for all
experiments and risk factors. The summary tables below give the bias, RMSE and
size ($\times100$), for the DGP with one strong $(\alpha_{M}=1)$ and two
semi-strong factors $\left(  \alpha_{H}=0.85,\text{ }\alpha_{S}=0.65\right)  $
for the two-step and the bias-corrected estimators of $\boldsymbol{\phi}%
=(\phi_{M},\phi_{H},\phi_{S})^{\prime},$ for different sample sizes. These are
given for the twelve experimental designs listed in Table \ref{TabExperiments}
below. For experiments 8 and 9 we also report the results with larger values
for the parameter of the pricing errors ($\alpha_{\eta}=0.50$) and the spatial
coefficients ($\rho_{\varepsilon}=0.85$), denoted as Experiments 8a and 9a,
respectively. Following each table the empirical power functions for the bias
corrected estimator of $\boldsymbol{\phi}$ are displayed for different sample
sizes. The threshold estimator of the covariance matrix described in Section
3.2 of the main paper is used in computing the standard errors of the tests.

\section{List of Monte Carlo Experiments}

The full list of Monte Carlo experimenst is provided in Table
\ref{TabExperiments}. Six designs, the odd numbered ones, have errors in the
return equations that are Gaussian, six, the even numbered ones, have errors
that are $t$ distributed with 5 degrees of freedom. Designs 3 and 4 add GARCH
effects in the factor errors to designs 1 and 2, respectively. Designs 5 and
6, add the pricing error, $\eta_{i}$, to designs 3 and 4, and designs 7 and 8
further add the missing factor, $g_{t}$, to the error of the return equations.
Designs 1-8 have a diagonal covariance matrix for the idiosyncratic errors,
$v_{it}$. Designs 9 and 10 introduce spatial errors in the idiosyncratic
errors, $v_{it}$, and continue to allow for GARCH effects, pricing errors, and
a missing factor. Designs 11 and 12 generate $v_{it}$ with a block covariance
matrix structure, instead of the spatial pattern assumed in designs 9 and 10.
All experiments are implemented using $R=2,000$ replications.

\floatstyle{plaintop} \restylefloat{table}\renewcommand{\thetable}{S-\arabic{table}}%

	\begin{table}[H]%

		\begin{center}%
			\begin{tabular}
				[c]{clllll}\hline\hline
				& {\footnotesize Error} & {\footnotesize GARCH} & {\footnotesize Pricing} &
				{\footnotesize Missing} & {\footnotesize Error}\\
				& {\footnotesize distribution} & {\footnotesize effects} &
				{\footnotesize errors} & {\footnotesize factor} & {\footnotesize covariance}%
				\\\hline
				{\footnotesize 1} & {\footnotesize Gaussian} & ${\footnotesize b}%
				_{k}{\footnotesize =0,c}_{k}{\footnotesize =0}$ & {\footnotesize No} &
				{\footnotesize No} & $\sigma_{ij}=0${\footnotesize , }$i\neq j$\\
				{\footnotesize 2} & $t(5)${\footnotesize \ } & ${\footnotesize b}%
				_{k}{\footnotesize =0,c}_{k}{\footnotesize =0}$ & {\footnotesize No} &
				{\footnotesize No} & $\sigma_{ij}=0${\footnotesize , }$i\neq j$\\
				{\footnotesize 3} & {\footnotesize Gaussian} & ${\footnotesize b}%
				_{k}{\footnotesize =0.8,c}_{k}{\footnotesize =0.1}$ & {\footnotesize No} &
				{\footnotesize No} & $\sigma_{ij}=0${\footnotesize , }$i\neq j$\\
				{\footnotesize 4} & $t(5)${\footnotesize \ } & ${\footnotesize b}%
				_{k}{\footnotesize =0.8,c}_{k}{\footnotesize =0.1}$ & {\footnotesize No} &
				{\footnotesize No} & $\sigma_{ij}=0${\footnotesize , }$i\neq j$\\
				{\footnotesize 5} & {\footnotesize Gaussian} & ${\footnotesize b}%
				_{k}{\footnotesize =0.8,c}_{k}{\footnotesize =0.1}$ & ${\footnotesize \alpha
				}_{\eta}{\footnotesize =0.3}$ & {\footnotesize No} & $\sigma_{ij}%
				=0${\footnotesize , }$i\neq j$\\
				{\footnotesize 6} & $t(5)${\footnotesize \ } & ${\footnotesize b}%
				_{k}{\footnotesize =0.8,c}_{k}{\footnotesize =0.1}$ & ${\footnotesize \alpha
				}_{\eta}{\footnotesize =0.3}$ & {\footnotesize No} & $\sigma_{ij}%
				=0${\footnotesize , }$i\neq j$\\
				{\footnotesize 7} & {\footnotesize Gaussian} & ${\footnotesize b}%
				_{k}{\footnotesize =0.8,c}_{k}{\footnotesize =0.1}$ & ${\footnotesize \alpha
				}_{\eta}{\footnotesize =0.3}$ & ${\footnotesize \alpha}_{\gamma}%
				{\footnotesize =0.5}$ & $\sigma_{ij}=0${\footnotesize , }$i\neq j$\\
				{\footnotesize 8} & $t(5)${\footnotesize \ } & ${\footnotesize b}%
				_{k}{\footnotesize =0.8,c}_{k}{\footnotesize =0.1}$ & ${\footnotesize \alpha
				}_{\eta}{\footnotesize =0.3}$ & ${\footnotesize \alpha}_{\gamma}%
				{\footnotesize =0.5}$ & $\sigma_{ij}=0${\footnotesize , }$i\neq j$\\
				{\footnotesize 9} & {\footnotesize Gaussian} & ${\footnotesize b}%
				_{k}{\footnotesize =0.8,c}_{k}{\footnotesize =0.1}$ & ${\footnotesize \alpha
				}_{\eta}{\footnotesize =0.3}$ & ${\footnotesize \alpha}_{\gamma}%
				{\footnotesize =0.5}$ & ${\footnotesize \rho}_{\varepsilon}%
				{\footnotesize =0.5}$\\
				{\footnotesize 10} & $t(5)${\footnotesize \ } & ${\footnotesize b}%
				_{k}{\footnotesize =0.8,c}_{k}{\footnotesize =0.1}$ & ${\footnotesize \alpha
				}_{\eta}{\footnotesize =0.3}$ & ${\footnotesize \alpha}_{\gamma}%
				{\footnotesize =0.5}$ & ${\footnotesize \rho}_{\varepsilon}%
				{\footnotesize =0.5}$\\
				{\footnotesize 11} & {\footnotesize Gaussian} & ${\footnotesize b}%
				_{k}{\footnotesize =0.8,c}_{k}{\footnotesize =0.1}$ & ${\footnotesize \alpha
				}_{\eta}{\footnotesize =0.3}$ & ${\footnotesize \alpha}_{\gamma}%
				{\footnotesize =0.5}$ & {\footnotesize Block}\\
				{\footnotesize 12} & ${\footnotesize t(5)}$ & ${\footnotesize b}%
				_{k}{\footnotesize =0.8,c}_{k}{\footnotesize =0.1}$ & ${\footnotesize \alpha
				}_{\eta}{\footnotesize =0.3}$ & ${\footnotesize \alpha}_{\gamma}%
				{\footnotesize =0.5}$ & {\footnotesize Block}\\\hline\hline
			\end{tabular}

		\end{center}
		
		{\footnotesize Notes: t-distributed errors are denoted by }$t(5)$%
		{\footnotesize , }$b_{k}$ {\footnotesize and }$c_{k}$ {\footnotesize are the
			parameters of the GARCH(1,1), }$\alpha_{\eta}$ {\footnotesize is the strength
			of the pricing errors, }$\alpha_{\gamma}$ {\footnotesize refers to the
			strength of the missing factor, }$\sigma_{ij}=0${\footnotesize , }$i\neq j$
		{\footnotesize means that the error covariance is diagonal, }$\rho
		_{\varepsilon}$ {\footnotesize is the coefficient of spatial error process,
			and "Block" means that the error covariance matrix is block diagonal. See he
			online supplement A for further details.}
		
		\caption{List of experimental designs and their parameter values}
		\label{TabExperiments}%
\end{table}%

The simulation design, as presented in Table \ref{TabExperiments}, is aligned
with the naming convention employed for tables and figures. Specifically, if a
table is denoted as Table S-A-EX, it signifies that the table pertains to
Experiment X when the GDP and the panel regressions correctly include one
strong and two semi-strong factors. On the other hand, a table labeled as
Table S-B-EX relates to Experiment X when the DGP includes one strong and two
semi-strong factors, comparing the results when strong and semi-strong factors
are included (correct specification) with the ones where only the strong
factor is included (incorrect specification). Likewise, Table S-C-EX provides
the summary results for Experiment X, when the DGP contains one strong and two
weak factors, comparing the results to the case when strong and weak factors
are included (correct specification) with the results obtained when the weak
factors are excluded (incorrect specification). The aforementioned
nomenclature also applies to the figures that present the empirical power functions.

\pagebreak%

\renewcommand{\thetable}{S-A-E\arabic{table}}\renewcommand{\thefigure}{S-A-E\arabic{figure}}%
\setcounter{table}{0}%
%

	\begin{table}[H]%
		\caption{Bias, RMSE and size for the two-step and bias-corrected (BC)
			estimators of $\phi$, for Experiment 1 with one strong and two semi-strong
			factors}\label{tab:s-a-e1}
		
		\begin{center}
			{\footnotesize
				\begin{tabular}
					[c]{rrrrrrrrrr}\hline\hline
					&  & \multicolumn{2}{c}{Bias(x100)} &  & \multicolumn{2}{c}{RMSE(x100)} &  &
					\multicolumn{2}{c}{Size(x100)}\\\cline{3-4}\cline{6-7}\cline{9-10}%
					$\phi_{M}$=$-0.49$, $\alpha_{M}$=$1$ & n & {Two Step} & BC &  & {Two Step} &
					BC &  & {Two Step} & BC\\\cline{2-4}\cline{6-7}\cline{9-10}%
					{$T= 60$} & {\ 100} & 0.51 & -1.55 &  & 26.19 & 86.14 &  & 11.10 & 4.00\\
					{} & {\ 500} & -0.06 & 0.28 &  & 17.05 & 14.33 &  & 28.25 & 5.40\\
					{} & {1,000} & -0.43 & 0.10 &  & 15.80 & 10.12 &  & 41.35 & 6.20\\
					{} & {3,000} & -0.86 & -0.01 &  & 14.77 & 5.56 &  & 60.90 & 5.70\\
					&  &  &  &  &  &  &  &  & \\
					{$T=120$} & {\ 100} & 0.24 & 1.84 &  & 17.47 & 91.35 &  & 7.70 & 4.70\\
					{} & {\ 500} & 0.18 & 0.15 &  & 9.34 & 8.99 &  & 13.95 & 6.15\\
					{} & {1,000} & -0.04 & 0.01 &  & 7.57 & 6.22 &  & 19.15 & 5.65\\
					{} & {3,000} & -0.32 & -0.06 &  & 6.27 & 3.49 &  & 36.65 & 4.95\\
					&  &  &  &  &  &  &  &  & \\
					{$T=240$} & {\ 100} & 0.24 & -0.15 &  & 12.11 & 13.11 &  & 5.40 & 4.55\\
					{} & {\ 500} & 0.31 & 0.09 &  & 5.75 & 5.75 &  & 8.00 & 4.85\\
					{} & {1,000} & 0.22 & 0.03 &  & 4.34 & 4.06 &  & 10.60 & 5.30\\
					{} & {3,000} & 0.07 & -0.00 &  & 3.06 & 2.35 &  & 18.80 & 5.10\\
					$\phi_{H}$=$-0.35$, $\alpha_{H}$=$0.85$ &  &  &  &  &  &  &  &  & \\
					{$T= 60$} & {\ 100} & 2.63 & -0.89 &  & 29.26 & 140.65 &  & 17.55 & 3.65\\
					{} & {\ 500} & 4.07 & -1.16 &  & 25.10 & 24.87 &  & 45.65 & 5.90\\
					{} & {1,000} & 4.58 & -0.12 &  & 25.19 & 18.25 &  & 57.25 & 6.30\\
					{} & {3,000} & 4.81 & 0.30 &  & 26.13 & 11.76 &  & 75.10 & 5.90\\
					&  &  &  &  &  &  &  &  & \\
					{$T=120$} & {\ 100} & 2.51 & 3.60 &  & 19.56 & 190.94 &  & 9.95 & 4.75\\
					{} & {\ 500} & 2.97 & -0.47 &  & 13.83 & 13.46 &  & 25.65 & 5.55\\
					{} & {1,000} & 3.34 & 0.02 &  & 13.68 & 10.20 &  & 39.85 & 6.00\\
					{} & {3,000} & 3.57 & 0.15 &  & 13.85 & 6.48 &  & 61.55 & 5.50\\
					&  &  &  &  &  &  &  &  & \\
					{$T=240$} & {\ 100} & 1.64 & -0.50 &  & 13.44 & 16.67 &  & 7.75 & 6.15\\
					{} & {\ 500} & 1.99 & -0.39 &  & 7.87 & 8.01 &  & 12.95 & 5.55\\
					{} & {1,000} & 2.29 & -0.15 &  & 7.28 & 6.04 &  & 23.05 & 5.40\\
					{} & {3,000} & 2.64 & 0.02 &  & 6.97 & 3.82 &  & 44.55 & 4.70\\
					$\phi_{S}$=$0.16$, $\alpha_{S}$=$0.65$ &  &  &  &  &  &  &  &  & \\
					{$T= 60$} & {\ 100} & -20.09 & 1.57 &  & 37.57 & 256.48 &  & 24.55 & 2.60\\
					{} & {\ 500} & -23.78 & 1.80 &  & 37.42 & 37.75 &  & 56.20 & 4.05\\
					{} & {1,000} & -25.54 & 1.05 &  & 38.85 & 32.14 &  & 68.60 & 5.45\\
					{} & {3,000} & -28.27 & 0.35 &  & 42.26 & 24.96 &  & 82.80 & 7.00\\
					&  &  &  &  &  &  &  &  & \\
					{$T=120$} & {\ 100} & -12.99 & -6.16 &  & 25.50 & 378.75 &  & 16.25 & 4.90\\
					{} & {\ 500} & -16.99 & 0.65 &  & 23.97 & 19.49 &  & 46.85 & 4.80\\
					{} & {1,000} & -18.80 & 0.56 &  & 25.41 & 16.40 &  & 62.60 & 5.55\\
					{} & {3,000} & -22.08 & 0.44 &  & 28.91 & 12.73 &  & 80.55 & 5.75\\
					&  &  &  &  &  &  &  &  & \\
					{$T=240$} & {\ 100} & -8.52 & 0.16 &  & 17.30 & 19.66 &  & 10.80 & 4.55\\
					{} & {\ 500} & -11.19 & 0.19 &  & 15.08 & 11.66 &  & 36.15 & 4.90\\
					{} & {1,000} & -12.71 & 0.34 &  & 15.72 & 9.27 &  & 55.90 & 4.60\\
					{} & {3,000} & -15.93 & 0.18 &  & 18.70 & 7.02 &  & 81.90 & 5.60\\\hline\hline
				\end{tabular}
			}
		\end{center}
		
		{\footnotesize Notes: The DGP for Experiment 1 allows for Gaussian errors, no
			GARCH effects, without pricing errors, no missing factors, and without
			spatial/block error cross dependence. For further details of the experiments,
			see Table }\ref{TabExperiments}.%
		
\end{table}%

\pagebreak

\begin{center}
	{\footnotesize
		\begin{figure}[h]%
			\centering
			\caption{Power functions of the bias-corrected estimators of $\phi_{M}\,$,
				$\phi_{H}$ and $\phi_{S}$ for Experiment 1 }%
			\includegraphics[
			height=5.8608in,
			width=6.4524in
			]%
			{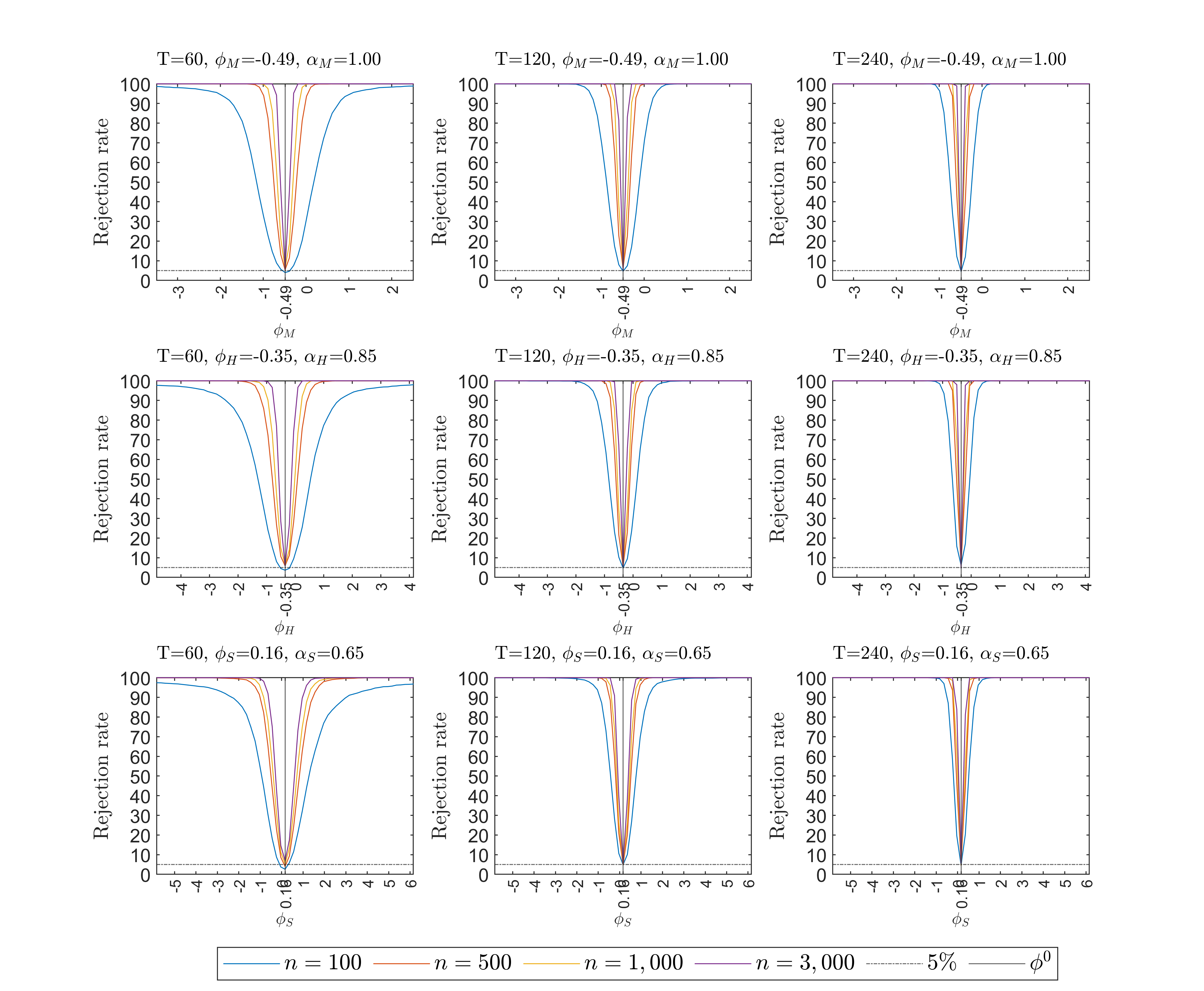}%
			\label{fig:s-a-e1}%
		\end{figure}
	}
\end{center}

{\footnotesize {Note: } See the notes to{ Table \ref{tab:s-a-e1}.}}

\pagebreak%

	\begin{table}[H]%
		\caption{Bias, RMSE and size for the two-step and bias-corrected (BC)
			estimators of $\phi$, for Experiment 2 with one strong and two semi-strong
			factors}\label{tab:s-a-e2}
		
		\begin{center}
			{\footnotesize
				\begin{tabular}
					[c]{rrrrrrrrrr}\hline\hline
					&  & \multicolumn{2}{c}{Bias(x100)} &  & \multicolumn{2}{c}{RMSE(x100)} &  &
					\multicolumn{2}{c}{Size(x100)}\\\cline{3-4}\cline{6-7}\cline{9-10}%
					$\phi_{M}$=$-0.49$, $\alpha_{M}$=$1$ & n & {Two Step} & BC &  & {Two Step} &
					BC &  & {Two Step} & BC\\\cline{2-4}\cline{6-7}\cline{9-10}%
					{$T= 60$} & {\ 100} & 0.65 & 0.27 &  & 26.62 & 45.08 &  & 11.75 & 3.70\\
					{} & {\ 500} & -0.07 & 0.27 &  & 17.26 & 14.47 &  & 28.25 & 5.75\\
					{} & {1,000} & -0.50 & 0.04 &  & 15.74 & 10.24 &  & 39.75 & 6.10\\
					{} & {3,000} & -0.87 & -0.05 &  & 14.80 & 5.72 &  & 60.70 & 5.90\\
					&  &  &  &  &  &  &  &  & \\
					{$T=120$} & {\ 100} & 0.42 & 0.02 &  & 17.50 & 20.45 &  & 7.05 & 4.90\\
					{} & {\ 500} & 0.25 & 0.23 &  & 9.47 & 9.16 &  & 14.80 & 6.60\\
					{} & {1,000} & -0.03 & 0.03 &  & 7.46 & 6.12 &  & 19.20 & 5.55\\
					{} & {3,000} & -0.36 & -0.09 &  & 6.31 & 3.52 &  & 36.70 & 5.65\\
					&  &  &  &  &  &  &  &  & \\
					{$T=240$} & {\ 100} & 0.27 & -0.14 &  & 12.17 & 13.18 &  & 6.20 & 5.05\\
					{} & {\ 500} & 0.34 & 0.11 &  & 5.84 & 5.81 &  & 8.15 & 5.05\\
					{} & {1,000} & 0.22 & 0.04 &  & 4.34 & 4.07 &  & 10.40 & 5.55\\
					{} & {3,000} & 0.04 & -0.03 &  & 3.07 & 2.36 &  & 19.15 & 5.25\\
					$\phi_{H}$=$-0.35$, $\alpha_{H}$=$0.85$ &  &  &  &  &  &  &  &  & \\
					{$T= 60$} & {\ 100} & 2.88 & -2.77 &  & 29.60 & 81.53 &  & 17.70 & 3.25\\
					{} & {\ 500} & 4.18 & -0.95 &  & 25.36 & 25.93 &  & 44.85 & 5.80\\
					{} & {1,000} & 4.62 & -0.02 &  & 25.27 & 18.42 &  & 57.40 & 5.00\\
					{} & {3,000} & 4.75 & 0.13 &  & 26.17 & 12.23 &  & 74.80 & 5.95\\
					&  &  &  &  &  &  &  &  & \\
					{$T=120$} & {\ 100} & 2.49 & -0.71 &  & 19.67 & 28.58 &  & 9.70 & 4.85\\
					{} & {\ 500} & 3.17 & -0.13 &  & 13.89 & 13.46 &  & 25.75 & 5.10\\
					{} & {1,000} & 3.30 & -0.05 &  & 13.75 & 10.34 &  & 40.05 & 6.00\\
					{} & {3,000} & 3.52 & 0.08 &  & 13.84 & 6.61 &  & 62.40 & 6.05\\
					&  &  &  &  &  &  &  &  & \\
					{$T=240$} & {\ 100} & 1.75 & -0.39 &  & 13.41 & 16.61 &  & 7.70 & 5.65\\
					{} & {\ 500} & 2.08 & -0.27 &  & 7.96 & 8.01 &  & 13.75 & 5.10\\
					{} & {1,000} & 2.23 & -0.23 &  & 7.22 & 6.08 &  & 22.20 & 5.15\\
					{} & {3,000} & 2.64 & 0.02 &  & 6.95 & 3.83 &  & 43.50 & 4.60\\
					$\phi_{S}$=$0.16$, $\alpha_{S}$=$0.65$ &  &  &  &  &  &  &  &  & \\
					{$T= 60$} & {\ 100} & -20.16 & 5.40 &  & 37.93 & 125.25 &  & 24.85 & 3.60\\
					{} & {\ 500} & -23.88 & 1.79 &  & 37.72 & 40.47 &  & 55.30 & 4.25\\
					{} & {1,000} & -25.44 & 1.10 &  & 38.75 & 33.04 &  & 67.85 & 5.35\\
					{} & {3,000} & -28.24 & 0.82 &  & 42.24 & 25.93 &  & 82.05 & 6.45\\
					&  &  &  &  &  &  &  &  & \\
					{$T=120$} & {\ 100} & -13.00 & 2.26 &  & 25.41 & 38.09 &  & 16.65 & 5.40\\
					{} & {\ 500} & -17.08 & 0.48 &  & 24.21 & 20.35 &  & 46.35 & 5.55\\
					{} & {1,000} & -18.80 & 0.53 &  & 25.43 & 16.73 &  & 62.75 & 5.50\\
					{} & {3,000} & -22.12 & 0.38 &  & 28.88 & 12.87 &  & 79.90 & 5.70\\
					&  &  &  &  &  &  &  &  & \\
					{$T=240$} & {\ 100} & -8.38 & 0.35 &  & 17.38 & 19.88 &  & 10.75 & 4.90\\
					{} & {\ 500} & -11.10 & 0.33 &  & 14.99 & 11.64 &  & 35.40 & 5.35\\
					{} & {1,000} & -12.66 & 0.43 &  & 15.60 & 9.15 &  & 55.40 & 4.55\\
					{} & {3,000} & -15.88 & 0.29 &  & 18.65 & 6.97 &  & 80.75 & 4.95\\\hline\hline
				\end{tabular}
			}
		\end{center}
		
		{\footnotesize Notes: The DGP for Experiment 2 allows for t(5) distributed
			errors, no GARCH effects, without pricing errors, no missing factors, and
			without spatial/block error cross dependence. For further details of the
			experiments, see Table \ref{TabExperiments}.}%
		
\end{table}%

\pagebreak

\begin{center}
	{\footnotesize
		\begin{figure}[h]%
			\centering
			\caption{Empirical Power Functions, experiment 2, for coefficient of the
				semi-strong factors}%
			\includegraphics[
			height=5.8608in,
			width=7.1122in
			]%
			{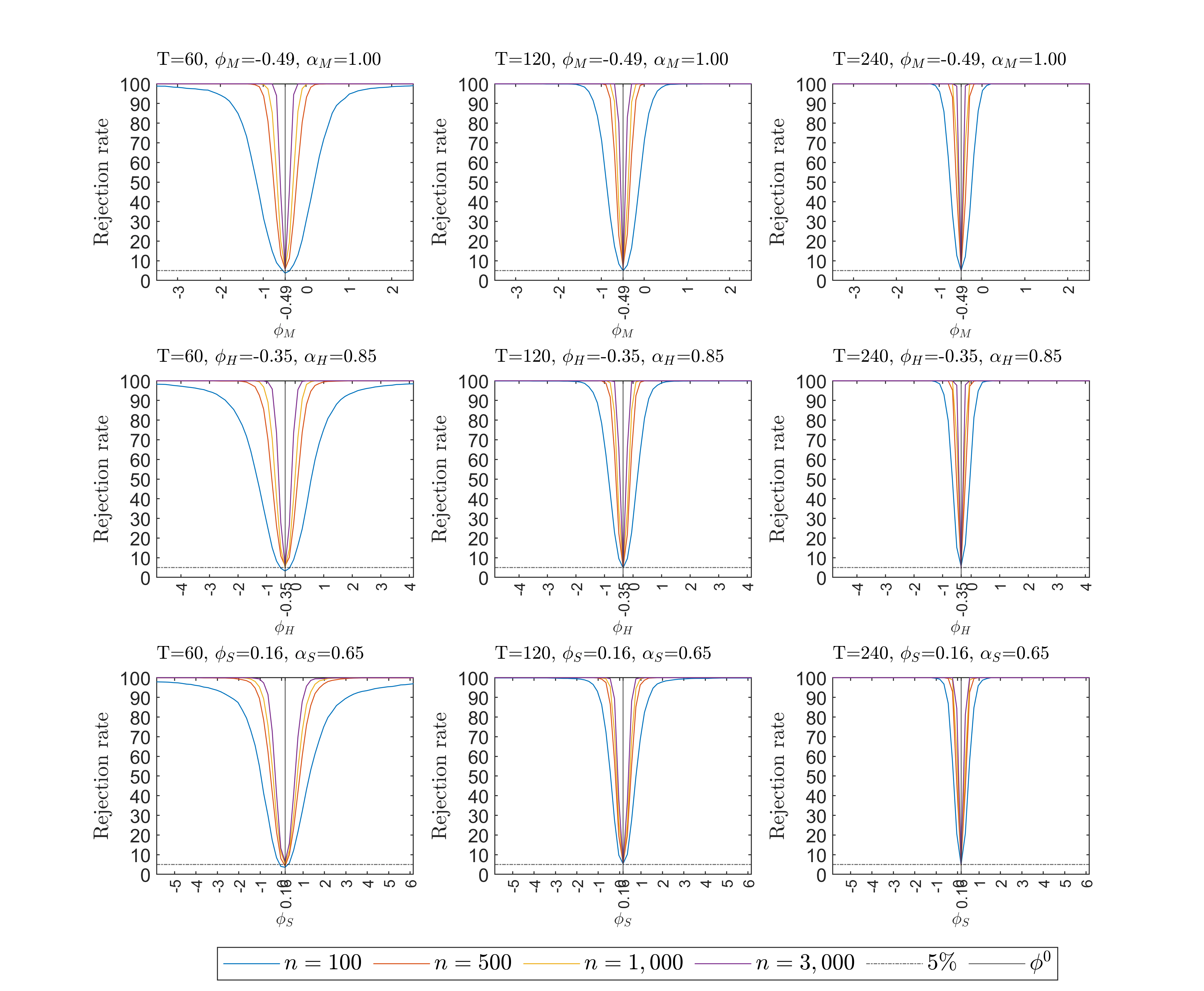}%
			\label{fig:s-a-e2}%
		\end{figure}
	}
\end{center}

{\footnotesize {Note: } See the notes to{ Table \ref{tab:s-a-e2}.}}

\pagebreak%

	\begin{table}[H]%
		\caption{Bias, RMSE and size for the two-step and bias-corrected (BC)
			estimators of $\phi$, for Experiment 3 with one strong and two semi-strong
			factors}\label{tab:s-a-e3}
		
		\begin{center}
			{\footnotesize
				\begin{tabular}
					[c]{rrrrrrrrrr}\hline\hline
					&  & \multicolumn{2}{c}{Bias(x100)} &  & \multicolumn{2}{c}{RMSE(x100)} &  &
					\multicolumn{2}{c}{Size(x100)}\\\cline{3-4}\cline{6-7}\cline{9-10}%
					$\phi_{M}$=$-0.49$, $\alpha_{M}$=$1$ & n & {Two Step} & BC &  & {Two Step} &
					BC &  & {Two Step} & BC\\\cline{2-4}\cline{6-7}\cline{9-10}%
					{$T= 60$} & {\ 100} & 0.51 & -6.00 &  & 26.18 & 227.89 &  & 11.05 & 3.40\\
					{} & {\ 500} & -0.14 & 0.29 &  & 17.03 & 14.71 &  & 28.20 & 5.65\\
					{} & {1,000} & -0.51 & 0.09 &  & 15.81 & 10.33 &  & 40.30 & 6.35\\
					{} & {3,000} & -0.92 & 0.00 &  & 14.83 & 5.65 &  & 61.15 & 5.65\\
					&  &  &  &  &  &  &  &  & \\
					{$T=120$} & {\ 100} & 0.21 & -0.19 &  & 17.45 & 20.63 &  & 7.85 & 4.50\\
					{} & {\ 500} & 0.13 & 0.15 &  & 9.35 & 9.05 &  & 13.80 & 6.25\\
					{} & {1,000} & -0.09 & 0.02 &  & 7.61 & 6.25 &  & 19.95 & 5.75\\
					{} & {3,000} & -0.37 & -0.05 &  & 6.31 & 3.50 &  & 36.65 & 4.90\\
					&  &  &  &  &  &  &  &  & \\
					{$T=240$} & {\ 100} & 0.25 & -0.13 &  & 12.15 & 13.17 &  & 5.35 & 4.45\\
					{} & {\ 500} & 0.31 & 0.10 &  & 5.77 & 5.78 &  & 8.00 & 5.20\\
					{} & {1,000} & 0.21 & 0.05 &  & 4.36 & 4.07 &  & 11.15 & 5.25\\
					{} & {3,000} & 0.06 & 0.00 &  & 3.07 & 2.36 &  & 18.80 & 5.20\\
					$\phi_{H}$=$-0.35$, $\alpha_{H}$=$0.85$ &  &  &  &  &  &  &  &  & \\
					{$T= 60$} & {\ 100} & 2.67 & -7.85 &  & 28.98 & 155.63 &  & 17.70 & 3.80\\
					{} & {\ 500} & 4.06 & -1.08 &  & 24.84 & 26.10 &  & 45.75 & 6.05\\
					{} & {1,000} & 4.57 & -0.15 &  & 24.93 & 18.93 &  & 56.85 & 6.35\\
					{} & {3,000} & 4.80 & 0.28 &  & 25.82 & 12.24 &  & 75.55 & 5.40\\
					&  &  &  &  &  &  &  &  & \\
					{$T=120$} & {\ 100} & 2.51 & -0.70 &  & 19.41 & 28.86 &  & 10.05 & 4.25\\
					{} & {\ 500} & 2.95 & -0.46 &  & 13.76 & 13.63 &  & 25.90 & 5.55\\
					{} & {1,000} & 3.31 & 0.05 &  & 13.56 & 10.28 &  & 39.70 & 5.70\\
					{} & {3,000} & 3.52 & 0.17 &  & 13.70 & 6.56 &  & 61.60 & 5.30\\
					&  &  &  &  &  &  &  &  & \\
					{$T=240$} & {\ 100} & 1.64 & -0.53 &  & 13.42 & 16.73 &  & 7.65 & 6.15\\
					{} & {\ 500} & 2.02 & -0.38 &  & 7.86 & 8.04 &  & 13.35 & 5.65\\
					{} & {1,000} & 2.32 & -0.13 &  & 7.24 & 6.05 &  & 22.35 & 5.45\\
					{} & {3,000} & 2.65 & 0.03 &  & 6.93 & 3.82 &  & 44.45 & 4.35\\
					$\phi_{S}$=$0.16$, $\alpha_{S}$=$0.65$ &  &  &  &  &  &  &  &  & \\
					{$T= 60$} & {\ 100} & -20.59 & 25.62 &  & 37.71 & 892.51 &  & 25.80 & 2.50\\
					{} & {\ 500} & -24.12 & 1.65 &  & 37.45 & 40.44 &  & 57.75 & 3.90\\
					{} & {1,000} & -25.89 & 1.16 &  & 38.88 & 34.07 &  & 68.95 & 5.40\\
					{} & {3,000} & -28.50 & 0.41 &  & 42.22 & 26.13 &  & 82.65 & 6.70\\
					&  &  &  &  &  &  &  &  & \\
					{$T=120$} & {\ 100} & -13.31 & 2.54 &  & 25.72 & 37.94 &  & 16.80 & 4.95\\
					{} & {\ 500} & -17.21 & 0.73 &  & 24.19 & 19.85 &  & 47.55 & 5.25\\
					{} & {1,000} & -19.04 & 0.61 &  & 25.61 & 16.70 &  & 63.75 & 5.65\\
					{} & {3,000} & -22.27 & 0.51 &  & 29.04 & 12.99 &  & 80.65 & 5.95\\
					&  &  &  &  &  &  &  &  & \\
					{$T=240$} & {\ 100} & -8.63 & 0.19 &  & 17.37 & 19.69 &  & 10.70 & 4.55\\
					{} & {\ 500} & -11.31 & 0.21 &  & 15.22 & 11.72 &  & 36.85 & 5.15\\
					{} & {1,000} & -12.87 & 0.33 &  & 15.91 & 9.31 &  & 56.35 & 4.75\\
					{} & {3,000} & -16.07 & 0.17 &  & 18.87 & 7.04 &  & 81.75 & 5.60\\\hline\hline
				\end{tabular}
			}
		\end{center}
		
		{\footnotesize Notes: The DGP for Experiment 3 allows for Gaussian errors,
			with GARCH effects, without pricing errors, no missing factors, and without
			spatial/block error cross dependence. For further details of the experiments,
			see Table \ref{TabExperiments}. }
\end{table}%

\pagebreak

\begin{center}
	{\footnotesize
		\begin{figure}[h]%
			\centering
			\caption{Empirical Power Functions, experiment 3, for coefficient of the
				semi-strong factors}%
			\includegraphics[
			height=5.8608in,
			width=7.1122in
			]%
			{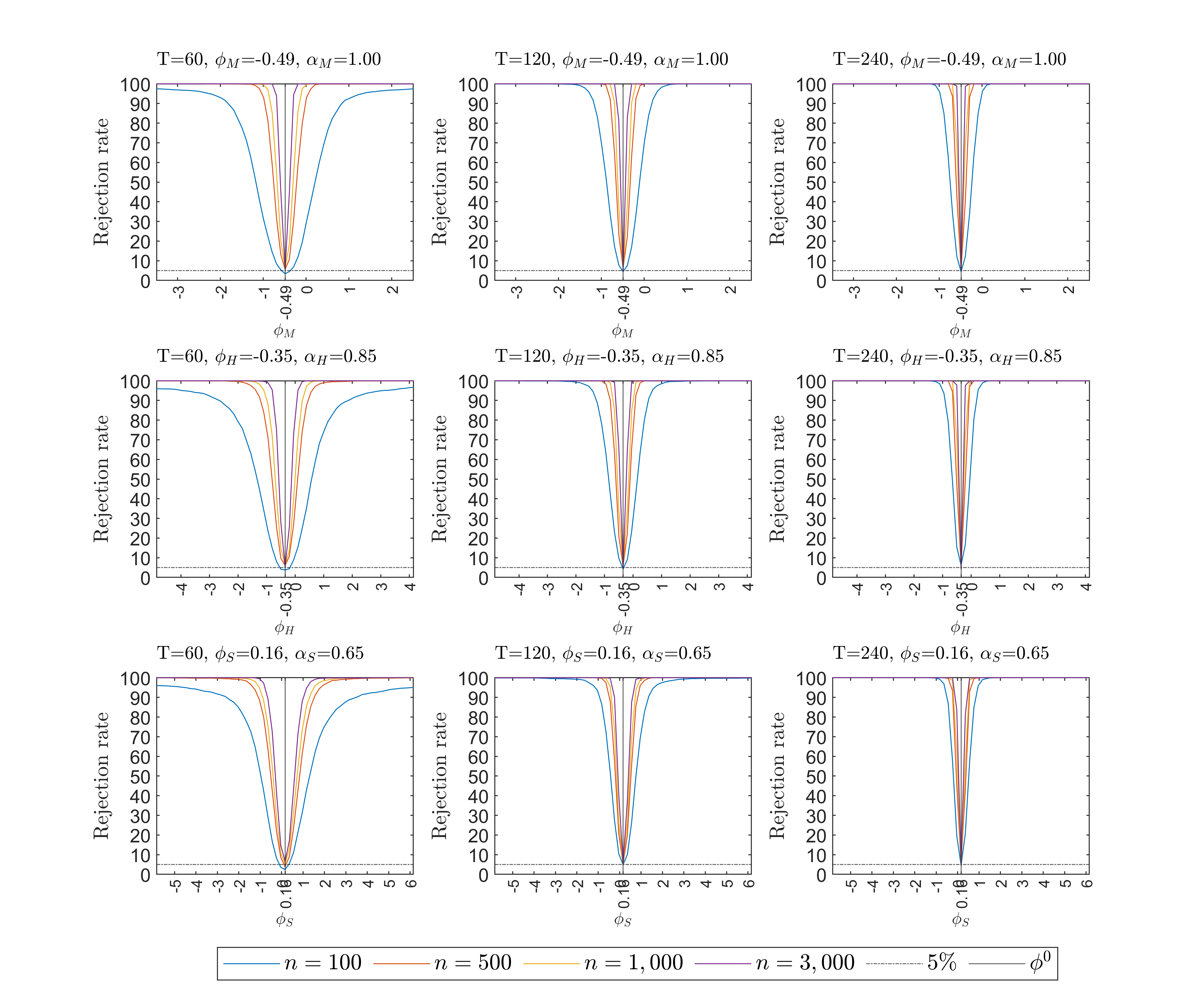}%
			\label{fig:s-a-e3}%
		\end{figure}
	}
\end{center}

{\footnotesize {Note: } See the notes to{ Table \ref{tab:s-a-e3}.}}

\pagebreak%

	\begin{table}[H]%
		\caption{Bias, RMSE and size for the two-step and bias-corrected (BC)
			estimators of $\phi$, for Experiment 4 with one strong and two semi-strong
			factors}\label{tab:s-a-e4}
		
		\begin{center}
			{\footnotesize
				\begin{tabular}
					[c]{rrrrrrrrrr}\hline\hline
					&  & \multicolumn{2}{c}{Bias(x100)} &  & \multicolumn{2}{c}{RMSE(x100)} &  &
					\multicolumn{2}{c}{Size(x100)}\\\cline{3-4}\cline{6-7}\cline{9-10}%
					$\phi_{M}$=$-0.49$, $\alpha_{M}$=$1$ & n & {Two Step} & BC &  & {Two Step} &
					BC &  & {Two Step} & BC\\\cline{2-4}\cline{6-7}\cline{9-10}%
					{$T= 60$} & {\ 100} & 0.59 & 0.17 &  & 26.59 & 107.83 &  & 12.15 & 3.75\\
					{} & {\ 500} & -0.16 & 0.23 &  & 17.26 & 14.75 &  & 28.05 & 5.65\\
					{} & {1,000} & -0.57 & 0.05 &  & 15.76 & 10.43 &  & 39.95 & 6.10\\
					{} & {3,000} & -0.93 & -0.04 &  & 14.85 & 5.82 &  & 60.70 & 5.95\\
					&  &  &  &  &  &  &  &  & \\
					{$T=120$} & {\ 100} & 0.37 & -0.86 &  & 17.48 & 44.37 &  & 7.05 & 4.95\\
					{} & {\ 500} & 0.20 & 0.21 &  & 9.49 & 9.24 &  & 15.15 & 6.45\\
					{} & {1,000} & -0.08 & 0.03 &  & 7.50 & 6.16 &  & 19.45 & 4.90\\
					{} & {3,000} & -0.40 & -0.08 &  & 6.35 & 3.53 &  & 36.70 & 5.85\\
					&  &  &  &  &  &  &  &  & \\
					{$T=240$} & {\ 100} & 0.27 & -0.13 &  & 12.20 & 13.22 &  & 6.35 & 4.85\\
					{} & {\ 500} & 0.34 & 0.12 &  & 5.86 & 5.83 &  & 8.80 & 5.10\\
					{} & {1,000} & 0.22 & 0.05 &  & 4.36 & 4.08 &  & 10.60 & 5.60\\
					{} & {3,000} & 0.03 & -0.02 &  & 3.08 & 2.36 &  & 19.15 & 5.25\\
					$\phi_{H}$=$-0.35$, $\alpha_{H}$=$0.85$ &  &  &  &  &  &  &  &  & \\
					{$T= 60$} & {\ 100} & 2.92 & -5.32 &  & 29.30 & 146.82 &  & 17.95 & 3.15\\
					{} & {\ 500} & 4.18 & -0.76 &  & 25.12 & 27.25 &  & 44.65 & 5.80\\
					{} & {1,000} & 4.61 & -0.06 &  & 25.02 & 19.04 &  & 56.50 & 4.85\\
					{} & {3,000} & 4.73 & 0.12 &  & 25.86 & 12.64 &  & 74.85 & 5.60\\
					&  &  &  &  &  &  &  &  & \\
					{$T=120$} & {\ 100} & 2.50 & -3.13 &  & 19.53 & 113.15 &  & 9.75 & 4.55\\
					{} & {\ 500} & 3.16 & -0.11 &  & 13.84 & 13.64 &  & 25.65 & 5.25\\
					{} & {1,000} & 3.27 & -0.02 &  & 13.64 & 10.41 &  & 40.10 & 6.05\\
					{} & {3,000} & 3.47 & 0.10 &  & 13.70 & 6.67 &  & 60.75 & 6.25\\
					&  &  &  &  &  &  &  &  & \\
					{$T=240$} & {\ 100} & 1.75 & -0.41 &  & 13.39 & 16.66 &  & 7.85 & 5.75\\
					{} & {\ 500} & 2.12 & -0.25 &  & 7.97 & 8.07 &  & 13.90 & 5.05\\
					{} & {1,000} & 2.25 & -0.22 &  & 7.17 & 6.09 &  & 21.70 & 5.40\\
					{} & {3,000} & 2.65 & 0.02 &  & 6.92 & 3.85 &  & 43.35 & 4.70\\
					$\phi_{S}$=$0.16$, $\alpha_{S}$=$0.65$ &  &  &  &  &  &  &  &  & \\
					{$T= 60$} & {\ 100} & -20.62 & 12.61 &  & 38.03 & 267.24 &  & 25.90 & 3.60\\
					{} & {\ 500} & -24.23 & 1.81 &  & 37.80 & 44.22 &  & 56.50 & 3.85\\
					{} & {1,000} & -25.79 & 1.07 &  & 38.78 & 34.41 &  & 68.90 & 5.45\\
					{} & {3,000} & -28.50 & 0.96 &  & 42.21 & 27.28 &  & 81.45 & 6.20\\
					&  &  &  &  &  &  &  &  & \\
					{$T=120$} & {\ 100} & -13.32 & 7.33 &  & 25.66 & 214.62 &  & 16.60 & 5.25\\
					{} & {\ 500} & -17.29 & 0.59 &  & 24.43 & 20.72 &  & 46.95 & 5.50\\
					{} & {1,000} & -19.05 & 0.56 &  & 25.63 & 16.96 &  & 62.30 & 5.80\\
					{} & {3,000} & -22.32 & 0.41 &  & 29.03 & 13.12 &  & 80.50 & 5.85\\
					&  &  &  &  &  &  &  &  & \\
					{$T=240$} & {\ 100} & -8.49 & 0.40 &  & 17.46 & 19.99 &  & 10.75 & 5.10\\
					{} & {\ 500} & -11.22 & 0.35 &  & 15.14 & 11.73 &  & 36.10 & 5.20\\
					{} & {1,000} & -12.80 & 0.44 &  & 15.77 & 9.19 &  & 55.95 & 4.40\\
					{} & {3,000} & -16.02 & 0.29 &  & 18.80 & 6.99 &  & 81.40 & 5.00\\\hline\hline
				\end{tabular}
			}
		\end{center}
		
		{\footnotesize Notes: The DGP for Experiment 4 allows for t(5) distributed
			errors, with GARCH effects, without pricing errors, no missing factors, and
			without spatial/block error cross dependence. For further details of the
			experiments, see Table \ref{TabExperiments}. }
\end{table}%

\pagebreak

\begin{center}
	{\footnotesize
		\begin{figure}[h]%
			\centering
			\caption{Empirical Power Functions, experiment 4, for coefficient of the
				semi-strong factors}%
			\includegraphics[
			height=5.8608in,
			width=7.1122in
			]%
			{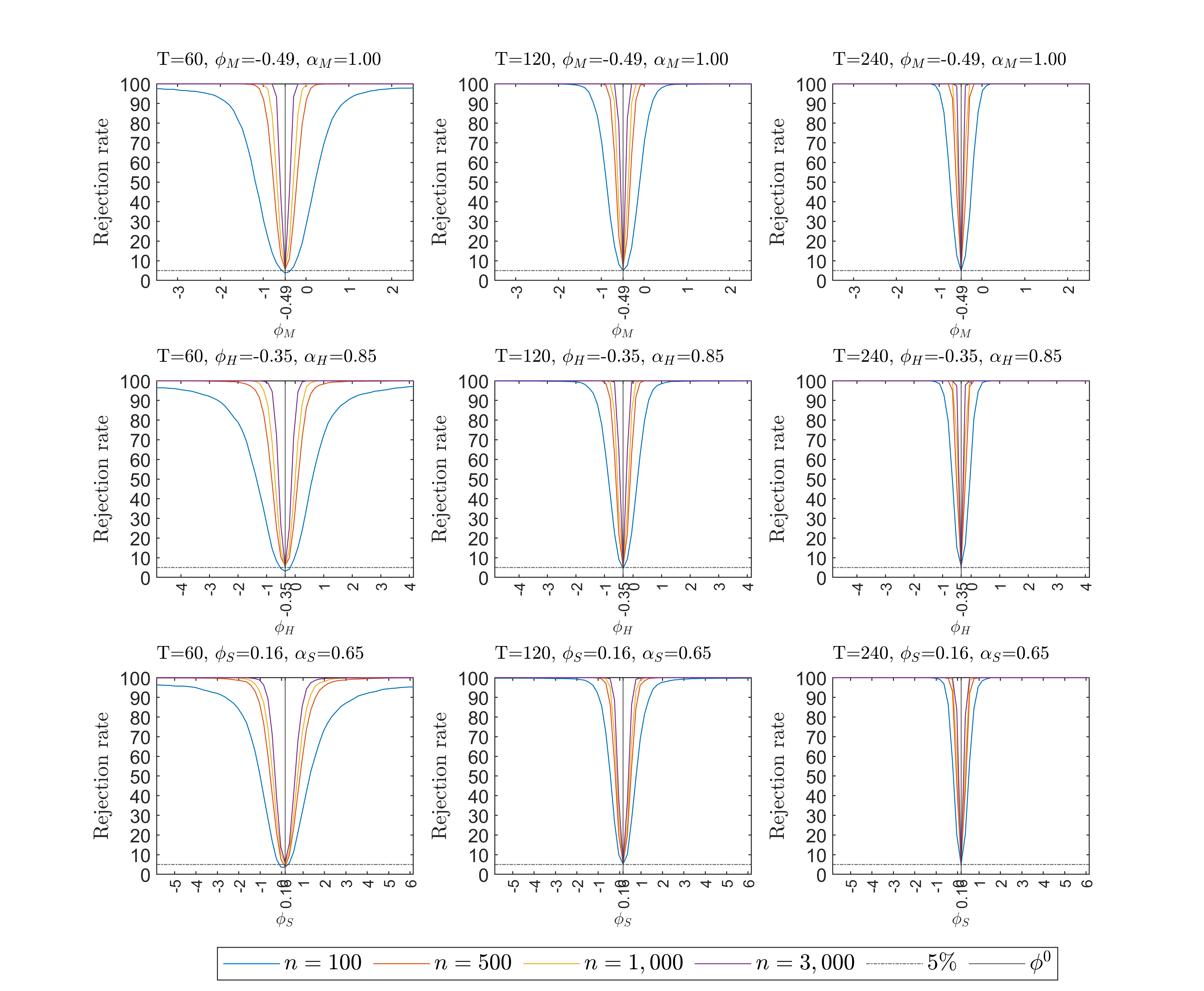}%
			\label{fig:s-a-e4}%
		\end{figure}
	}
\end{center}

{\footnotesize {Note: } See the notes to{ Table \ref{tab:s-a-e4}.}}

\pagebreak%

	\begin{table}[H]%

		\caption{Bias, RMSE and size for the two-step and bias-corrected (BC)
			estimators of $\phi$, for Experiment 5 with one strong and two semi-strong
			factors }\label{tab:s-a-e5}
		
		\begin{center}
			{\footnotesize
				\begin{tabular}
					[c]{rrrrrrrrrr}\hline\hline
					&  & \multicolumn{2}{c}{Bias(x100)} &  & \multicolumn{2}{c}{RMSE(x100)} &  &
					\multicolumn{2}{c}{Size(x100)}\\\cline{3-4}\cline{6-7}\cline{9-10}%
					$\phi_{M}$=$-0.49$, $\alpha_{M}$=$1$ & n & {Two Step} & BC &  & {Two Step} &
					BC &  & {Two Step} & BC\\\cline{2-4}\cline{6-7}\cline{9-10}%
					{$T= 60$} & {\ 100} & 0.42 & -7.58 &  & 26.27 & 285.21 &  & 11.35 & 3.75\\
					{} & {\ 500} & -0.15 & 0.28 &  & 17.05 & 14.79 &  & 28.10 & 5.80\\
					{} & {1,000} & -0.49 & 0.10 &  & 15.81 & 10.35 &  & 40.35 & 6.25\\
					{} & {3,000} & -0.92 & -0.00 &  & 14.83 & 5.65 &  & 61.10 & 5.75\\
					&  &  &  &  &  &  &  &  & \\
					{$T=120$} & {\ 100} & 0.06 & -0.35 &  & 17.63 & 20.82 &  & 7.85 & 4.95\\
					{} & {\ 500} & 0.13 & 0.14 &  & 9.40 & 9.15 &  & 14.85 & 6.45\\
					{} & {1,000} & -0.07 & 0.04 &  & 7.60 & 6.28 &  & 20.00 & 5.90\\
					{} & {3,000} & -0.37 & -0.06 &  & 6.32 & 3.50 &  & 37.00 & 4.80\\
					&  &  &  &  &  &  &  &  & \\
					{$T=240$} & {\ 100} & 0.08 & -0.31 &  & 12.48 & 13.53 &  & 6.30 & 5.35\\
					{} & {\ 500} & 0.31 & 0.10 &  & 5.83 & 5.84 &  & 8.35 & 5.35\\
					{} & {1,000} & 0.23 & 0.07 &  & 4.38 & 4.11 &  & 10.95 & 5.20\\
					{} & {3,000} & 0.06 & 0.00 &  & 3.08 & 2.36 &  & 18.75 & 5.25\\
					$\phi_{H}$=$-0.35$, $\alpha_{H}$=$0.85$ &  &  &  &  &  &  &  &  & \\
					{$T= 60$} & {\ 100} & 2.58 & -8.49 &  & 28.97 & 174.83 &  & 18.20 & 3.30\\
					{} & {\ 500} & 4.10 & -1.01 &  & 24.82 & 26.10 &  & 45.25 & 6.15\\
					{} & {1,000} & 4.58 & -0.12 &  & 24.96 & 18.91 &  & 57.15 & 6.50\\
					{} & {3,000} & 4.80 & 0.28 &  & 25.82 & 12.25 &  & 75.60 & 5.60\\
					&  &  &  &  &  &  &  &  & \\
					{$T=120$} & {\ 100} & 2.44 & -0.79 &  & 19.68 & 28.94 &  & 10.35 & 4.70\\
					{} & {\ 500} & 2.98 & -0.43 &  & 13.78 & 13.79 &  & 25.75 & 5.60\\
					{} & {1,000} & 3.32 & 0.07 &  & 13.57 & 10.29 &  & 40.00 & 6.05\\
					{} & {3,000} & 3.52 & 0.17 &  & 13.70 & 6.56 &  & 61.85 & 5.45\\
					&  &  &  &  &  &  &  &  & \\
					{$T=240$} & {\ 100} & 1.57 & -0.63 &  & 13.97 & 17.37 &  & 8.25 & 6.75\\
					{} & {\ 500} & 2.04 & -0.35 &  & 7.92 & 8.17 &  & 13.85 & 5.75\\
					{} & {1,000} & 2.33 & -0.12 &  & 7.28 & 6.09 &  & 22.70 & 5.50\\
					{} & {3,000} & 2.66 & 0.04 &  & 6.94 & 3.83 &  & 44.70 & 4.55\\
					$\phi_{S}$=$0.16$, $\alpha_{S}$=$0.65$ &  &  &  &  &  &  &  &  & \\
					{$T= 60$} & {\ 100} & -20.57 & 31.92 &  & 37.88 & 1112.05 &  & 25.45 & 2.55\\
					{} & {\ 500} & -24.15 & 1.61 &  & 37.49 & 40.53 &  & 57.90 & 3.95\\
					{} & {1,000} & -25.88 & 1.16 &  & 38.88 & 34.12 &  & 68.80 & 5.70\\
					{} & {3,000} & -28.51 & 0.40 &  & 42.23 & 26.13 &  & 82.65 & 6.60\\
					&  &  &  &  &  &  &  &  & \\
					{$T=120$} & {\ 100} & -13.21 & 2.55 &  & 26.09 & 37.89 &  & 16.70 & 5.50\\
					{} & {\ 500} & -17.24 & 0.70 &  & 24.27 & 20.01 &  & 47.90 & 5.40\\
					{} & {1,000} & -19.03 & 0.62 &  & 25.61 & 16.77 &  & 63.50 & 5.85\\
					{} & {3,000} & -22.26 & 0.53 &  & 29.04 & 13.01 &  & 80.55 & 6.05\\
					&  &  &  &  &  &  &  &  & \\
					{$T=240$} & {\ 100} & -8.55 & 0.29 &  & 18.03 & 20.67 &  & 11.85 & 6.25\\
					{} & {\ 500} & -11.35 & 0.16 &  & 15.31 & 11.89 &  & 37.35 & 5.85\\
					{} & {1,000} & -12.87 & 0.33 &  & 15.91 & 9.37 &  & 56.40 & 4.75\\
					{} & {3,000} & -16.07 & 0.18 &  & 18.86 & 7.07 &  & 81.75 & 5.65\\\hline\hline
				\end{tabular}
			}
		\end{center}
		
		{\footnotesize Notes: The DGP for Experiment 5 allows for Gaussian errors,
			with GARCH effects, with pricing errors ($\alpha_{\eta}=0.3$), no missing
			factors, and without spatial/block error cross dependence. For further details
			of the experiments, see Table \ref{TabExperiments}. }
\end{table}%

\pagebreak

\begin{center}
	{\footnotesize
		\begin{figure}[h]%
			\centering
			\caption{Empirical Power Functions, experiment 5, for coefficient of the
				semi-strong factors}%
			\includegraphics[
			height=5.8608in,
			width=7.1122in
			]%
			{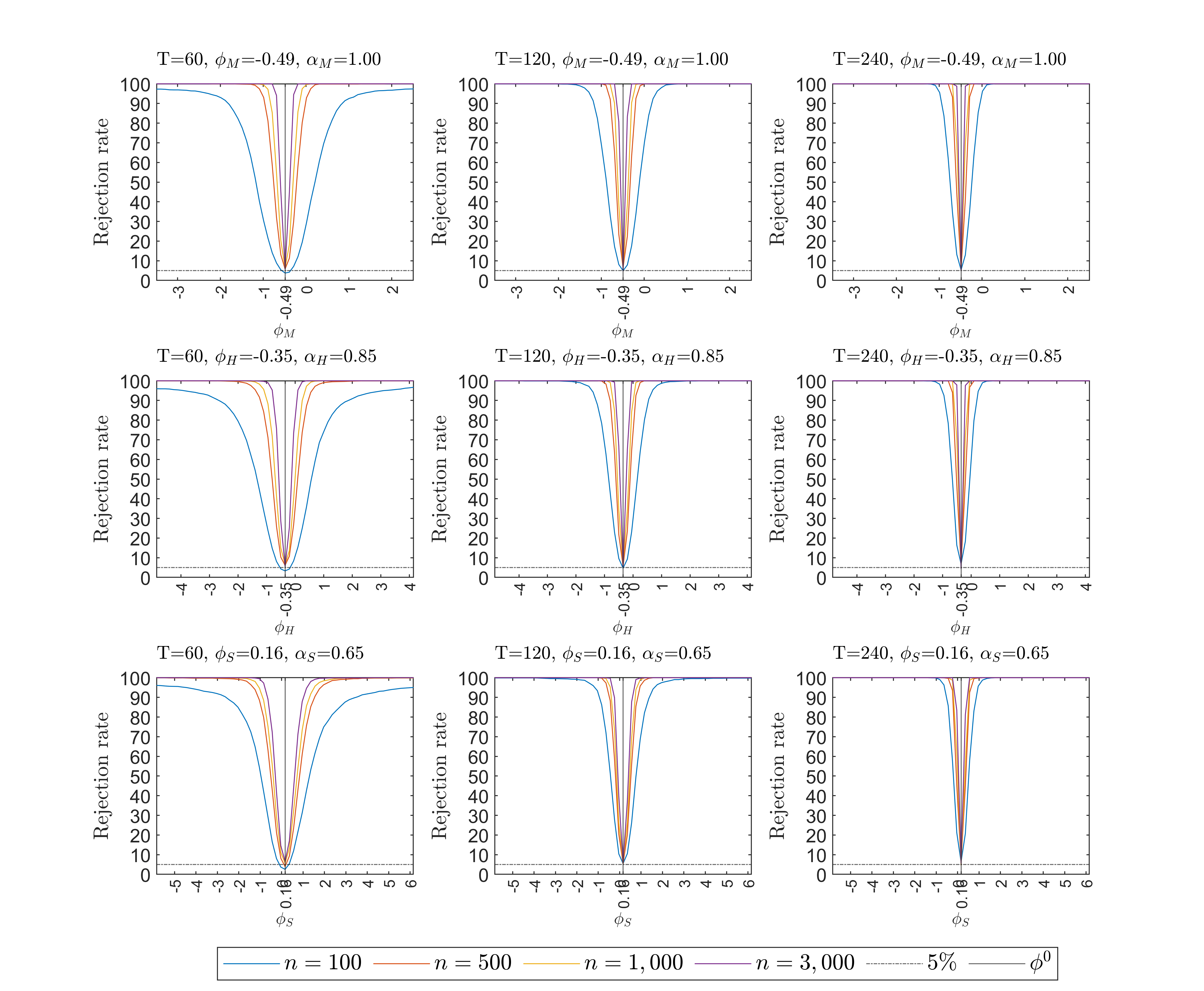}%
			\label{fig:s-a-e5}%
		\end{figure}
	}
\end{center}

{\footnotesize {Note: } See the notes to{ Table \ref{tab:s-a-e5}.}}

\pagebreak%

	\begin{table}[H]%

		\caption{Bias, RMSE and size for the two-step and bias-corrected (BC)
			estimators of $\phi$, for Experiment 6 with one strong and two semi-strong
			factors }\label{tab:s-a-e6}
		
		\begin{center}
			{\footnotesize
				\begin{tabular}
					[c]{rrrrrrrrrr}\hline\hline
					&  & \multicolumn{2}{c}{Bias(x100)} &  & \multicolumn{2}{c}{RMSE(x100)} &  &
					\multicolumn{2}{c}{Size(x100)}\\\cline{3-4}\cline{6-7}\cline{9-10}%
					$\phi_{M}$=$-0.49$, $\alpha_{M}$=$1$ & n & {Two Step} & BC &  & {Two Step} &
					BC &  & {Two Step} & BC\\\cline{2-4}\cline{6-7}\cline{9-10}%
					{$T= 60$} & {\ 100} & 0.49 & 0.05 &  & 26.69 & 103.65 &  & 12.25 & 3.95\\
					{} & {\ 500} & -0.17 & 0.24 &  & 17.27 & 14.81 &  & 28.25 & 5.95\\
					{} & {1,000} & -0.56 & 0.06 &  & 15.75 & 10.44 &  & 39.60 & 6.00\\
					{} & {3,000} & -0.93 & -0.05 &  & 14.85 & 5.82 &  & 60.75 & 5.80\\
					&  &  &  &  &  &  &  &  & \\
					{$T=120$} & {\ 100} & 0.22 & -0.76 &  & 17.65 & 33.47 &  & 7.45 & 5.15\\
					{} & {\ 500} & 0.20 & 0.22 &  & 9.54 & 9.32 &  & 15.30 & 6.95\\
					{} & {1,000} & -0.06 & 0.05 &  & 7.49 & 6.18 &  & 19.50 & 5.30\\
					{} & {3,000} & -0.40 & -0.09 &  & 6.36 & 3.54 &  & 37.05 & 5.70\\
					&  &  &  &  &  &  &  &  & \\
					{$T=240$} & {\ 100} & 0.11 & -0.30 &  & 12.57 & 13.63 &  & 7.55 & 6.10\\
					{} & {\ 500} & 0.33 & 0.11 &  & 5.91 & 5.88 &  & 8.60 & 5.10\\
					{} & {1,000} & 0.24 & 0.07 &  & 4.38 & 4.13 &  & 10.30 & 5.80\\
					{} & {3,000} & 0.03 & -0.03 &  & 3.09 & 2.37 &  & 19.25 & 5.45\\
					$\phi_{H}$=$-0.35$, $\alpha_{H}$=$0.85$ &  &  &  &  &  &  &  &  & \\
					{$T= 60$} & {\ 100} & 2.83 & -5.32 &  & 29.34 & 144.70 &  & 18.05 & 3.40\\
					{} & {\ 500} & 4.22 & -0.67 &  & 25.10 & 27.28 &  & 44.70 & 5.45\\
					{} & {1,000} & 4.62 & -0.03 &  & 25.04 & 19.02 &  & 56.30 & 4.80\\
					{} & {3,000} & 4.73 & 0.12 &  & 25.86 & 12.64 &  & 74.70 & 5.65\\
					&  &  &  &  &  &  &  &  & \\
					{$T=120$} & {\ 100} & 2.42 & -2.45 &  & 19.82 & 78.69 &  & 11.40 & 4.90\\
					{} & {\ 500} & 3.19 & -0.06 &  & 13.87 & 13.81 &  & 25.85 & 5.45\\
					{} & {1,000} & 3.28 & -0.00 &  & 13.65 & 10.44 &  & 40.80 & 6.35\\
					{} & {3,000} & 3.47 & 0.11 &  & 13.70 & 6.67 &  & 61.05 & 6.10\\
					&  &  &  &  &  &  &  &  & \\
					{$T=240$} & {\ 100} & 1.68 & -0.51 &  & 13.94 & 17.33 &  & 8.35 & 6.70\\
					{} & {\ 500} & 2.15 & -0.22 &  & 8.05 & 8.22 &  & 14.70 & 5.60\\
					{} & {1,000} & 2.26 & -0.20 &  & 7.23 & 6.14 &  & 21.75 & 5.65\\
					{} & {3,000} & 2.65 & 0.03 &  & 6.93 & 3.86 &  & 43.85 & 5.00\\
					$\phi_{S}$=$0.16$, $\alpha_{S}$=$0.65$ &  &  &  &  &  &  &  &  & \\
					{$T= 60$} & {\ 100} & -20.59 & 12.67 &  & 38.25 & 257.83 &  & 25.55 & 3.70\\
					{} & {\ 500} & -24.27 & 1.73 &  & 37.83 & 44.57 &  & 56.60 & 3.95\\
					{} & {1,000} & -25.78 & 1.10 &  & 38.79 & 34.52 &  & 68.85 & 5.30\\
					{} & {3,000} & -28.50 & 0.95 &  & 42.22 & 27.27 &  & 81.40 & 6.05\\
					&  &  &  &  &  &  &  &  & \\
					{$T=120$} & {\ 100} & -13.22 & 5.85 &  & 26.06 & 146.46 &  & 16.95 & 5.40\\
					{} & {\ 500} & -17.32 & 0.55 &  & 24.50 & 20.89 &  & 47.05 & 5.55\\
					{} & {1,000} & -19.04 & 0.58 &  & 25.64 & 17.05 &  & 62.35 & 5.95\\
					{} & {3,000} & -22.32 & 0.42 &  & 29.03 & 13.16 &  & 80.25 & 5.90\\
					&  &  &  &  &  &  &  &  & \\
					{$T=240$} & {\ 100} & -8.41 & 0.50 &  & 18.09 & 20.93 &  & 12.15 & 6.40\\
					{} & {\ 500} & -11.25 & 0.31 &  & 15.25 & 11.95 &  & 36.35 & 5.70\\
					{} & {1,000} & -12.80 & 0.44 &  & 15.77 & 9.26 &  & 55.35 & 5.05\\
					{} & {3,000} & -16.02 & 0.29 &  & 18.80 & 7.02 &  & 81.00 & 5.10\\\hline\hline
				\end{tabular}
			}
		\end{center}
		
		{\footnotesize Notes: The DGP for Experiment 6 allows for t(5) distributed
			errors, with GARCH effects, with pricing errors ($\alpha_{\eta}=0.3$), no
			missing factors, and without spatial/block error cross dependence. For further
			details of the experiments, see Table \ref{TabExperiments}.}
\end{table}%

\pagebreak

\begin{center}
	{\footnotesize
		\begin{figure}[h]%
			\centering
			\caption{Empirical Power Functions, experiment 6, for coefficient of the
				semi-strong factors}%
			\includegraphics[
			height=5.8608in,
			width=7.1122in
			]%
			{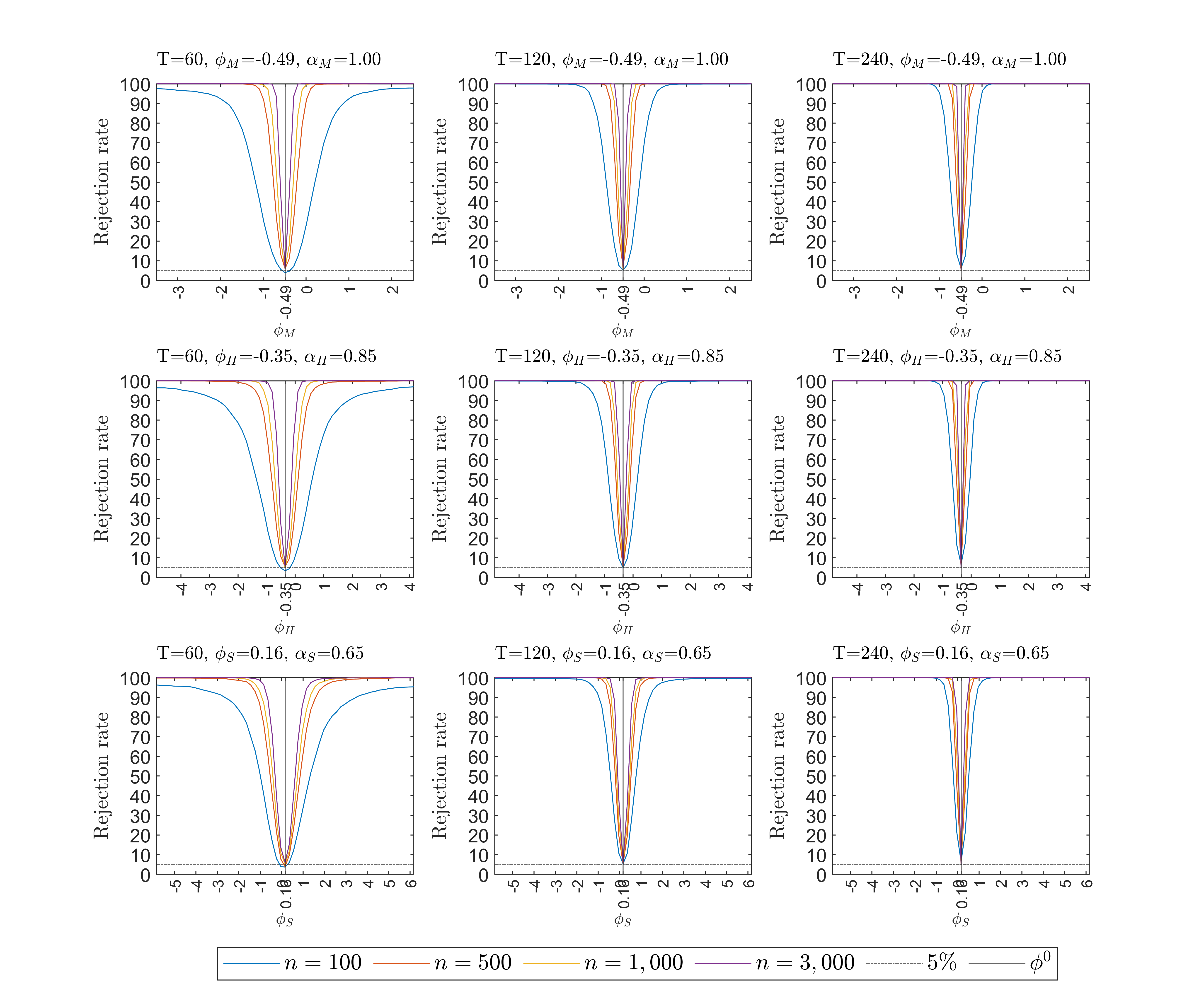}%
			\label{fig:s-a-e6}%
		\end{figure}
	}
\end{center}

{\footnotesize {Note: } See the notes to{ Table \ref{tab:s-a-e6}.}}

\pagebreak%

	\begin{table}[H]%

		\caption{Bias, RMSE and size for the two-step and bias-corrected (BC)
			estimators of $\phi$, for Experiment 7 with one strong and two semi-strong
			factors}\label{tab:s-a-e7}
		
		\begin{center}
			{\footnotesize
				\begin{tabular}
					[c]{rrrrrrrrrr}\hline\hline
					&  & \multicolumn{2}{c}{Bias(x100)} &  & \multicolumn{2}{c}{RMSE(x100)} &  &
					\multicolumn{2}{c}{Size(x100)}\\\cline{3-4}\cline{6-7}\cline{9-10}%
					$\phi_{M}$=$-0.49$, $\alpha_{M}$=$1$ & n & {Two Step} & BC &  & {Two Step} &
					BC &  & {Two Step} & BC\\\cline{2-4}\cline{6-7}\cline{9-10}%
					{$T= 60$} & {\ 100} & -0.33 & -14.31 &  & 26.54 & 555.11 &  & 11.80 & 3.20\\
					{} & {\ 500} & -0.50 & -0.18 &  & 17.67 & 15.13 &  & 29.80 & 6.40\\
					{} & {1,000} & -0.73 & -0.15 &  & 15.95 & 10.41 &  & 42.20 & 6.25\\
					{} & {3,000} & -0.91 & 0.04 &  & 14.87 & 5.81 &  & 61.20 & 5.90\\
					&  &  &  &  &  &  &  &  & \\
					{$T=120$} & {\ 100} & -0.13 & -0.49 &  & 17.44 & 20.30 &  & 7.55 & 5.05\\
					{} & {\ 500} & 0.01 & 0.01 &  & 9.49 & 8.86 &  & 14.20 & 5.85\\
					{} & {1,000} & -0.11 & 0.04 &  & 7.74 & 6.31 &  & 20.50 & 6.60\\
					{} & {3,000} & -0.27 & 0.08 &  & 6.40 & 3.60 &  & 36.65 & 5.65\\
					&  &  &  &  &  &  &  &  & \\
					{$T=240$} & {\ 100} & -0.00 & -0.39 &  & 12.82 & 14.02 &  & 7.30 & 6.75\\
					{} & {\ 500} & 0.15 & -0.08 &  & 6.14 & 6.13 &  & 9.85 & 6.95\\
					{} & {1,000} & 0.14 & -0.03 &  & 4.52 & 4.22 &  & 11.90 & 5.80\\
					{} & {3,000} & 0.07 & 0.01 &  & 3.05 & 2.32 &  & 18.30 & 5.15\\
					$\phi_{H}$=$-0.35$, $\alpha_{H}$=$0.85$ &  &  &  &  &  &  &  &  & \\
					{$T= 60$} & {\ 100} & 3.05 & -222.60 &  & 29.40 & 9435.95 &  & 17.55 & 2.70\\
					{} & {\ 500} & 4.42 & -0.75 &  & 25.03 & 26.57 &  & 46.25 & 4.95\\
					{} & {1,000} & 4.43 & -0.56 &  & 25.18 & 19.01 &  & 58.75 & 5.45\\
					{} & {3,000} & 4.63 & -0.15 &  & 25.85 & 12.14 &  & 76.15 & 5.75\\
					&  &  &  &  &  &  &  &  & \\
					{$T=120$} & {\ 100} & 2.45 & -0.89 &  & 19.84 & 28.94 &  & 10.55 & 4.95\\
					{} & {\ 500} & 3.24 & -0.00 &  & 13.84 & 13.44 &  & 26.75 & 5.95\\
					{} & {1,000} & 3.30 & 0.05 &  & 13.47 & 9.91 &  & 41.30 & 5.20\\
					{} & {3,000} & 3.29 & -0.25 &  & 13.67 & 6.56 &  & 62.45 & 5.75\\
					&  &  &  &  &  &  &  &  & \\
					{$T=240$} & {\ 100} & 1.67 & -0.55 &  & 13.68 & 16.80 &  & 7.45 & 6.10\\
					{} & {\ 500} & 2.28 & -0.04 &  & 8.00 & 8.00 &  & 13.35 & 5.50\\
					{} & {1,000} & 2.42 & -0.02 &  & 7.22 & 6.00 &  & 22.60 & 5.10\\
					{} & {3,000} & 2.54 & -0.13 &  & 6.85 & 3.82 &  & 44.00 & 5.25\\
					$\phi_{S}$=$0.16$, $\alpha_{S}$=$0.65$ &  &  &  &  &  &  &  &  & \\
					{$T= 60$} & {\ 100} & -20.23 & 62.14 &  & 37.82 & 2347.80 &  & 25.15 & 2.70\\
					{} & {\ 500} & -24.55 & 1.18 &  & 37.58 & 46.58 &  & 59.25 & 4.40\\
					{} & {1,000} & -26.04 & 0.76 &  & 38.91 & 33.25 &  & 69.25 & 5.25\\
					{} & {3,000} & -28.68 & 0.06 &  & 42.34 & 26.75 &  & 82.30 & 6.90\\
					&  &  &  &  &  &  &  &  & \\
					{$T=120$} & {\ 100} & -12.90 & 2.52 &  & 25.93 & 38.03 &  & 16.35 & 5.70\\
					{} & {\ 500} & -17.48 & 0.44 &  & 24.54 & 21.10 &  & 50.00 & 6.00\\
					{} & {1,000} & -19.25 & 0.09 &  & 25.81 & 16.81 &  & 63.35 & 5.45\\
					{} & {3,000} & -22.34 & 0.29 &  & 29.12 & 13.12 &  & 80.15 & 6.15\\
					&  &  &  &  &  &  &  &  & \\
					{$T=240$} & {\ 100} & -8.13 & 0.84 &  & 17.77 & 20.91 &  & 11.30 & 5.65\\
					{} & {\ 500} & -11.31 & 0.30 &  & 15.26 & 11.79 &  & 37.05 & 5.70\\
					{} & {1,000} & -12.84 & 0.38 &  & 15.96 & 9.45 &  & 56.40 & 5.15\\
					{} & {3,000} & -16.09 & 0.21 &  & 18.91 & 7.26 &  & 80.95 & 6.30\\\hline\hline
				\end{tabular}
			}
		\end{center}
		
		{\footnotesize Notes: The DGP for Experiment 7 allows for Gaussian errors,
			with GARCH effects, with pricing errors ($\alpha_{\eta}=0.3$), with one weak
			missing factor ($\alpha_{\gamma}=0.5$), and without spatial/block error cross
			dependence. For further details of the experiments, see \ref{TabExperiments}.
		}
\end{table}%

\pagebreak

\begin{center}
	{\footnotesize
		\begin{figure}[h]%
			\centering
			\caption{Empirical Power Functions, experiment 7 for coefficient of the
				semi-strong factors}%
			\includegraphics[
			height=5.8608in,
			width=7.1122in
			]%
			{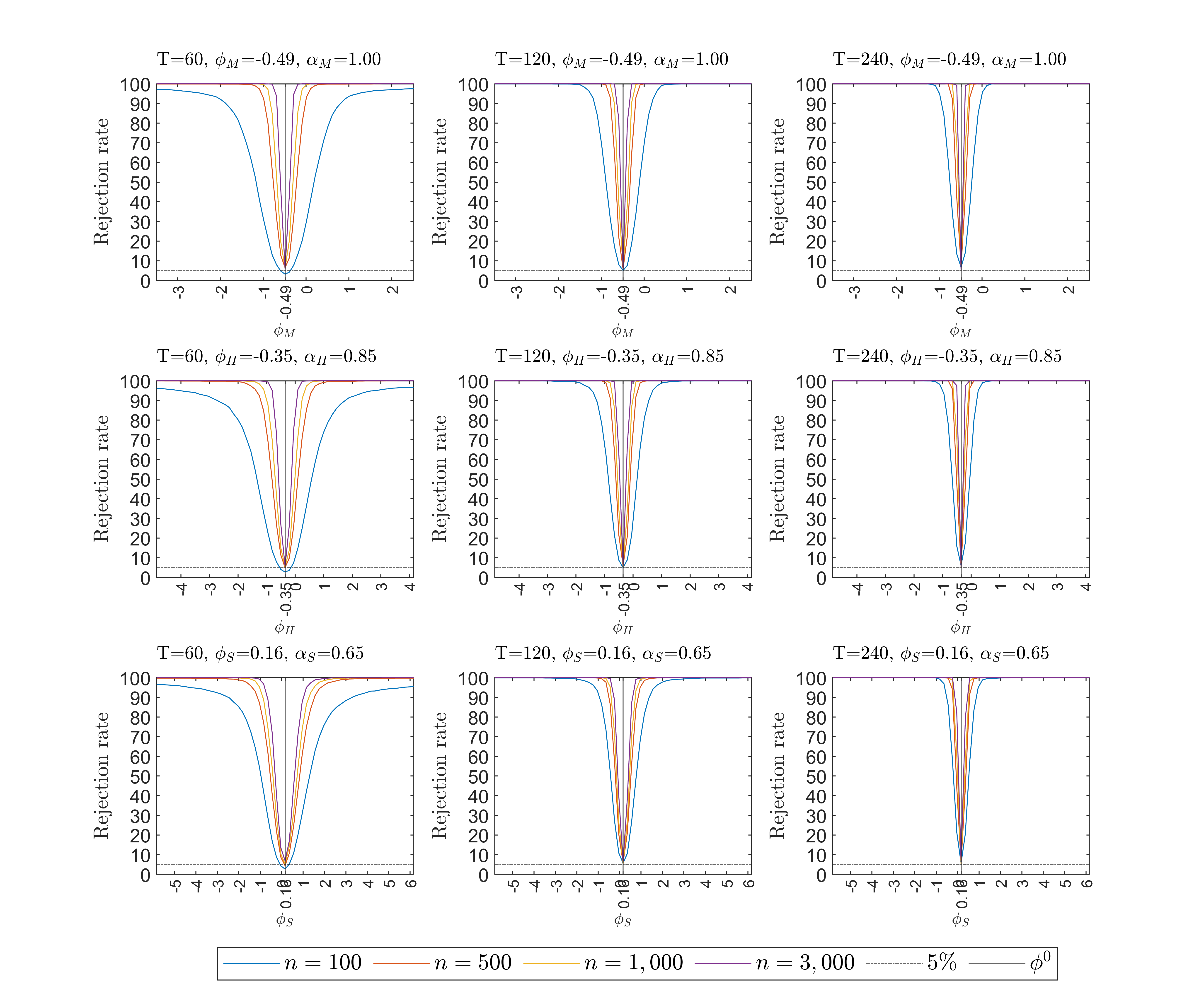}%
			\label{fig:s-a-e7}%
		\end{figure}
	}
\end{center}

{\footnotesize {Note: } See the notes to{ Table \ref{tab:s-a-e7}.}}

\pagebreak%

	\begin{table}[H]%

		\caption{Bias, RMSE and size for the two-step and bias-corrected (BC)
			estimators of $\phi$, for Experiment 8 with one strong and two semi-strong
			factors}\label{tab:s-a-e8}
		
		\begin{center}
			{\footnotesize
				\begin{tabular}
					[c]{rrrrrrrrrr}\hline\hline
					&  & \multicolumn{2}{c}{Bias(x100)} &  & \multicolumn{2}{c}{RMSE(x100)} &  &
					\multicolumn{2}{c}{Size(x100)}\\\cline{3-4}\cline{6-7}\cline{9-10}%
					$\phi_{M}$=$-0.49$, $\alpha_{M}$=$1$ & n & {Two Step} & BC &  & {Two Step} &
					BC &  & {Two Step} & BC\\\cline{2-4}\cline{6-7}\cline{9-10}%
					{$T= 60$} & {\ 100} & -0.27 & 9.06 &  & 27.12 & 251.34 &  & 11.45 & 3.75\\
					{} & {\ 500} & -0.74 & -0.67 &  & 17.81 & 16.71 &  & 29.75 & 6.30\\
					{} & {1,000} & -0.64 & 0.01 &  & 16.06 & 10.38 &  & 40.90 & 5.70\\
					{} & {3,000} & -0.85 & 0.12 &  & 14.90 & 5.96 &  & 59.65 & 6.15\\
					&  &  &  &  &  &  &  &  & \\
					{$T=120$} & {\ 100} & 0.04 & -0.32 &  & 17.49 & 20.53 &  & 7.35 & 4.75\\
					{} & {\ 500} & -0.06 & -0.04 &  & 9.49 & 8.89 &  & 13.95 & 5.80\\
					{} & {1,000} & 0.01 & 0.19 &  & 7.78 & 6.40 &  & 20.90 & 6.10\\
					{} & {3,000} & -0.24 & 0.14 &  & 6.39 & 3.63 &  & 37.25 & 5.90\\
					&  &  &  &  &  &  &  &  & \\
					{$T=240$} & {\ 100} & 0.09 & -0.29 &  & 12.84 & 14.07 &  & 7.10 & 6.35\\
					{} & {\ 500} & 0.20 & -0.01 &  & 6.08 & 6.06 &  & 9.50 & 6.05\\
					{} & {1,000} & 0.17 & 0.01 &  & 4.52 & 4.25 &  & 11.55 & 6.30\\
					{} & {3,000} & 0.06 & 0.01 &  & 3.07 & 2.35 &  & 18.15 & 5.15\\
					$\phi_{H}$=$-0.35$, $\alpha_{H}$=$0.85$ &  &  &  &  &  &  &  &  & \\
					{$T= 60$} & {\ 100} & 3.41 & 39.71 &  & 29.84 & 1913.29 &  & 17.65 & 3.15\\
					{} & {\ 500} & 4.52 & -1.03 &  & 25.21 & 30.79 &  & 44.70 & 4.95\\
					{} & {1,000} & 4.47 & -0.47 &  & 25.14 & 19.34 &  & 57.65 & 5.85\\
					{} & {3,000} & 4.59 & -0.21 &  & 25.88 & 12.19 &  & 75.20 & 5.45\\
					&  &  &  &  &  &  &  &  & \\
					{$T=120$} & {\ 100} & 2.68 & -0.54 &  & 20.19 & 28.91 &  & 11.20 & 5.05\\
					{} & {\ 500} & 3.38 & 0.31 &  & 13.97 & 13.39 &  & 26.60 & 5.35\\
					{} & {1,000} & 3.36 & 0.14 &  & 13.38 & 9.90 &  & 40.55 & 4.90\\
					{} & {3,000} & 3.33 & -0.18 &  & 13.64 & 6.52 &  & 61.95 & 5.20\\
					&  &  &  &  &  &  &  &  & \\
					{$T=240$} & {\ 100} & 1.77 & -0.40 &  & 13.95 & 17.05 &  & 7.45 & 6.00\\
					{} & {\ 500} & 2.38 & 0.12 &  & 8.01 & 7.92 &  & 14.30 & 4.70\\
					{} & {1,000} & 2.58 & 0.20 &  & 7.25 & 5.99 &  & 23.15 & 4.75\\
					{} & {3,000} & 2.54 & -0.13 &  & 6.81 & 3.84 &  & 42.85 & 4.95\\
					$\phi_{S}$=$0.16$, $\alpha_{S}$=$0.65$ &  &  &  &  &  &  &  &  & \\
					{$T= 60$} & {\ 100} & -20.01 & -141.72 &  & 37.89 & 5932.49 &  & 25.40 &
					3.15\\
					{} & {\ 500} & -24.52 & 3.22 &  & 37.63 & 77.87 &  & 57.05 & 4.45\\
					{} & {1,000} & -26.14 & 0.57 &  & 39.14 & 34.26 &  & 69.30 & 5.10\\
					{} & {3,000} & -28.54 & 0.84 &  & 42.33 & 26.81 &  & 81.00 & 5.85\\
					&  &  &  &  &  &  &  &  & \\
					{$T=120$} & {\ 100} & -13.01 & 2.45 &  & 26.01 & 36.96 &  & 16.80 & 5.60\\
					{} & {\ 500} & -17.52 & 0.34 &  & 24.66 & 20.94 &  & 48.35 & 6.55\\
					{} & {1,000} & -19.32 & -0.11 &  & 25.88 & 16.81 &  & 64.05 & 5.80\\
					{} & {3,000} & -22.35 & 0.26 &  & 29.14 & 13.22 &  & 79.70 & 6.20\\
					&  &  &  &  &  &  &  &  & \\
					{$T=240$} & {\ 100} & -8.14 & 0.79 &  & 17.57 & 20.56 &  & 11.15 & 5.85\\
					{} & {\ 500} & -11.42 & 0.12 &  & 15.36 & 11.78 &  & 37.65 & 5.55\\
					{} & {1,000} & -13.01 & 0.12 &  & 16.02 & 9.32 &  & 55.70 & 5.20\\
					{} & {3,000} & -16.10 & 0.20 &  & 18.91 & 7.40 &  & 80.10 & 6.70\\\hline\hline
				\end{tabular}
			}
		\end{center}
		
		{\footnotesize Notes: The DGP for Experiment 8 allows for t(5) distributed
			errors, with GARCH effects, with pricing errors ($\alpha_{\eta}=0.3$), with
			one weak missing factor ($\alpha_{\gamma}=0.5$), and without spatial/block
			error cross dependence. For further details of the experiments, see
			\ref{TabExperiments}.}%
		
\end{table}%

\pagebreak

\renewcommand{\thetable}{S-A-E8a}%

	\begin{table}[H]%

		\caption{Bias, RMSE and size for the two-step and bias-corrected (BC)
			estimators of $\phi$ for Experiment 8 with one strong and two semi-strong
			factors}\label{tab:s-a-e8a}
		
		\begin{center}
			{\footnotesize
				\begin{tabular}
					[c]{rrrrrrrrrr}\hline\hline
					&  & \multicolumn{2}{c}{Bias(x100)} &  & \multicolumn{2}{c}{RMSE(x100)} &  &
					\multicolumn{2}{c}{Size(x100)}\\\cline{3-4}\cline{6-7}\cline{9-10}%
					$\phi_{M}$=$-0.49$, $\alpha_{M}$=$1$ & n & Two Step & BC &  & Two Step & BC &
					& Two Step & BC\\\cline{2-4}\cline{6-7}\cline{9-10}%
					{$T=60$} & { 100} & -0.16 & 12.00 &  & 27.46 & 356.79 &  & 12.85 & 4.05\\
					& { 500} & -0.77 & -0.71 &  & 17.89 & 16.74 &  & 29.55 & 6.50\\
					& {1,000} & -0.66 & -0.02 &  & 16.09 & 10.45 &  & 40.95 & 5.70\\
					& {3,000} & -0.86 & 0.10 &  & 14.91 & 5.99 &  & 59.85 & 6.20\\
					&  &  &  &  &  &  &  &  & \\
					{$T=120$} & { 100} & 0.19 & -0.16 &  & 18.02 & 21.26 &  & 8.25 & 5.30\\
					& { 500} & -0.09 & -0.08 &  & 9.66 & 9.09 &  & 14.80 & 6.00\\
					& {1,000} & -0.01 & 0.16 &  & 7.84 & 6.51 &  & 20.95 & 6.70\\
					& {3,000} & -0.25 & 0.12 &  & 6.40 & 3.65 &  & 38.10 & 6.35\\
					&  &  &  &  &  &  &  &  & \\
					{$T=240$} & { 100} & 0.27 & -0.10 &  & 13.79 & 15.14 &  & 8.85 & 8.70\\
					& { 500} & 0.19 & -0.03 &  & 6.34 & 6.36 &  & 11.05 & 7.05\\
					& {1,000} & 0.14 & -0.03 &  & 4.63 & 4.40 &  & 12.05 & 7.50\\
					& {3,000} & 0.05 & -0.00 &  & 3.09 & 2.39 &  & 18.85 & 5.70\\
					$\phi_{H}$=$-0.35$, $\alpha_{H}$=$0.85$ &  &  &  &  &  &  &  &  & \\
					{$T=60$} & { 100} & 3.10 & 55.07 &  & 29.92 & 2828.62 &  & 18.85 & 3.10\\
					& { 500} & 4.53 & -0.98 &  & 25.32 & 30.87 &  & 45.10 & 5.45\\
					& {1,000} & 4.46 & -0.51 &  & 25.19 & 19.42 &  & 58.20 & 5.45\\
					& {3,000} & 4.59 & -0.21 &  & 25.88 & 12.26 &  & 75.20 & 5.75\\
					&  &  &  &  &  &  &  &  & \\
					{$T=120$} & { 100} & 2.39 & -0.98 &  & 20.59 & 29.45 &  & 11.90 & 5.85\\
					& { 500} & 3.39 & 0.32 &  & 14.17 & 13.75 &  & 27.35 & 5.45\\
					& {1,000} & 3.34 & 0.10 &  & 13.45 & 10.09 &  & 40.55 & 5.00\\
					& {3,000} & 3.33 & -0.18 &  & 13.63 & 6.58 &  & 62.30 & 5.20\\
					&  &  &  &  &  &  &  &  & \\
					{$T=240$} & { 100} & 1.52 & -0.70 &  & 14.82 & 18.29 &  & 10.60 & 8.55\\
					& { 500} & 2.38 & 0.12 &  & 8.27 & 8.27 &  & 15.65 & 6.05\\
					& {1,000} & 2.55 & 0.15 &  & 7.35 & 6.23 &  & 23.20 & 5.65\\
					& {3,000} & 2.54 & -0.13 &  & 6.83 & 3.90 &  & 42.75 & 5.95\\
					$\phi_{S}$=$0.16$, $\alpha_{S}$=$0.65$ &  &  &  &  &  &  &  &  & \\
					{$T=60$} & { 100} & -19.86 & -210.04 &  & 38.24 & 8784.76 &  & 26.15 & 3.40\\
					& { 500} & -24.49 & 3.17 &  & 37.68 & 75.99 &  & 56.80 & 4.95\\
					& {1,000} & -26.11 & 0.71 &  & 39.11 & 34.49 &  & 69.05 & 4.95\\
					& {3,000} & -28.55 & 0.78 &  & 42.34 & 26.88 &  & 81.20 & 6.05\\
					&  &  &  &  &  &  &  &  & \\
					{$T=120$} & { 100} & -12.88 & 2.68 &  & 26.48 & 38.62 &  & 17.60 & 5.85\\
					& { 500} & -17.41 & 0.56 &  & 24.62 & 21.19 &  & 47.90 & 6.25\\
					& {1,000} & -19.28 & -0.02 &  & 25.86 & 17.01 &  & 63.85 & 6.00\\
					& {3,000} & -22.36 & 0.23 &  & 29.18 & 13.36 &  & 79.65 & 6.75\\
					&  &  &  &  &  &  &  &  & \\
					{$T=240$} & { 100} & -8.03 & 0.96 &  & 18.42 & 22.23 &  & 13.25 & 7.35\\
					& { 500} & -11.34 & 0.24 &  & 15.47 & 12.20 &  & 36.80 & 6.50\\
					& {1,000} & -12.95 & 0.21 &  & 16.03 & 9.63 &  & 56.70 & 6.20\\
					& {3,000} & -16.10 & 0.19 &  & 18.96 & 7.59 &  & 80.25 & 7.25\\\hline\hline
				\end{tabular}
			}
		\end{center}
		
		{\footnotesize Notes: The DGP for Experiment 8 allows for t(5) distributed
			errors, with GARCH effects, with pricing errors ($\alpha_{\eta}=0.5$), with
			one weak missing factor ($\alpha_{\gamma}=0.5$), and without spatial/block
			error cross dependence. For further details of the experiments, see
			\ref{TabExperiments}.}%
		
\end{table}%

\pagebreak

{\footnotesize
	\begin{figure}[h]%
		\centering
		\caption{Empirical Power Functions, experiment 8 for coefficient of the
			semi-strong factors}%
		\includegraphics[
		height=5.8608in,
		width=7.1122in
		]%
		{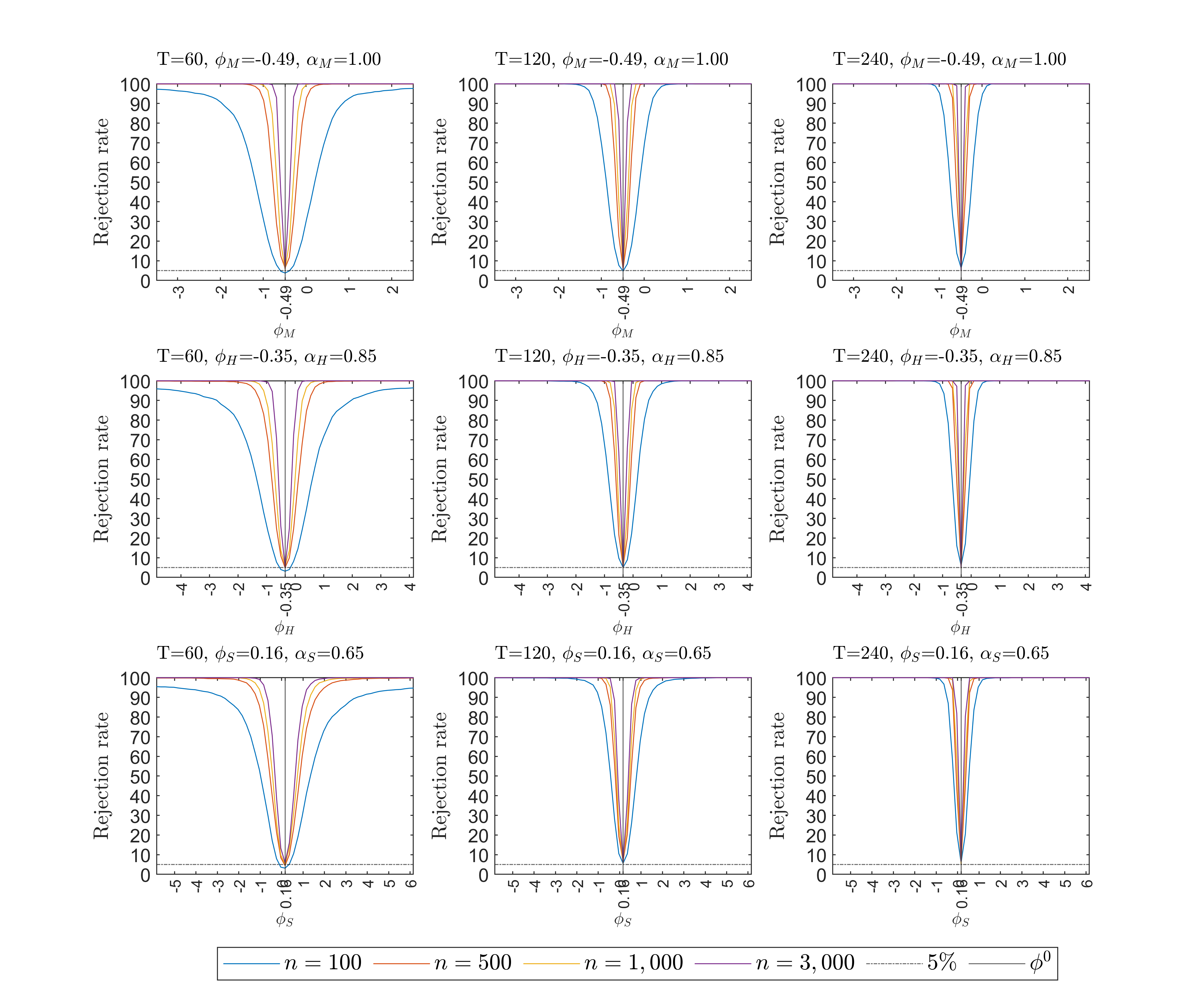}%
		\label{fig:s-a-e8}%
	\end{figure}
	{Note: } See the notes to{ Table \ref{tab:s-a-e8}.} }

{\footnotesize \pagebreak}

\renewcommand{\thefigure}{S-A-E8a}

{\footnotesize
	\begin{figure}[h]%
		\centering
		\caption{Empirical Power Functions, experiment 8a for coefficient of the
			semi-strong factors}%
		\includegraphics[
		height=5.8608in,
		width=7.1122in
		]%
		{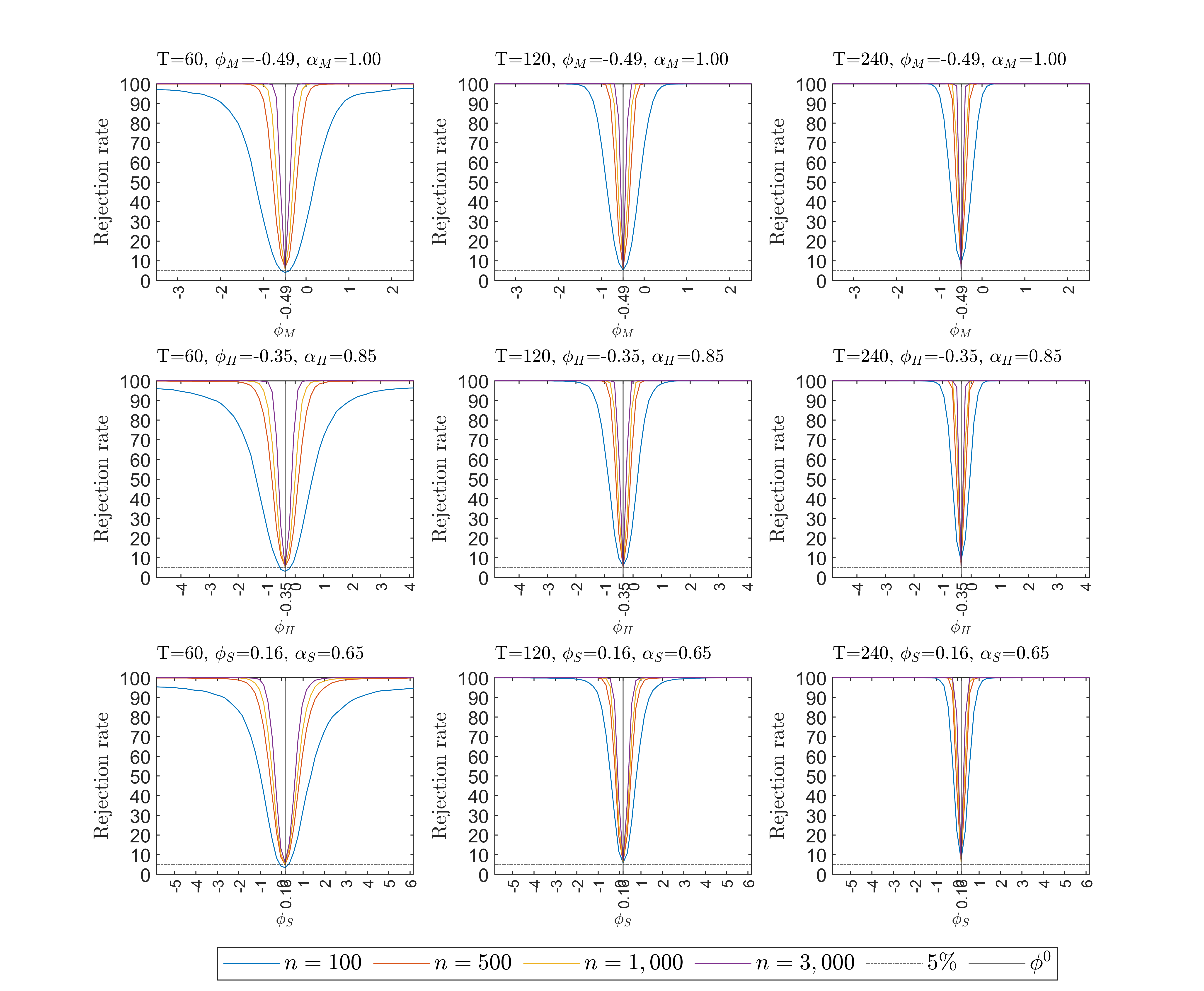}%
		\label{fig:s-a-e8a}%
	\end{figure}
	{Note: } See the notes to{ Table \ref{tab:s-a-e8a}.} }

{\footnotesize \pagebreak}

\renewcommand{\thetable}{S-A-E9}%

	\begin{table}[H]%

		\caption{Bias, RMSE and size for the two-step and bias-corrected (BC)
			estimators of $\phi$, for Experiment 9 with one strong and two semi-strong
			factors}\label{tab:s-a-e9}
		
		\begin{center}
			{\footnotesize
				\begin{tabular}
					[c]{rrrrrrrrrr}\hline\hline
					&  & \multicolumn{2}{c}{Bias(x100)} &  & \multicolumn{2}{c}{RMSE(x100)} &  &
					\multicolumn{2}{c}{Size(x100)}\\\cline{3-4}\cline{6-7}\cline{9-10}%
					$\phi_{M}$=$-0.49$, $\alpha_{M}$=$1$ & n & {Two Step} & BC &  & {Two Step} &
					BC &  & {Two Step} & BC\\\cline{2-4}\cline{6-7}\cline{9-10}%
					{$T= 60$} & {\ 100} & -0.51 & -15.47 &  & 27.16 & 662.68 &  & 11.80 & 3.40\\
					{} & {\ 500} & -0.55 & 1.29 &  & 18.02 & 84.39 &  & 31.25 & 7.75\\
					{} & {1,000} & -0.85 & -0.47 &  & 16.24 & 11.31 &  & 43.00 & 8.50\\
					{} & {3,000} & -0.99 & -0.16 &  & 14.97 & 6.25 &  & 61.30 & 7.75\\
					&  &  &  &  &  &  &  &  & \\
					{$T=120$} & {\ 100} & -0.17 & -0.59 &  & 17.92 & 21.49 &  & 7.15 & 4.50\\
					{} & {\ 500} & 0.03 & -0.01 &  & 9.70 & 9.20 &  & 13.95 & 6.15\\
					{} & {1,000} & -0.17 & -0.07 &  & 7.88 & 6.53 &  & 20.35 & 6.30\\
					{} & {3,000} & -0.33 & -0.01 &  & 6.45 & 3.71 &  & 36.90 & 6.50\\
					&  &  &  &  &  &  &  &  & \\
					{$T=240$} & {\ 100} & -0.12 & -0.54 &  & 12.73 & 14.02 &  & 6.80 & 6.85\\
					{} & {\ 500} & 0.13 & -0.12 &  & 6.18 & 6.19 &  & 9.65 & 6.55\\
					{} & {1,000} & 0.11 & -0.08 &  & 4.57 & 4.30 &  & 11.75 & 6.60\\
					{} & {3,000} & 0.06 & -0.00 &  & 3.09 & 2.37 &  & 17.90 & 5.45\\
					$\phi_{H}$=$-0.35$, $\alpha_{H}$=$0.85$ &  &  &  &  &  &  &  &  & \\
					{$T= 60$} & {\ 100} & 2.80 & 3.36 &  & 30.96 & 683.86 &  & 19.20 & 3.90\\
					{} & {\ 500} & 4.18 & 2.34 &  & 25.47 & 189.26 &  & 47.10 & 8.65\\
					{} & {1,000} & 4.41 & -0.83 &  & 25.43 & 22.18 &  & 59.30 & 9.35\\
					{} & {3,000} & 4.63 & -0.28 &  & 25.90 & 14.12 &  & 76.45 & 10.60\\
					&  &  &  &  &  &  &  &  & \\
					{$T=120$} & {\ 100} & 2.27 & -1.02 &  & 20.81 & 33.62 &  & 10.25 & 4.50\\
					{} & {\ 500} & 3.04 & -0.39 &  & 14.24 & 14.82 &  & 24.60 & 5.70\\
					{} & {1,000} & 3.17 & -0.24 &  & 13.77 & 11.21 &  & 38.65 & 6.10\\
					{} & {3,000} & 3.24 & -0.40 &  & 13.74 & 7.47 &  & 60.05 & 7.75\\
					&  &  &  &  &  &  &  &  & \\
					{$T=240$} & {\ 100} & 1.62 & -0.61 &  & 14.29 & 17.94 &  & 7.20 & 5.75\\
					{} & {\ 500} & 2.25 & -0.09 &  & 8.26 & 8.60 &  & 13.10 & 5.35\\
					{} & {1,000} & 2.36 & -0.10 &  & 7.47 & 6.57 &  & 21.10 & 5.65\\
					{} & {3,000} & 2.57 & -0.10 &  & 6.99 & 4.27 &  & 41.70 & 5.85\\
					$\phi_{S}$=$0.16$, $\alpha_{S}$=$0.65$ &  &  &  &  &  &  &  &  & \\
					{$T= 60$} & {\ 100} & -19.57 & 38.19 &  & 40.89 & 1634.08 &  & 26.25 & 4.75\\
					{} & {\ 500} & -24.60 & -17.77 &  & 38.37 & 944.75 &  & 59.35 & 10.05\\
					{} & {1,000} & -26.04 & 1.99 &  & 39.30 & 41.96 &  & 70.30 & 11.50\\
					{} & {3,000} & -28.69 & 0.77 &  & 42.55 & 32.04 &  & 82.85 & 11.70\\
					&  &  &  &  &  &  &  &  & \\
					{$T=120$} & {\ 100} & -12.73 & 4.51 &  & 29.32 & 56.35 &  & 14.45 & 5.50\\
					{} & {\ 500} & -17.45 & 0.82 &  & 25.67 & 26.62 &  & 44.50 & 8.70\\
					{} & {1,000} & -19.16 & 0.54 &  & 26.31 & 20.96 &  & 58.65 & 8.40\\
					{} & {3,000} & -22.27 & 0.58 &  & 29.33 & 16.06 &  & 77.55 & 8.95\\
					&  &  &  &  &  &  &  &  & \\
					{$T=240$} & {\ 100} & -8.13 & 1.07 &  & 21.26 & 26.72 &  & 10.60 & 6.50\\
					{} & {\ 500} & -11.34 & 0.35 &  & 16.58 & 15.28 &  & 31.25 & 7.40\\
					{} & {1,000} & -12.81 & 0.50 &  & 16.52 & 11.98 &  & 46.35 & 6.40\\
					{} & {3,000} & -16.03 & 0.35 &  & 19.08 & 9.09 &  & 75.35 & 7.25\\\hline\hline
				\end{tabular}
			}
		\end{center}
		
		{\footnotesize Notes: The DGP for Experiment 9 allows for Gaussian errors,
			with GARCH effects, with pricing errors ($\alpha_{\eta}=0.3$), with one weak
			missing factor ($\alpha_{\gamma}=0.5$), and with spatial error cross
			dependence ($\rho_{\varepsilon}=0.5$). For further details of the experiments,
			see Table \ref{TabExperiments}. }
\end{table}%

{\footnotesize \pagebreak}

\renewcommand{\thetable}{S-A-E9a}%

	\begin{table}[H]%

		\caption{Bias, RMSE and size for the two-step and bias-corrected (BC)
			estimators of $\phi$ for Experiment 9 with one strong and two semi-strong
			factors}\label{tab:s-a-e9a}
		
		\begin{center}
			{\footnotesize
				\begin{tabular}
					[c]{rrrrrrrrrr}\hline\hline
					&  & \multicolumn{2}{c}{Bias(x100)} &  & \multicolumn{2}{c}{RMSE(x100)} &  &
					\multicolumn{2}{c}{Size(x100)}\\\cline{3-4}\cline{6-7}\cline{9-10}%
					$\phi_{M}$=$-0.49$, $\alpha_{M}$=$1$ & n & Two Step & BC &  & Two Step & BC &
					& Two Step & BC\\\cline{2-4}\cline{6-7}\cline{9-10}%
					{$T=60$} & { 100} & -0.48 & 24.91 &  & 29.91 & 1802.71 &  & 10.25 & 3.05\\
					& { 500} & -0.58 & -1.22 &  & 19.37 & 92.46 &  & 26.95 & 6.75\\
					& {1,000} & -0.92 & -0.95 &  & 17.00 & 17.30 &  & 37.20 & 8.60\\
					& {3,000} & -1.03 & -0.35 &  & 15.24 & 8.41 &  & 55.50 & 8.65\\
					&  &  &  &  &  &  &  &  & \\
					{$T=120$} & { 100} & -0.02 & -1.97 &  & 19.95 & 35.08 &  & 8.60 & 6.20\\
					& { 500} & 0.03 & -0.10 &  & 10.58 & 10.55 &  & 13.75 & 6.10\\
					& {1,000} & -0.24 & -0.22 &  & 8.38 & 7.44 &  & 17.85 & 6.55\\
					& {3,000} & -0.40 & -0.12 &  & 6.67 & 4.26 &  & 33.75 & 6.90\\
					&  &  &  &  &  &  &  &  & \\
					{$T=240$} & { 100} & 0.02 & -0.44 &  & 13.81 & 15.39 &  & 8.65 & 7.40\\
					& { 500} & 0.12 & -0.14 &  & 6.62 & 6.70 &  & 10.25 & 6.35\\
					& {1,000} & 0.06 & -0.15 &  & 4.86 & 4.72 &  & 10.85 & 7.45\\
					& {3,000} & 0.04 & -0.03 &  & 3.26 & 2.62 &  & 16.95 & 6.45\\
					$\phi_{H}$=$-0.35$, $\alpha_{H}$=$0.85$ &  &  &  &  &  &  &  &  & \\
					{$T=60$} & { 100} & 1.63 & 29.05 &  & 36.11 & 1916.43 &  & 14.35 & 1.95\\
					& { 500} & 3.76 & 1.48 &  & 27.36 & 303.36 &  & 36.10 & 7.25\\
					& {1,000} & 4.26 & -1.84 &  & 26.21 & 39.83 &  & 48.45 & 8.20\\
					& {3,000} & 4.56 & -0.62 &  & 26.15 & 22.78 &  & 69.25 & 11.05\\
					&  &  &  &  &  &  &  &  & \\
					{$T=120$} & { 100} & 1.79 & -3.12 &  & 24.84 & 50.47 &  & 9.30 & 4.50\\
					& { 500} & 2.85 & -0.87 &  & 16.01 & 19.66 &  & 18.60 & 5.60\\
					& {1,000} & 3.03 & -0.57 &  & 14.77 & 15.08 &  & 28.80 & 6.90\\
					& {3,000} & 3.13 & -0.68 &  & 14.02 & 10.07 &  & 49.45 & 8.40\\
					&  &  &  &  &  &  &  &  & \\
					{$T=240$} & { 100} & 1.21 & -1.21 &  & 17.34 & 22.52 &  & 8.30 & 6.75\\
					& { 500} & 2.22 & -0.15 &  & 9.58 & 10.83 &  & 11.00 & 5.95\\
					& {1,000} & 2.27 & -0.24 &  & 8.46 & 8.56 &  & 18.20 & 6.15\\
					& {3,000} & 2.60 & -0.08 &  & 7.42 & 5.55 &  & 33.25 & 7.10\\
					$\phi_{S}$=$0.16$, $\alpha_{S}$=$0.65$ &  &  &  &  &  &  &  &  & \\
					{$T=60$} & { 100} & -18.49 & 110.60 &  & 50.00 & 7440.55 &  & 17.75 & 2.35\\
					& { 500} & -24.61 & 18.11 &  & 41.10 & 1183.88 &  & 44.00 & 6.75\\
					& {1,000} & -26.08 & 6.93 &  & 40.81 & 118.48 &  & 58.05 & 8.75\\
					& {3,000} & -28.71 & 2.78 &  & 43.19 & 66.65 &  & 75.45 & 11.30\\
					&  &  &  &  &  &  &  &  & \\
					{$T=120$} & { 100} & -12.19 & 13.59 &  & 37.47 & 192.49 &  & 10.25 & 4.45\\
					& { 500} & -17.28 & 2.13 &  & 28.78 & 40.29 &  & 30.50 & 8.05\\
					& {1,000} & -19.06 & 1.36 &  & 27.98 & 31.28 &  & 44.05 & 8.90\\
					& {3,000} & -22.25 & 1.02 &  & 30.06 & 23.69 &  & 66.95 & 8.90\\
					&  &  &  &  &  &  &  &  & \\
					{$T=240$} & { 100} & -8.15 & 1.74 &  & 28.35 & 38.76 &  & 8.85 & 5.65\\
					& { 500} & -11.28 & 0.68 &  & 19.74 & 22.51 &  & 19.25 & 7.65\\
					& {1,000} & -12.72 & 0.82 &  & 18.22 & 17.60 &  & 31.20 & 7.15\\
					& {3,000} & -15.99 & 0.51 &  & 19.76 & 13.26 &  & 62.25 & 8.00\\\hline\hline
				\end{tabular}
			}
		\end{center}
		
		{\footnotesize Notes: The DGP for Experiment 9 allows for Gaussian errors,
			with GARCH effects, with pricing errors ($\alpha_{\eta}=0.5$), with one weak
			missing factor ($\alpha_{\gamma}=0.5$), and with spatial error cross
			dependence ($\rho_{\varepsilon}=0.85$). For further details of the
			experiments, see Table \ref{TabExperiments}. }
\end{table}%

{\footnotesize \pagebreak}

\renewcommand{\thefigure}{S-A-E9} {\footnotesize
	\begin{figure}[h]%
		\centering
		\caption{Empirical Power Functions, experiment 9 for coefficient of the
			semi-strong factors}%
		\includegraphics[
		height=5.8608in,
		width=7.1122in
		]%
		{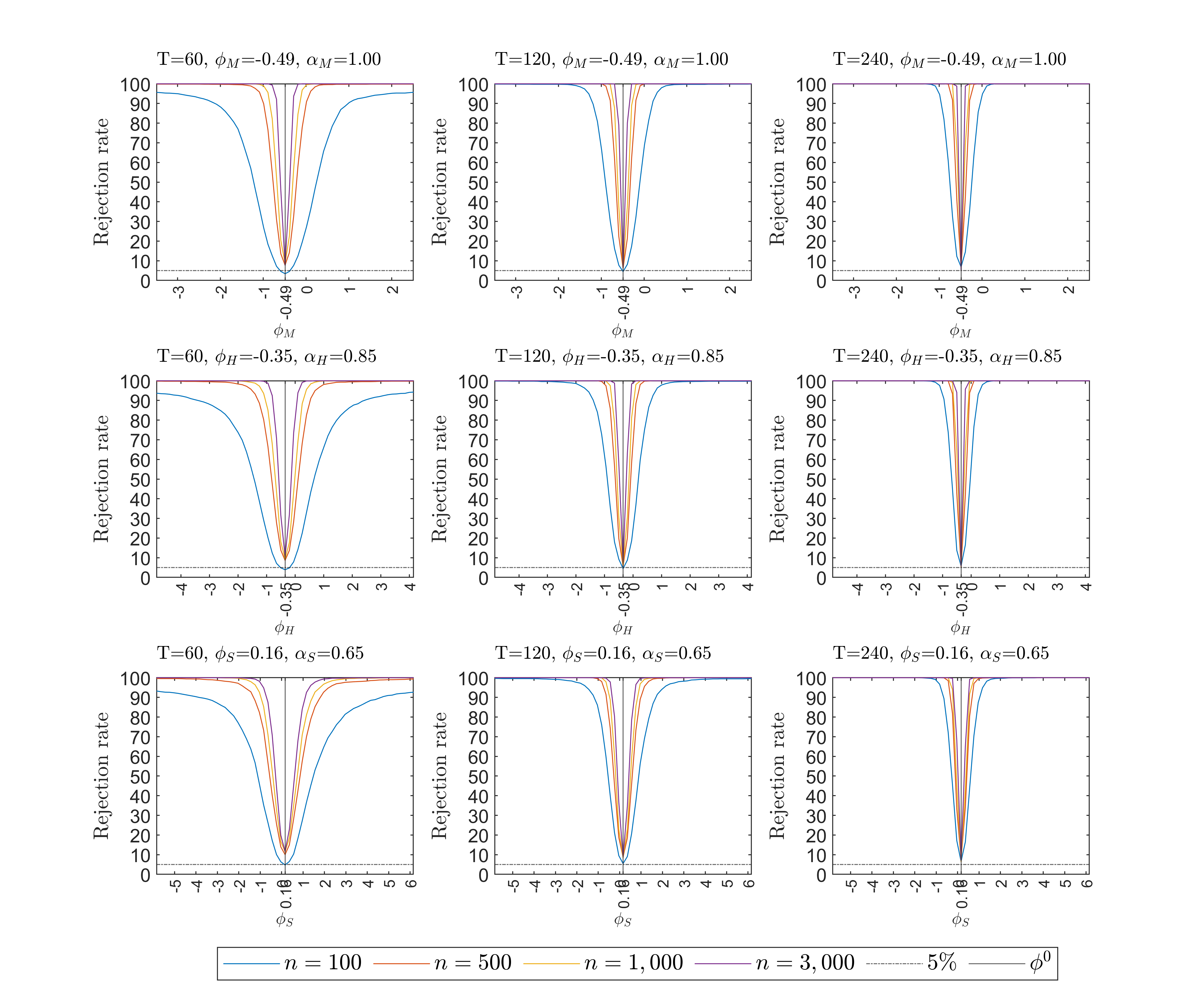}%
		\label{fig:s-a-e9}%
	\end{figure}
	{Note: } See the notes to{ Table \ref{tab:s-a-e9}.} }

{\footnotesize \pagebreak}

\renewcommand{\thefigure}{S-A-E9a}{\footnotesize
	\begin{figure}[h]%
		\centering
		\caption{Empirical Power Functions, experiment 9a for coefficient of the
			semi-strong factors}%
		\includegraphics[
		height=5.8608in,
		width=7.1122in
		]%
		{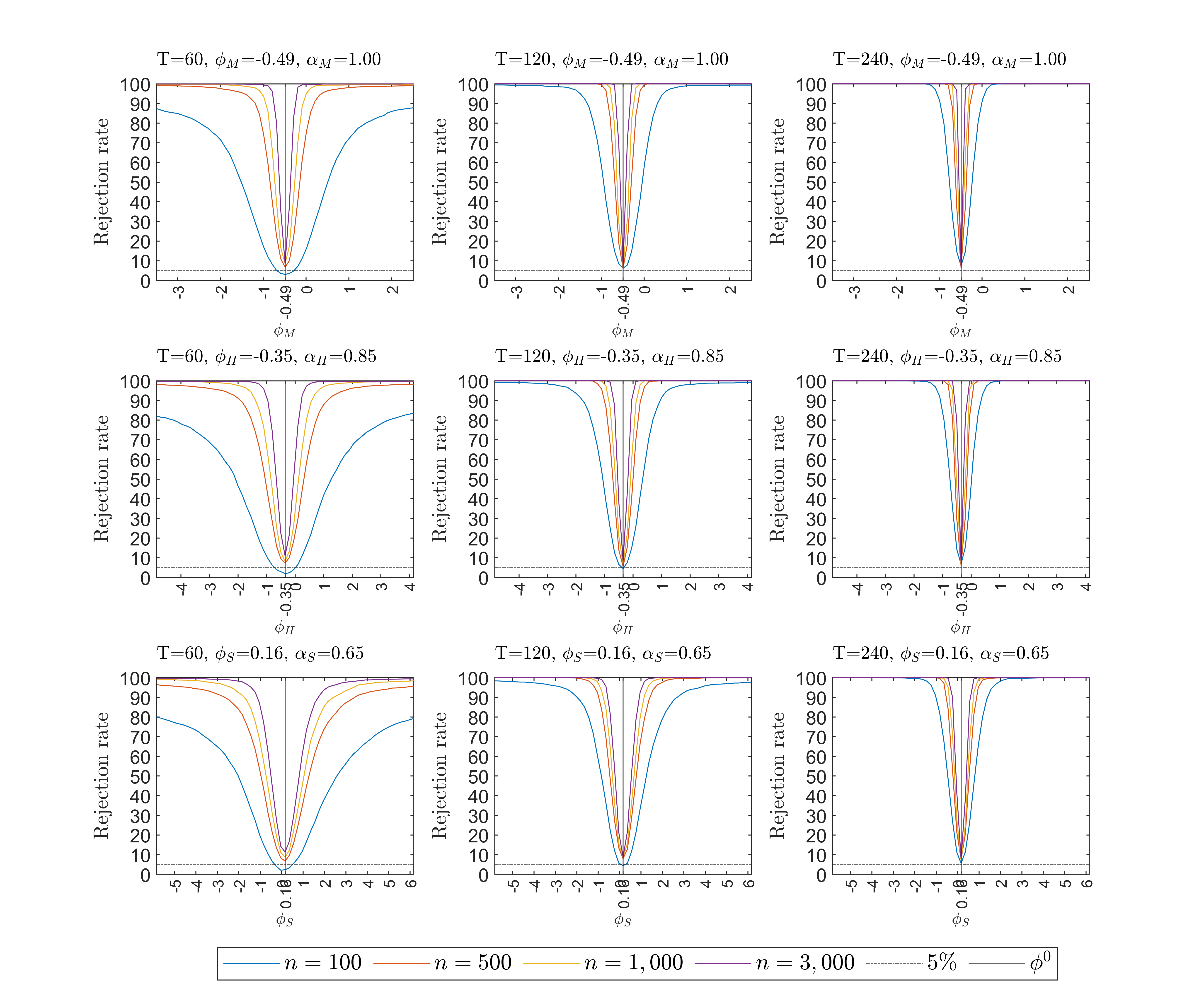}%
		\label{s-a-e9a}%
	\end{figure}
	{Note: } See the notes to{ Table \ref{tab:s-a-e9a}.} }

{\footnotesize \pagebreak}

\renewcommand{\thetable}{S-A-E10}%

	\begin{table}[H]%

		\caption{Bias, RMSE and size for the two-step and bias-corrected (BC)
			estimators of $\phi$, for Experiment 10 with one strong and two semi-strong
			factors}\label{tab:s-a-e10}
		
		\begin{center}
			{\footnotesize
				\begin{tabular}
					[c]{rrrrrrrrrr}\hline\hline
					&  & \multicolumn{2}{c}{Bias(x100)} &  & \multicolumn{2}{c}{RMSE(x100)} &  &
					\multicolumn{2}{c}{Size(x100)}\\\cline{3-4}\cline{6-7}\cline{9-10}%
					$\phi_{M}$=$-0.49$, $\alpha_{M}$=$1$ & n & {Two Step} & BC &  & {Two Step} &
					BC &  & {Two Step} & BC\\\cline{2-4}\cline{6-7}\cline{9-10}%
					{$T= 60$} & {\ 100} & -0.39 & -1.61 &  & 27.55 & 151.91 &  & 11.70 & 3.75\\
					{} & {\ 500} & -0.82 & -0.69 &  & 18.11 & 17.07 &  & 30.40 & 7.70\\
					{} & {1,000} & -0.73 & -0.31 &  & 16.33 & 11.30 &  & 42.50 & 7.05\\
					{} & {3,000} & -0.92 & -0.06 &  & 14.97 & 6.36 &  & 60.05 & 8.00\\
					&  &  &  &  &  &  &  &  & \\
					{$T=120$} & {\ 100} & 0.02 & -1.86 &  & 18.04 & 69.05 &  & 7.65 & 4.80\\
					{} & {\ 500} & -0.07 & -0.09 &  & 9.78 & 9.29 &  & 15.05 & 6.10\\
					{} & {1,000} & 0.01 & 0.15 &  & 7.90 & 6.62 &  & 19.70 & 6.75\\
					{} & {3,000} & -0.29 & 0.05 &  & 6.42 & 3.72 &  & 37.55 & 5.80\\
					&  &  &  &  &  &  &  &  & \\
					{$T=240$} & {\ 100} & 0.01 & -0.39 &  & 12.87 & 14.17 &  & 6.30 & 6.35\\
					{} & {\ 500} & 0.17 & -0.05 &  & 6.13 & 6.15 &  & 9.25 & 6.05\\
					{} & {1,000} & 0.17 & 0.00 &  & 4.55 & 4.30 &  & 11.30 & 6.40\\
					{} & {3,000} & 0.06 & 0.00 &  & 3.12 & 2.40 &  & 17.65 & 5.75\\
					$\phi_{H}$=$-0.35$, $\alpha_{H}$=$0.85$ &  &  &  &  &  &  &  &  & \\
					{$T= 60$} & {\ 100} & 3.25 & -8.09 &  & 31.16 & 504.30 &  & 19.35 & 3.65\\
					{} & {\ 500} & 4.29 & -1.32 &  & 25.64 & 33.61 &  & 46.05 & 8.10\\
					{} & {1,000} & 4.50 & -0.68 &  & 25.36 & 22.95 &  & 58.25 & 9.70\\
					{} & {3,000} & 4.57 & -0.37 &  & 25.89 & 14.12 &  & 76.05 & 9.50\\
					&  &  &  &  &  &  &  &  & \\
					{$T=120$} & {\ 100} & 2.48 & -1.13 &  & 21.16 & 32.42 &  & 11.15 & 4.50\\
					{} & {\ 500} & 3.15 & -0.09 &  & 14.31 & 14.78 &  & 25.55 & 5.70\\
					{} & {1,000} & 3.25 & -0.07 &  & 13.66 & 11.19 &  & 38.65 & 5.55\\
					{} & {3,000} & 3.31 & -0.26 &  & 13.70 & 7.36 &  & 60.50 & 6.70\\
					&  &  &  &  &  &  &  &  & \\
					{$T=240$} & {\ 100} & 1.78 & -0.40 &  & 14.58 & 18.18 &  & 8.25 & 6.05\\
					{} & {\ 500} & 2.35 & 0.08 &  & 8.31 & 8.60 &  & 13.35 & 5.20\\
					{} & {1,000} & 2.53 & 0.12 &  & 7.51 & 6.58 &  & 21.05 & 5.00\\
					{} & {3,000} & 2.57 & -0.10 &  & 6.96 & 4.28 &  & 41.35 & 6.00\\
					$\phi_{S}$=$0.16$, $\alpha_{S}$=$0.65$ &  &  &  &  &  &  &  &  & \\
					{$T= 60$} & {\ 100} & -19.35 & 25.35 &  & 40.86 & 579.34 &  & 25.25 & 4.60\\
					{} & {\ 500} & -24.56 & 2.56 &  & 38.40 & 68.65 &  & 57.50 & 9.85\\
					{} & {1,000} & -26.17 & 1.78 &  & 39.59 & 43.66 &  & 69.65 & 10.70\\
					{} & {3,000} & -28.53 & 1.55 &  & 42.55 & 32.38 &  & 81.95 & 11.90\\
					&  &  &  &  &  &  &  &  & \\
					{$T=120$} & {\ 100} & -12.71 & 10.19 &  & 29.68 & 304.30 &  & 14.70 & 6.05\\
					{} & {\ 500} & -17.42 & 0.86 &  & 25.76 & 26.78 &  & 43.95 & 8.75\\
					{} & {1,000} & -19.30 & 0.15 &  & 26.41 & 20.74 &  & 59.65 & 8.50\\
					{} & {3,000} & -22.31 & 0.47 &  & 29.35 & 16.12 &  & 77.35 & 8.95\\
					&  &  &  &  &  &  &  &  & \\
					{$T=240$} & {\ 100} & -8.09 & 1.05 &  & 21.11 & 26.50 &  & 10.55 & 6.65\\
					{} & {\ 500} & -11.52 & 0.08 &  & 16.67 & 15.20 &  & 30.35 & 7.15\\
					{} & {1,000} & -12.98 & 0.22 &  & 16.56 & 11.84 &  & 47.50 & 6.80\\
					{} & {3,000} & -16.05 & 0.32 &  & 19.11 & 9.25 &  & 75.50 & 8.00\\\hline\hline
				\end{tabular}
			}
		\end{center}
		
		{\footnotesize Notes: The DGP for Experiment 10 allows for t(5) distributed
			errors, with GARCH effects, with pricing errors ($\alpha_{\eta}=0.3$), with
			one weak missing factor ($\alpha_{\gamma}=0.5$), and with spatial error cross
			dependence ($\rho_{\varepsilon}=0.5$). For further details of the experiments,
			see Table \ref{TabExperiments}.}
\end{table}%

{\footnotesize \pagebreak}

\renewcommand{\thefigure}{S-A-E10}{\footnotesize
	\begin{figure}[h]%
		\centering
		\caption{Empirical Power Functions, experiment 10 for coefficient of the
			semi-strong factors}%
		\includegraphics[
		height=5.8608in,
		width=7.1122in
		]%
		{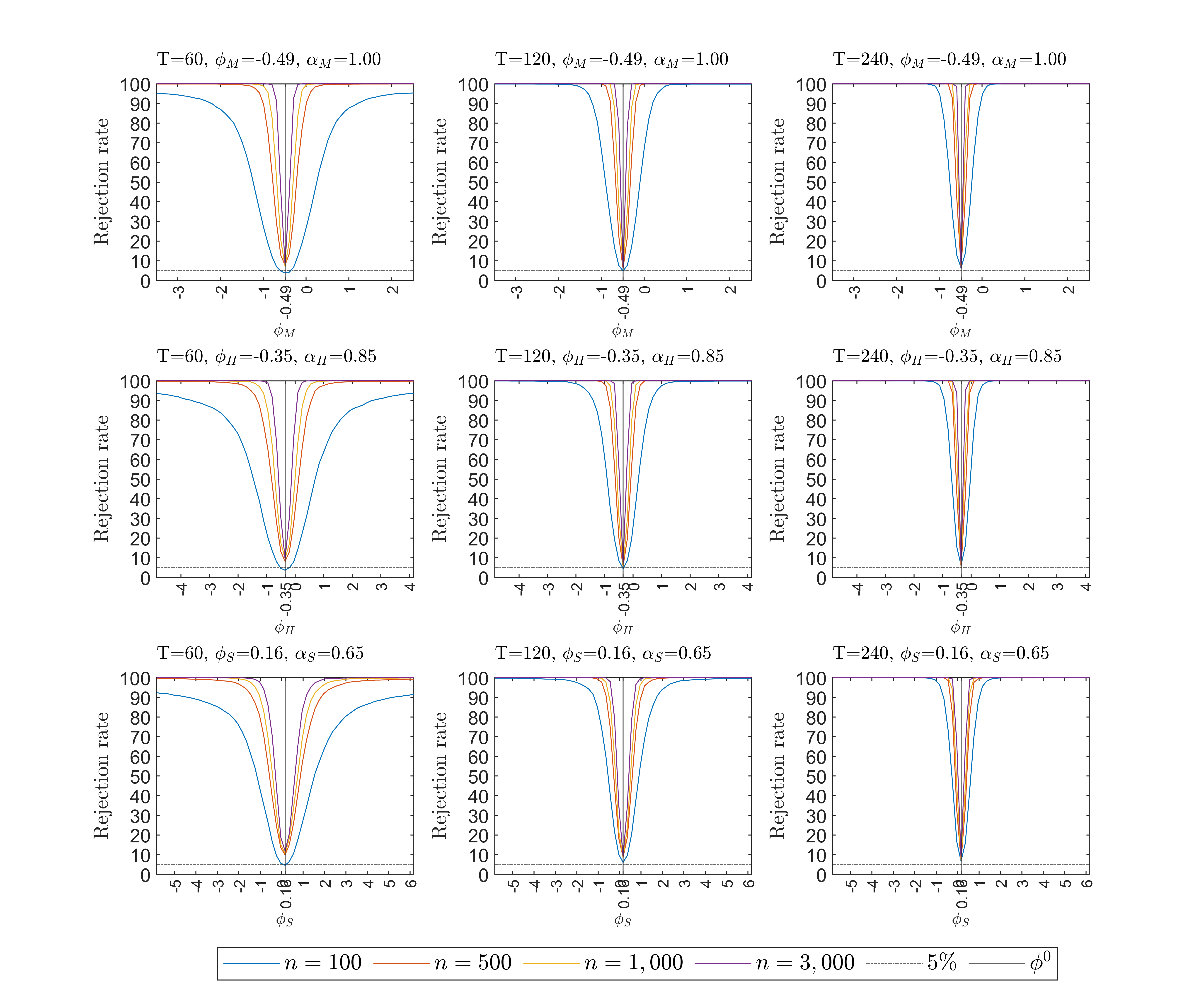}%
		\label{fig:s-a-e10}%
	\end{figure}
	{Note: } See the notes to{ Table \ref{tab:s-a-e10}.} }

{\footnotesize \pagebreak}

\renewcommand{\thetable}{S-A-E11}%

	\begin{table}[H]%

		\caption{Bias, RMSE and size for the two-step and bias-corrected (BC)
			estimators of $\phi$, for Experiment 11 with one strong and two semi-strong
			factors }\label{tab:s-a-e11}
		
		\begin{center}
			{\footnotesize
				\begin{tabular}
					[c]{rrrrrrrrrr}\hline\hline
					&  & \multicolumn{2}{c}{Bias(x100)} &  & \multicolumn{2}{c}{RMSE(x100)} &  &
					\multicolumn{2}{c}{Size(x100)}\\\cline{3-4}\cline{6-7}\cline{9-10}%
					$\phi_{M}$=$-0.49$, $\alpha_{M}$=$1$ & n & {Two Step} & BC &  & {Two Step} &
					BC &  & {Two Step} & BC\\\cline{2-4}\cline{6-7}\cline{9-10}%
					{$T= 60$} & {\ 100} & -0.05 & -0.72 &  & 26.31 & 47.59 &  & 11.30 & 3.40\\
					{} & {\ 500} & -0.33 & -0.07 &  & 17.63 & 15.12 &  & 29.40 & 6.70\\
					{} & {1,000} & -0.71 & -0.22 &  & 16.06 & 10.46 &  & 41.95 & 6.35\\
					{} & {3,000} & -0.85 & 0.05 &  & 14.87 & 5.91 &  & 60.85 & 7.00\\
					&  &  &  &  &  &  &  &  & \\
					{$T=120$} & {\ 100} & 0.10 & -0.46 &  & 17.08 & 19.81 &  & 7.15 & 4.55\\
					{} & {\ 500} & 0.13 & 0.10 &  & 9.51 & 8.93 &  & 13.90 & 6.40\\
					{} & {1,000} & -0.11 & -0.01 &  & 7.80 & 6.28 &  & 20.65 & 6.15\\
					{} & {3,000} & -0.21 & 0.09 &  & 6.41 & 3.60 &  & 37.25 & 6.20\\
					&  &  &  &  &  &  &  &  & \\
					{$T=240$} & {\ 100} & 0.20 & -0.24 &  & 12.33 & 13.43 &  & 6.15 & 5.30\\
					{} & {\ 500} & 0.24 & 0.01 &  & 6.12 & 6.11 &  & 10.05 & 6.75\\
					{} & {1,000} & 0.13 & -0.07 &  & 4.52 & 4.19 &  & 11.75 & 5.85\\
					{} & {3,000} & 0.11 & 0.02 &  & 3.05 & 2.32 &  & 18.35 & 4.75\\
					$\phi_{H}$=$-0.35$, $\alpha_{H}$=$0.85$ &  &  &  &  &  &  &  &  & \\
					{$T= 60$} & {\ 100} & 2.50 & -2.25 &  & 29.66 & 103.01 &  & 18.75 & 3.35\\
					{} & {\ 500} & 3.74 & -0.07 &  & 24.90 & 25.61 &  & 45.55 & 5.45\\
					{} & {1,000} & 3.60 & -0.52 &  & 24.93 & 18.83 &  & 59.20 & 6.05\\
					{} & {3,000} & 3.98 & 0.01 &  & 25.84 & 12.59 &  & 75.45 & 7.90\\
					&  &  &  &  &  &  &  &  & \\
					{$T=120$} & {\ 100} & 1.81 & -0.44 &  & 19.83 & 27.88 &  & 10.00 & 4.95\\
					{} & {\ 500} & 2.49 & 0.09 &  & 13.69 & 13.12 &  & 26.00 & 6.15\\
					{} & {1,000} & 2.45 & 0.01 &  & 13.32 & 10.12 &  & 40.65 & 6.00\\
					{} & {3,000} & 2.63 & 0.01 &  & 13.61 & 6.86 &  & 62.05 & 7.20\\
					&  &  &  &  &  &  &  &  & \\
					{$T=240$} & {\ 100} & 0.97 & -0.57 &  & 13.36 & 16.38 &  & 6.90 & 5.55\\
					{} & {\ 500} & 1.60 & -0.05 &  & 7.76 & 8.01 &  & 13.50 & 5.40\\
					{} & {1,000} & 1.72 & 0.01 &  & 6.95 & 5.93 &  & 20.35 & 5.20\\
					{} & {3,000} & 1.85 & -0.07 &  & 6.58 & 3.89 &  & 41.35 & 5.95\\
					$\phi_{S}$=$0.16$, $\alpha_{S}$=$0.65$ &  &  &  &  &  &  &  &  & \\
					{$T= 60$} & {\ 100} & -20.60 & -0.57 &  & 38.15 & 212.41 &  & 27.80 & 2.75\\
					{} & {\ 500} & -23.77 & 2.07 &  & 36.92 & 42.46 &  & 58.25 & 5.55\\
					{} & {1,000} & -25.88 & 1.00 &  & 39.00 & 34.14 &  & 69.55 & 6.15\\
					{} & {3,000} & -28.31 & 0.81 &  & 41.99 & 28.50 &  & 82.20 & 7.90\\
					&  &  &  &  &  &  &  &  & \\
					{$T=120$} & {\ 100} & -12.85 & 2.34 &  & 26.13 & 35.00 &  & 17.60 & 4.80\\
					{} & {\ 500} & -16.99 & 0.79 &  & 23.96 & 20.38 &  & 47.50 & 5.65\\
					{} & {1,000} & -19.15 & 0.29 &  & 25.77 & 16.81 &  & 64.20 & 5.85\\
					{} & {3,000} & -21.97 & 0.81 &  & 28.73 & 13.50 &  & 80.80 & 7.00\\
					&  &  &  &  &  &  &  &  & \\
					{$T=240$} & {\ 100} & -7.60 & 1.18 &  & 17.18 & 20.04 &  & 10.50 & 5.15\\
					{} & {\ 500} & -10.88 & 0.56 &  & 14.80 & 11.65 &  & 35.40 & 5.55\\
					{} & {1,000} & -12.61 & 0.61 &  & 15.89 & 9.46 &  & 55.65 & 5.80\\
					{} & {3,000} & -15.72 & 0.52 &  & 18.57 & 7.12 &  & 80.50 & 5.35\\\hline\hline
				\end{tabular}
			}
		\end{center}
		
		{\footnotesize Notes: The DGP for Experiment 11 allows for Gaussian errors,
			with GARCH effects, with pricing errors ($\alpha_{\eta}=0.3$), with one weak
			missing factor ($\alpha_{\gamma}=0.5$), and with block error cross dependence.
			For further details of the experiments, see Table \ref{TabExperiments}. }
\end{table}%

{\footnotesize \pagebreak}

\renewcommand{\thefigure}{S-A-E11} {\footnotesize
	\begin{figure}[h]%
		\centering
		\caption{Empirical Power Functions, experiment 11 for coefficient of the
			semi-strong factors}%
		\includegraphics[
		height=5.8608in,
		width=7.1122in
		]%
		{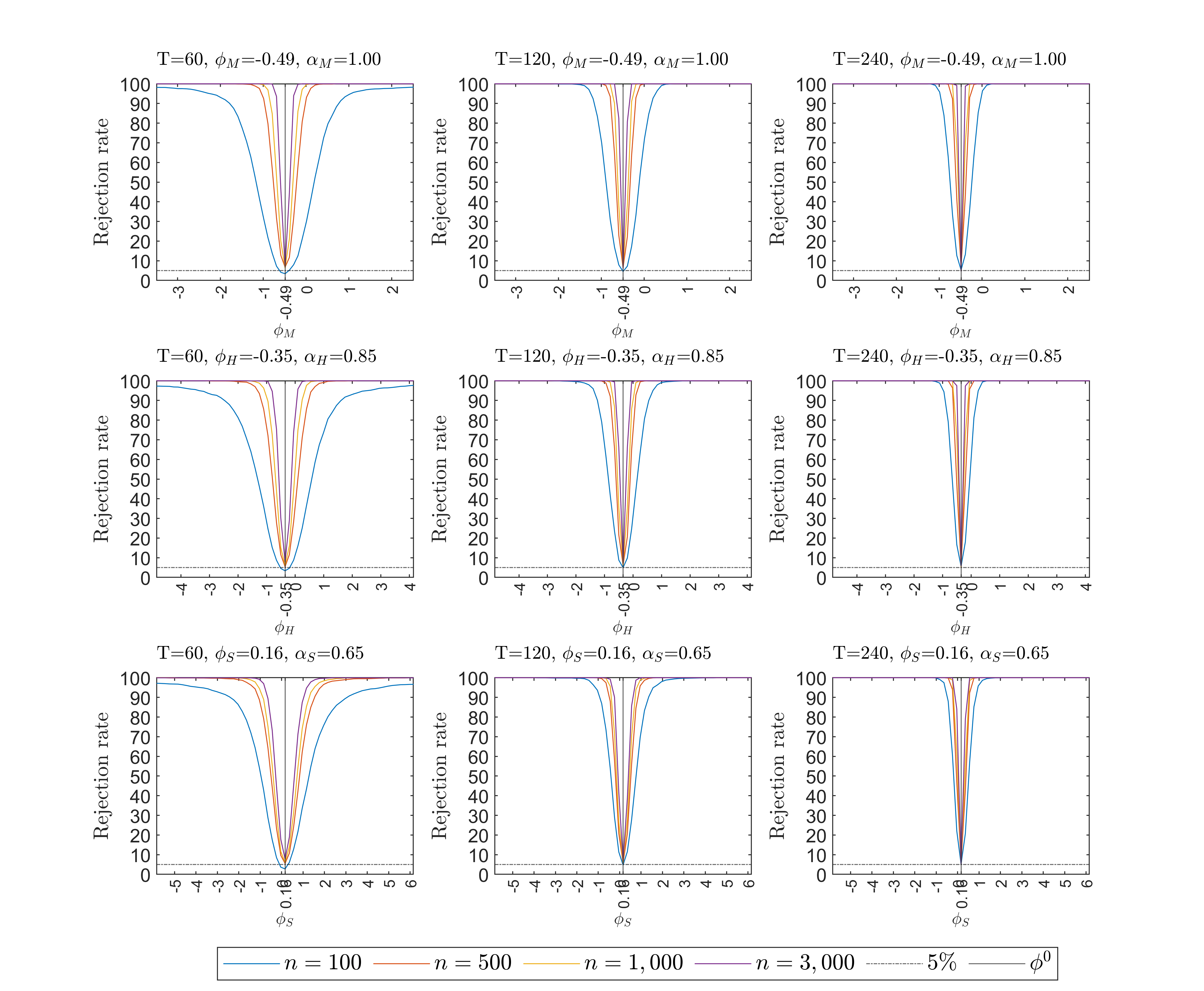}%
		\label{fig:s-a-e11}%
	\end{figure}
	{Note: } See the notes to{ Table \ref{tab:s-a-e11}.} }

{\footnotesize \pagebreak}

\renewcommand{\thetable}{S-A-E12}%

	\begin{table}[H]%

		\caption{Bias, RMSE and size for the two-step and bias-corrected (BC)
			estimators of $\phi$, for Experiment 12 with one strong and two semi-strong
			factors }\label{tab:s-a-e12}
		
		\begin{center}
			{\footnotesize
				\begin{tabular}
					[c]{rrrrrrrrrr}\hline\hline
					&  & \multicolumn{2}{c}{Bias(x100)} &  & \multicolumn{2}{c}{RMSE(x100)} &  &
					\multicolumn{2}{c}{Size(x100)}\\\cline{3-4}\cline{6-7}\cline{9-10}%
					$\phi_{M}$=$-0.49$, $\alpha_{M}$=$1$ & n & {Two Step} & BC &  & {Two Step} &
					BC &  & {Two Step} & BC\\\cline{2-4}\cline{6-7}\cline{9-10}%
					{$T= 60$} & {\ 100} & 0.09 & -0.14 &  & 27.00 & 129.71 &  & 11.65 & 4.10\\
					{} & {\ 500} & -0.57 & -0.25 &  & 17.72 & 15.76 &  & 29.85 & 6.10\\
					{} & {1,000} & -0.60 & -0.04 &  & 16.17 & 10.48 &  & 41.65 & 5.65\\
					{} & {3,000} & -0.80 & 0.10 &  & 14.89 & 6.07 &  & 60.25 & 6.75\\
					&  &  &  &  &  &  &  &  & \\
					{$T=120$} & {\ 100} & 0.22 & -0.35 &  & 17.19 & 20.06 &  & 6.90 & 4.15\\
					{} & {\ 500} & 0.07 & 0.06 &  & 9.47 & 8.91 &  & 13.70 & 5.40\\
					{} & {1,000} & 0.04 & 0.16 &  & 7.81 & 6.35 &  & 20.50 & 6.00\\
					{} & {3,000} & -0.18 & 0.14 &  & 6.40 & 3.62 &  & 37.75 & 5.60\\
					&  &  &  &  &  &  &  &  & \\
					{$T=240$} & {\ 100} & 0.30 & -0.13 &  & 12.34 & 13.45 &  & 6.50 & 5.90\\
					{} & {\ 500} & 0.30 & 0.07 &  & 6.02 & 6.01 &  & 8.90 & 6.30\\
					{} & {1,000} & 0.17 & -0.02 &  & 4.50 & 4.21 &  & 11.40 & 5.80\\
					{} & {3,000} & 0.10 & 0.02 &  & 3.09 & 2.35 &  & 18.15 & 5.05\\
					$\phi_{H}$=$-0.35$, $\alpha_{H}$=$0.85$ &  &  &  &  &  &  &  &  & \\
					{$T= 60$} & {\ 100} & 2.62 & 9.84 &  & 29.60 & 432.72 &  & 18.55 & 2.75\\
					{} & {\ 500} & 3.90 & 0.37 &  & 25.08 & 28.76 &  & 45.25 & 5.35\\
					{} & {1,000} & 3.57 & -0.62 &  & 24.92 & 19.22 &  & 58.20 & 5.95\\
					{} & {3,000} & 3.94 & -0.10 &  & 25.92 & 12.64 &  & 74.50 & 6.90\\
					&  &  &  &  &  &  &  &  & \\
					{$T=120$} & {\ 100} & 1.82 & -0.38 &  & 19.87 & 27.99 &  & 9.25 & 4.65\\
					{} & {\ 500} & 2.59 & 0.31 &  & 13.82 & 13.27 &  & 25.70 & 5.85\\
					{} & {1,000} & 2.45 & 0.02 &  & 13.31 & 10.26 &  & 39.70 & 5.95\\
					{} & {3,000} & 2.70 & 0.11 &  & 13.61 & 6.85 &  & 61.20 & 7.15\\
					&  &  &  &  &  &  &  &  & \\
					{$T=240$} & {\ 100} & 0.92 & -0.65 &  & 13.61 & 16.79 &  & 7.15 & 4.90\\
					{} & {\ 500} & 1.65 & 0.02 &  & 7.79 & 8.06 &  & 12.50 & 5.05\\
					{} & {1,000} & 1.77 & 0.07 &  & 7.03 & 6.03 &  & 21.70 & 5.55\\
					{} & {3,000} & 1.87 & -0.05 &  & 6.57 & 3.93 &  & 41.05 & 6.05\\
					$\phi_{S}$=$0.16$, $\alpha_{S}$=$0.65$ &  &  &  &  &  &  &  &  & \\
					{$T= 60$} & {\ 100} & -20.25 & 1.01 &  & 38.24 & 407.11 &  & 26.15 & 3.00\\
					{} & {\ 500} & -23.75 & 1.77 &  & 36.94 & 55.06 &  & 55.90 & 5.30\\
					{} & {1,000} & -26.03 & 0.80 &  & 39.26 & 35.51 &  & 69.75 & 6.50\\
					{} & {3,000} & -28.18 & 1.33 &  & 41.98 & 27.55 &  & 81.60 & 7.85\\
					&  &  &  &  &  &  &  &  & \\
					{$T=120$} & {\ 100} & -13.05 & 2.21 &  & 26.45 & 36.34 &  & 18.15 & 5.05\\
					{} & {\ 500} & -16.97 & 0.78 &  & 24.00 & 20.16 &  & 47.40 & 5.40\\
					{} & {1,000} & -19.18 & 0.24 &  & 25.77 & 17.03 &  & 64.35 & 6.10\\
					{} & {3,000} & -21.96 & 0.84 &  & 28.73 & 13.65 &  & 80.55 & 7.20\\
					&  &  &  &  &  &  &  &  & \\
					{$T=240$} & {\ 100} & -7.73 & 1.03 &  & 17.33 & 20.17 &  & 11.10 & 5.40\\
					{} & {\ 500} & -10.89 & 0.50 &  & 14.85 & 11.67 &  & 34.35 & 5.40\\
					{} & {1,000} & -12.73 & 0.44 &  & 15.92 & 9.49 &  & 54.85 & 5.60\\
					{} & {3,000} & -15.75 & 0.49 &  & 18.58 & 7.25 &  & 80.80 & 5.65\\\hline\hline
				\end{tabular}
			}
		\end{center}
		
		{\footnotesize Notes: The DGP for Experiment 12 allows for t(5) distributed
			errors, with GARCH effects, with pricing errors ($\alpha_{\eta}=0.3$), with
			one weak missing factor ($\alpha_{\gamma}=0.5$), and with block error cross
			dependence. For further details of the experiments, see \ref{TabExperiments}.
		}
\end{table}%

{\footnotesize \pagebreak}

\renewcommand{\thefigure}{S-A-E12}{\footnotesize
	\begin{figure}[h]%
		\centering
		\caption{Empirical Power Functions, experiment 12 for coefficient of the
			semi-strong factors}%
		\includegraphics[
		height=5.8608in,
		width=7.1122in
		]%
		{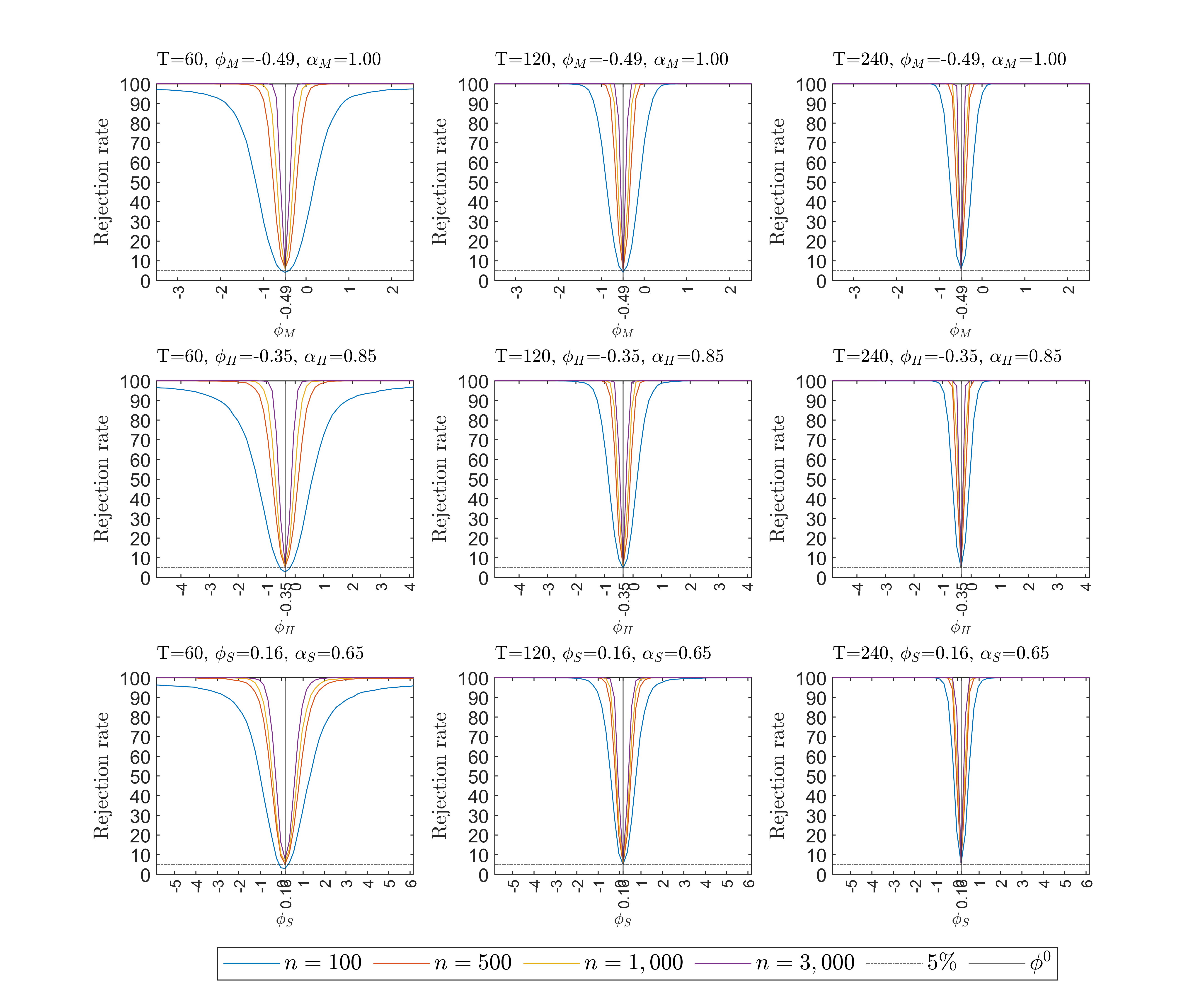}%
		\label{fig:s-a-e12}%
	\end{figure}
	{Note: } See the notes to{ Table \ref{tab:s-a-e12}.} }

{\footnotesize \pagebreak}

\subsection{Estimators of $\phi$ for one strong and two semi-strong factors,
	threshold estimator of the covariance matrix with and without
	misspecification}

This subsection examines the effects of either erroneously omitting
semi-strong factors or appropriately including semi-strong factors on the
small sample characteristics of the bias-corrected (BC) estimator of $\phi
_{M}$, the coefficient associated with the strong factor. Each table presents
results from three of the 12 experiments detailed in Table S-1, along with
their respective empirical power functions.

\setcounter{figure}{0} \renewcommand{\thetable}{S-B-E1-3} \renewcommand{\thefigure}{S-B-E\arabic{figure}}%

	\begin{table}[H]%

		\caption{Bias, RMSE and size for the bias-corrected estimators of $\phi_{M}$=$-0.49$, $\alpha_{M}$=$1$ with and without semi-strong factors $\left(
			\alpha_{H}=0.85,\alpha_{S}=0.65\right)  $ included in the regression for the
			cases of experiments 1, 2 and 3}\label{tab:s-b1-3}
		
		\begin{center}
			{\footnotesize
				\begin{tabular}
					[c]{rrrrrrrrrr}\hline\hline
					&  & \multicolumn{2}{c}{Bias(x100)} &  & \multicolumn{2}{c}{RMSE(x100)} &  &
					\multicolumn{2}{c}{Size(x100)}\\\cline{3-4}\cline{6-7}\cline{9-10}%
					Experiment 1 & n & With & Without &  & With & Without &  & With &
					Without\\\cline{2-4}\cline{6-7}\cline{9-10}
					&  & \multicolumn{2}{c}{semi-strong factors} &  &
					\multicolumn{2}{c}{semi-strong factors} &  & \multicolumn{2}{c}{semi-strong
						factors}\\\cline{2-4}\cline{6-7}\cline{9-10}%
					{$T=60$} & {\ 100} & -1.55 & 5.63 &  & 86.14 & 31.56 &  & 4.00 & 7.60\\
					& {\ 500} & 0.28 & 3.72 &  & 14.33 & 14.60 &  & 5.40 & 8.80\\
					& {1,000} & 0.10 & 2.96 &  & 10.12 & 11.02 &  & 6.20 & 10.85\\
					& {3,000} & -0.01 & 1.96 &  & 5.56 & 6.98 &  & 5.70 & 13.50\\
					&  &  &  &  &  &  &  &  & \\
					{$T=120$} & {\ 100} & 1.84 & 4.96 &  & 91.35 & 20.49 &  & 4.70 & 7.50\\
					& {\ 500} & 0.15 & 3.36 &  & 8.99 & 9.93 &  & 6.15 & 10.50\\
					& {1,000} & 0.01 & 2.72 &  & 6.22 & 7.40 &  & 5.65 & 12.15\\
					& {3,000} & -0.06 & 1.88 &  & 3.49 & 4.71 &  & 4.95 & 15.80\\
					&  &  &  &  &  &  &  &  & \\
					{$T=240$} & {\ 100} & -0.15 & 5.00 &  & 13.11 & 14.46 &  & 4.55 & 7.90\\
					& {\ 500} & 0.09 & 3.33 &  & 5.75 & 7.03 &  & 4.85 & 11.90\\
					& {1,000} & 0.03 & 2.73 &  & 4.06 & 5.33 &  & 5.30 & 14.40\\
					& {3,000} & -0.00 & 1.95 &  & 2.35 & 3.52 &  & 5.10 & 19.95\\
					Experiment 2 &  &  &  &  &  &  &  &  & \\
					{$T=60$} & {\ 100} & 0.27 & 6.04 &  & 45.08 & 32.20 &  & 3.70 & 7.25\\
					& {\ 500} & 0.27 & 3.69 &  & 14.47 & 14.60 &  & 5.75 & 7.75\\
					& {1,000} & 0.04 & 2.85 &  & 10.24 & 11.07 &  & 6.10 & 11.10\\
					& {3,000} & -0.05 & 1.94 &  & 5.72 & 7.03 &  & 5.90 & 14.45\\
					&  &  &  &  &  &  &  &  & \\
					{$T=120$} & {\ 100} & 0.02 & 5.24 &  & 20.45 & 20.58 &  & 4.90 & 8.30\\
					& {\ 500} & 0.23 & 3.46 &  & 9.16 & 10.04 &  & 6.60 & 10.35\\
					& {1,000} & 0.03 & 2.73 &  & 6.12 & 7.37 &  & 5.55 & 11.35\\
					& {3,000} & -0.09 & 1.83 &  & 3.52 & 4.70 &  & 5.65 & 16.10\\
					&  &  &  &  &  &  &  &  & \\
					{$T=240$} & {\ 100} & -0.14 & 5.08 &  & 13.18 & 14.44 &  & 5.05 & 8.75\\
					& {\ 500} & 0.11 & 3.39 &  & 5.81 & 7.06 &  & 5.05 & 11.65\\
					& {1,000} & 0.04 & 2.74 &  & 4.07 & 5.35 &  & 5.55 & 15.50\\
					& {3,000} & -0.03 & 1.92 &  & 2.36 & 3.51 &  & 5.25 & 20.55\\
					Experiment 3 &  &  &  &  &  &  &  &  & \\
					{$T=60$} & {\ 100} & -6.00 & 5.70 &  & 227.89 & 31.96 &  & 3.40 & 7.55\\
					& {\ 500} & 0.29 & 3.68 &  & 14.71 & 14.81 &  & 5.65 & 9.10\\
					& {1,000} & 0.09 & 2.94 &  & 10.33 & 11.25 &  & 6.35 & 11.55\\
					& {3,000} & 0.00 & 1.96 &  & 5.65 & 7.18 &  & 5.65 & 14.55\\
					&  &  &  &  &  &  &  &  & \\
					{$T=120$} & {\ 100} & -0.19 & 4.89 &  & 20.63 & 20.51 &  & 4.50 & 7.45\\
					& {\ 500} & 0.15 & 3.32 &  & 9.05 & 9.98 &  & 6.25 & 10.65\\
					& {1,000} & 0.02 & 2.68 &  & 6.25 & 7.46 &  & 5.75 & 12.15\\
					& {3,000} & -0.05 & 1.85 &  & 3.50 & 4.75 &  & 4.90 & 15.20\\
					&  &  &  &  &  &  &  &  & \\
					{$T=240$} & {\ 100} & -0.13 & 4.96 &  & 13.17 & 14.48 &  & 4.45 & 7.90\\
					& {\ 500} & 0.10 & 3.32 &  & 5.78 & 7.05 &  & 5.20 & 11.85\\
					& {1,000} & 0.05 & 2.72 &  & 4.07 & 5.34 &  & 5.25 & 14.10\\
					& {3,000} & 0.00 & 1.93 &  & 2.36 & 3.53 &  & 5.20 & 19.90\\\hline\hline
				\end{tabular}
			}
		\end{center}
		
		{\footnotesize \noindent Notes: The DGP includes one strong $\alpha_{M}=1$ and
			two semi-strong ($\alpha_{H}=0.85$, $\alpha_{S}=0.65$) factors, the regression
			with the two semi-strong factors includes them, the regression without
			excludes them. For further details of the experiments, see
			\ref{TabExperiments}. }
\end{table}%

{\footnotesize \pagebreak}

{\footnotesize
	\begin{figure}[ph]%
		\centering
		\caption{Empirical Power Functions, experiment 1, for coefficient of the
			$\phi_{M}$ factor with and without misspecification}%
		\includegraphics[
		height=5.8608in,
		width=7.1122in
		]%
		{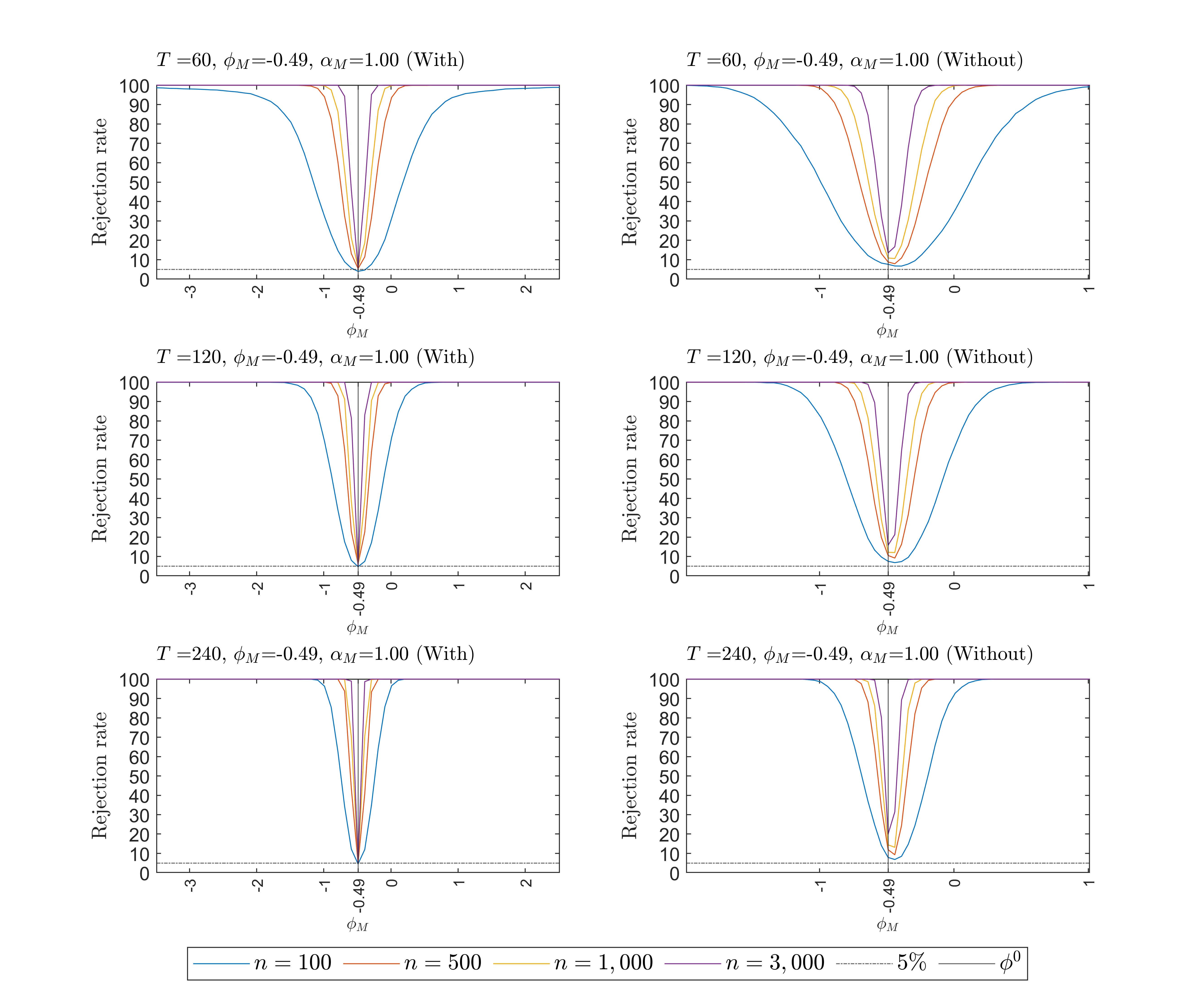}%
		\label{fig:s-b-e1}%
	\end{figure}
	{Note: } See the notes to{ Table \ref{tab:s-b1-3}.} }

{\footnotesize \pagebreak}

{\footnotesize
	\begin{figure}[ph]%
		\centering
		\caption{Empirical Power Functions, experiment 2, for coefficient of the
			$\phi_{M}$ factor with and without misspecification}%
		\includegraphics[
		height=5.8608in,
		width=7.1114in
		]%
		{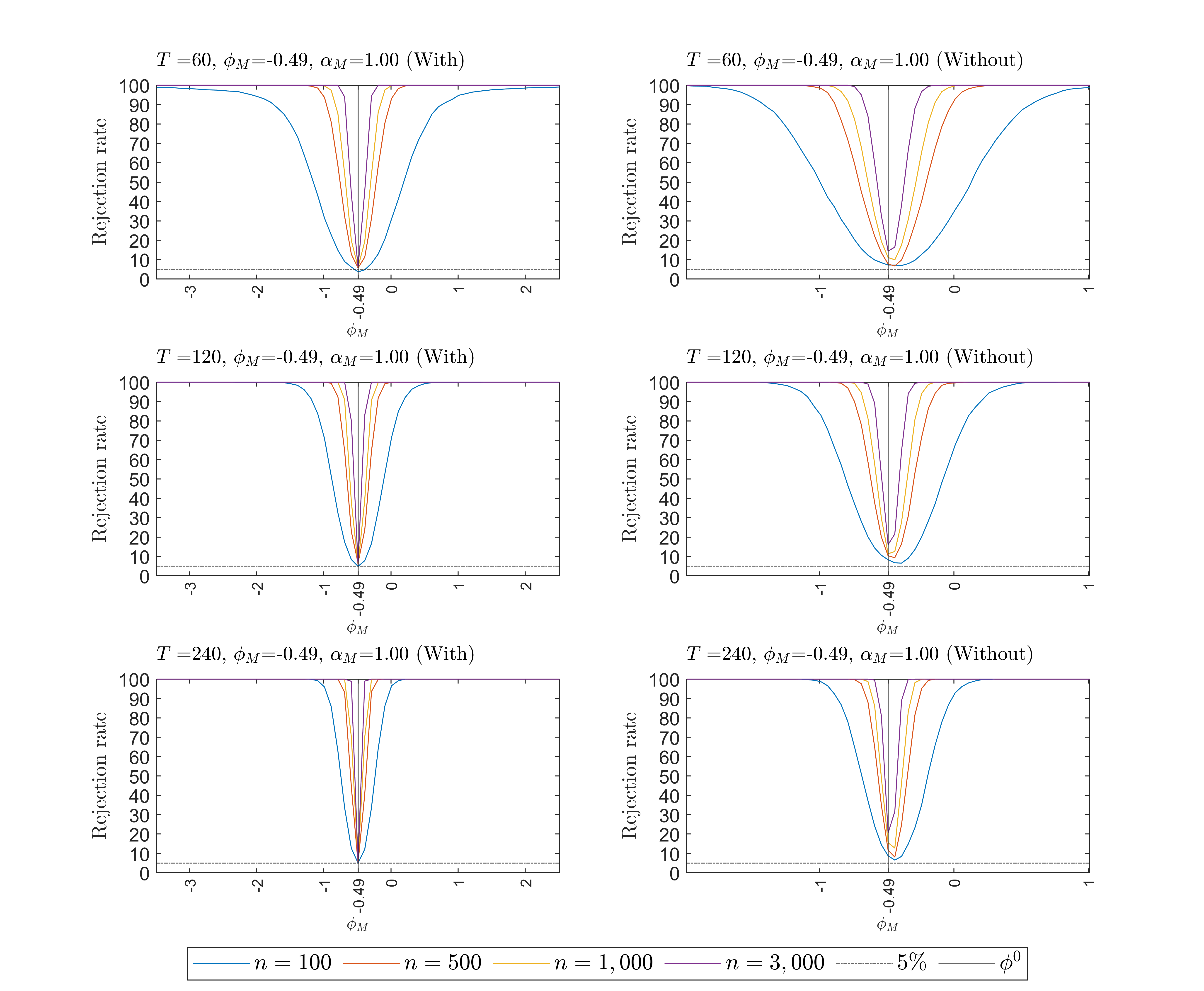}%
		\label{fig:s-b-e2}%
	\end{figure}
	{Note: } See the notes to{ Table \ref{tab:s-b1-3}.} }

{\footnotesize \pagebreak}

{\footnotesize
	\begin{figure}[ph]%
		\centering
		\caption{Empirical Power Functions, experiment 3, for coefficient of the
			$\phi_{M}$ factor with and without misspecification}%
		\includegraphics[
		height=5.8608in,
		width=7.1114in
		]%
		{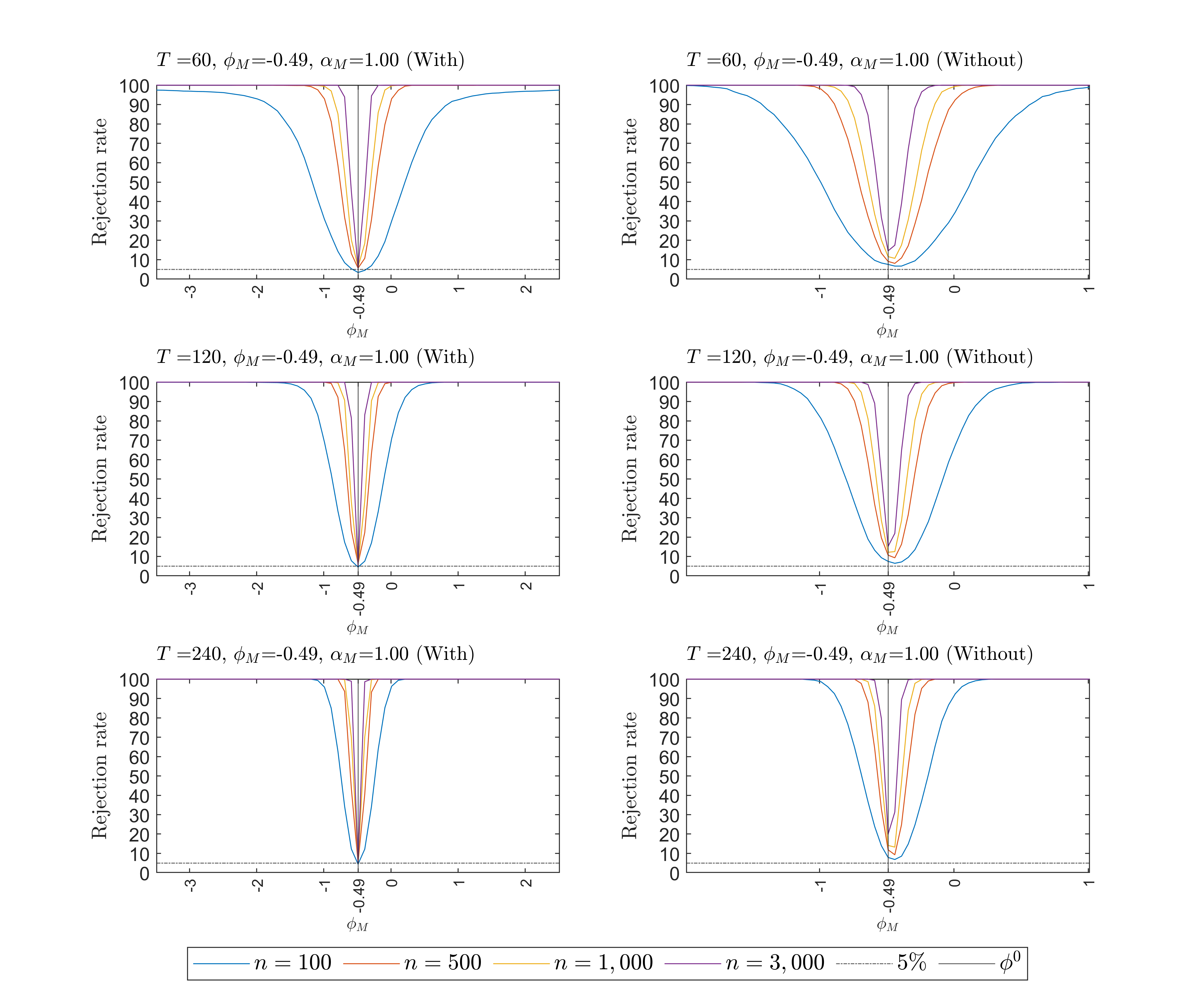}%
		\label{fig:s-b-e3}%
	\end{figure}
	{Note: } See the notes to{ Table \ref{tab:s-b1-3}.} }

\renewcommand{\thetable}{S-B-E4-6}%

	\begin{table}[H]%

		\caption{Bias, RMSE and size for the bias-corrected estimators of $\phi_{M}$
			=$-0.49$, $\alpha_{M}$=$1$ with and without semi-strong factors $\left(
			\alpha_{H}=0.85,\alpha_{S}=0.65\right)  $ included in the regression for the
			cases of experiments 4, 5 and 6}\label{tab:s-b-e4-6}
		
		\begin{center}
			{\footnotesize
				\begin{tabular}
					[c]{rrrrrrrrrr}\hline\hline
					&  & \multicolumn{2}{c}{Bias(x100)} &  & \multicolumn{2}{c}{RMSE(x100)} &  &
					\multicolumn{2}{c}{Size(x100)}\\\cline{3-4}\cline{6-7}\cline{9-10}%
					Experiment 4 & n & With & Without &  & With & Without &  & With &
					Without\\\cline{2-4}\cline{6-7}\cline{9-10}
					&  & \multicolumn{2}{c}{semi-strong factors} &  &
					\multicolumn{2}{c}{semi-strong factors} &  & \multicolumn{2}{c}{semi-strong
						factors}\\\cline{2-4}\cline{6-7}\cline{9-10}%
					{$T=60$} & {\ 100} & 0.17 & 6.08 &  & 107.83 & 32.68 &  & 3.75 & 7.25\\
					& {\ 500} & 0.23 & 3.64 &  & 14.75 & 14.88 &  & 5.65 & 7.85\\
					& {1,000} & 0.05 & 2.84 &  & 10.43 & 11.29 &  & 6.10 & 10.85\\
					& {3,000} & -0.04 & 1.94 &  & 5.82 & 7.21 &  & 5.95 & 14.15\\
					&  &  &  &  &  &  &  &  & \\
					{$T=120$} & {\ 100} & -0.86 & 5.15 &  & 44.37 & 20.62 &  & 4.95 & 8.10\\
					& {\ 500} & 0.21 & 3.42 &  & 9.24 & 10.11 &  & 6.45 & 10.75\\
					& {1,000} & 0.03 & 2.69 &  & 6.16 & 7.41 &  & 4.90 & 11.05\\
					& {3,000} & -0.08 & 1.81 &  & 3.53 & 4.75 &  & 5.85 & 15.55\\
					&  &  &  &  &  &  &  &  & \\
					{$T=240$} & {\ 100} & -0.13 & 5.03 &  & 13.22 & 14.46 &  & 4.85 & 8.70\\
					& {\ 500} & 0.12 & 3.37 &  & 5.83 & 7.07 &  & 5.10 & 11.15\\
					& {1,000} & 0.05 & 2.72 &  & 4.08 & 5.36 &  & 5.60 & 15.50\\
					& {3,000} & -0.02 & 1.90 &  & 2.36 & 3.52 &  & 5.25 & 20.00\\
					Experiment 5 &  &  &  &  &  &  &  &  & \\
					{$T=60$} & {\ 100} & -7.58 & 5.70 &  & 285.21 & 31.96 &  & 3.75 & 7.55\\
					& {\ 500} & 0.28 & 3.68 &  & 14.79 & 14.86 &  & 5.80 & 9.05\\
					& {1,000} & 0.10 & 2.92 &  & 10.35 & 11.25 &  & 6.25 & 11.80\\
					& {3,000} & -0.00 & 1.97 &  & 5.65 & 7.19 &  & 5.75 & 14.70\\
					&  &  &  &  &  &  &  &  & \\
					{$T=120$} & {\ 100} & -0.35 & 4.89 &  & 20.82 & 20.51 &  & 4.95 & 7.45\\
					& {\ 500} & 0.14 & 3.32 &  & 9.15 & 10.03 &  & 6.45 & 10.30\\
					& {1,000} & 0.04 & 2.67 &  & 6.28 & 7.48 &  & 5.90 & 12.20\\
					& {3,000} & -0.06 & 1.86 &  & 3.50 & 4.76 &  & 4.80 & 15.45\\
					&  &  &  &  &  &  &  &  & \\
					{$T=240$} & {\ 100} & -0.31 & 4.96 &  & 13.53 & 14.48 &  & 5.35 & 7.90\\
					& {\ 500} & 0.10 & 3.31 &  & 5.84 & 7.08 &  & 5.35 & 11.50\\
					& {1,000} & 0.07 & 2.71 &  & 4.11 & 5.36 &  & 5.20 & 14.10\\
					& {3,000} & 0.00 & 1.94 &  & 2.36 & 3.53 &  & 5.25 & 20.30\\
					Experiment 6 &  &  &  &  &  &  &  &  & \\
					{$T=60$} & {\ 100} & 0.05 & 6.08 &  & 103.65 & 32.68 &  & 3.95 & 7.25\\
					& {\ 500} & 0.24 & 3.63 &  & 14.81 & 14.93 &  & 5.95 & 8.25\\
					& {1,000} & 0.06 & 2.82 &  & 10.44 & 11.30 &  & 6.00 & 11.00\\
					& {3,000} & -0.05 & 1.95 &  & 5.82 & 7.22 &  & 5.80 & 14.40\\
					&  &  &  &  &  &  &  &  & \\
					{$T=120$} & {\ 100} & -0.76 & 5.15 &  & 33.47 & 20.62 &  & 5.15 & 8.10\\
					& {\ 500} & 0.22 & 3.42 &  & 9.32 & 10.16 &  & 6.95 & 10.65\\
					& {1,000} & 0.05 & 2.68 &  & 6.18 & 7.43 &  & 5.30 & 11.25\\
					& {3,000} & -0.09 & 1.82 &  & 3.54 & 4.75 &  & 5.70 & 15.40\\
					&  &  &  &  &  &  &  &  & \\
					{$T=240$} & {\ 100} & -0.30 & 5.03 &  & 13.63 & 14.46 &  & 6.10 & 8.70\\
					& {\ 500} & 0.11 & 3.37 &  & 5.88 & 7.10 &  & 5.10 & 11.85\\
					& {1,000} & 0.07 & 2.71 &  & 4.13 & 5.38 &  & 5.80 & 15.50\\
					& {3,000} & -0.03 & 1.91 &  & 2.37 & 3.52 &  & 5.45 & 20.40\\\hline\hline
				\end{tabular}
			}
		\end{center}
		
		{\footnotesize \noindent Notes: The DGP includes one strong $\alpha_{M}=1$ and
			two semi-strong ($\alpha_{H}=0.85$, $\alpha_{S}=0.65$) factors, the regression
			with the two semi-strong factors includes them, the regression without
			excludes them. For further details of the experiments, see
			\ref{TabExperiments}. }
\end{table}%

{\footnotesize \pagebreak}

{\footnotesize
	\begin{figure}[ph]%
		\centering
		\caption{ Empirical Power Functions, experiment 4, for coefficient of the
			$\phi_{M}$ factor with and without misspecification}%
		\includegraphics[
		height=5.8608in,
		width=7.1122in
		]%
		{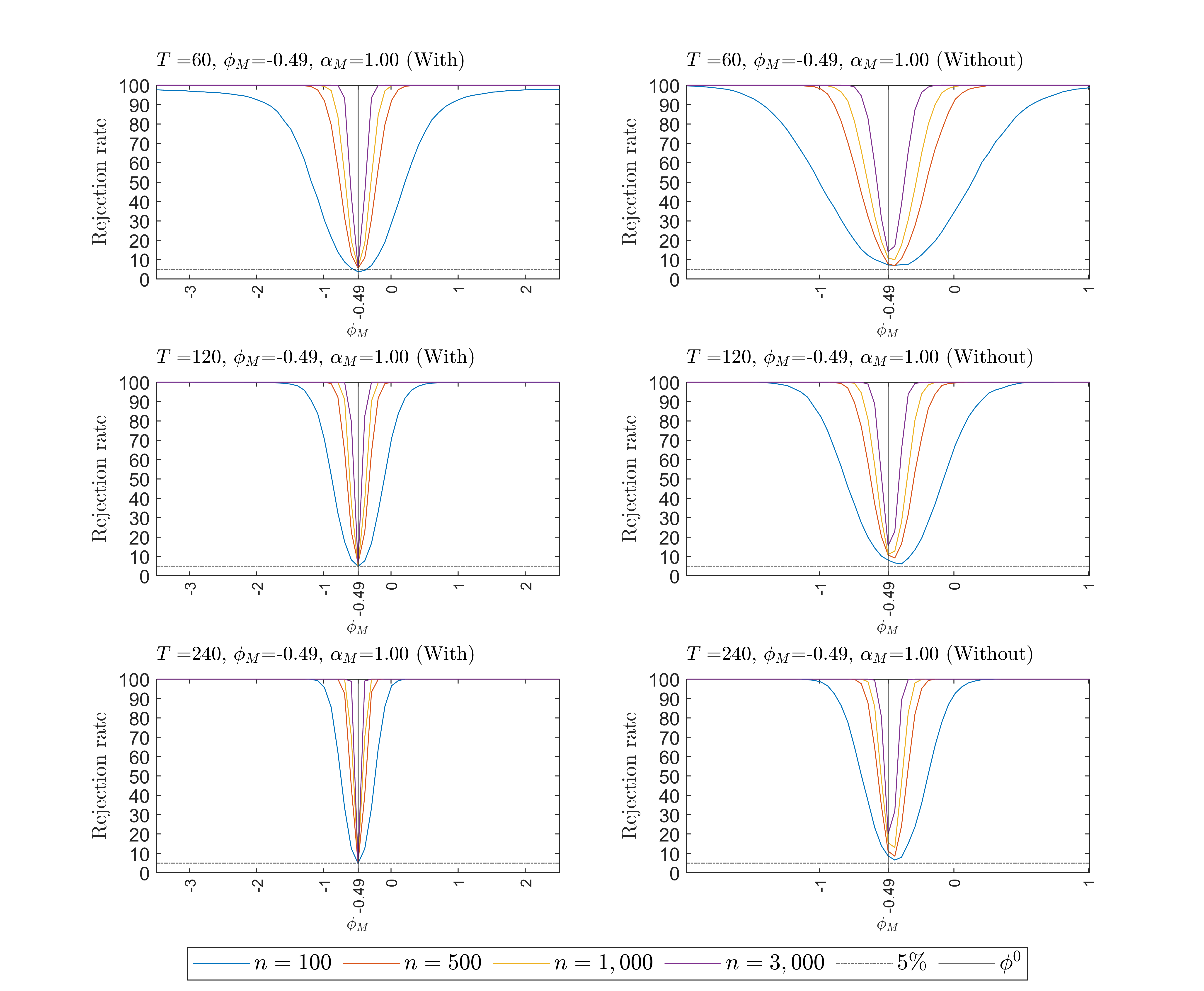}%
		\label{fig:s-b-e4}%
	\end{figure}
	{Note: } See the notes to{ Table \ref{tab:s-b-e4-6}.} }

{\footnotesize \pagebreak}

{\footnotesize
	\begin{figure}[ph]%
		\centering
		\caption{Empirical Power Functions, experiment 5, for coefficient of the
			$\phi_{M}$ factor with and without misspecification}%
		\includegraphics[
		height=5.8608in,
		width=7.1114in
		]%
		{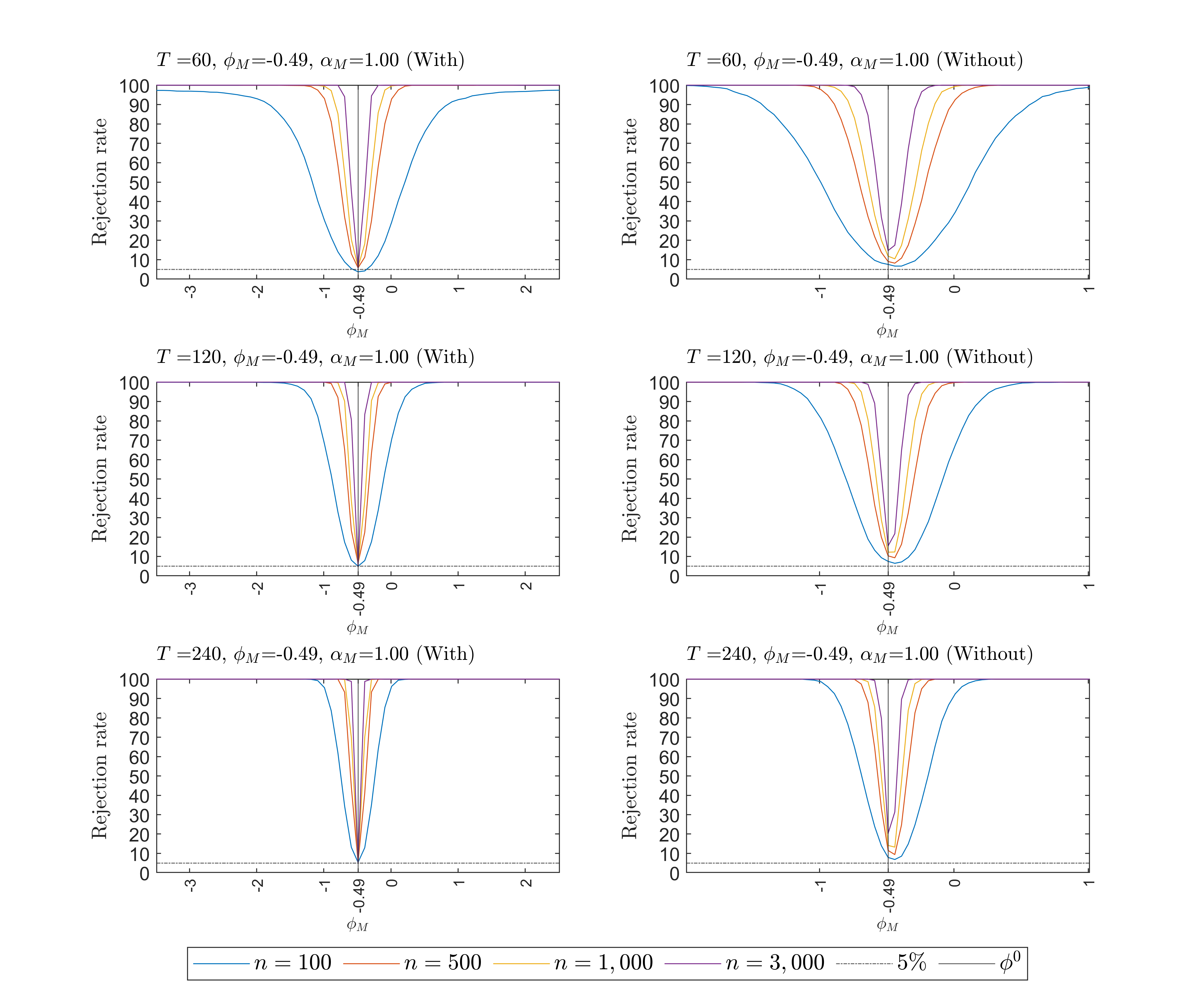}%
		\label{fig:s-b-e5}%
	\end{figure}
	{Note: } See the notes to{ Table \ref{tab:s-b-e4-6}.} }

{\footnotesize \pagebreak}

{\footnotesize
	\begin{figure}[ph]%
		\centering
		\caption{Empirical Power Functions, experiment 6, for coefficient of the
			$\phi_{M}$ factor with and without misspecification}%
		\includegraphics[
		height=5.8608in,
		width=7.1122in
		]%
		{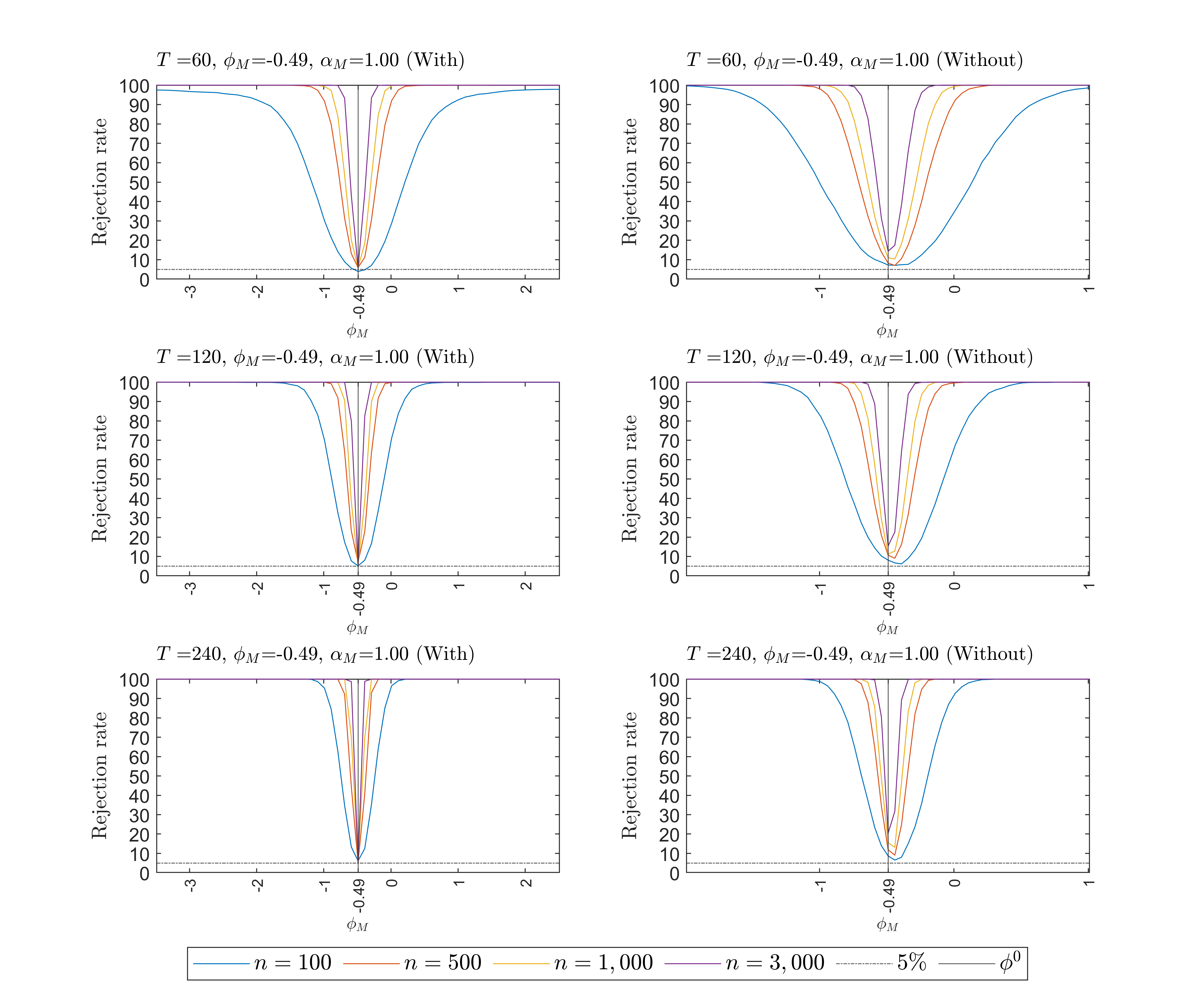}%
		\label{fig:s-b-e6}%
	\end{figure}
	{Note: } See the notes to{ Table \ref{tab:s-b-e4-6}.} }

\renewcommand{\thetable}{S-B-E7-9}%

	\begin{table}[H]%

		\caption{Bias, RMSE and size for the bias-corrected estimators of $\phi_{M}$
			=$-0.49$, $\alpha_{M}$=$1$ with and without semi-strong factors $\left(
			\alpha_{H}=0.85,\alpha_{S}=0.65\right)  $ included in the regression for the
			cases of experiments 7, 8 and 9}\label{tab:s-b-e7-9}
		
		\begin{center}
			{\footnotesize
				\begin{tabular}
					[c]{rrrrrrrrrr}\hline\hline
					&  & \multicolumn{2}{c}{Bias(x100)} &  & \multicolumn{2}{c}{RMSE(x100)} &  &
					\multicolumn{2}{c}{Size(x100)}\\\cline{3-4}\cline{6-7}\cline{9-10}%
					Experiment 7 & n & With & Without &  & With & Without &  & With &
					Without\\\cline{2-4}\cline{6-7}\cline{9-10}
					&  & \multicolumn{2}{c}{semi-strong factors} &  &
					\multicolumn{2}{c}{semi-strong factors} &  & \multicolumn{2}{c}{semi-strong
						factors}\\\cline{2-4}\cline{6-7}\cline{9-10}%
					{$T=60$} & {\ 100} & -14.31 & 4.76 &  & 555.11 & 31.07 &  & 3.20 & 5.35\\
					& {\ 500} & -0.18 & 3.25 &  & 15.13 & 15.20 &  & 6.40 & 9.65\\
					& {1,000} & -0.15 & 2.62 &  & 10.41 & 11.29 &  & 6.25 & 11.35\\
					& {3,000} & 0.04 & 1.99 &  & 5.81 & 7.17 &  & 5.90 & 14.95\\
					&  &  &  &  &  &  &  &  & \\
					{$T=120$} & {\ 100} & -0.49 & 4.70 &  & 20.30 & 19.98 &  & 5.05 & 6.55\\
					& {\ 500} & 0.01 & 3.15 &  & 8.86 & 9.80 &  & 5.85 & 10.15\\
					& {1,000} & 0.04 & 2.62 &  & 6.31 & 7.45 &  & 6.60 & 11.85\\
					& {3,000} & 0.08 & 1.99 &  & 3.60 & 4.80 &  & 5.65 & 17.05\\
					&  &  &  &  &  &  &  &  & \\
					{$T=240$} & {\ 100} & -0.39 & 4.91 &  & 14.02 & 14.48 &  & 6.75 & 8.50\\
					& {\ 500} & -0.08 & 3.13 &  & 6.13 & 7.18 &  & 6.95 & 12.85\\
					& {1,000} & -0.03 & 2.67 &  & 4.22 & 5.45 &  & 5.80 & 15.15\\
					& {3,000} & 0.01 & 1.97 &  & 2.32 & 3.50 &  & 5.15 & 20.90\\
					Experiment 8 &  &  &  &  &  &  &  &  & \\
					{$T=60$} & {\ 100} & 9.06 & 4.80 &  & 251.34 & 31.95 &  & 3.75 & 6.85\\
					& {\ 500} & -0.67 & 2.88 &  & 16.71 & 15.13 &  & 6.30 & 9.95\\
					& {1,000} & 0.01 & 2.75 &  & 10.38 & 11.32 &  & 5.70 & 10.25\\
					& {3,000} & 0.12 & 2.09 &  & 5.96 & 7.25 &  & 6.15 & 14.70\\
					&  &  &  &  &  &  &  &  & \\
					{$T=120$} & {\ 100} & -0.32 & 4.83 &  & 20.53 & 20.18 &  & 4.75 & 6.25\\
					& {\ 500} & -0.04 & 3.05 &  & 8.89 & 9.73 &  & 5.80 & 10.20\\
					& {1,000} & 0.19 & 2.76 &  & 6.40 & 7.60 &  & 6.10 & 12.70\\
					& {3,000} & 0.14 & 2.01 &  & 3.63 & 4.81 &  & 5.90 & 17.05\\
					&  &  &  &  &  &  &  &  & \\
					{$T=240$} & {\ 100} & -0.29 & 5.00 &  & 14.07 & 14.58 &  & 6.35 & 9.40\\
					& {\ 500} & -0.01 & 3.19 &  & 6.06 & 7.15 &  & 6.05 & 12.05\\
					& {1,000} & 0.01 & 2.68 &  & 4.25 & 5.46 &  & 6.30 & 15.90\\
					& {3,000} & 0.01 & 1.96 &  & 2.35 & 3.52 &  & 5.15 & 20.95\\
					Experiment 9 &  &  &  &  &  &  &  &  & \\
					{$T=60$} & {\ 100} & -15.47 & 4.04 &  & 662.68 & 33.36 &  & 3.40 & 6.95\\
					& {\ 500} & 1.29 & 2.84 &  & 84.39 & 15.74 &  & 7.75 & 11.40\\
					& {1,000} & -0.47 & 2.26 &  & 11.31 & 11.50 &  & 8.50 & 12.25\\
					& {3,000} & -0.16 & 1.69 &  & 6.25 & 7.20 &  & 7.75 & 15.35\\
					&  &  &  &  &  &  &  &  & \\
					{$T=120$} & {\ 100} & -0.59 & 4.09 &  & 21.49 & 21.20 &  & 4.50 & 7.00\\
					& {\ 500} & -0.01 & 2.95 &  & 9.20 & 9.95 &  & 6.15 & 9.15\\
					& {1,000} & -0.07 & 2.40 &  & 6.53 & 7.37 &  & 6.30 & 10.80\\
					& {3,000} & -0.01 & 1.74 &  & 3.71 & 4.68 &  & 6.50 & 14.30\\
					&  &  &  &  &  &  &  &  & \\
					{$T=240$} & {\ 100} & -0.54 & 4.06 &  & 14.02 & 14.60 &  & 6.85 & 8.50\\
					& {\ 500} & -0.12 & 2.85 &  & 6.19 & 7.11 &  & 6.55 & 10.80\\
					& {1,000} & -0.08 & 2.42 &  & 4.30 & 5.27 &  & 6.60 & 13.35\\
					& {3,000} & -0.00 & 1.75 &  & 2.37 & 3.37 &  & 5.45 & 17.50\\\hline\hline
				\end{tabular}
			}
		\end{center}
		
		{\footnotesize \noindent Notes: The DGP includes one strong $\alpha_{M}=1$ and
			two semi-strong ($\alpha_{H}=0.85$, $\alpha_{S}=0.65$) factors, the regression
			with the two semi-strong factors includes them, the regression without
			excludes them. For further details of the experiments, see
			\ref{TabExperiments}. }
\end{table}%

{\footnotesize \pagebreak}

{\footnotesize
	\begin{figure}[ph]%
		\centering
		\caption{Empirical Power Functions, experiment 7, for coefficient of the
			$\phi_{M}$ factor with and without misspecification}%
		\includegraphics[
		height=5.8608in,
		width=7.1114in
		]%
		{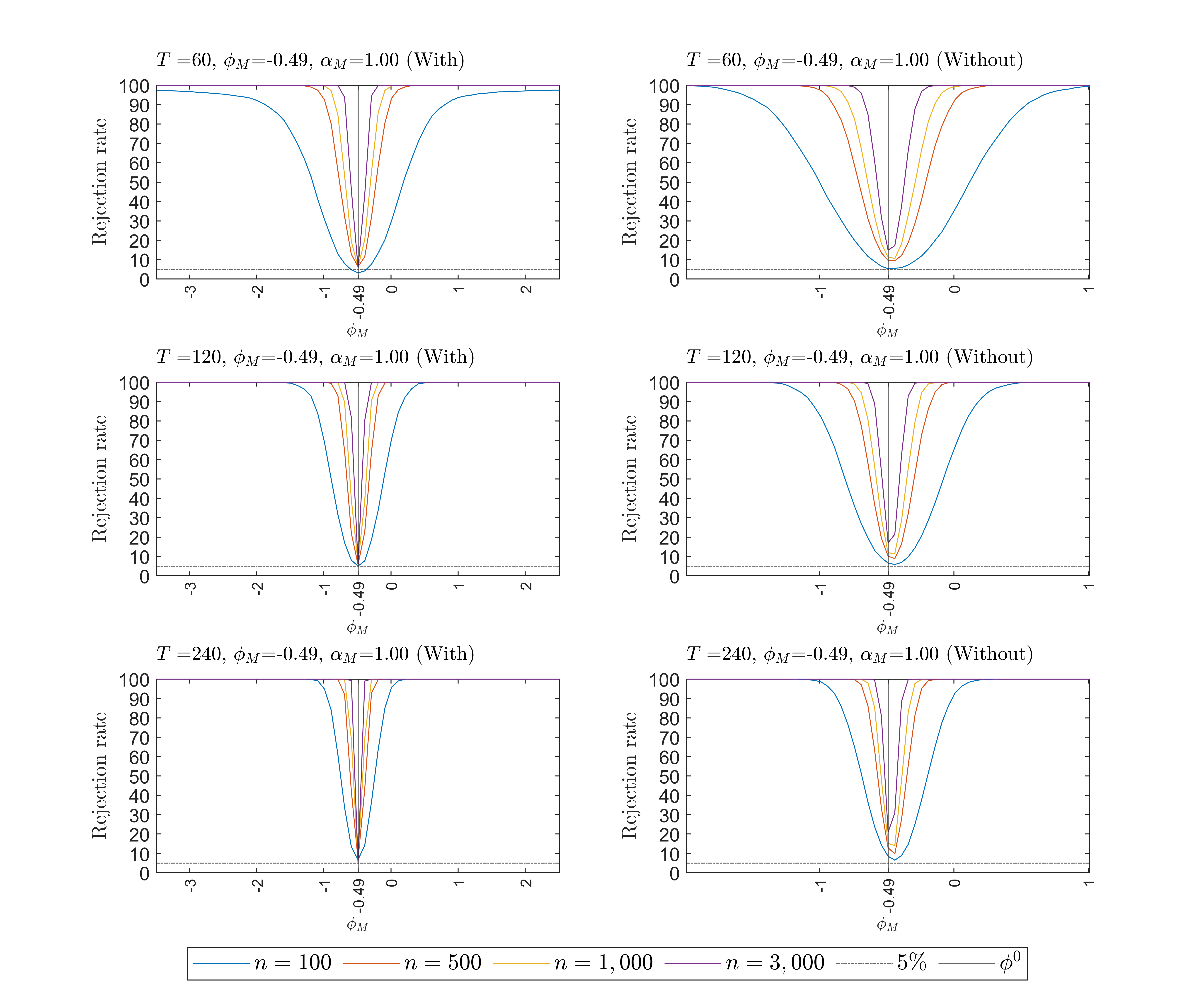}%
		\label{fig:s-b-e7}%
	\end{figure}
	{Note: } See the notes to{ Table \ref{tab:s-b-e7-9}.} }

{\footnotesize \pagebreak}

{\footnotesize
	\begin{figure}[ph]%
		\centering
		\caption{Empirical Power Functions, experiment 8, for coefficient of the
			$\phi_{M}$ factor with and without misspecification}%
		\includegraphics[
		height=5.8608in,
		width=7.1122in
		]%
		{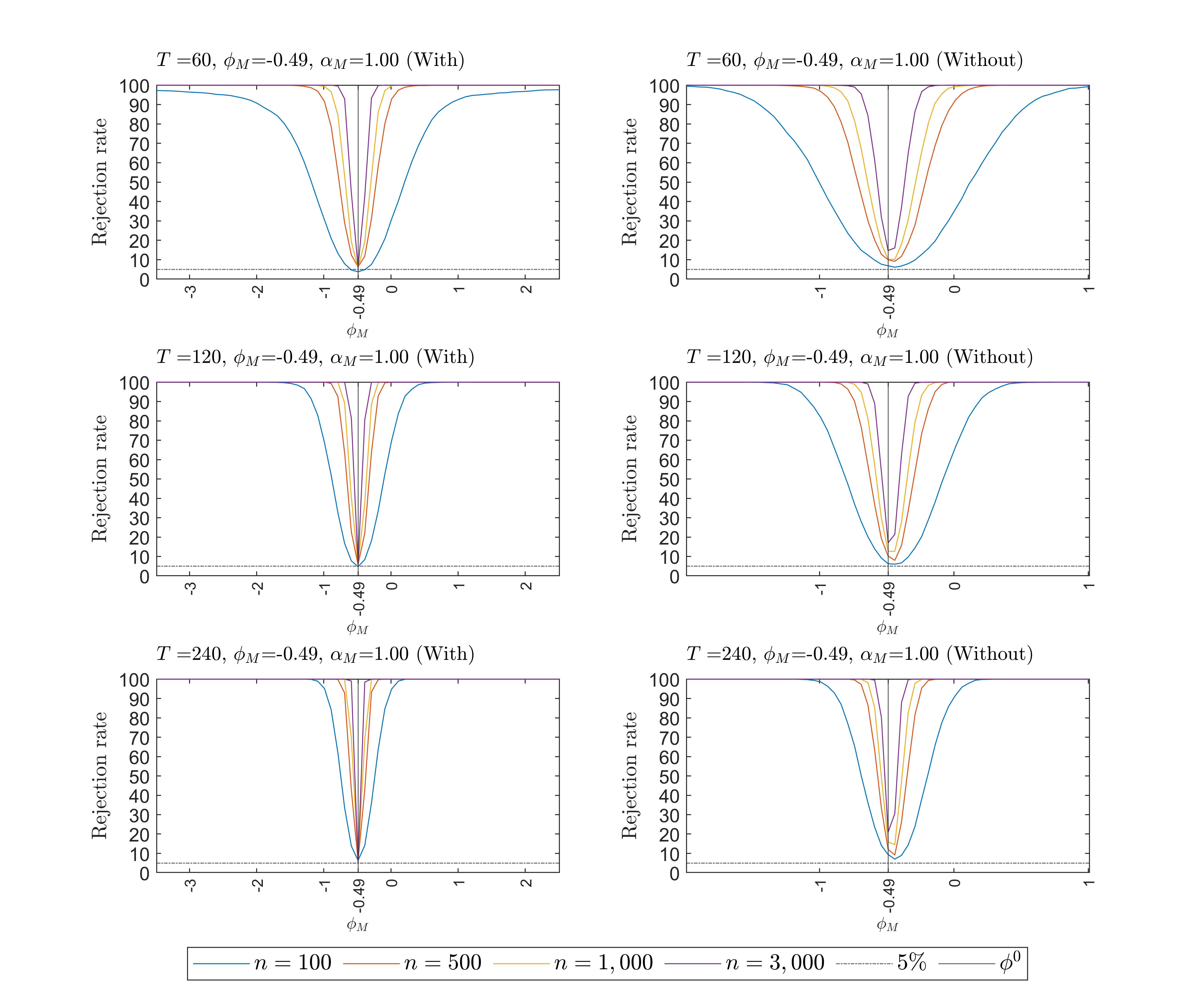}%
		\label{fig:s-b-e8}%
	\end{figure}
	{Note: } See the notes to{ Table \ref{tab:s-b-e7-9}.} }

{\footnotesize \pagebreak}

{\footnotesize
	\begin{figure}[ph]%
		\centering
		\caption{Empirical Power Functions, experiment 9, for coefficient of the
			$\phi_{M}$ factor with and without misspecification}%
		\includegraphics[
		height=5.8608in,
		width=7.1114in
		]%
		{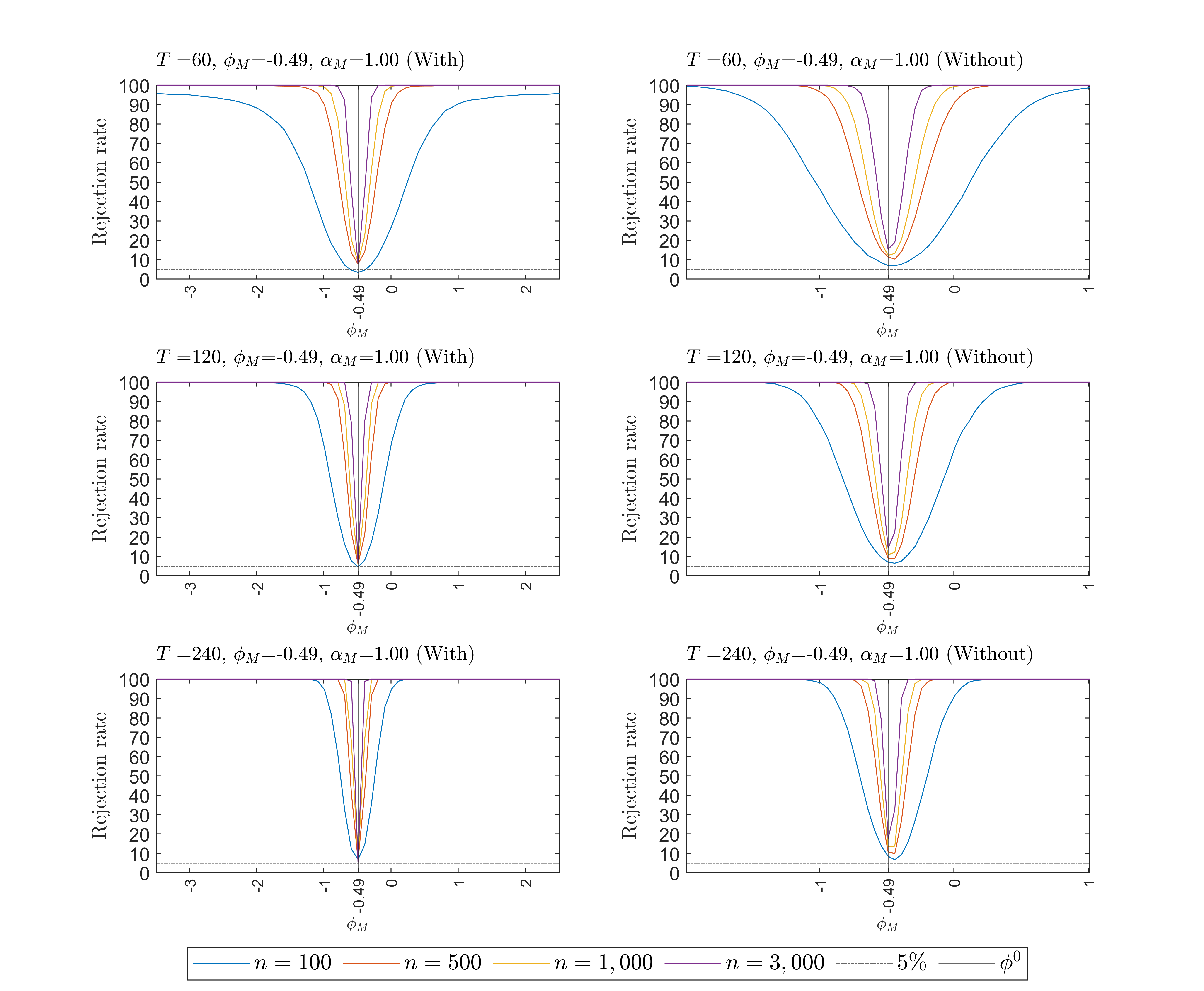}%
		\label{fig:s-b-e9}%
	\end{figure}
	{Note: } See the notes to{ Table \ref{tab:s-b-e7-9}.} }

\renewcommand{\thetable}{S-B-E10-12}%

	\begin{table}[H]%

		\caption{Bias, RMSE and size for the bias-corrected estimators of $\phi_{M}$
			=$-0.49$, $\alpha_{M}$=$1$ with and without semi-strong factors $\left(
			\alpha_{H}=0.85,\alpha_{S}=0.65\right)  $ included in the regression for the
			cases of experiments 10, 11 and 12}\label{tab:s-b-e10-12}
		
		\begin{center}
			{\footnotesize
				\begin{tabular}
					[c]{rrrrrrrrrr}\hline\hline
					&  & \multicolumn{2}{c}{Bias(x100)} &  & \multicolumn{2}{c}{RMSE(x100)} &  &
					\multicolumn{2}{c}{Size(x100)}\\\cline{3-4}\cline{6-7}\cline{9-10}%
					Experiment 10 & n & With & Without &  & With & Without &  & With &
					Without\\\cline{2-4}\cline{6-7}\cline{9-10}
					&  & \multicolumn{2}{c}{semi-strong factors} &  &
					\multicolumn{2}{c}{semi-strong factors} &  & \multicolumn{2}{c}{semi-strong
						factors}\\\cline{2-4}\cline{6-7}\cline{9-10}%
					{$T=60$} & {\ 100} & -1.61 & 4.07 &  & 151.91 & 33.92 &  & 3.75 & 7.45\\
					& {\ 500} & -0.69 & 2.46 &  & 17.07 & 15.76 &  & 7.70 & 11.35\\
					& {1,000} & -0.31 & 2.39 &  & 11.30 & 11.47 &  & 7.05 & 10.90\\
					& {3,000} & -0.06 & 1.79 &  & 6.36 & 7.31 &  & 8.00 & 14.95\\
					&  &  &  &  &  &  &  &  & \\
					{$T=120$} & {\ 100} & -1.86 & 4.22 &  & 69.05 & 21.53 &  & 4.80 & 7.25\\
					& {\ 500} & -0.09 & 2.81 &  & 9.29 & 9.92 &  & 6.10 & 9.85\\
					& {1,000} & 0.15 & 2.59 &  & 6.62 & 7.51 &  & 6.75 & 11.40\\
					& {3,000} & 0.05 & 1.76 &  & 3.72 & 4.69 &  & 5.80 & 14.55\\
					&  &  &  &  &  &  &  &  & \\
					{$T=240$} & {\ 100} & -0.39 & 4.20 &  & 14.17 & 14.82 &  & 6.35 & 8.60\\
					& {\ 500} & -0.05 & 2.87 &  & 6.15 & 7.09 &  & 6.05 & 10.80\\
					& {1,000} & 0.00 & 2.46 &  & 4.30 & 5.28 &  & 6.40 & 13.65\\
					& {3,000} & 0.00 & 1.75 &  & 2.40 & 3.40 &  & 5.75 & 18.80\\
					Experiment 11 &  &  &  &  &  &  &  &  & \\
					{$T=60$} & {\ 100} & -0.72 & 4.82 &  & 47.59 & 31.18 &  & 3.40 & 5.90\\
					& {\ 500} & -0.07 & 3.31 &  & 15.12 & 15.30 &  & 6.70 & 10.30\\
					& {1,000} & -0.22 & 2.56 &  & 10.46 & 11.28 &  & 6.35 & 11.45\\
					& {3,000} & 0.05 & 1.96 &  & 5.91 & 7.26 &  & 7.00 & 16.00\\
					&  &  &  &  &  &  &  &  & \\
					{$T=120$} & {\ 100} & -0.46 & 4.75 &  & 19.81 & 20.00 &  & 4.55 & 6.60\\
					& {\ 500} & 0.10 & 3.25 &  & 8.93 & 9.83 &  & 6.40 & 10.00\\
					& {1,000} & -0.01 & 2.59 &  & 6.28 & 7.41 &  & 6.15 & 11.95\\
					& {3,000} & 0.09 & 1.99 &  & 3.60 & 4.83 &  & 6.20 & 16.70\\
					&  &  &  &  &  &  &  &  & \\
					{$T=240$} & {\ 100} & -0.24 & 4.94 &  & 13.43 & 14.51 &  & 5.30 & 8.45\\
					& {\ 500} & 0.01 & 3.22 &  & 6.11 & 7.23 &  & 6.75 & 12.95\\
					& {1,000} & -0.07 & 2.64 &  & 4.19 & 5.42 &  & 5.85 & 14.65\\
					& {3,000} & 0.02 & 1.97 &  & 2.32 & 3.50 &  & 4.75 & 20.55\\
					Experiment 12 &  &  &  &  &  &  &  &  & \\
					{$T=60$} & {\ 100} & -0.14 & 4.87 &  & 129.71 & 32.05 &  & 4.10 & 6.75\\
					& {\ 500} & -0.25 & 2.94 &  & 15.76 & 15.19 &  & 6.10 & 9.45\\
					& {1,000} & -0.04 & 2.71 &  & 10.48 & 11.34 &  & 5.65 & 10.40\\
					& {3,000} & 0.10 & 2.05 &  & 6.07 & 7.34 &  & 6.75 & 15.20\\
					&  &  &  &  &  &  &  &  & \\
					{$T=120$} & {\ 100} & -0.35 & 4.86 &  & 20.06 & 20.21 &  & 4.15 & 6.35\\
					& {\ 500} & 0.06 & 3.16 &  & 8.91 & 9.76 &  & 5.40 & 10.35\\
					& {1,000} & 0.16 & 2.76 &  & 6.35 & 7.57 &  & 6.00 & 12.75\\
					& {3,000} & 0.14 & 2.01 &  & 3.62 & 4.82 &  & 5.60 & 16.95\\
					&  &  &  &  &  &  &  &  & \\
					{$T=240$} & {\ 100} & -0.13 & 5.03 &  & 13.45 & 14.61 &  & 5.90 & 9.25\\
					& {\ 500} & 0.07 & 3.27 &  & 6.01 & 7.18 &  & 6.30 & 12.40\\
					& {1,000} & -0.02 & 2.66 &  & 4.21 & 5.41 &  & 5.80 & 15.70\\
					& {3,000} & 0.02 & 1.96 &  & 2.35 & 3.52 &  & 5.05 & 20.65\\\hline\hline
				\end{tabular}
			}
		\end{center}
		
		{\footnotesize \noindent Notes: The DGP includes one strong $\alpha_{M}=1$ and
			two semi-strong ($\alpha_{H}=0.85$, $\alpha_{S}=0.65$) factors, the regression
			with the two semi-strong factors includes them, the regression without
			excludes them. For further details of the experiments, see
			\ref{TabExperiments}. }%
		
\end{table}%

{\footnotesize \pagebreak}

{\footnotesize
	\begin{figure}[ph]%
		\centering
		\caption{Empirical Power Functions, experiment 10, for coefficient of the
			$\phi_{M}$ factor with and without misspecification}%
		\includegraphics[
		height=5.8608in,
		width=7.1114in
		]%
		{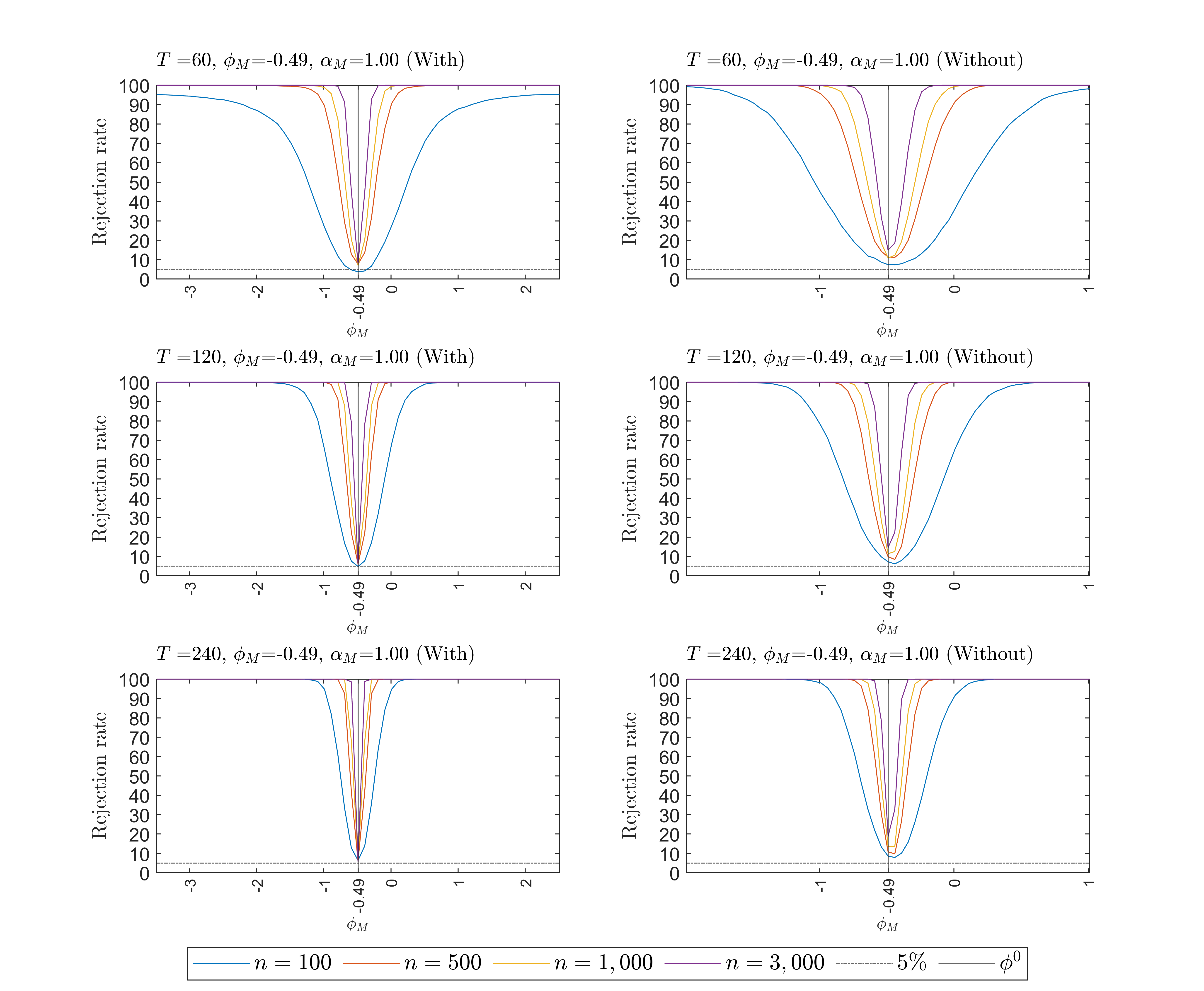}%
		\label{fig:s-b-e10}%
	\end{figure}
	{Note: } See the notes to{ Table \ref{tab:s-b-e10-12}.} }

{\footnotesize \pagebreak}

{\footnotesize
	\begin{figure}[ph]%
		\centering
		\caption{Empirical Power Functions, experiment 11, for coefficient of the
			$\phi_{M}$ factor with and without misspecification}%
		\includegraphics[
		height=5.8608in,
		width=7.1114in
		]%
		{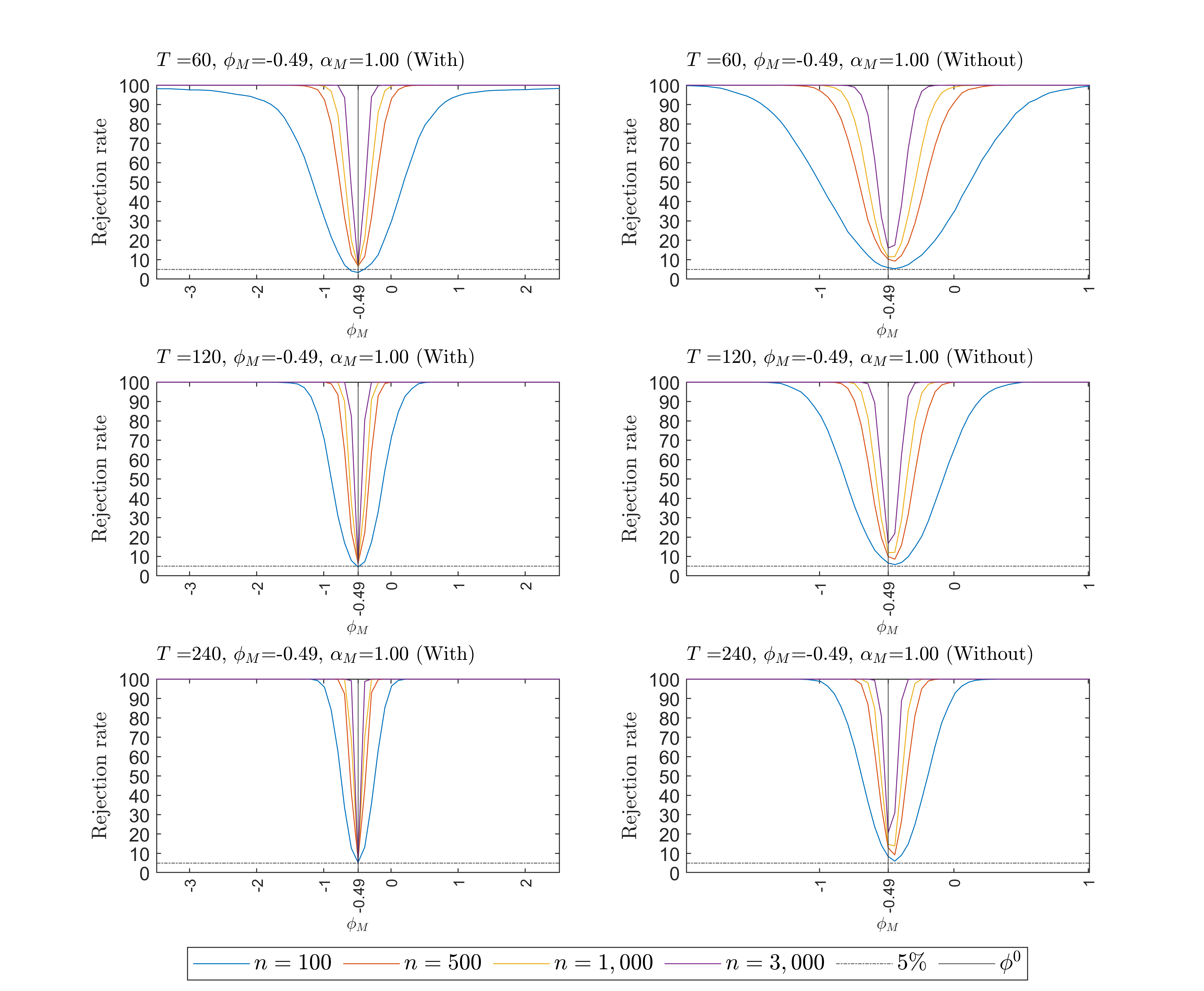}%
		\label{fig:s-b-e11}%
	\end{figure}
	{Note: } See the notes to{ Table \ref{tab:s-b-e10-12}.} }

{\footnotesize \pagebreak}

{\footnotesize
	\begin{figure}[h]%
		\centering
		\caption{Empirical Power Functions, experiment 12, for coefficient of the
			$\phi_{M}$ factor with and without misspecification}%
		\includegraphics[
		height=5.8608in,
		width=7.1122in
		]%
		{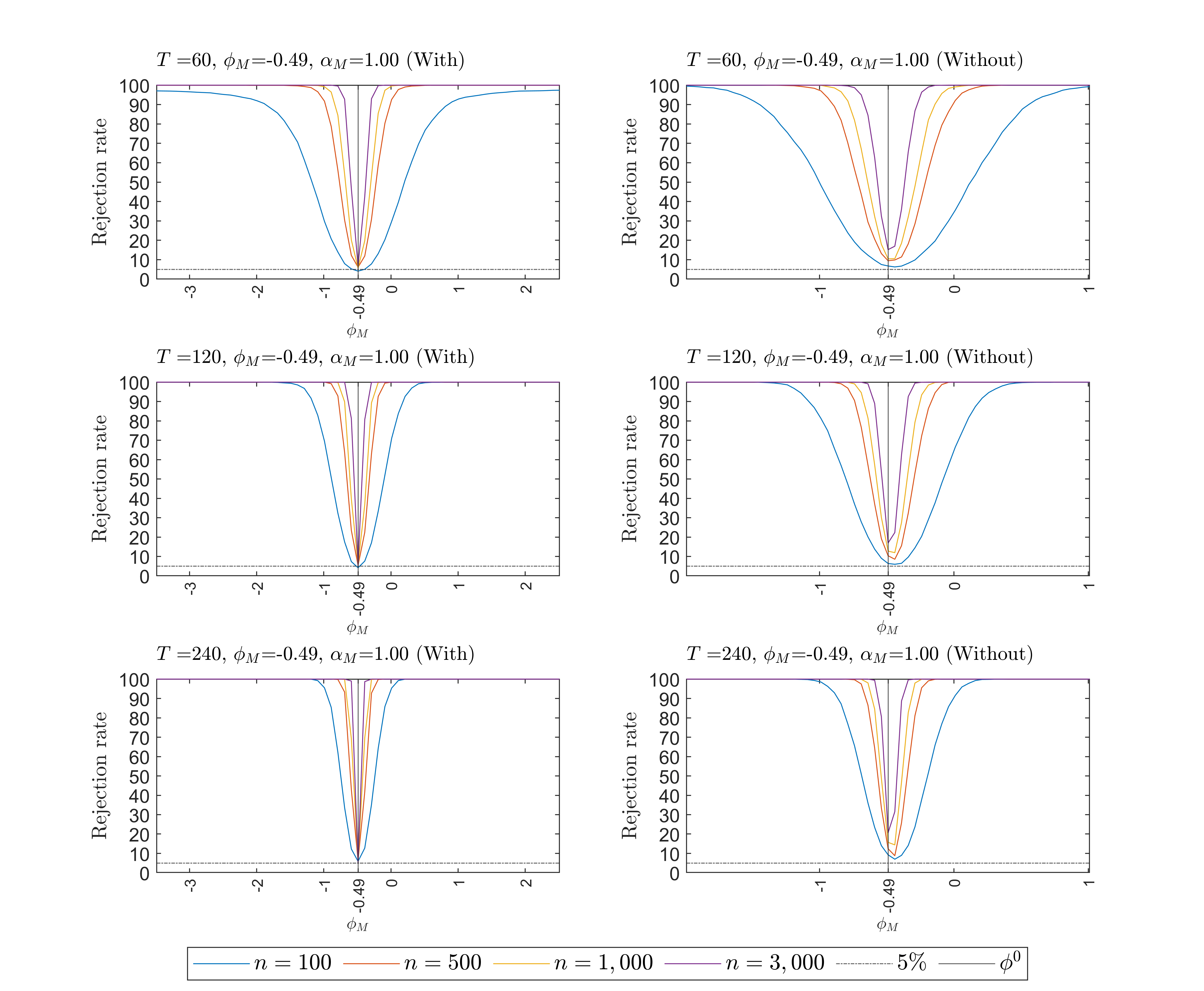}%
		\label{fig:s-b-e12}%
	\end{figure}
	{Note: } See the notes to{ Table \ref{tab:s-b-e10-12}.} }

{\footnotesize \pagebreak}

\subsection{Estimators of $\phi$ for one strong and two weak factors,
	threshold estimator of the covariance matrix with and without
	misspecification}

This subsection examines the effects of incorrectly omitting weak factors
($\alpha_{S}=\alpha_{H}=1/2$) and correctly including them on the small sample
properties of the bias-corrected (BC) estimator $\phi_{M}$ associated with the
strong factor ($\alpha_{M}=1$). Each table presents results from three of the
12 experiments detailed in Table S-1, along with their respective empirical
power functions.

\setcounter{figure}{0} \renewcommand{\thetable}{S-C-E1-3} \renewcommand{\thefigure}{S-C-E\arabic{figure}}%

	\begin{table}[H]%

		\caption{Bias, RMSE and size for the bias-corrected estimators of $\phi_{M}$ =
			$-0.49$, $\alpha_{M}$=$1$ with and without weak factors included in the
			regression for the cases of experiments 1,2 and 3}\label{tab:s-c-e1-3}
		
		\begin{center}
			{\footnotesize
				\begin{tabular}
					[c]{rrrrrrrrrr}\hline\hline
					&  & \multicolumn{2}{c}{Bias(x100)} &  & \multicolumn{2}{c}{RMSE(x100)} &  &
					\multicolumn{2}{c}{Size(x100)}\\\cline{3-4}\cline{6-7}\cline{9-10}%
					Experiment 1 & n & With & Without &  & With & Without &  & With &
					Without\\\cline{2-4}\cline{6-7}\cline{9-10}
					&  & \multicolumn{2}{c}{weak factors} &  & \multicolumn{2}{c}{weak factors} &
					& \multicolumn{2}{c}{weak factors}\\\cline{2-10}%
					{$T=60$} & {\ 100} & 22.05 & 3.67 &  & 454.21 & 30.53 &  & 1.15 & 6.10\\
					& {\ 500} & -10.55 & 1.66 &  & 332.83 & 13.13 &  & 2.25 & 6.15\\
					& {1,000} & -1.33 & 1.10 &  & 382.27 & 9.48 &  & 1.65 & 6.65\\
					& {3,000} & -19.15 & 0.58 &  & 932.38 & 5.38 &  & 2.00 & 6.20\\
					&  &  &  &  &  &  &  &  & \\
					{$T=120$} & {\ 100} & 6.52 & 3.04 &  & 237.70 & 19.49 &  & 2.90 & 5.80\\
					& {\ 500} & 0.73 & 1.37 &  & 35.05 & 8.75 &  & 3.20 & 7.30\\
					& {1,000} & 4.61 & 0.94 &  & 152.17 & 6.12 &  & 3.60 & 6.50\\
					& {3,000} & 0.02 & 0.50 &  & 6.94 & 3.42 &  & 2.90 & 5.00\\
					&  &  &  &  &  &  &  &  & \\
					{$T=240$} & {\ 100} & 8.56 & 2.91 &  & 385.16 & 13.50 &  & 3.45 & 6.45\\
					& {\ 500} & 0.22 & 1.24 &  & 7.22 & 5.80 &  & 4.00 & 5.80\\
					& {1,000} & 0.06 & 0.90 &  & 4.23 & 4.14 &  & 4.00 & 6.50\\
					& {3,000} & 0.01 & 0.51 &  & 2.41 & 2.38 &  & 4.90 & 5.25\\
					Experiment 2 &  &  &  &  &  &  &  &  & \\
					{$T=60$} & {\ 100} & -13.53 & 4.08 &  & 774.85 & 31.16 &  & 1.40 & 6.00\\
					& {\ 500} & -11.15 & 1.63 &  & 755.67 & 13.33 &  & 2.10 & 5.65\\
					& {1,000} & 10.51 & 1.01 &  & 256.15 & 9.61 &  & 2.00 & 6.55\\
					& {3,000} & -0.10 & 0.56 &  & 71.71 & 5.50 &  & 2.25 & 6.60\\
					&  &  &  &  &  &  &  &  & \\
					{$T=120$} & {\ 100} & 468.29 & 3.25 &  & 20351.06 & 19.57 &  & 2.25 & 6.60\\
					& {\ 500} & 5.00 & 1.46 &  & 141.50 & 8.89 &  & 3.85 & 7.55\\
					& {1,000} & -0.54 & 0.97 &  & 16.07 & 6.06 &  & 2.80 & 5.75\\
					& {3,000} & 0.00 & 0.45 &  & 9.71 & 3.45 &  & 3.35 & 5.45\\
					&  &  &  &  &  &  &  &  & \\
					{$T=240$} & {\ 100} & -1.38 & 2.96 &  & 67.20 & 13.49 &  & 4.15 & 7.50\\
					& {\ 500} & 0.04 & 1.28 &  & 6.83 & 5.85 &  & 4.15 & 5.60\\
					& {1,000} & 0.11 & 0.91 &  & 4.59 & 4.15 &  & 4.65 & 7.15\\
					& {3,000} & -0.03 & 0.48 &  & 2.46 & 2.39 &  & 5.15 & 6.10\\
					Experiment 3 &  &  &  &  &  &  &  &  & \\
					{$T=60$} & {\ 100} & 11.17 & 3.75 &  & 686.97 & 30.84 &  & 1.00 & 6.25\\
					& {\ 500} & -10.50 & 1.65 &  & 377.36 & 13.28 &  & 2.00 & 6.15\\
					& {1,000} & -4.65 & 1.09 &  & 120.83 & 9.59 &  & 1.50 & 6.15\\
					& {3,000} & -4.51 & 0.59 &  & 134.45 & 5.45 &  & 1.90 & 6.25\\
					&  &  &  &  &  &  &  &  & \\
					{$T=120$} & {\ 100} & 5.69 & 3.01 &  & 267.53 & 19.52 &  & 2.65 & 5.65\\
					& {\ 500} & -5.01 & 1.36 &  & 140.66 & 8.79 &  & 3.15 & 7.45\\
					& {1,000} & -0.91 & 0.94 &  & 23.35 & 6.14 &  & 3.35 & 6.45\\
					& {3,000} & -0.35 & 0.50 &  & 20.92 & 3.43 &  & 2.70 & 4.75\\
					&  &  &  &  &  &  &  &  & \\
					{$T=240$} & {\ 100} & 1.04 & 2.89 &  & 38.32 & 13.51 &  & 3.45 & 6.35\\
					& {\ 500} & 0.32 & 1.24 &  & 8.13 & 5.82 &  & 3.95 & 5.90\\
					& {1,000} & 0.08 & 0.90 &  & 4.42 & 4.14 &  & 4.00 & 6.45\\
					& {3,000} & 0.04 & 0.51 &  & 2.59 & 2.39 &  & 4.85 & 5.45\\\hline\hline
				\end{tabular}
			}
		\end{center}
		
		{\footnotesize \noindent Notes: The DGP includes one strong factor,
			$\alpha_{M}=1$, and two weak factors, $\alpha_{H}=\alpha_{S}=0.5$. The results
			labelled "With" include the weak factors correctly, and those labelled
			"Without" exclude the weak factors. For further details of the experiments,
			see Table \ref{TabExperiments}. }
\end{table}%

{\footnotesize \pagebreak}

{\footnotesize
	\begin{figure}[ph]%
		\centering
		\caption{Empirical Power Functions, experiment 1, for coefficient of the
			$\phi_{M}$ factor with and without misspecification}%
		\includegraphics[
		height=5.8608in,
		width=7.1122in
		]%
		{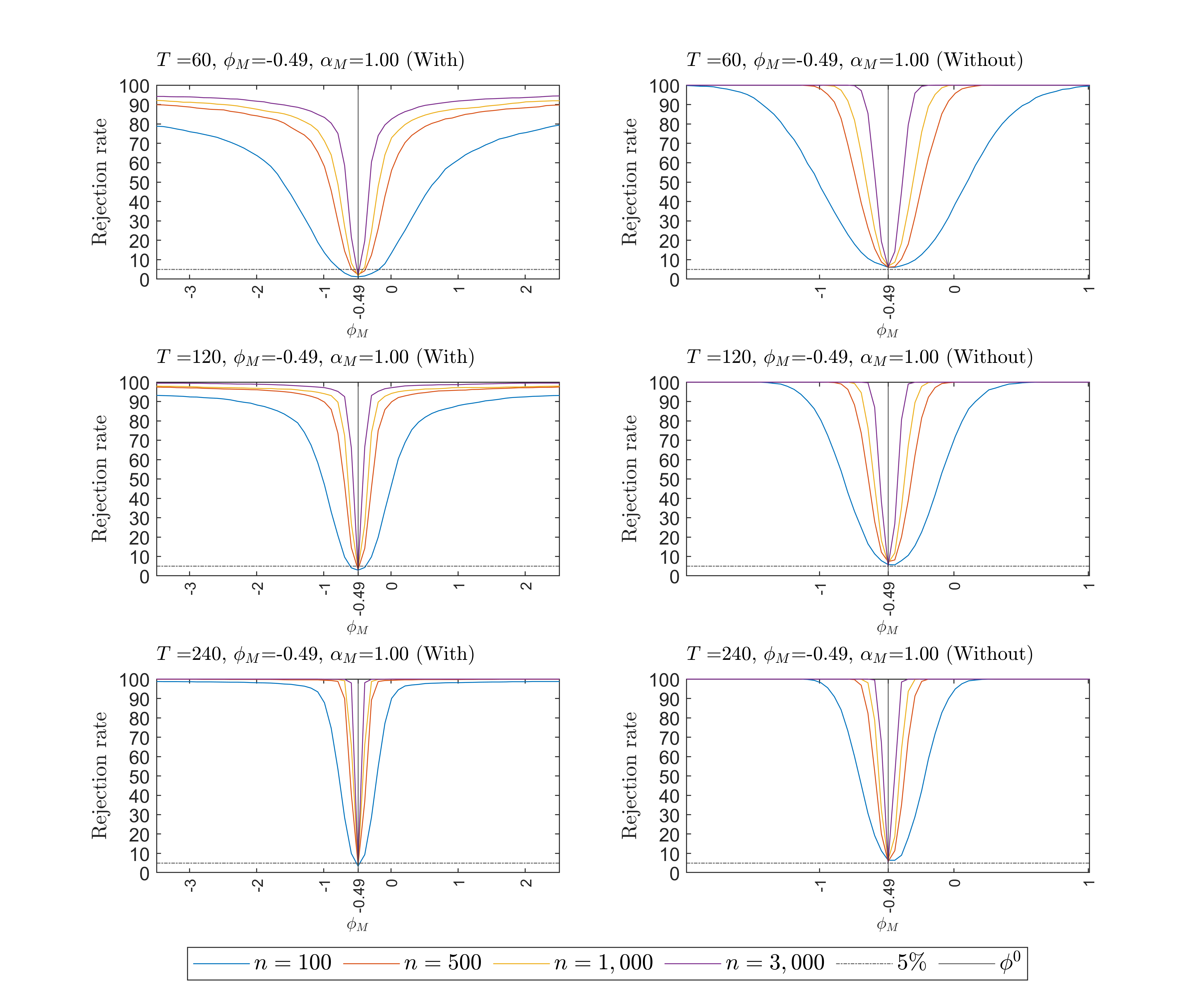}%
		\label{fig:s-c-e1}%
	\end{figure}
	{Note: } See the notes to{ Table \ref{tab:s-c-e1-3}.} }

{\footnotesize \pagebreak}

{\footnotesize
	\begin{figure}[ph]%
		\centering
		\caption{Empirical Power Functions, experiment 2, for coefficient of the
			$\phi_{M}$ factor with and without misspecification}%
		\includegraphics[
		height=5.8608in,
		width=7.1114in
		]%
		{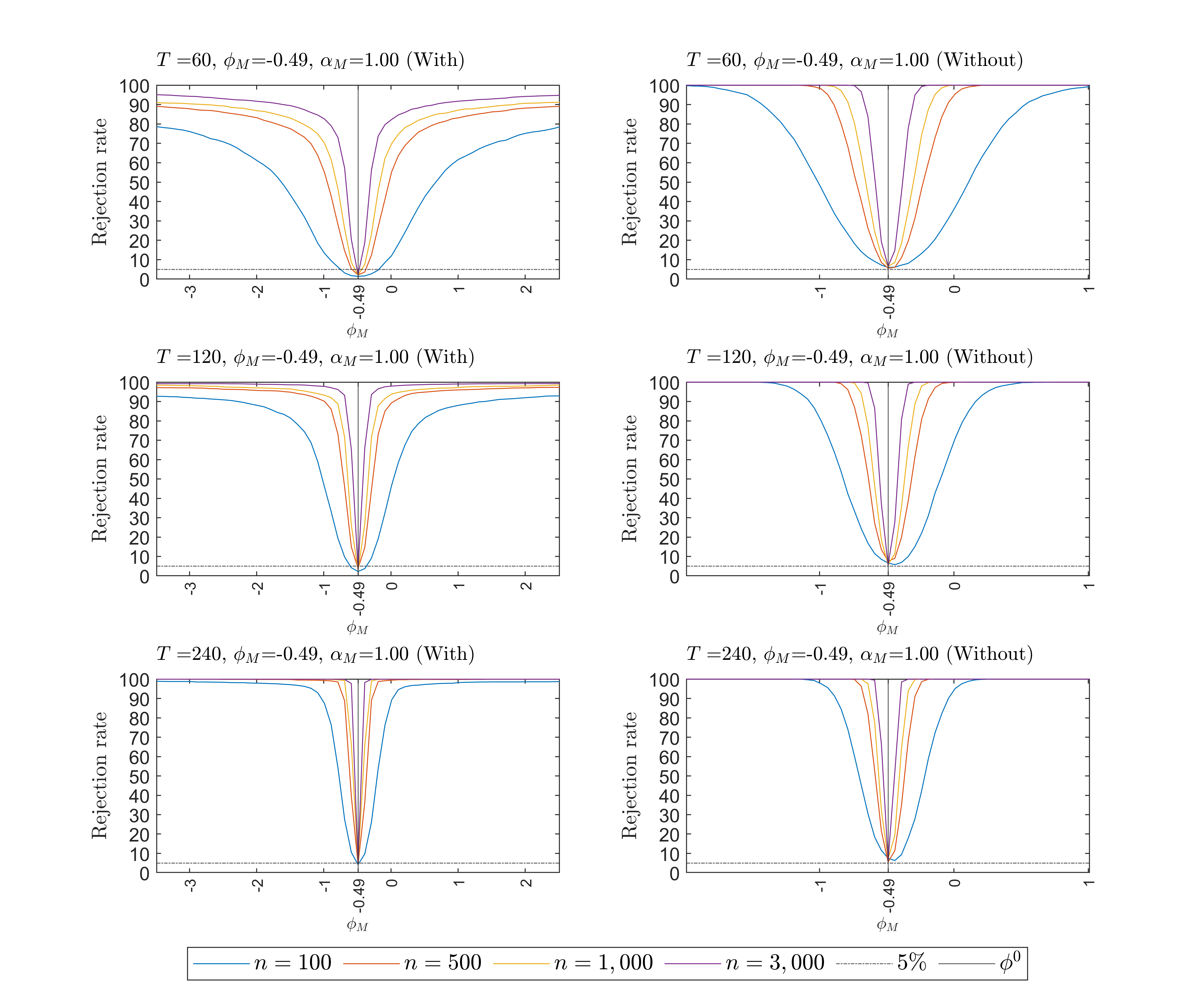}%
		\label{fig:s-c-e2}%
	\end{figure}
	{Note: } See the notes to{ Table \ref{tab:s-c-e1-3}.} }

{\footnotesize \pagebreak}

{\footnotesize
	\begin{figure}[h]%
		\centering
		\caption{Empirical Power Functions, experiment 3, for coefficient of the
			$\phi_{M}$ factor with and without misspecification}%
		\includegraphics[
		height=5.8608in,
		width=7.1122in
		]%
		{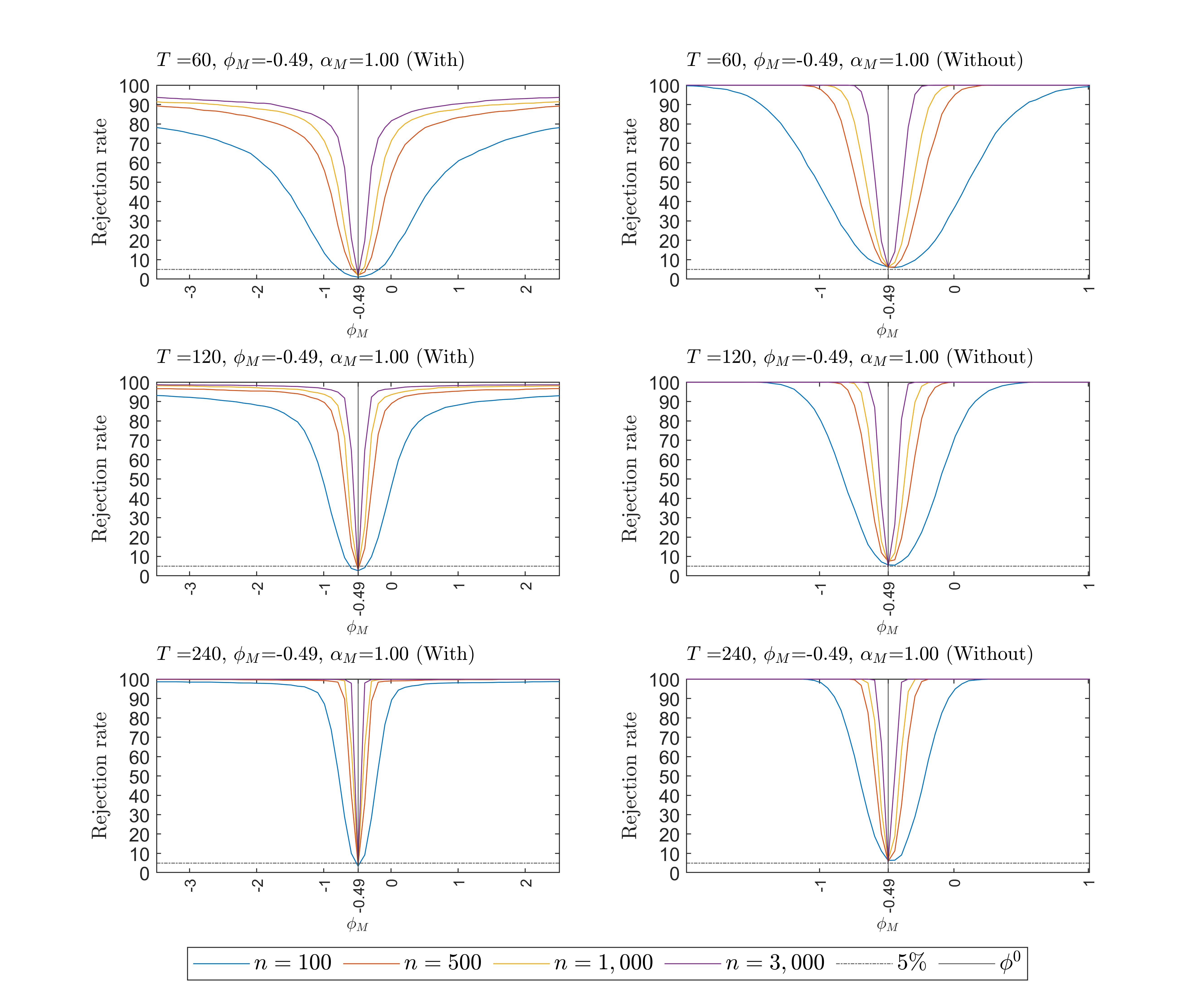}%
		\label{fig:s-c-e3}%
	\end{figure}
	{Note: } See the notes to{ Table \ref{tab:s-c-e1-3}.} }

{\footnotesize \pagebreak}

\renewcommand{\thetable}{S-C-E4-6}%

	\begin{table}[H]%

		\caption{Bias, RMSE and size for the bias-corrected estimators of $\phi_{M}$ =
			$-0.49$, $\alpha_{M}$=$1$ with and without weak factors included in the
			regression for the cases of experiments 4, 5 and 6}\label{tab:s-c-e4-6}
		
		\begin{center}
			{\footnotesize
				\begin{tabular}
					[c]{rrrrrrrrrr}\hline\hline
					&  & \multicolumn{2}{c}{Bias(x100)} &  & \multicolumn{2}{c}{RMSE(x100)} &  &
					\multicolumn{2}{c}{Size(x100)}\\\cline{3-4}\cline{6-7}\cline{9-10}%
					Experiment 4 & n & With & Without &  & With & Without &  & With &
					Without\\\cline{2-4}\cline{6-7}\cline{9-10}
					&  & \multicolumn{2}{c}{weak factors} &  & \multicolumn{2}{c}{weak factors} &
					& \multicolumn{2}{c}{weak factors}\\\cline{2-10}%
					{$T=60$} & {\ 100} & -24.42 & 4.12 &  & 1165.53 & 31.54 &  & 1.35 & 6.15\\
					& {\ 500} & -11.84 & 1.61 &  & 302.88 & 13.49 &  & 1.95 & 5.55\\
					& {1,000} & 5.16 & 1.01 &  & 229.86 & 9.72 &  & 1.90 & 6.75\\
					& {3,000} & -8.17 & 0.57 &  & 450.22 & 5.57 &  & 2.05 & 6.60\\
					&  &  &  &  &  &  &  &  & \\
					{$T=120$} & {\ 100} & -4.28 & 3.22 &  & 452.39 & 19.63 &  & 2.30 & 6.30\\
					& {\ 500} & -0.54 & 1.44 &  & 26.55 & 8.94 &  & 3.80 & 7.85\\
					& {1,000} & 0.98 & 0.96 &  & 31.58 & 6.08 &  & 2.75 & 5.80\\
					& {3,000} & -0.17 & 0.45 &  & 8.84 & 3.46 &  & 3.25 & 5.40\\
					&  &  &  &  &  &  &  &  & \\
					{$T=240$} & {\ 100} & -0.11 & 2.94 &  & 31.04 & 13.51 &  & 3.95 & 7.55\\
					& {\ 500} & -0.50 & 1.28 &  & 27.63 & 5.85 &  & 4.05 & 5.90\\
					& {1,000} & -0.05 & 0.91 &  & 5.23 & 4.16 &  & 4.65 & 6.90\\
					& {3,000} & 0.02 & 0.48 &  & 3.36 & 2.39 &  & 4.85 & 6.05\\
					Experiment 5 &  &  &  &  &  &  &  &  & \\
					{$T=60$} & {\ 100} & 12.72 & 3.75 &  & 679.29 & 30.84 &  & 0.95 & 6.25\\
					& {\ 500} & -10.07 & 1.66 &  & 364.25 & 13.34 &  & 2.10 & 6.15\\
					& {1,000} & -5.01 & 1.07 &  & 122.63 & 9.60 &  & 1.45 & 5.90\\
					& {3,000} & -4.43 & 0.60 &  & 133.49 & 5.47 &  & 1.90 & 6.25\\
					&  &  &  &  &  &  &  &  & \\
					{$T=120$} & {\ 100} & 6.87 & 3.01 &  & 277.10 & 19.52 &  & 2.55 & 5.65\\
					& {\ 500} & -4.67 & 1.36 &  & 127.57 & 8.85 &  & 3.65 & 7.60\\
					& {1,000} & -0.85 & 0.92 &  & 24.13 & 6.16 &  & 3.60 & 6.40\\
					& {3,000} & -0.35 & 0.51 &  & 20.50 & 3.44 &  & 2.65 & 5.20\\
					&  &  &  &  &  &  &  &  & \\
					{$T=240$} & {\ 100} & 0.54 & 2.89 &  & 31.90 & 13.51 &  & 4.15 & 6.35\\
					& {\ 500} & 0.29 & 1.25 &  & 7.91 & 5.86 &  & 3.85 & 6.00\\
					& {1,000} & 0.10 & 0.89 &  & 4.44 & 4.16 &  & 4.30 & 6.75\\
					& {3,000} & 0.04 & 0.52 &  & 2.58 & 2.40 &  & 5.20 & 5.45\\
					Experiment 6 &  &  &  &  &  &  &  &  & \\
					{$T=60$} & {\ 100} & -27.78 & 4.12 &  & 905.91 & 31.54 &  & 1.45 & 6.15\\
					& {\ 500} & -11.62 & 1.60 &  & 305.30 & 13.55 &  & 2.05 & 5.95\\
					& {1,000} & 5.28 & 0.99 &  & 231.07 & 9.73 &  & 1.65 & 6.50\\
					& {3,000} & -7.16 & 0.58 &  & 412.98 & 5.59 &  & 2.05 & 6.95\\
					&  &  &  &  &  &  &  &  & \\
					{$T=120$} & {\ 100} & -2.66 & 3.22 &  & 451.26 & 19.63 &  & 2.30 & 6.30\\
					& {\ 500} & -0.54 & 1.44 &  & 26.52 & 9.00 &  & 4.15 & 7.85\\
					& {1,000} & 0.92 & 0.94 &  & 30.26 & 6.10 &  & 2.90 & 6.30\\
					& {3,000} & -0.18 & 0.46 &  & 8.91 & 3.47 &  & 3.25 & 5.55\\
					&  &  &  &  &  &  &  &  & \\
					{$T=240$} & {\ 100} & -0.18 & 2.94 &  & 29.74 & 13.51 &  & 5.10 & 7.55\\
					& {\ 500} & -0.49 & 1.29 &  & 27.08 & 5.89 &  & 4.40 & 6.30\\
					& {1,000} & -0.02 & 0.90 &  & 5.12 & 4.17 &  & 5.20 & 6.80\\
					& {3,000} & 0.03 & 0.49 &  & 3.47 & 2.40 &  & 4.50 & 6.45\\\hline\hline
				\end{tabular}
			}
		\end{center}
		
		{\footnotesize \noindent Notes: The DGP includes one strong $\alpha_{M}=1$ and
			two weak $(\alpha_{H}=\alpha_{S}=0.5)$ factors, the regression with weak
			factors includes them, the regression without excludes them. For further
			details of the experiments, see Table \ref{TabExperiments}. }
\end{table}%

{\footnotesize \pagebreak}

{\footnotesize
	\begin{figure}[ph]%
		\centering
		\caption{Empirical Power Functions, experiment 4, for coefficient of the
			$\phi_{M}$ factor with and without misspecification}%
		\includegraphics[
		height=5.8608in,
		width=7.1122in
		]%
		{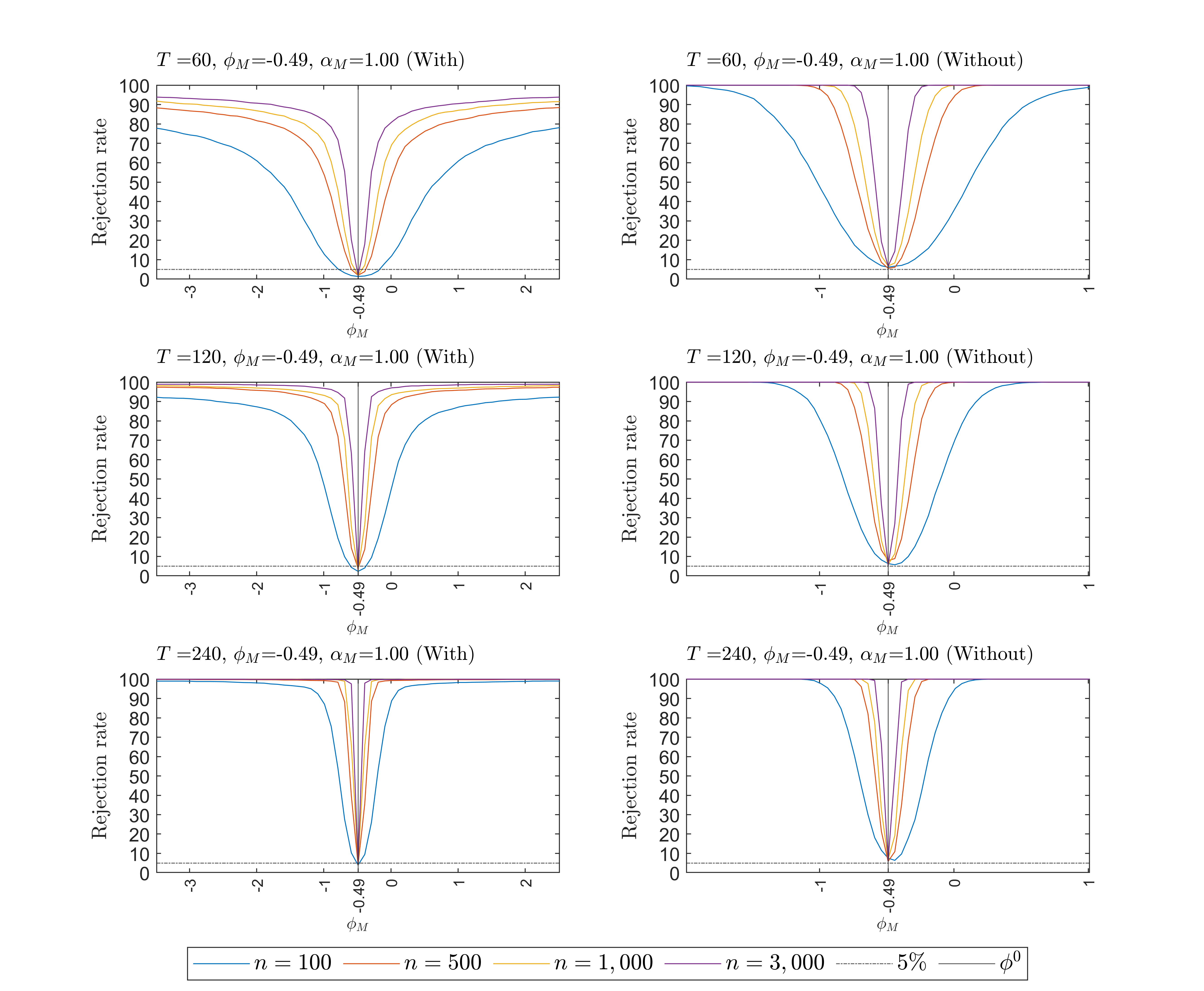}%
		\label{fig:s-c-e4}%
	\end{figure}
	{Note: } See the notes to{ Table \ref{tab:s-c-e4-6}.} }

{\footnotesize \pagebreak}

{\footnotesize
	\begin{figure}[ph]%
		\centering
		\caption{Empirical Power Functions, experiment 5, for coefficient of the
			$\phi_{M}$ factor with and without misspecification}%
		\includegraphics[
		height=5.8608in,
		width=7.1114in
		]%
		{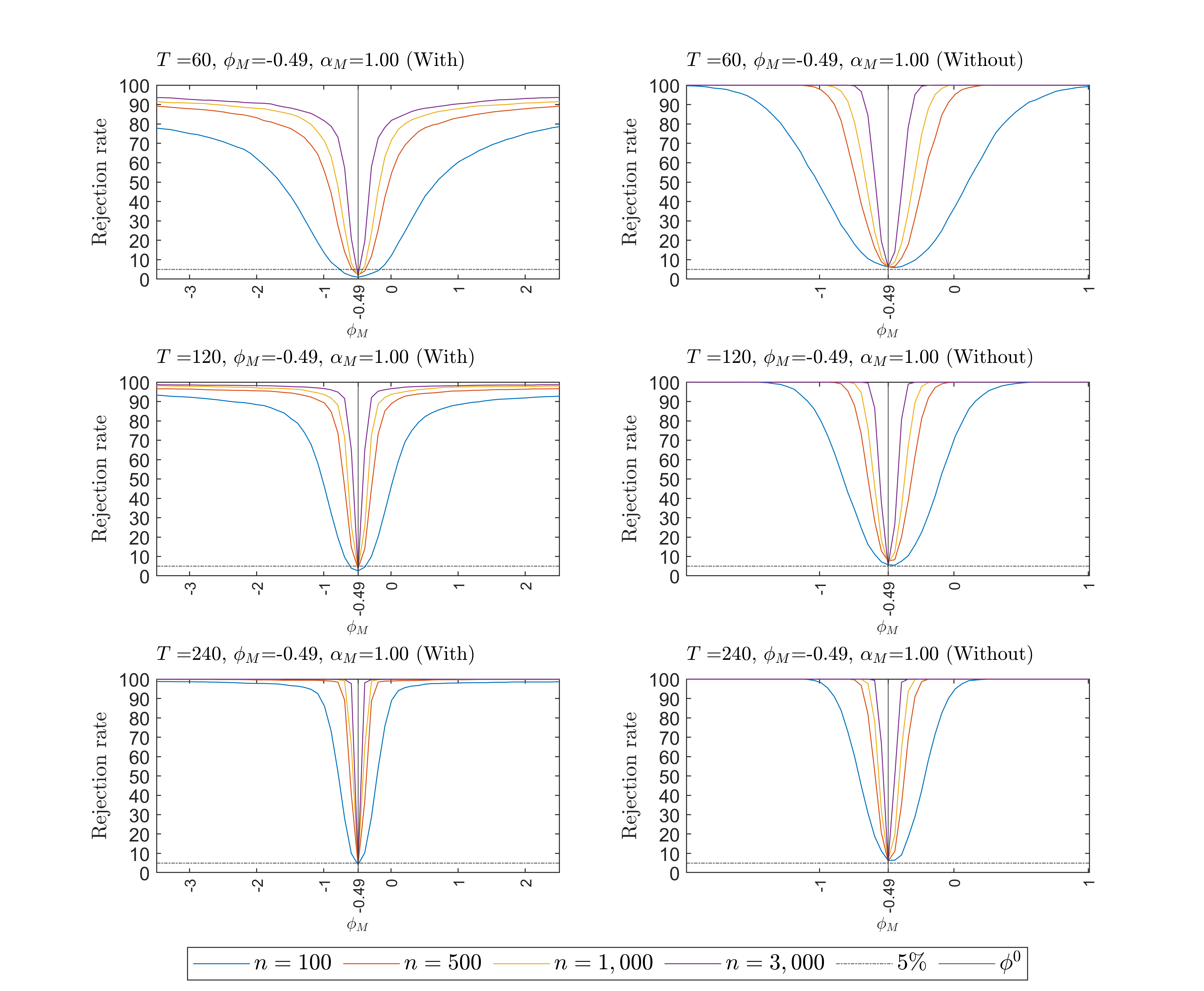}%
		\label{fig:s-c-e5}%
	\end{figure}
	{Note: } See the notes to{ Table \ref{tab:s-c-e4-6}.} }

{\footnotesize \pagebreak}

{\footnotesize
	\begin{figure}[h]%
		\centering
		\caption{Empirical Power Functions, experiment 6, for coefficient of the
			$\phi_{M}$ factor with and without misspecification}%
		\includegraphics[
		height=5.8608in,
		width=7.1122in
		]%
		{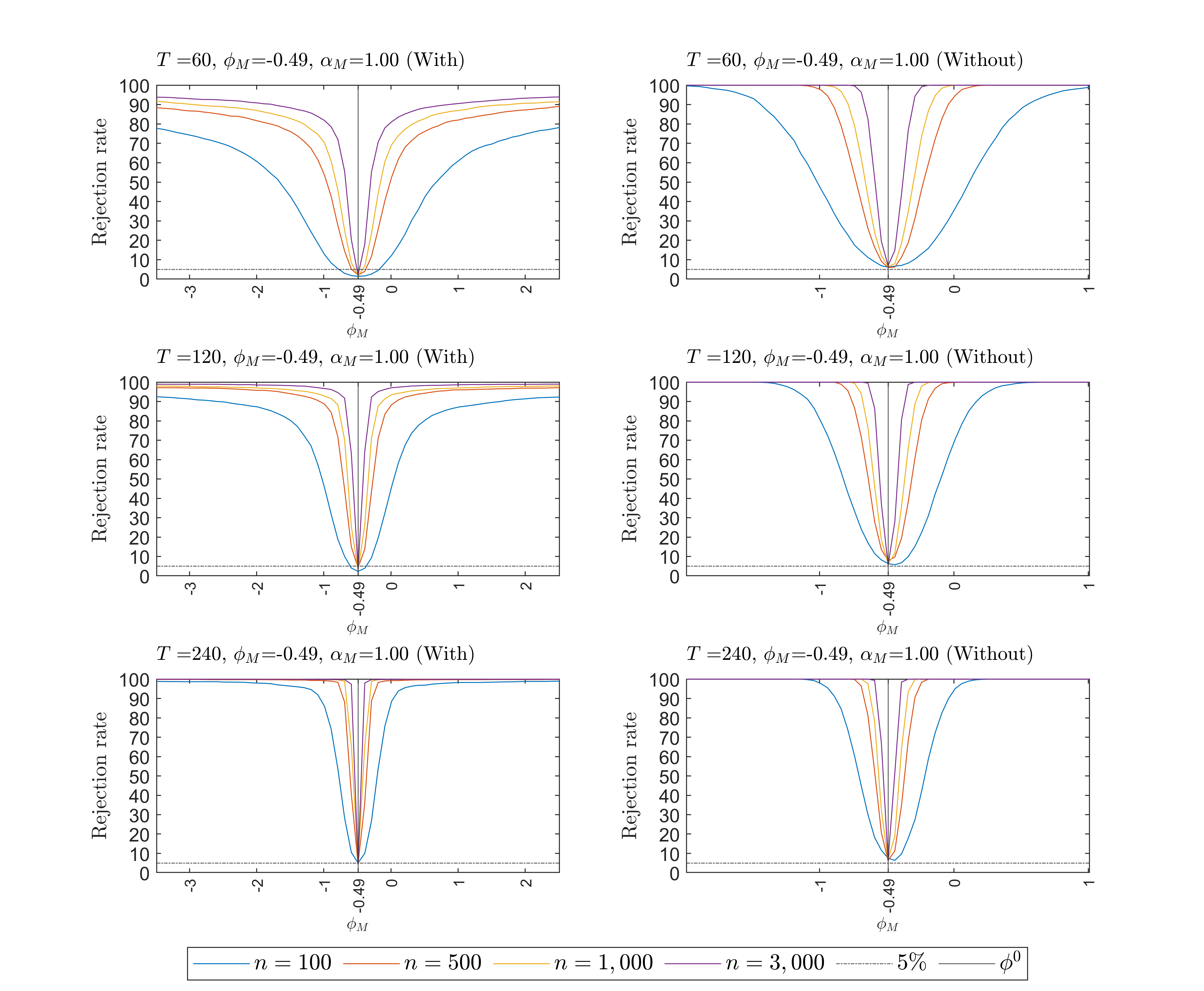}%
		\label{fig:s-c-e6}%
	\end{figure}
	{Note: } See the notes to{ Table \ref{tab:s-c-e4-6}.} }

{\footnotesize \pagebreak}

\renewcommand{\thetable}{S-C-E7-9}%

	\begin{table}[H]%

		\caption{Bias, RMSE and size for the bias-corrected estimators of $\phi_{M}$ =
			$-0.49$, $\alpha_{M}$=$1$ with and without weak factors included in the
			regression for the cases of experiments 7, 8 and 9}\label{tab:s-c-e7-9}
		
		\begin{center}
			{\footnotesize
				\begin{tabular}
					[c]{rrrrrrrrrr}\hline\hline
					&  & \multicolumn{2}{c}{Bias(x100)} &  & \multicolumn{2}{c}{RMSE(x100)} &  &
					\multicolumn{2}{c}{Size(x100)}\\\cline{3-4}\cline{6-7}\cline{9-10}%
					Experiment 7 & n & With & Without &  & With & Without &  & With &
					Without\\\cline{2-4}\cline{6-7}\cline{9-10}
					&  & \multicolumn{2}{c}{weak factors} &  & \multicolumn{2}{c}{weak factors} &
					& \multicolumn{2}{c}{weak factors}\\\cline{2-10}%
					{$T=60$} & {\ 100} & 3.08 & 2.82 &  & 859.39 & 30.18 &  & 1.30 & 5.35\\
					& {\ 500} & 24.88 & 1.20 &  & 840.21 & 13.62 &  & 1.85 & 6.65\\
					& {1,000} & 0.24 & 0.84 &  & 79.72 & 9.64 &  & 1.35 & 6.25\\
					& {3,000} & -4.91 & 0.61 &  & 528.52 & 5.51 &  & 2.20 & 6.40\\
					&  &  &  &  &  &  &  &  & \\
					{$T=120$} & {\ 100} & 4.42 & 2.67 &  & 216.43 & 19.15 &  & 2.10 & 6.00\\
					& {\ 500} & 1.74 & 1.23 &  & 147.10 & 8.59 &  & 2.95 & 6.05\\
					& {1,000} & -3.59 & 0.89 &  & 154.11 & 6.12 &  & 3.75 & 5.90\\
					& {3,000} & -0.84 & 0.61 &  & 41.40 & 3.52 &  & 3.75 & 5.75\\
					&  &  &  &  &  &  &  &  & \\
					{$T=240$} & {\ 100} & -0.94 & 2.74 &  & 22.50 & 13.27 &  & 4.95 & 5.80\\
					& {\ 500} & -0.83 & 1.11 &  & 38.11 & 6.04 &  & 6.15 & 7.15\\
					& {1,000} & -0.06 & 0.84 &  & 4.85 & 4.25 &  & 4.95 & 7.00\\
					& {3,000} & -0.01 & 0.55 &  & 2.46 & 2.38 &  & 4.70 & 6.50\\
					Experiment 8 &  &  &  &  &  &  &  &  & \\
					{$T=60$} & {\ 100} & -18.57 & 2.85 &  & 1090.60 & 31.04 &  & 1.30 & 5.70\\
					& {\ 500} & 3.08 & 0.84 &  & 248.52 & 13.70 &  & 1.90 & 6.10\\
					& {1,000} & 11.56 & 0.98 &  & 511.88 & 9.66 &  & 1.60 & 5.60\\
					& {3,000} & -0.24 & 0.71 &  & 45.00 & 5.66 &  & 2.05 & 6.50\\
					&  &  &  &  &  &  &  &  & \\
					{$T=120$} & {\ 100} & 15.36 & 2.83 &  & 488.96 & 19.33 &  & 2.40 & 5.00\\
					& {\ 500} & -8.41 & 1.12 &  & 424.57 & 8.53 &  & 3.55 & 5.70\\
					& {1,000} & -1.21 & 1.04 &  & 42.58 & 6.25 &  & 3.65 & 6.25\\
					& {3,000} & -0.68 & 0.64 &  & 43.26 & 3.56 &  & 3.60 & 6.00\\
					&  &  &  &  &  &  &  &  & \\
					{$T=240$} & {\ 100} & -4.70 & 2.85 &  & 184.48 & 13.30 &  & 5.10 & 6.30\\
					& {\ 500} & 0.07 & 1.17 &  & 7.15 & 5.98 &  & 5.45 & 6.30\\
					& {1,000} & -0.07 & 0.85 &  & 4.44 & 4.27 &  & 5.10 & 6.60\\
					& {3,000} & 0.01 & 0.54 &  & 2.41 & 2.40 &  & 4.80 & 6.70\\
					Experiment 9 &  &  &  &  &  &  &  &  & \\
					{$T=60$} & {\ 100} & 16.90 & 1.71 &  & 1051.65 & 32.31 &  & 1.35 & 6.65\\
					& {\ 500} & -10.51 & 0.95 &  & 739.68 & 14.35 &  & 2.35 & 7.50\\
					& {1,000} & -2.96 & 0.53 &  & 160.31 & 10.26 &  & 2.65 & 8.75\\
					& {3,000} & -0.09 & 0.43 &  & 106.32 & 5.80 &  & 2.70 & 8.35\\
					&  &  &  &  &  &  &  &  & \\
					{$T=120$} & {\ 100} & 61.02 & 1.74 &  & 2471.10 & 20.37 &  & 2.05 & 6.25\\
					& {\ 500} & -9.78 & 1.08 &  & 365.78 & 8.87 &  & 2.70 & 6.00\\
					& {1,000} & 0.35 & 0.73 &  & 77.55 & 6.36 &  & 3.00 & 7.30\\
					& {3,000} & 1.23 & 0.46 &  & 48.78 & 3.59 &  & 3.85 & 6.55\\
					&  &  &  &  &  &  &  &  & \\
					{$T=240$} & {\ 100} & 0.51 & 1.73 &  & 52.46 & 13.89 &  & 4.55 & 6.90\\
					& {\ 500} & -0.22 & 0.91 &  & 8.62 & 6.08 &  & 5.75 & 6.40\\
					& {1,000} & -0.20 & 0.71 &  & 7.44 & 4.27 &  & 5.50 & 7.15\\
					& {3,000} & -0.04 & 0.43 &  & 2.52 & 2.37 &  & 4.30 & 5.60\\\hline\hline
				\end{tabular}
			}
		\end{center}
		
		{\footnotesize \noindent Notes: The DGP includes one strong $\alpha_{M}=1$ and
			two weak $(\alpha_{H}=\alpha_{S}=0.5)$ factors, the regression with weak
			factors includes them, the regression without excludes them. For further
			details of the experiments, see Table \ref{TabExperiments}. }
\end{table}%

{\footnotesize \pagebreak}

{\footnotesize
	\begin{figure}[ph]%
		\centering
		\caption{Empirical Power Functions, experiment 7, for coefficient of the
			$\phi_{M}$ factor with and without misspecification}%
		\includegraphics[
		height=5.8608in,
		width=7.1114in
		]%
		{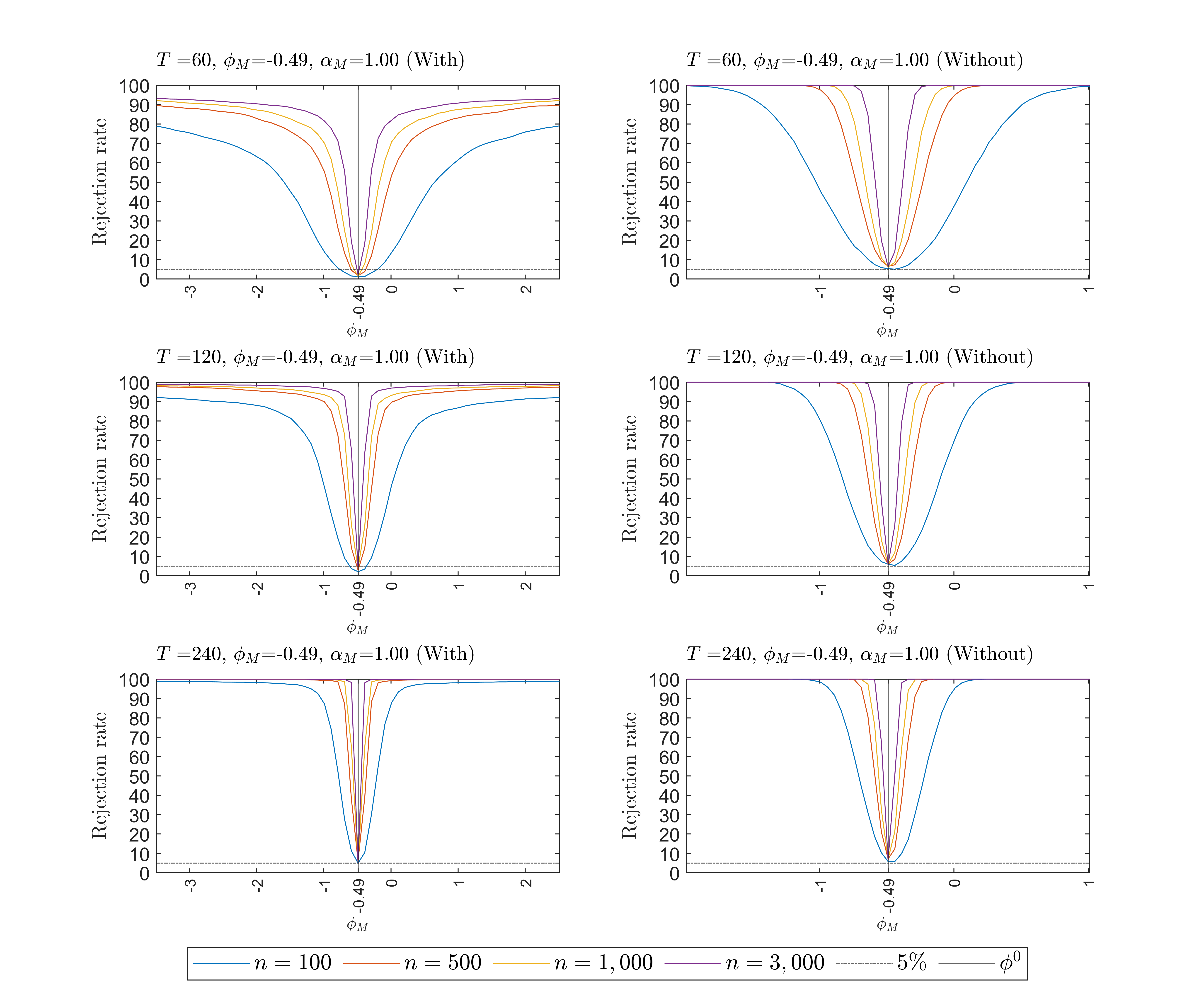}%
		\label{fig:s-c-e7}%
	\end{figure}
	{Note: } See the notes to{ Table \ref{tab:s-c-e7-9}.} }

{\footnotesize \pagebreak}

{\footnotesize
	\begin{figure}[ph]%
		\centering
		\caption{Empirical Power Functions, experiment 8, for coefficient of the
			$\phi_{M}$ factor with and without misspecification}%
		\includegraphics[
		height=5.8608in,
		width=7.1114in
		]%
		{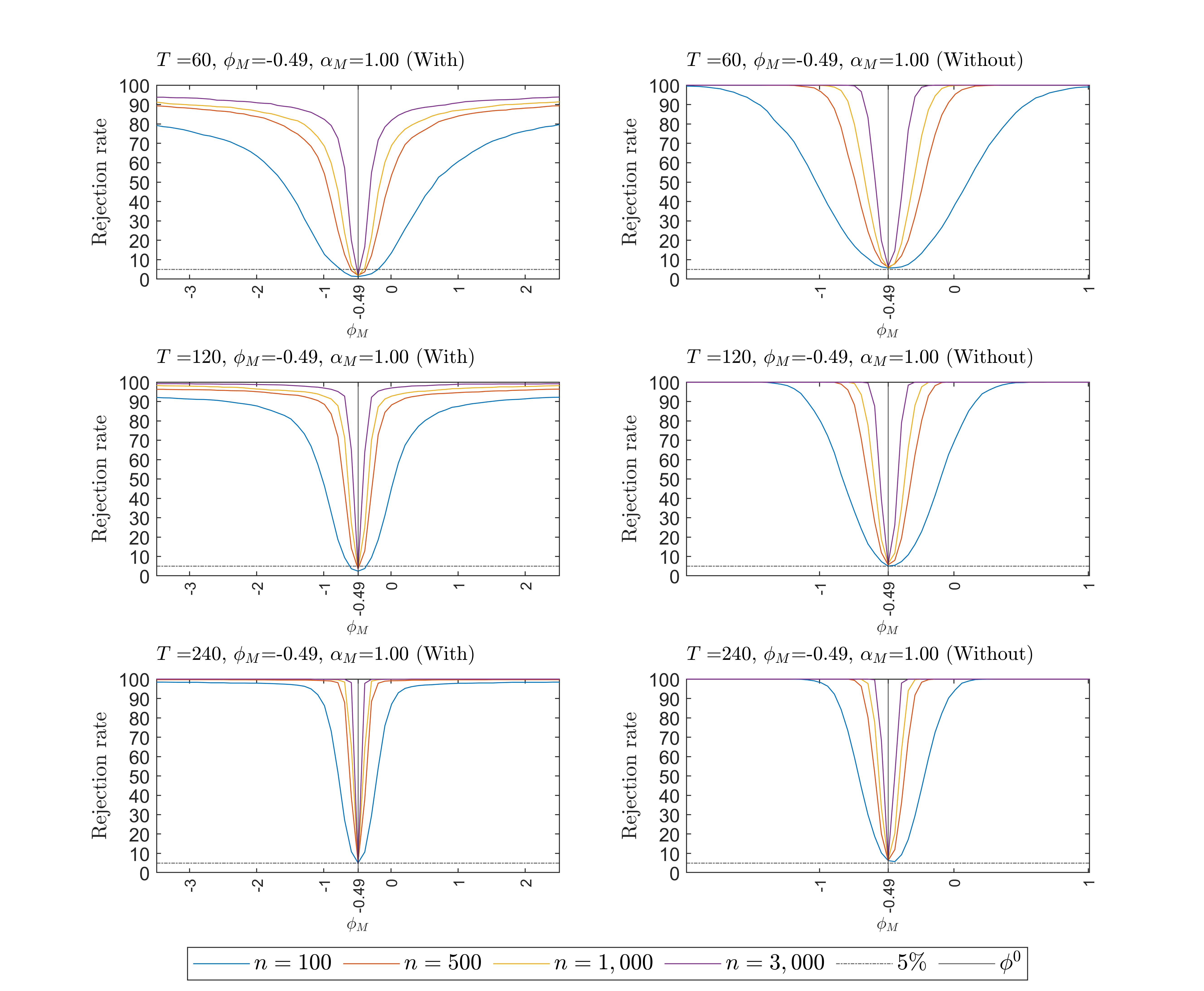}%
		\label{fig:s-c-e8}%
	\end{figure}
	{Note: } See the notes to{ Table \ref{tab:s-c-e7-9}.} }

{\footnotesize \pagebreak}

{\footnotesize
	\begin{figure}[ph]%
		\centering
		\caption{Empirical Power Functions, experiment 9, for coefficient of the
			$\phi_{M}$ factor with and without misspecification}%
		\includegraphics[
		height=5.8608in,
		width=7.1114in
		]%
		{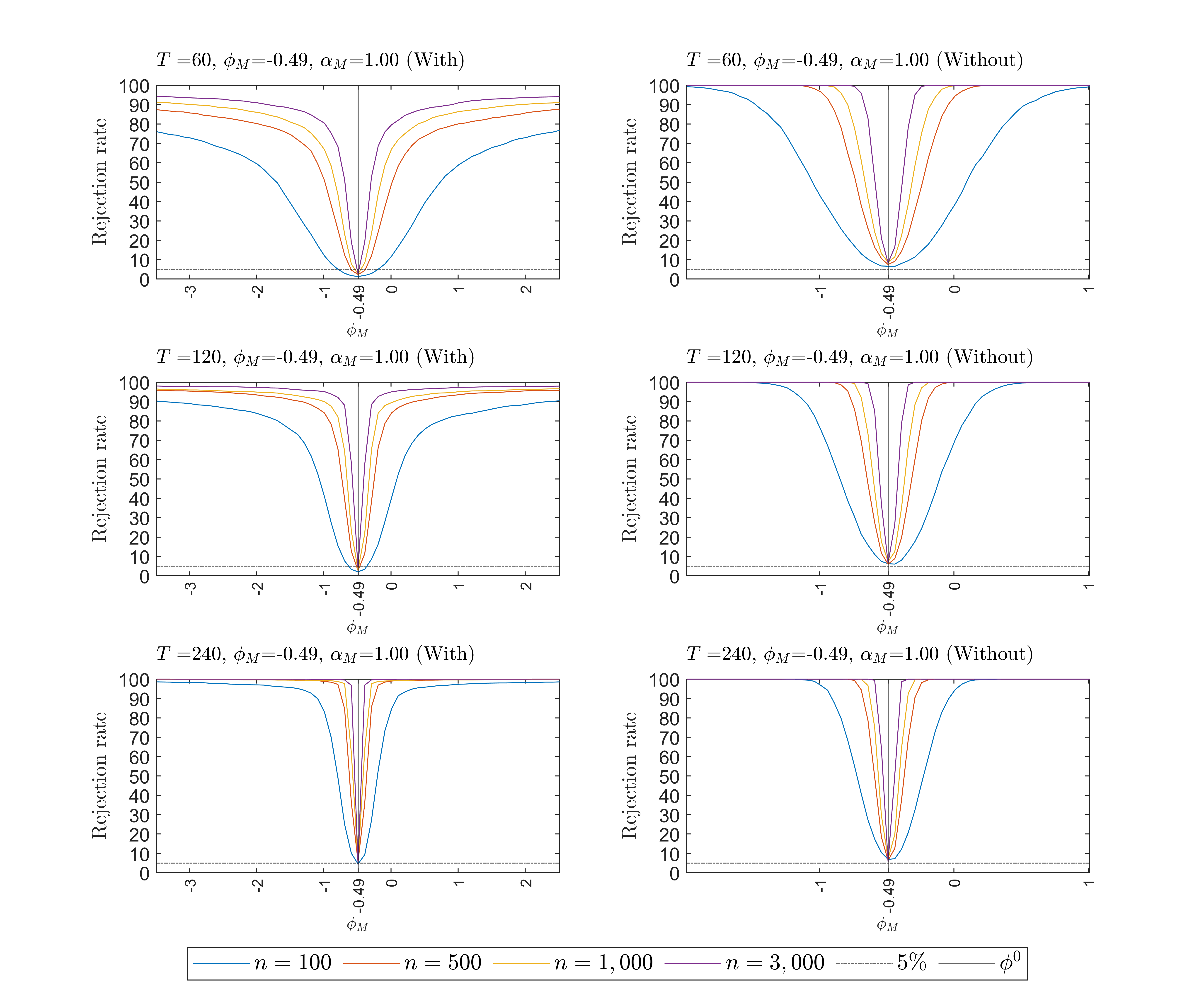}%
		\label{fig:s-c-e9}%
	\end{figure}
	{Note: } See the notes to{ Table \ref{tab:s-c-e7-9}.} }

{\footnotesize \pagebreak}

\renewcommand{\thetable}{S-C-E10-12}%

	\begin{table}[H]%

		\caption{Bias, RMSE and size for the bias-corrected estimators of $\phi_{M}$ =
			$-0.49$, $\alpha_{M}$=$1$ with and without weak factors included in the
			regression for the cases of experiments 10, 11 and 12}\label{tab:s-c-e10-12}
		
		\begin{center}
			{\footnotesize
				\begin{tabular}
					[c]{rrrrrrrrrr}\hline\hline
					&  & \multicolumn{2}{c}{Bias(x100)} &  & \multicolumn{2}{c}{RMSE(x100)} &  &
					\multicolumn{2}{c}{Size(x100)}\\\cline{3-4}\cline{6-7}\cline{9-10}%
					Experiment 10 & n & With & Without &  & With & Without &  & With &
					Without\\\cline{2-4}\cline{6-7}\cline{9-10}
					&  & \multicolumn{2}{c}{weak factors} &  & \multicolumn{2}{c}{weak factors} &
					& \multicolumn{2}{c}{weak factors}\\\cline{2-10}%
					{$T=60$} & {\ 100} & -1.30 & 1.80 &  & 581.15 & 33.08 &  & 1.55 & 6.95\\
					& {\ 500} & 19.94 & 0.54 &  & 738.70 & 14.45 &  & 2.10 & 8.20\\
					& {1,000} & 1.31 & 0.69 &  & 76.46 & 10.23 &  & 1.85 & 7.35\\
					& {3,000} & 1.02 & 0.53 &  & 333.48 & 5.91 &  & 2.70 & 8.30\\
					&  &  &  &  &  &  &  &  & \\
					{$T=120$} & {\ 100} & 27.62 & 1.91 &  & 810.64 & 20.78 &  & 1.90 & 6.30\\
					& {\ 500} & -4.30 & 0.95 &  & 195.46 & 8.90 &  & 2.40 & 6.25\\
					& {1,000} & 7.68 & 0.93 &  & 212.90 & 6.46 &  & 3.45 & 7.25\\
					& {3,000} & -0.46 & 0.49 &  & 34.61 & 3.60 &  & 3.20 & 5.50\\
					&  &  &  &  &  &  &  &  & \\
					{$T=240$} & {\ 100} & 0.01 & 1.87 &  & 109.19 & 14.08 &  & 4.60 & 7.30\\
					& {\ 500} & 0.27 & 0.94 &  & 19.59 & 6.08 &  & 5.20 & 6.50\\
					& {1,000} & 0.41 & 0.76 &  & 21.39 & 4.28 &  & 5.50 & 6.55\\
					& {3,000} & -0.04 & 0.76 &  & 2.57 & 2.40 &  & 4.70 & 7.10\\
					Experiment 11 &  &  &  &  &  &  &  &  & \\
					{$T=60$} & {\ 100} & 17.29 & 2.89 &  & 696.53 & 30.28 &  & 1.45 & 5.40\\
					& {\ 500} & -5.92 & 1.27 &  & 510.19 & 13.71 &  & 2.00 & 6.80\\
					& {1,000} & 45.15 & 0.78 &  & 2091.69 & 9.66 &  & 2.00 & 6.45\\
					& {3,000} & 2.77 & 0.59 &  & 75.60 & 5.61 &  & 2.65 & 7.05\\
					&  &  &  &  &  &  &  &  & \\
					{$T=120$} & {\ 100} & 4.02 & 2.73 &  & 135.30 & 19.16 &  & 2.55 & 6.20\\
					& {\ 500} & 0.44 & 1.34 &  & 24.35 & 8.61 &  & 3.50 & 5.60\\
					& {1,000} & 10.16 & 0.86 &  & 415.98 & 6.09 &  & 3.85 & 6.10\\
					& {3,000} & 0.10 & 0.62 &  & 7.88 & 3.54 &  & 4.65 & 6.70\\
					&  &  &  &  &  &  &  &  & \\
					{$T=240$} & {\ 100} & -0.31 & 2.77 &  & 14.31 & 13.28 &  & 4.10 & 6.15\\
					& {\ 500} & 0.07 & 1.19 &  & 6.78 & 6.06 &  & 5.45 & 6.90\\
					& {1,000} & -0.06 & 0.81 &  & 4.40 & 4.23 &  & 5.45 & 6.65\\
					& {3,000} & 0.01 & 0.55 &  & 2.37 & 2.38 &  & 4.65 & 6.50\\
					Experiment 12 &  &  &  &  &  &  &  &  & \\
					{$T=60$} & {\ 100} & 0.33 & 2.93 &  & 531.55 & 31.13 &  & 1.35 & 5.75\\
					& {\ 500} & -16.77 & 0.90 &  & 496.29 & 13.75 &  & 1.70 & 5.60\\
					& {1,000} & 17.22 & 0.93 &  & 766.48 & 9.68 &  & 1.65 & 6.15\\
					& {3,000} & -0.91 & 0.68 &  & 162.80 & 5.75 &  & 2.40 & 7.00\\
					&  &  &  &  &  &  &  &  & \\
					{$T=120$} & {\ 100} & -11.47 & 2.88 &  & 529.74 & 19.36 &  & 2.55 & 5.15\\
					& {\ 500} & -16.06 & 1.23 &  & 555.75 & 8.55 &  & 3.55 & 5.80\\
					& {1,000} & 0.99 & 1.03 &  & 20.02 & 6.22 &  & 3.50 & 6.50\\
					& {3,000} & -2.21 & 0.64 &  & 100.80 & 3.57 &  & 4.00 & 5.90\\
					&  &  &  &  &  &  &  &  & \\
					{$T=240$} & {\ 100} & -0.27 & 2.88 &  & 19.39 & 13.31 &  & 4.20 & 6.30\\
					& {\ 500} & 0.10 & 1.24 &  & 6.30 & 5.99 &  & 5.35 & 6.25\\
					& {1,000} & -0.05 & 0.83 &  & 4.37 & 4.23 &  & 5.10 & 6.70\\
					& {3,000} & 0.01 & 0.54 &  & 2.39 & 2.41 &  & 4.95 & 6.50\\\hline\hline
				\end{tabular}
			}
		\end{center}
		
		{\footnotesize \noindent Notes: The DGP includes one strong $\alpha_{M}=1$ and
			two weak $(\alpha_{H}=\alpha_{S}=0.5)$ factors, the regression with weak
			factors includes them, the regression without excludes them. For further
			details of the experiments, see Table \ref{TabExperiments}.}%
		
\end{table}%

{\footnotesize \pagebreak}

{\footnotesize
	\begin{figure}[ph]%
		\centering
		\caption{Empirical Power Functions, experiment 10, for coefficient of the
			$\phi_{M}$ factor with and without misspecification}%
		\includegraphics[
		height=5.8608in,
		width=7.1114in
		]%
		{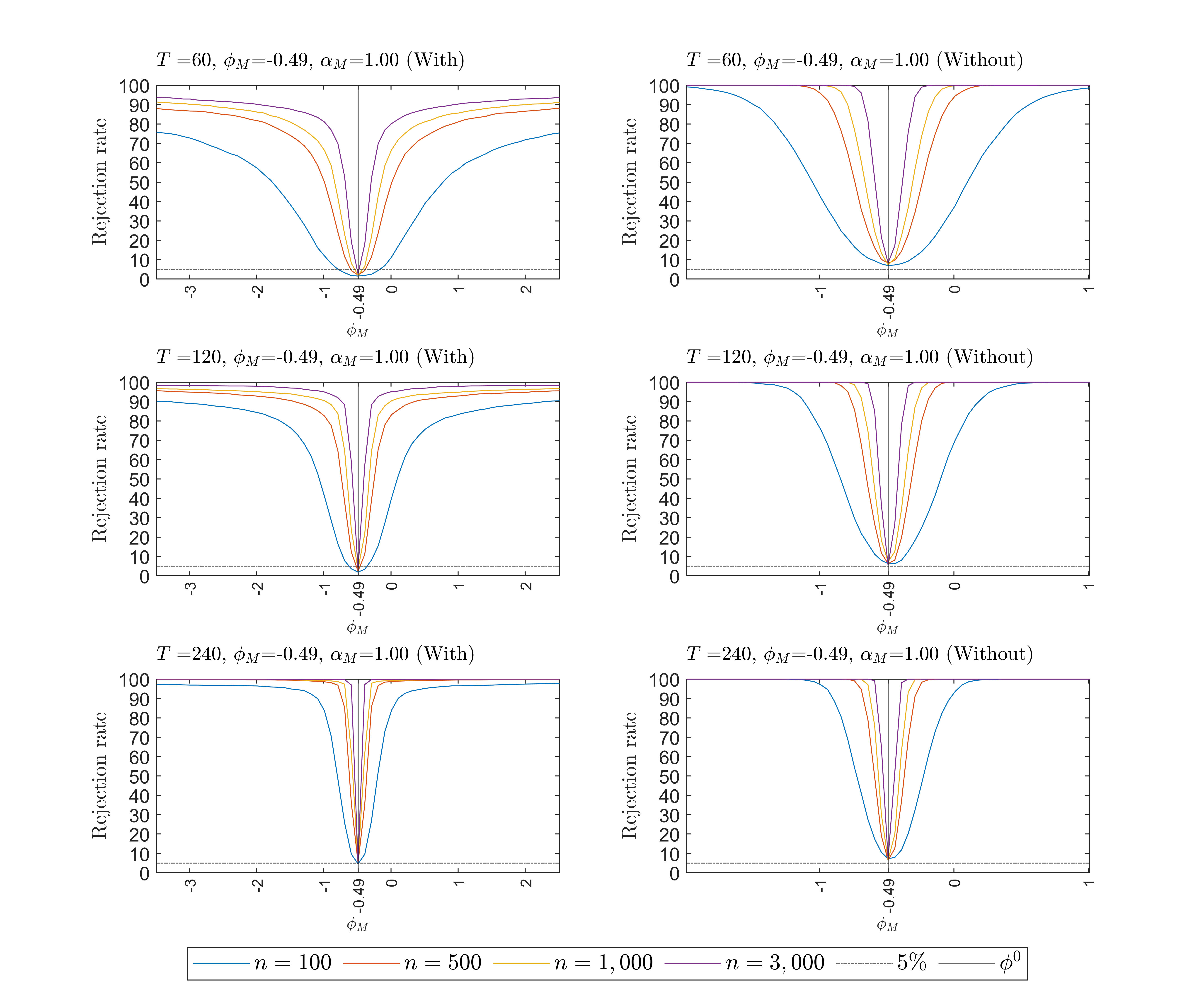}%
		\label{fig:s-c-e10}%
	\end{figure}
	{Note: } See the notes to{ Table \ref{tab:s-c-e10-12}.} }

{\footnotesize \pagebreak}

{\footnotesize
	\begin{figure}[ph]%
		\centering
		\caption{Empirical Power Functions, experiment 11, for coefficient of the
			$\phi_{M}$ factor with and without misspecification}%
		\includegraphics[
		height=5.8608in,
		width=7.1114in
		]%
		{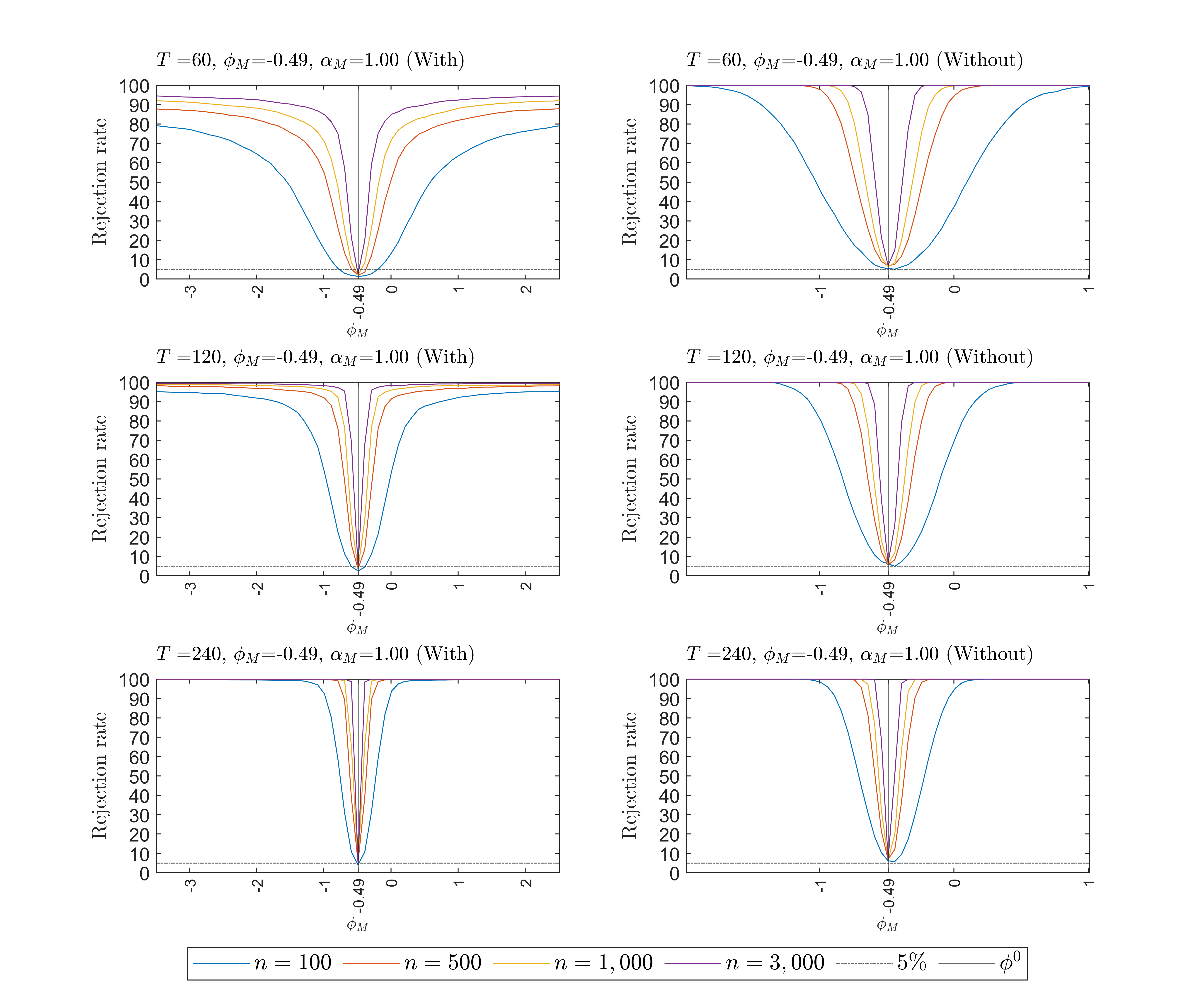}%
		\label{fig:s-c-e11}%
	\end{figure}
	{Note: } See the notes to{ Table \ref{tab:s-c-e10-12}.} }

{\footnotesize \pagebreak}

{\footnotesize
	\begin{figure}[ph]%
		\centering
		\caption{Empirical Power Functions, experiment 12, for coefficient of the
			$\phi_{M}$ factor with and without misspecification}%
		\includegraphics[
		height=5.8608in,
		width=7.1114in
		]%
		{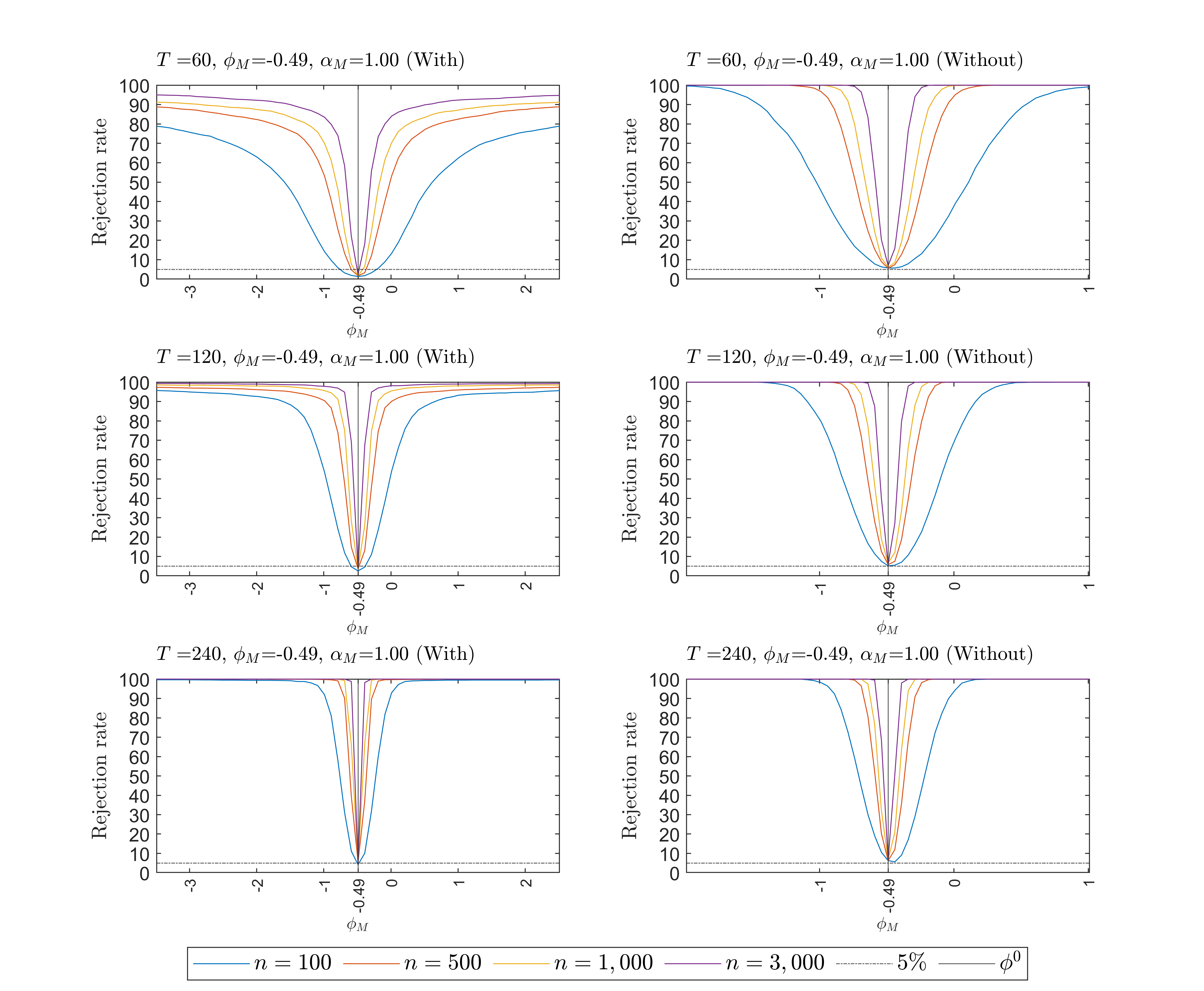}%
		\label{fig:s-c-e12}%
	\end{figure}
	{Note: } See the notes to{ Table \ref{tab:s-c-e10-12}.} }
\end{document}